 \nonstopmode 
\documentclass{Paper}
\newif\ifcolour \colourfalse
\newif\ifshort \shortfalse

\newrgbcolor{col}{1 0.5 0.5}

\input{macros}
\def \scsc#1{\mbox{\scriptsize \sc #1}}
\def \scit#1{\mbox{\scriptsize \it #1}}

\def \xsind{\mbox{\scriptsize {\bf x}{\sc s}}}
\def \xhind{\mbox{\scriptsize {\bf x}{\sc h}}}
\def \xind{\mbox{\scriptsize {\bf x}}}
\def \sind{\mbox{\scriptsize {\sc s}}}
\def \hind{\mbox{\scriptsize {\sc h}}}

\def \slangle {\raise.5\point\hbox{\small $\langle$}}
\def \srangle {\raise.5\point\hbox{\small $\rangle$}}

\def \stackrel#1#2{\setbox155=\hbox{$#1$}\setbox156=\hbox{$#2$}%
\mathrel{\copy156\kern-.5\wd156\kern-.5\wd155\raise4.5\point\hbox{\copy155}\kern-.5\wd155\kern.5\wd156}}

\newfont{\itbf}{ptmbo at 10.25pt}
\newfont{\itsf}{phvro at 10.25pt}
\newfont{\sfbf}{phvb at 10.25pt}
\newfont{\oldmath}{cmsy10 at 11pt}
\newfont{\largeoldmath}{cmsy10 at 15pt}

\def \mysl#1{\mbox{\rm\textsl{#1}}\kern0\point}
\def \mytiny{\fontsize{6.5}{9}\selectfont}
\def \tinyfont{\fontsize{6pt}{6.6pt}\selectfont}
\def \emptyline	{\vskip 10\point plus 5\point minus 4\point \noindent}

\makeatletter 

\newif\if@om
\def \omfalse{\let\if@om\iffalse}
\def \omtrue{\let\if@om\iftrue}

\newif\if@lom
\def \lomfalse{\let\if@lom\iffalse}
\def \lomtrue{\let\if@lom\iftrue}

\newif\if@Xom
\def \Xomfalse{\let\if@Xom\iffalse}
\def \Xomtrue{\let\if@Xom\iftrue}
\Xomtrue

\newskip \adjustablepoint \adjustablepoint .9\point

\usepackage{pstricks}
\usepackage{pst-node}
\usepackage{pdftricks}
\usepackage{graphicx} 
\usepackage{qsymbols}
\usepackage{inference}
\usepackage{relsize}

\Comment{
}

\def \Cont[#1]{\mbox{\small 
\sf C}[\ssk{#1}\ssk]}
\def \ContH[#1]{\mbox{\small 
\sf C}_{\scsc{h}}[\ssk{#1}\ssk]}

\def \typeout#1{\@latex@warning{#1}}

\def \Exists {\exists\,}

\def \ifnextchar{\@ifnextchar}

\def \qskp {\hspace*{.25\point}}
\def \hskp {\hspace*{.5\point}}
\def \skp {\hspace*{1\point}}
\def \psk {\hspace*{1\point}}
\def \skipper {\hspace*{1.5\point}}
\def \Bogskipleft	{~\kern-1\point~}
\def \Bogskipright	{\ \kern-1\point\ }
\def \skipperleft	{\hspace*{2\point}}
\def \skipperright {\hspace*{2\point}}

\def \dquad {\quad \quad}

\def \qquad {\quad \quad \quad \quad}

\def \BoUndnccurve[#1]#2#3{\psset{linecolor=red}%
\nccurve[#1]{#2}{#3}\psset{linecolor=ruleyellow}}

\def \Bf#1{\textbf{\Red #1}}
\def \Rm#1{\textrm{\Black #1}}
\def \It#1{\mbox{\rmfamily \it \Itcol #1\hspace*{1\point}}}
\def \Sf#1{\textsf{\Green #1}}
\def \Sfb#1{\mbox {\sfb #1}}
\def \Sc#1{\textsc{\Purple #1}}

\def \mynewrgbcolor#1#2{%
\ifmycolour\newrgbcolor{#1}{#2}\else \newrgbcolor{#1}{0 0 0}\fi}

\mynewrgbcolor{ltermcol}{.8 0 .8}
\mynewrgbcolor{termcol}{1 0 0}
\mynewrgbcolor{typecol}{0 .50 0}
\mynewrgbcolor{rulegreencol}{.1 1 .1}
\mynewrgbcolor{DarkGreencol}{0.1 .5 0.1}
\mynewrgbcolor{Greencol}{0.1 .5 0.1}
\mynewrgbcolor{Yellowcol}{.9 .9 0}
\mynewrgbcolor{Itcol}{.9 .5 0}
\mynewrgbcolor{Blackcol}{0 0 0}
\mynewrgbcolor{Bluecol}{0.35 0.35 1}
\mynewrgbcolor{Orangecol}{1 0.35 0.25}
\mynewrgbcolor{Purplecol}{.75 0 .75}
\mynewrgbcolor{Whitecol}{1 1 1}
\mynewrgbcolor{Greencol}{0 .25 0}
\mynewrgbcolor{yellowcol}{1 1 0}
\mynewrgbcolor{orangecol}{1 0.75 0.5}
\mynewrgbcolor{bluecol}{0 0 1}
\mynewrgbcolor{Redcol}{1 0.05 0.05}
\mynewrgbcolor{lightgreencol}{0.75 1 0.75}
\mynewrgbcolor{lightyellowcol}{1 1 .75}
\mynewrgbcolor{lightbluecol}{.75 .75 1}
\mynewrgbcolor{lightredcol}{1 .5 .5}
\mynewrgbcolor{purplecol}{.9 0.1 .9}
\mynewrgbcolor{whitecol}{1 1 1}
\mynewrgbcolor{blackcol}{0 0 0}
\mynewrgbcolor{neutralcol}{0 0 0}

\newrgbcolor{Redcol}{1 0.05 0.05}
\newrgbcolor{Bluecol}{0.05 0.05 1}

\mynewrgbcolor{ruleyellow}{1 1 .5}
\mynewrgbcolor{rulegreen}{.1 1 .1}

\mynewrgbcolor{foreground}{0 0 0}
\mynewrgbcolor{background}{1 1 1}

\def \ltermcolour{\ifmycolour\ltermcol\fi}
\def \termcolour{\ifmycolour\termcol\fi}
\def \Yellow{\ifmycolour\Yellowcol\fi}
\def \typecolour{\ifmycolour\typecol\fi}

\def \Green{\ifmycolour\Greencol\fi}
\def \Yellow{\ifmycolour\Yellowcol\fi}
\def \Black{\ifmycolour\Blackcol\fi}

\def \Purple{\ifmycolour\Purplecol\fi}

\def \orange{\ifmycolour\orangecol\fi}
\def \blue{\ifmycolour\bluecol\fi}
\def \Red{\ifmycolour\Redcol\fi}

\def \black{\ifmycolour\blackcol\fi}

\def \specialcol{\termcolour}

\def \mysemcolour {\mynewrgbcolor{thiscol}{1 0 1}\thiscol}
\ifmycolour \mynewrgbcolor{semcolour}{.9 0 0}\else \mynewrgbcolor{semcolour}{0 0 0}\fi

\psset{linecolor=ruleyellow,linewidth=0.3\point,arrows=->,nodesep=0.05}
\def \PSFrame(#1,#2){\mbox{\psframe[unit=\unitlength,linewidth=.5\point,framearc=.5
](#1,#2)}}
\def \PSBox(#1,#2){\mbox{\psframe[unit=\unitlength,linewidth=.5\point,framearc=0
](#1,#2)}}

\def \Pr[#1]{~[\,#1\,]}
\def \ele {\mathbin{\in}}
\def \notele {\mathbin{\not\nobreak\in}}

\def \subssymb {(\tvar \mapsto C)}
\def \subsone#1{\subssymb\,(#1)}
\def \subs{\@ifnextchar\bgroup{\subsone}{\subssymb}}

\def \unifytwo#1#2{\setbox33=\hbox{$#1$}\setbox34=\hbox{$#2$}
\ifdim\wd33<12\point\ifdim\wd34<12\point {\It {unify}}\ #1\ #2
\else {\It {unify}}\ #1\ (#2) \fi
\else \ifdim\wd34<12\point {\It {unify}\ (#1)\ #2}
\else {\It {unify}}\ (#1)\ (#2) \fi \fi }

\def \unify{\@ifnextchar{\bgroup}{\unifytwo}{\It {unify}}}

\def \eqA	{\mathrel{{\sim}\kern-2\point_{\cal A}}}
\def \eqAmu	{\mathrel{{\sim}\kern-2\point_{{\cal A}_{`m}}}}
\def \eqAx	{\mathrel{{\sim}\kern-2\point_{{\cal A}_{\xhind}}}}
\def \eqAxs	{\mathrel{{\sim}\kern-2\point_{{\cal A}_{\xsind}}}}
\def \equ	{\mathrel{\sim}}

\def \beq	{\mathop{\geq}}
\def \seqr {{\typecolour \leq}}

\def \seqU {\mathbin{\typecolour \leq_{\cup}}}
\def \seqI {\mathbin{\typecolour \leq_{\cap}}}

\def \seq	{\mathop{\seqr}}

\def \seqsr {{\leq_{\scsc{s}}}}
\def \seqs	{\skipperleft\seqsr\skipperright}

\def \Union {\mathbin{\cup}}
\def \Inter {\mathbin{\cap}}
\def \And	{\mathrel{\&}}
\def \Or	{\mathrel{\vee}}
\def \bottom	{{\perp}}

\def \Conj	{\mathord{\wedge}}

\def \longarrow#1{\psset{unit=1\point,linewidth=0.35\point}%
\psline{cc-cc}(0,0)(#1,0)%
\hspace*{#1}%
\pscurve{cc-cc}(0,0)(-1.1,.4)(-2.5,2.2)%
\pscurve{cc-cc}(0,0)(-1.1,-.4)(-2.5,-2.2)}

\def \Rel#1#2{\setbox155=\hbox{\scriptsize $#1$}\setbox156=\hbox{\scriptsize $#2$}%
\,\raise5.5pt\copy155\kern-1\wd155%
\raise3pt\hbox{\longarrow{1\wd155}}
\kern-.5\wd155\kern-.5\wd156\raise-2pt\copy156\kern-.5\wd156\kern.5\wd155\,}

\def \dots	{\mbox	{$\cdots$}}
\def \Skip	{\hspace*{2\point}}
\def \Quad {\hspace*{1cm}}

\def \pCsymb{\It{pC}}
\def \pCarg#1{\pCsymb\psk{\termcolour #1}}
\def \pC{\@ifnextchar\bgroup{\pCarg}{\pCsymb}}

\def \pCgsymb{\It{pC}\kern-1pt_{\mbox{\mysemcolour \scriptsize \sc g}}}
\def \pCgarg#1{\pCgsymb\,({\termcolour #1})}
\def \pCg{\@ifnextchar\bgroup{\pCgarg}{\pCgsymb}}

\def \SetAppBr(#1){{\cal A}(#1)}
\def \SetApp{\@ifnextchar({\SetAppBr}{{\cal A}}}
\def \AppMu	{{\termcolour {\cal A}_{`m}}}
\def \SetAppMuBr(#1){\AppMu(#1)}
\def \SetAppMu{\@ifnextchar({\SetAppMuBr}{\AppMu}}
\def \AppS	{{\termcolour \cal A}_{\sind}}
\def \SetAppSBr(#1){\AppS(#1)}
\def \SetAppS{\@ifnextchar({\SetAppSBr}{\AppS}}
\def \AppXS	{{\termcolour \cal A}_{\xsind}}
\def \SetAppXSBr(#1){\AppXS(#1)}
\def \SetAppXS{\@ifnextchar({\SetAppSBr}{\AppXS}}
\def \AppX	{{\termcolour {\cal A}_{\xhind}}}
\def \SetAppXBr(#1){\AppX(#1)}
\def \SetAppX{\@ifnextchar({\SetAppXBr}{\AppX}}

\def \AppS	{{\termcolour \cal A}_{\sind}}

\def \bred	{\mathrel{\Red \bredr}}

\def \bredr {\mbox	{\Red $\rightarrow_{`b}$}}

\def \dbredr	{\arrow \kern-7\point\bredr}
\def \dbred {\mathrel{\dbredr}}
\def \eqb	{=_{`b}}
\def \eqbmu	{=_{`b\kern-.5\point`m}}
\def \eqbmx	{=_{`b\kern-.5\point`m\xsind}}
\def \LC	{\Sc{lc}}

\def \diverges {\,{\raise1\point\hbox{\Red \small $\uparrow$}}}
\def \converges {\,{\raise1\point\hbox{\Red \small $\downarrow$}}}
\def \divergesPi {\,{\raise1\point\hbox{\Red \small $\uparrow_{`p}$}}}
\def \convergesbeta {\,{\raise1\point\hbox{\Red \small $\downarrow_{`b}$}}}
\def \convergesbmu {\,{\raise1\point\hbox{\Red \small $\downarrow_{`b`m}$}}}
\def \divergesbeta {\,{\raise1\point\hbox{\Red \small $\uparrow_{`b}$}}}
\def \divergesxs {\,{\raise1\point\hbox{\Red \small $\uparrow$}_{\kern-.5\point\xsind}}}
\def \divergess {\,{\raise1\point\hbox{\Red \small $\uparrow$}_{\kern-.5\point\sind}}}
\def \divergesh {\,{\raise1\point\hbox{\Red \small $\uparrow$}_{\kern-.5\point\hind}}}
\def \divergesxh {\,{\raise1\point\hbox{\Red \small $\uparrow$}_{\kern-.5\point\xhind}}}
\def \divergesS {\,{\raise1\point\hbox{\Red \small $\uparrow_{\sind}$}}}
\def \convergesS {\,{\raise1\point\hbox{\Red \small $\downarrow_{\sind}$}}}
\def \convergesxs {\,{\raise1\point\hbox{\Red \small $\downarrow_{\xsind}$}}}

\newcommand{\indexot}[2]{{#1}\skp{\in}\skp{#2}}

\def \iotn {\indexot{i}{\n}}

\def \Top {\mbox{\typecolour $\top$}}
\def \Bottom {\mbox{\typecolour $\perp$}}
\def \Types	{{\Red \cal T}}
\def \IntTypes	{{\Red {\cal T}_{\cap}}}
\def \UnTypes	{{\Red {\cal T}_{\cup}}}

\def \Tproper	{\Types\kern-3\point_{\textit{\scriptsize \Yellow p}}}
\def \Tstrict	{\Types_{\textit{\scriptsize s}}}

\def \dom(#1){\textit{dom}\,(#1)}

\def \k	{\underline{k}}

\def \m	{\underline{m}}
\def \n	{\underline{n}}

\def \p	{\underline{p}}

\def \arrow	{\mathord{\rightarrow}}
\def \prod	{\mathord{\times}}

\def \arr {\arrow}
\def \Arrow	{{\Rightarrow}}
\def \tvar	{\varphi}

\def \intsymb	{\mbox {\small $\cap$}} 
\def \inter {\mathord{\hskp{\intsymb}\hskp}}

\def \intTwo#1(#2){\intsymb_{\ul{#1}}\qskp(#2)}
\def \intOne#1{\@ifnextchar({\intTwo{#1}}{\intsymb_{\ul{#1}}\qskp}}
\def \int{\@ifnextchar\bgroup{\intOne}{\inter}}

\def \unsymb	{{\cup}} 
\def \union {\mathord{\hskp\type {\unsymb} \hskp}}
\def \compunion	{\setbox172=\hbox{$\cup$}%
\copy172\kern-.75\wd172^c\kern.25\wd172}

\def \unTwo#1(#2){\unsymb\kern-.5\point_{\ul{#1}}\qskp({#2})}
\def \unOne#1{\@ifnextchar({\unTwo{#1}}{\unsymb\kern-.5\point_{\ul{#1}}\qskp}}
\def \un{\@ifnextchar\bgroup{\unOne}{\union}}

\def \BuildInt({\@ifnextchar{n}{\BuildIntn(}{\@ifnextchar{m}{\BuildIntm(}{\BuildIntAux(}}}
\def \BuildIntAux(#1)#2{\int{#1}{#2}_l}
\def \BuildIntn(#1)#2{\int{n}{#2}_i}
\def \BuildIntm(#1)#2{\int{m}{#2}_j}
\def \BuildUn({\@ifnextchar{n}{\BuildUnn(}{\@ifnextchar{m}{\BuildUnm(}{\BuildUnA(}}}
\def \BuildUnA(#1)#2{\un{#1}{#2}_l}
\def \BuildUnn(#1)#2{\un{n}{#2}_i}
\def \BuildUnm(#1)#2{\un{m}{#2}_j}

\def \AoI#1{\BuildInt(#1){\type A }}
\def \AoU#1{\BuildUn(#1){\type A }}

\def \AoIn {\AoI{n}}

\def \AoUn {\AoU{n}}
\def \AoInfull {A_1 \inter \dots \inter A_n}
\def \AoUnfull {A_1 \union \dots \union A_n}

\def \BoI#1{\BuildInt(#1){B}}
\def \BoU#1{\BuildUn(#1){B}}

\def \CoI#1{\BuildInt(#1){C}}

\def \CoU#1{\BuildUn(#1){C}}

\def \CoIn {\CoI{n}}

\def \CoUn {\CoU{n}}

\def \DoI#1{\BuildInt(#1){D}}

\def \Dot#1{D_i\,(\forall i \ele #1)}

\def \DoIn {\DoI{n}}

\def \Dotn {\Dot{n}}

\def \DeloU#1{\un{#1}`D_i}

\def \DeloUn {\DeloU{n}}

\def \GUnd_#1{`G\kern-1.5pt_{#1}}
\def \GCom,{`G\kern-1.5pt,}
\def \G{\@ifnextchar_{\GUnd}{\@ifnextchar,}{\GCom}{`G}}
\def \Ga_#1{`G'\kern-1.5pt_{#1}}

\def \GamoI#1{{\int{#1}\G_i}}

\def \GamoIn {{\GamoI{n}}}

\def \GamoIn {{\GamoI{n}}}

\def \intI	{{\Red \intsymb \textsl{I}}}
\def \intE	{{\Red \intsymb \textsl{E}}}

\def \unL	{{\Red \unsymb \textsl{L}}}

\def \unR	{{\Red \unsymb \textsl{R}}}
\def \unE	{{\Red \unsymb \textsl{E}}}
\def \Ax {{\Red \mbox{\sl Ax}}}
\def \arrI	{{\Red \arrow \mbox{\sl I}}}
\def \arrE	{{\Red \arrow \mbox{\sl E}}}

\def \D {{\orange \mathcal D}}
\def \dcol {\,{::}\,}

\def \arrL	{{\arr}\textsl{L}}
\def \arrR	{{\arr}\textsl{R}}
\def \ArrL	{\Arr\textsl{L}}
\def \ArrR	{\Arr\textsl{R}}
\def \intL	{\intsymb\textsl{L}}

\def \intR	{\intsymb\textsl{R}}

\def \Weak	{\textsl{W}}

\def \Except	{{\setminus}}
\def \Iff	{\mathbin{\,{\Leftarrow}\kern-6\point{\Rightarrow}\,}}
\def \Thenone#1{\mathrel{{\Rightarrow}\,(#1)}}
\def \Then{\@ifnextchar{\bgroup}{\Thenone}{\mathrel{\Rightarrow}}}

\def \Implyone#1{\mathrel{\Rightarrow\,(#1)}}
\def \Imply{\@ifnextchar{\bgroup}{\Implyone}{\mathrel{\Rightarrow}}}

\def \X {{\Red \mbox{$\mathcal{X}$}}}
\def \Xi {{\Red{\mathcal X}^{\textit{\footnotesize i}}}} 

\def \diagX {{\Red ^d \kern-.5\point \mathcal X}} 
\def \copyX {{\Red ^{\scriptsize \copyright} \kern-.5\point \mathcal X}} 

\setbox134=\hbox{\leavevmode\raise0.5\point\hbox{${\mysemcolour	\lceil}\kern-3\point{\mysemcolour \lceil}\kern-.5pt$}}
\setbox135=\hbox{\raise0.5\point\hbox{$\kern-.5pt{\mysemcolour	\rfloor}\kern-3\point{\mysemcolour	\rfloor}$}}

\def \LeftTop#1{{\hspace{1.5\point}
\setbox66=\hbox{$#1$}%
\ifdim\ht66>5\point \setlength{\unitlength}{.1\point} 
	\psset{unit=.125\ht66,linewidth=0.4\point,linecolor=black}%
\else \setlength{\unitlength}{.2\point} 
	\psset{unit=.8\point,linewidth=0.4\point,linecolor=black}%
\fi
\ifmycolour\psset{linecolor=semcolour}\fi
\setbox67=\hbox{\put(-3,2)%
{\psline{c-c}(0,0)(3,0)%
 \psline{cc-cc}(0,0)(0,-8)%
 \psline{cc-cc}(1.5,0)(1.5,-6)%
}}%
\raise7.25\point\box67
\kern2.4\point}}

\def \RightBot{{\kern2.75\point
\psset{unit=1\point,linewidth=0.4\point,linecolor=black}%
\ifmycolour\psset{linecolor=semcolour}\fi
\hbox{\put(0,-2)%
{\psline{c-c}(0,0)(-3,0)%
 \psline{cc-cc}(-1.35,0)(-1.35,6)%
 \psline{cc-cc}(0,0)(0,8)%
}}\kern2\point}}

\def \Rightbot{{\kern3.25\point
\psset{unit=1\point,linewidth=0.4\point,linecolor=black}%
\ifmycolour\psset{linecolor=semcolour}\fi
\hbox{\put(0,-2)%
{\psline{c-c}(0,0)(-3,0)%
\psline{cc-cc}(-1.35,0)(-1.35,5)%
\psline{cc-cc}(0,0)(0,5)%
}}\kern1\point}}

\def \OldSem[#1]{{\semcolour [\kern-1.75\point[}\kern.6\point \lterm{#1} \kern.6\point{\semcolour ]\kern-1.75\point]} }

\def \x {\X}
\def \semEl#1{\LeftTop{#1}{#1}\RightBot}
\def \semThree#1#2#3{\setbox136=\hbox{{\scriptsize \Yellow $#2$}}\setbox137=\hbox{#3}%
	(#1)_{\copy136}\kern-\wd136%
	\raise5\point\box137}
\def \SemThree#1#2#3{%
	\setbox136=\hbox{{\scriptsize \termcolour ${#2}$}}%
	\setbox137=\hbox{#3}%
	\semEl{#1}_{\kern-1\point\termcolour \copy136}
	\ifdim\wd136>1pt\kern-\wd136\raise5\point\copy137
	\ifdim\wd136>\wd137 \kern-\wd137\kern\wd136\fi
	\fi}
\def \SemTwo#1#2{\@ifnextchar{\bgroup}{\SemThree{#1}{#2}}{\semEl{#1}\kern-1\point_{#2}}}
\def \Sem#1{\@ifnextchar{\bgroup}{\SemTwo{#1}}{\semEl{#1}}}
\def \SemL#1#2{\SemThree{{\ltermcolour #1}}{#2}{\mbox{\mysemcolour	\scriptsize $\lambda$}}}
\def \semL [#1] #2 {\SemL{#1}{#2}}

\def \SemGr#1{\Sem{#1}\kern-1\point_{\mbox{\scriptsize \sc g}}}

\def \ftn{\mbox{\Red \textsc{\scriptsize n}}}
\def \ftv{\mbox{\Red \textsc{\scriptsize v}}}

\def \SemN#1{\semEl{#1}\kern-1\point_{\ftn}}
\def \SemV#1{\semEl{#1}\kern-1\point_{\ftv}}

\def \SemN#1{\semEl{#1}\kern-1\point_{\ftn}}
\def \SemV#1{\semEl{#1}\kern-1\point_{\ftv}}

 \def\NF#1{\mbox{\itbf #1}}

 \def \LmuNFa'{\NF{N}\skp'}
 \def \LmuNF{\@ifnextchar{'}{\LmuNFa}{\NF{N}}}

 \def \LmuHNFa'{\NF{H}\skp'}
 \def \LmuHNF {\@ifnextchar{'}{\LmuHNFa}{\NF{H}}}

 \def \Id<#1>{#1}
 \def \VarBrTerm<(#1)>{#1}
 \def \VarTermAux<#1>{\@ifnextchar({\VarBrTerm<#1>}{\Id<#1>}}
 \def \VarTerm<#1>{\if@om {\global \omfalse} \VarTermAux<#1> {\global \omtrue} \else \Id<#1>\fi}

 \def \LTerm<#1>{\LTermAux<#1>}

 \def\RemBr<(#1)>{\LTermAux<#1>}

 \def \NoBrLTermAux<{\@ifnextchar({\RemBr<}{\LTermAux<}} %
 \def \LTermAux<%
{\ltermcolour 
	 \@ifnextchar l{\LTermAbs<}%
	{\@ifnextchar a{\LTermapp<}%
	{\@ifnextchar X{\LTermXpp<}%
	{\@ifnextchar d{\LTermDpp<}%
	{\@ifnextchar e{\LTermEpp<}%
	{\@ifnextchar t{\LTermTpp<}%
	{\@ifnextchar v{\LTermVar<}%
	{\@ifnextchar S{\LTermSub<}%
	{\@ifnextchar s{\LTermsub<}%
	{\@ifnextchar ({\LTermBr<}%
	{\@ifnextchar X{\LTermMultiSub<}%
	{\@ifnextchar C{\LTermContSub<}%
	{\@ifnextchar c{\LTermContsub<}%
	{\@ifnextchar n{\LTermName<}%
	{\@ifnextchar m{\LTermMu<}%
	{\@ifnextchar B{\LTermBot<}%
		{\VarTerm<}%
}}}}}}}}}}}}}}}}

 \def \LTermapp<a #1 #2>{\LTermAux<#1> \LTermAux<#2>}
 \def \LTermXpp<X #1 #2>{\LTermAux<#1> \LTermAux<#2>}
 \def \LTermDpp<d #1 #2>{\LTermAux<#1> \LTermAux<#2>}
 \def \LTermEpp<e #1 #2>{\LTermAux<#1> \LTermAux<#2>}
 \def \LTermTpp<t #1 #2>{\LTermAux<#1> \LTermAux<#2>}
 \def \LTermAbs<l #1 . #2>{`l#1 . \LTermAux<#2> }
 \def \LTermBr<(#1)>{(\LTermAux<#1>)}
 \def \LTermVar<v #1>{#1}
 \def \LTermSub<S #1 #2 := #3>{\LTermAux<#1> \hskp \slangle{#2}{\,:=\,}\LTermAux<#3>\srangle } 
\def \LTermsub<s #1 #2 := #3>{\LTermAux<#1> \exsub #2 := {\LTermAux<#3>} }
 \def \LTermContSub<C #1 #2 := #3 . #4>{\LTermAux<#1> \hskp \slangle{#2}{\,:=\,}\LTermAux<#3>{`.}#4\srangle }
 \def \LTermContsub<c #1 #2 := #3 . #4>{\LTermAux<#1> \hskp \slangle{#2}{\,:=\,}\LTermAux<#3>{`.}#4\srangle }
 \def \LTermName<n #1 #2>{ [#1] \LTermAux<#2> }
 \def \LTermMu<m #1 . #2>{ `m #1 . \LTermAux<#2> }
 \def \LTermBot<B>{\bot}

\newbox\@indexbox

\newcounter{myindbeta} 
\def \newmybeta{\ifnum\themyindbeta=-1$b'$\else \ifnum\themyindbeta=0{$`b$}\else${`b}_{\themyindbeta}$\fi\fi\stepcounter{myindbeta}}
\newcounter{myindgamma} 
\def \newmygamma{\ifnum\themyindgamma=-1$`g'$\else \ifnum\themyindgamma=0{$`g$}\else${`g}_{\themyindgamma}$\fi\fi\stepcounter{myindgamma}}
\newcounter{myinddelta} 
\def \newmydelta{\ifnum\themyinddelta=-1$`d'$\else \ifnum\themyinddelta=0{$`d$}\else${`d}_{\themyinddelta}$\fi\fi\stepcounter{myinddelta}}
\newcounter{myindrho} 
\def \newmyrho{\ifnum\themyindrho=-1$`r'$\else \ifnum\themyindrho=0{$`r$}\else${`r}_{\themyindrho}$\fi\fi\stepcounter{myindrho}}

 \def \resetLvariables{%
\setcounter{myindbeta}{0}%
\setcounter{myindgamma}{0}%
\setcounter{myinddelta}{0}%
\setcounter{myindrho}{0}%
}

 \def \XHLSem[#1] #2 {\setbox199=\hbox{$#2$}%
{\Sem{\LTerm<#1>}}\kern-1\point\raise5\point\hbox{\tiny $\cal X$}\kern-6\point_{#2} 
\kern2\point}

 \def \SHTerm [#1] #2 {\XHLSem{\LTerm<#1>} {#2} }

 \def \XHLTerm [#1] #2 {\if@Xom \Xomfalse \XHLTermAux[#1] #2 \Xomtrue \resetLvariables 
 \else \XHLTermAux[#1] #2 \fi}

 \def \XHLTermAux[
	{\@ifnextchar L{\XHLTermAbs[}
	{\@ifnextchar A{\XHLTermApp[}
	{\@ifnextchar D{\XHLTermDpp[}
	{\@ifnextchar V{\XHLTermVar[}
	{\@ifnextchar S{\XHLTermSub[}
	{\@ifnextchar T{\XHLTermTub[}
	{\@ifnextchar({\XHLTermBr[}
	{\@ifnextchar<{\XHLTermLTerm[}
		{\XHLSem[}
}}}}}}}}

 \def \XHLTermApp[A #1 #2] #3 {
\@ifnextchar<{\XHLTermAppPar [A {#1} {#2}] {#3} }{\XHLTermAppNoPar [A {#1} {#2}] {#3} }}

 \def \XHLTermAppNoPar[A #1 #2] #3 {
\cut { \XHLTermAux [#1] `g } `g + x {\imp { \XHLTermAux [#2] `b } `b [x] z \caps<z,#3> } }

 \def \XHLTermAppPar[A #1 #2] #3 <#4>{P
\cut { \XHLTermAux [#1] `g } `g + x {\imp { \XHLTermAux [#2] {#4} } {#4} [x] z \caps<z,#3> } }

 \def \XHLTermDpp[D #1 #2] #3 {
\cutdl { \XHLTermAux [#1] `g } `g + x { \imp { \XHLTermAux [#2] `b } `b [x] z \caps<z,#3> } }

 \def \XHLTermAbs[L #1 . #2] #3 {
\@ifnextchar<{\XHLTermAbsPar [L {#1} . {#2}] {#3} }{\XHLTermAbsNoPar [L {#1} . {#2}] {#3} }}

 \def \XHLTermAbsNoPar[L #1 . #2] #3 {
\exp #1 {\XHLTermAux [#2] `d } `d . #3 }

 \def \XHLTermAbsPar[L #1 . #2] #3 <#4>{
\exp #1 {\XHLTermAux [#2] #4 } #4 . #3 }

 \def \XHLTermVar[V #1] #2 {
 \caps<#1,#2> }

 \def \XHLTermSub[S #1 #2 := #3] #4 {
 \cut {\XHLTermAux [#3] `b } `b + {#2} {\XHLTermAux [#1] #4 } }

 \def \XHLTermBr[(#1)] #2 {
\XHLTermAux [#1] #2 }

 \def \XHLTermLTerm[<#1>] #2 {
\XHLSem [#1] #2 }

\def \Neg#1{\overline{#1}} 
\def \NegBr(#1){({#1)^o}}
\def \NegNoBr#1{{{#1}^o}}
\def \Neg{\@ifnextchar({\NegBr}{\NegNoBr}}

\def \SemGN#1{\slangle{#1}\srangle\kern-1\point_{\ftn}^{l}}
\def \SemDN#1{\slangle{#1}\srangle\kern-1\point_{\ftn}^{r}}
\def \SemGV#1{\slangle{#1}\srangle\kern-1\point_{\ftv}^{l}}
\def \SemDV#1{\slangle{#1}\srangle\kern-1\point_{\ftv}^{r}}

\def \SemLx[#1] #2 {\SemThree{#1}{#2}{\mbox{\mysemcolour	\tiny $`l$\sc x}}}

\def \LmmtSemN#1{\semEl{#1}\kern-1\point_{\ftn}}
\def \LmmtSemV#1{\semEl{#1}\kern-1\point_{\ftv}}

\newcounter{sind}0
\newcounter{tind}0
\def \newsigma{\if@om\omfalse`s\setcounter{sind}0\setcounter{tind}0\else\stepcounter{sind}`s_{\thesind}\fi}
\def \newtau{\ifnum\thetind=0\stepcounter{tind}`t\else`t_{\thetind}\fi}

\def \lefttop {{\leavevmode \kern2\adjustablepoint
\psset{unit=.9\adjustablepoint,linewidth=0.4\adjustablepoint}
\ifmycolour\psset{linecolor=semcolour}\fi
\hbox{\psline{c-c}(0,9)(3,9)\psline{cc-cc}(1.5,9)(1.5,-2)\psline{cc-cc}(0,9)(0,-2)\psline{c-c}(0,-2)(3,-2)}\kern3.5\adjustablepoint}}

\def \righttop{{\kern2\adjustablepoint
\psset{unit=.9\adjustablepoint,linewidth=0.4\adjustablepoint}
\ifmycolour\psset{linecolor=semcolour}\fi
\hbox{\psline{c-c}(1.5,9)(-1.5,9)\psline{cc-cc}(0,9)(0,2)\psline{cc-cc}(1.5,9)(1.5,2)}\kern1\adjustablepoint}}

\def \leftsem {\leavevmode \kern1\point
\psset{unit=.9\point,linewidth=0.4\point,linecolor=black}%
\ifmycolour\psset{linecolor=semcolour}\fi
\hbox{\psline{c-c}(0,9)(3,9)\psline{cc-cc}(1.5,9)(1.5,-2)\psline{cc-cc}(0,9)(0,-2)\psline{c-c}(0,-2)(3,-2)}\kern3.5\point}

\def \rightsem{\kern2\point
\psset{unit=.9\point,linewidth=0.4\point,linecolor=black}%
\ifmycolour\psset{linecolor=semcolour}\fi
\hbox{\psline{c-c}(1.5,9)(-1.5,9)\psline{cc-cc}(0,9)(0,-2)\psline{cc-cc}(1.5,9)(1.5,-2)\psline{c-c}(1.5,-2)(-1.5,-2)}\kern1\point}

\def \ContSub #1 #2 := #3 . #4 { {#1} \hskp \slangle{#2}{\,:=\,}{#3}{`.}{#4}\srangle }
\def \Sub #1 #2 := #3 { {#1} \hskp \slangle{#2}{\,:=\,}{#3}\srangle }
\def \VSub #1 #2 := #3 {{#1} `<\Vec{{#2}{\,:=\,}{#3}}`> }
\def \exsub #1 := #2 {\lterm{ `<{#1}{\,:=\,}{#2}`> }}
\def \imsub #1/#2 {[#1/#2]}
\def \clo#1#2#3{\Sub {#1} {#2} := {#3} }

\def \Lx {\mbox{$\Red `l${\bf\Red x}}}

\def \LK{\textsc{lk}}

\setbox139=\hbox{{\Yellow \raise0\point\hbox{\rotatebox{20}{$\downarrow$}}\kern-2.5pt\raise1.72\point\hbox{\rotatebox{-20}{$\downarrow$}}}}
 \def \eqX	{\stackrel{\mbox{\scriptsize $\cal X$}}{=}}
\def \eqXCBN	{\mathbin{\eqX\kern-2pt_{\ftn}}}
\def \eqXCBV	{\mathbin{\eqX\kern-2pt_{\ftv}}}

\def \crXCBN	{\mathbin{\downarrow}\kern-.5\point_{\ftn}}
\def \crXCBV	{\mathbin{\downarrow}\kern-.5\point_{\ftv}}

\def \openL{`L\kern-1\point^{\mbox{\scriptsize \sc o}}}
\def \closedL{`L\kern-1\point^{\mbox{\scriptsize \sc c}}}
\def \BothTerm{\mathrel{\stackrel{\raise 1.5\point \hbox{\tiny{$\bullet$}}}{\approx}}}

\def \mutilde	{\tilde{`m}}
\def \mt	{\mutilde}

\def \lmmt	{{\Red \overline{`l}`m\mutilde}}

\def \lmu	{{\Red `l`m}}

\def \ML	{{\Red \textsc{ml}}}
\def \redLK {\mathrel{{\rightarrow}\kern-1\point_{\cal LK}}}
\def \rtcredLK {\mathrel{{\rightarrow}\kern-1pt^*\kern-3\point_{\scriptLK}}}
\def \eqX	{\stackrel{{\cal X}}{ = }}
\def \redX {\mathrel{{\rightarrow}\kern-2\point_{\cal X}}}

\def \rtcredX {\mathrel{{\rightarrow}\kern-1pt^*\kern-3\point_{\cal X}}}
\def\scriptLK{\mbox{\tiny $\cal LK$}}
\def \red {\mathrel{\Red \arrow}}

\def \rtcred {\mathrel{\Red \arrow\kern-1\point^*}}
\def \dred {\Red \rightarrow\kern-.75em\rightarrow}

\def \redCBV {\mathrel{\Red \red_{\ftv}}}
\def \redCBN {\mathrel{\Red \red_{\ftn}}}

\def \redlmu {\mathrel{\arrow\kern-1\point_{`l`m}}}

 \def \eqb	{\mathrel{=_{`b}}}

 \def \redbmu {\mathrel{\Red \arrow_{`b\kern-.5\point`m}}}
 \def \rtcredbmu {\mathrel{\Red \arrow^{\ast}_{`b\kern-.5\point`m}}}
 \def \redlazy {\mathrel{\Red \arrow_{\kern-.5pt\mbox{\scriptsize \sc l}}}}
 \def \redom {\mathrel{\Red \arrow_{\kern-.5pt\mbox{\scriptsize \sc om}}}}
 \def \rtcredlazy {\mathrel{\Red \arrow^{\ast}_{\kern-.5pt\mbox{\scriptsize \sc l}}}}
 \def \reds {\mathrel{\Red \arrow_{\kern-.5\point\sind}}}
 \def \redh {\mathrel{\Red \arrow_{\kern-.5\point\hind}}}
 \def \rtcredh {\mathrel{\Red \arrow^{*}_{\kern-.5\point\mbox{\hind}}}}
 \def \rtcredsub {\mathrel{\Red \arrow^{*}_{\kern-.5\point\mbox{\scriptsize {\rm :=}}}}}
 \def \redxh {\mathrel{\Red \arrow_{\kern-.5\point\mbox{\xhind}}}}
 \def \tcredxh {\mathrel{\Red \arrow_{\kern-.5\point\xhind}}}
 \def \rtcredxh {\mathrel{\Red \arrow^{\ast}_{\kern-.5\point\xhind}}}
 \def \eqxh {\mathrel{\Red =_{\kern-.5\point\xhind}}}
 \def \redx {\mathrel{\Red \arrow_{\kern-.5\point\mbox{\scriptsize {\bf x}}}}}
 
 \def \redxsub {\mathrel{\Red \arrow_{\kern-.5\point\mbox{\scriptsize {\rm :=}}}}}
 \def \tcredx {\mathrel{\Red \arrow^{+}_{\kern-.5\point\mbox{\scriptsize {\bf x}}}}}
 \def \rtcredx {\mathrel{\Red \arrow^{*}_{\kern-.5\point\mbox{\scriptsize {\bf x}}}}}
 \def \rtcredxsub {\mathrel{\Red \arrow^{*}_{\kern-.5\point\mbox{\scriptsize {\rm :=}}}}}
 \def \redxV {\mathrel{\Red \arrow_{\kern-.5\point\mbox{\scriptsize {\bf x}{\sc v}}}}}
 \def \redxl {\mathrel{\Red \arrow_{\kern-.5\point\mbox{\scriptsize {\bf x}{\sc l}}}}}
 \def \redxs {\mathrel{\Red \arrow_{\kern-.5\point\xsind}}}
 \def \eqxs {\mathrel{\Red =_{\kern-.5\point\xsind}}}
 \def \crxs {\mathbin{\downarrow_{\kern-.5\point\xsind}}}

 \def \rtcredxs {\mathrel{\Red \arrow^{*}_{\kern-.5\point\xsind}}}
 \def \redxsa {\mathrel{\Red \arrow'_{\kern-.5\point\xsind}}}
 \def \tcreds {\mathrel{\Red \arrow^{+}_{\kern-.5\point\sind}}}
 \def \tcredl {\mathrel{\Red \arrow^{+}_{\kern-.5\point\mbox{\scriptsize {\sc l}}}}}
 \def \tcredxl {\mathrel{\Red \arrow^{+}_{\kern-.5\point\mbox{\scriptsize {\bf x}{\sc l}}}}}
 \def \tcredxs {\mathrel{\Red \arrow^{+}_{\kern-.5\point\xsind}}}
 \def \rtcreds {\mathrel{\Red \arrow^{\ast}_{\kern-.5\point\sind}}}
 \def \rtcredl {\mathrel{\Red \arrow^{\ast}_{\kern-.5\point\mbox{\scriptsize {\sc l}}}}}
 \def \rtcredx {\mathrel{\Red \arrow^{\ast}_{\kern-.5\point\mbox{\scriptsize {\bf x}}}}}
 \def \rtcredxl {\mathrel{\Red \arrow^{\ast}_{\kern-.5\point\mbox{\scriptsize {\bf x}{\sc l}}}}}
 \def \rtcredxs {\mathrel{\Red \arrow^{\ast}_{\kern-.5\point\xsind}}}
\def \redes {\mathrel{\Red \arrow_{\kern-.5\point\mbox{\scriptsize \sc es}}}}
 \def \dreds {\mathrel{\Red \arrow\kern-8\point\arrow_{\kern-.5\point\sind}}}

\def \Pair<#1,#2>{`<{#1},{#2}`>}

\def \Group<#1>{\slangle{#1}\srangle}

\def \BP#1{\bp(#1)}
\def \BS#1{\bs(#1)}
\def \bp(#1){{\typecolour \It{bp}({\termcolour #1})}}
\def \bs(#1){{\typecolour \It{bs}({\termcolour #1})}}
\def \BV#1{\textit{bv}(#1)}
\def \bv(#1){\BV{#1}}

\def \FC#1{\fc(#1)}
\def \FP#1{\fp({#1})}
\def \FS#1{{\fs({#1})}}
\def \fp(#1){\textrm{\it fp}({\termcolour #1})}
\def \fs(#1){\textrm{\it fs}\skp({\termcolour #1})}
\def \fc(#1){\textrm{\it fc}\skp({\termcolour #1})}
\def \FV#1{\textrm{\it fv}\skp({\termcolour #1})}
\def \fv(#1){\FV{#1}}
\def \HN#1{\textrm{\it hn}\skp({\termcolour #1})}
\def \hn(#1){\HN{#1}}
\def \HV#1{\textrm{\it hv}\skp({\termcolour #1})}
\def \hv(#1){\HV{#1}}

\def \Mid
{\hspace*{2.5pt}{\black \mid}\hspace*{2.5pt}}

\def \Langle{\kern1\adjustablepoint
\psset{unit=1\adjustablepoint,linewidth=0.4\adjustablepoint,linecolor=ltermcol}%
\put(0,2.5){\psline{c-c}(0,0)(1.7,4.9)%
\pscurve{c-c}(1.7,4.9)(.7,0)(1.7,-4.9)
\psline{c-c}(0,0)(1.7,-4.9)
}\kern3\adjustablepoint}
\def \Rangle{\psset{unit=1\adjustablepoint,linewidth=0.4\adjustablepoint,linecolor=ltermcol}%
\kern3\adjustablepoint
\put(0,2.5){\psline{c-c}(0,0)(-1.7,4.9)%
\pscurve{c-c}(-1.7,4.9)(-.8,0)(-1.7,-4.9)+
\psline{c-c}(0,0)(-1.7,-4.9)
}\kern\adjustablepoint}
\def \PiCell#1#2{\setbox171=\hbox{$#1$}\setbox172=\hbox{$#2$}
{\ltermcolour
\Langle 
	\ifdim\wd171>24\point\skp{\box171}\skp\else{\box171}\fi
\kern1pt{\mid}\kern1pt
	\ifdim\wd172>24\point\skp{\box172}\skp\else{\box172}\fi
\Rangle 
}}
\def \Cell#1#2{\slangle#1\skp{\mid}\skp#2\srangle}
\def \cell<#1|#2>{\Cell{#1}{#2}}

\def \repl #1[#2/#3]{\lterm{#1}{\skp\Black [}\lterm{#2}{\Black /}\lterm{#3}{\Black ]} }
\def \term#1 {{\termcolour #1}}
\def \ass#1:#2 {{\termcolour #1}{\black :}{\typecol #2} }
\def \lass#1:#2 {\lterm {#1}{\black :}{\typecol #2} }

\def \type #1 {{\typecolour #1}}
\def \lterm#1{{\ltermcolour #1}}
\def \lambdaterm #1 {{\ltermcolour #1}}
\def \Rule(#1){{\blue (#1) }}

\def \Stat#1:#2{{\ltermcolour #1}{:}{\typecolour #2}}
\def \stat#1#2{{\ltermcolour #1}{:}{\typecolour #2}}
\def \lstat#1#2{{\ltermcolour #1}\,{:\,}{\typecolour #2}}

\def \turn	{{\black {\vdash}}}
\def \Turn {\mathrel{\turn}}
\def \TurnSimple {\mathrel{\turn}}
\def \TurnD {\mathrel{\turn\kern-3\point_{\mbox{\tinyfont $\textit{\black D}$}}}}
\def \TurnI {\mathrel{\turn\kern-3\point_{\mbox{\tinyfont $\black \cap$}}}}
\def \TurnL {\mathrel{\turn\kern-.5\point_{`l}}}
\def \TurnLK {\mathrel{\Yellow \turn\kern-2\point_{\ftsc{lk}}}}
\def \TurnLKarr {\mathrel{\Yellow \turn\kern-2\point_{\ftsc{lk}({\rightarrow})}}}
\def \TurnR {\mathrel{\Yellow \turn\kern-2\point_{\ftsc{r}}}}
\def \TurnS {\mathrel{\turn\kern-2\point_{\ftsc{s}}}}
\def \TurnLN {\mathrel{\Yellow \turn\kern-2\point_{\ftsc{ln}}}}
\def \TurnLV {\mathrel{\Yellow \turn\kern-2\point_{\ftsc{lv}}}}

\def \TurnLmmt {\mathrel{\black \turn_{\sclmmt}}}
\def \TurnLmmtSN {\mathrel{\turn_{\sclmmt}^{\ftsc{sn}}}}
\def \TurnLmmtN {\mathrel{\turn\kern-2\point_{\black \ftsc{n}}}}
\def \TurnLmmtV {\mathrel{\turn\kern-2\point_{\black \ftsc{v}}}}
\def \TurnLmmtmin {\mathrel{\turn\kern-2\point_{\black \ftsc{m}}}}
\def \TurnLmmtmin {\mathrel{\turn}}
\def \TurnLmux {\mathrel{\turn\kern-\point_{\xind}}}
\def \TurnLmu {\mathrel{\turn\kern-\point_{`l`m}}}
\def \TurnLmuIU {\mathrel{\,\turn\kern-2\point_{`l`m}^{\,\raise 1\point\hbox{\tinyfont $\cap\cup$}}\,}}
\def \TurnLx {\mathrel{\turn_{`l\mbox{\scriptsize \sf x}}}}
\def \TurnMCL {\mathrel{\Yellow \turn\kern-2\point_{\ftsc{mcl}}}}
\def \TurnN {\mathrel{\Yellow \turn\kern-2\point_{\ftsc{n}}}}
\def \TurnNf {\mathrel{\Yellow \turn\kern-2\point_{\ftrm{n}}}}
\def \TurnV {\mathrel{\Yellow \turn\kern-2\point_{\ftsc{v}}}}
\def \TurnVf {\mathrel{\Yellow \turn\kern-2\point_{\ftrm{v}}}}
\def \TurnX {\mathrel{\turn\kern-2\point_{\cal X}}}
\def \TurnP {\mathrel{\turn\kern-3\point_p}}
\def \TurnM {\mathrel{\turn\kern-4\point_{\mbox{\tinyfont $\cal M$}}}}
\def \TurnMc {\mathrel{\turn\kern-4\point_{\mbox{\tinyfont $\cal M^{\tinysc{c}}$}}}}
\def \TurnMe {\mathrel{\turn\kern-4\point_{\mbox{\tinyfont $\cal M^{\tinysc{e}}$}}}}
\def \TurnMIU {\mathrel{\turn\kern-4\point_{\mbox{\tinyfont $\cal M^{\cap\cup}$}}}}
\def \TurnMI {\mathrel{\turn\kern-4\point_{\mbox{\tinyfont $\cal M^{\cap}$}}}}

\def \sclmmt {\mbox{
\tiny $\overline{`l}`m\mt$}}

\def \Der#1#2{\@ifnextchar{\bgroup}{\DerThree{#1}{#2}}{\DerTwo{#1}{#2}}}
\def \DerTwo#1#2{{\typecolour #1 \Turn #2}}
\def \DerThree#1#2#3{{\typecolour #2 \Turn #3}}
\def \Stoup#1#2{\lstat{#1}{#2}}
\def \StoupR#1#2{\lstat{#1}{#2} \Mid}
\def \StoupR#1#2{\framebox{$#1\,{:}\,#2$}~}
\def \StoupR#1#2{\StoupFrame{\lstat{#1}{#2}}}

\def \Stoup#1#2{\lstat{#1}{#2}}

\def \StoupR#1#2{\lstat{#1}{#2}\Skip\Mid}
\def \DerLmuThree#1#2#3{#1 \TurnLmu \Stoup{#2}{#3} }
\def \DerLmuFour#1#2#3#4{#1 \TurnLmu \StoupR{#2}{#3} #4}
\def \DerLmu#1#2#3{\@ifnextchar{\bgroup}{\DerLmuFour{#1}{#2}{#3}}{\DerLmuThree{#1}{#2}{#3}}}
\def \derLmuFour #1 |- #2 : #3 | #4 {\type{#1} \TurnLmu \Stoup{#2}{#3} \Mid \type{#4} }
\def \derLmuNormal #1 |- #2 : #3 {\@ifnextchar|{\derLmuFour #1 |- #2 : #3 }{#1 \TurnLmu \Stoup{#2}{#3} }}
\def \derLmuCmd #1 |- #2 | #3 {\type{#1} \TurnLmu \lterm{#2} \Mid \type{#3} }
\def \derLmuBasis #1 |- #2 {\@ifnextchar:{\derLmuNormal #1 |- #2 }{\derLmuCmd #1 |- #2 }}
\def \derLmu {\@ifnextchar|{\derLmuBasis {\emptyset} }{\derLmuBasis }}

\def \derLmuIU #1 {\def \TurnLmu {\TurnLmuIU}\derLmu {#1} }
\def \derLmux #1 {\def \TurnLmu {\TurnLmux}\derLmu {#1} }
\def \derLmuP  #1 {\def \TurnLmu {\skipper\TurnP\skipper}\derLmu {#1} }
\def \derLmuS  #1 {\def \TurnLmu {\TurnS}\derLmu {#1} }

\def \derMCLFour #1 |- #2 | #3 {#1 \TurnMCL {#2} \mid {#3} }
\def \derMCLNormal #1 |- #2 {\@ifnextchar|{\derMCLFour #1 |- #2 }{#1 \TurnMCL {#2} }}
\def \derdeMCL #1 |- | #2 {#1 \TurnMCL \mid {#2} }
\def \derMCL #1 |- {\@ifnextchar|{\derdeMCL #1 |- }{\derMCLNormal {#1} |- }}
\def \DerLmmt#1#2{\@ifnextchar{\bgroup}{\DerLmmtThree{#1}{#2}}{\DerLmmtTwo{#1}{#2}}}
\def \DerLmmtTwo#1#2{{\typecolour #1 \TurnLmmt #2}}
\def \DerLmmtThree#1#2#3{#1~{:}~{\typecolour #2 \TurnLmmt #3}}

\def \derLmmtTwo #1 |- #2 {{\typecolour {#1} \TurnLmmt {#2}}}
\def \derLmmtCommand #1 : #2 |- #3 {{\ltermcolour {#1}}~{\black :}~{\typecolour {#2} \TurnLmmt {#3}}}
\def \derLmmtTerm #1 |- #2 : #3 | #4 {{\typecolour {#1} \TurnLmmt {\ltermcolour #2}\skipper{\black :}\skipper{#3} ~{\black \mid}~ {#4} }}
\def \derLmmtContext #1 | #2 : #3 |- #4 {{\typecolour {#1}~{\black \mid}~{\ltermcolour #2}\skipper{\black :}\skipper{#3} \TurnLmmt {#4} }}
\def \derLmmtActivated #1 |{\@ifnextchar-{\derLmmtTerm {#1} |}{\derLmmtContext {#1} | }}
\def \derLmmtNamedDer #1 :: {{#1} \dcol \derLmmt }
\def \derLmmtColon #1 :{\@ifnextchar:{\derLmmtNamedDer #1 :}{\derLmmtCommand {#1} : }}
\def \derLmmt #1 {\@ifnextchar:{\derLmmtColon {#1} }{\derLmmtActivated {#1} }}
\def \derLmmtV #1 {\def \TurnLmmt{\TurnLmmtV}\derLmmt {#1} }
\def \derLmmtN #1 {\def \TurnLmmt{\TurnLmmtN}\derLmmt {#1} }
\def \derLmmtSN #1 {\def \TurnLmmt{\TurnLmmtSN}\derLmmt {#1} }
\def \derLmmtmin #1 {\def \TurnLmmt{\TurnLmmtmin} \derLmmt {#1} }
\def \derLmmtM #1 {\def \TurnLmmt{\TurnM} \derLmmt {#1} }
\def \derLmmtMc #1 {\def \TurnLmmt{\TurnMc} \derLmmt {#1} }

\def \derL #1 |- #2 : #3 {{\typecolour #1 \TurnL {\ltermcolour #2}\,{\black :}\,#3}}
\def \derLx #1 |- #2 : #3 {{\typecolour #1 \TurnLx {\ltermcolour #2}\,{\black :}\,#3}}
\def \derI #1 |- #2 : #3 {{\typecolour #1 \TurnI {\ltermcolour #2}\,{\black :}\,#3}}
\def \der #1 |- #2 : #3 {{\typecolour #1 \TurnSimple {\ltermcolour #2}\,{\black :}\,#3}}

\def \DerX#1#2#3{\PDer{#1}{#2}{#3}}

\def \derLK #1 |- #2 {{\typecolour #1}	\TurnLK	{\typecolour #2}}
\def \derLKarr #1 |- #2 {{\typecolour #1}	\TurnLKarr	{\typecolour #2}}
\def \deriv #1 |- #2 {{\typecolour #1}	\Turn	{\typecolour #2}}

\setbox161=\hbox{$\cdot$}
\setbox162=\hbox{{\neutral ~\raise-2.5\point\copy161\kern-\wd161\raise2.5\point\copy161\kern.5\point \raise0\point\copy161~}}
\def \DP {\copy162}
\def \DPa {{\neutral \raise-.7\point\hbox{~\large :~}}}

\def \PDer#1#2#3{{\termcolour #1} \DP {\typecolour #2} \Turn {\typecolour #3}}
\def \pDer#1:#2 |- #3 {{\termcolour #1} \DP {\typecolour {#2} \Turn {#3}}}
\def \InfPDerFour#1#2#3#4{ \InfBox{#1}{ \PDer{#2}{#3}{#4} } }
\def \InfPDer#1#2#3{\@ifnextchar{\bgroup}{\InfPDerFour{#1}{#2}{#3}}{\InfPDerFour{}{#1}{#2}{#3}}}
\def \PDerN#1#2#3{\typecolour {\termcolour #1 }\DP {\typecolour #2} \TurnN {\typecolour #3}}

\def \derXaux #1 : #2 |- #3 {{\termcolour #1 } \DP {\typecolour #2}	\TurnX	{\typecolour #3}}
\def \derX #1 : {\@ifnextchar|{\derXaux {#1} : {} }{\derXaux {#1} : }}

\def \derXR #1 : #2 |- #3 {{\termcolour #1 } \DP {\typecolour #2}	\TurnR	{\typecolour #3}}
\def \derXRN #1 : #2 |- #3 {{\termcolour #1 } \DP {\typecolour #2}	\TurnLN	{\typecolour #3}}
\def \derXRV #1 : #2 |- #3 {{\termcolour #1 } \DP {\typecolour #2}	\TurnLV	{\typecolour #3}}
\def \derXN #1 : #2 |- #3 {{\termcolour #1 } \DP {\typecolour #2}	\TurnN	{\typecolour #3}}
\def \derXNf #1 : #2 |- #3 {{\termcolour #1 } \DP {\typecolour #2}	\TurnNf	{\typecolour #3}}
\def \derXV #1 : #2 |- #3 {{\termcolour #1 } \DP {\typecolour #2}	\TurnV	{\typecolour #3}}
\def \derXVf #1 : #2 |- #3 {{\termcolour #1 } \DP {\typecolour #2}	\TurnVf	{\typecolour #3}}
\def \PDerV#1#2#3{\typecolour {\termcolour #1 }\DP {\typecolour #2} \TurnV {\typecolour #3}}
\def \InfPDerFourN#1#2#3#4{ \InfBox{#1}{ \PDerN{#2}{#3}{#4} } }
\def \InfPDerN#1#2#3{\@ifnextchar{\bgroup}{\InfPDerFourN{#1}{#2}{#3}}{\InfPDerFour{}{#1}{#2}{#3}}}
\def \InfPDerFourV#1#2#3#4{ \InfBox{#1}{ \PDerV{#2}{#3}{#4} } }
\def \InfPDerV#1#2#3{\@ifnextchar{\bgroup}{\InfPDerFourV{#1}{#2}{#3}}{\InfPDerFour{}{#1}{#2}{#3}}}

\def \derPure #1 : #2 |- #3 {{\termcolour #1 } \DP {\typecolour #2}	\Turn_{\kern-2pt\ftsc{p}}	{\typecolour #3}}
\def \derPureN #1 : #2 |- #3 {{\termcolour #1 } \DP {\typecolour #2}	\Turn_{\kern-2pt\ftsc{pn}}	{\typecolour #3}}
\def \derPureV #1 : #2 |- #3 {{\termcolour #1 } \DP {\typecolour #2}	\Turn_{\kern-2pt\ftsc{pv}}	{\typecolour #3}}

\def \derDP #1 : #2 |- #3 {{\termcolour #1 } \DP {\typecolour #2}	\Turn	{\typecolour #3}}
\def \LogicDer #1 |- #2 { {\typecolour #1}	\Turn	{\typecolour #2}}

\def \InfDerX #1 :: #2 : #3 |- #4 {\InfPDer{#1}{#2}{#3}{#4}}
\def \InfDer #1 :: #2 |- #3 : #4 {\InfLDer{#1}{#2}{#3}{#4}}

\def \XIn#1#2{{\widehat{#1}}\hspace*{.75pt}#2}
\def \XOut#1#2{#1\hspace*{1pt}{\widehat{#2}}}

\def \Dagger	{\mathrel{\mbox{\oldmath\symbol{121}}}}
\def \RedDagger	{\mathrel{\raise-1\point\hbox{
 \largeoldmath \symbol{121}}}}
\def \DaggerA{{\Yellow \Dagger\kern-3\point_{\mbox{\scriptsize \sc a}}}}
\def \DaggerLog{\kern-1\point{\Yellow \Dagger\kern-1.5\point_{\mbox{\scriptsize \sc l}\kern-1.5\point}}}

\def \DaggerLeft{\mathop{\kern-1\point{\specialcol \raise2\point\hbox{\rotatebox{-30}{$\Dagger$}}}\kern-1\point}}
\def \DaggerRight{\mathop{\kern-1\point{\specialcol \raise0\point\hbox{\rotatebox{30}{$\Dagger$}}}\kern-1\point}}

\def \DaggerV{\mathop{\,\Yellow \Dagger_{\mbox{\scriptsize \sc v}}}}
\def \DaggerN{\mathop{\,\Yellow \Dagger_{\mbox{\scriptsize \sc n}}}}
\def \DaggerL{\DaggerLeft}
\def \DaggerR{\DaggerRight}

\def \expT#1#2#3#4{\XOut{\XIn{#1}{#2}}{#3}\kern.2mm{`.}\kern.2mm#4}
\def \medT#1#2#3#4#5{\setbox111=\hbox{$#1$}\setbox112=\hbox{$#5$}%
\ifdim\wd111<10pt%
	\ifdim\wd112<10pt%
		\XOut{\box111}{#2}\,[#3]\,\XIn{#4}{\box112}%
	\else
		\XOut{\box111}{#2}~[#3]~\XIn{#4}{\box112}
	\fi
\else
	\XOut{\box111}{#2}~[#3]~\XIn{#4}{\box112}
\fi}
\def \medTdl#1#2#3#4#5{{\termcolour %
	\XOut{(#1)}{#2}~[#3] }\quad \\ \hfill {\termcolour \XIn{#4}{(#5)} }
}

\def \cutT#1#2#3#4{\XOut{#1}{#2} \Dagger \XIn{#3}{#4}}
\def \ColcutT#1#2#3#4{\XOut{#1}{#2} \RedDagger \XIn{#3}{#4}}
\def \cutlog#1#2#3#4{\XOut{#1}{\orange #2}{\orange \Dagger \XIn{#3}{#4}}}
\def \cutA#1#2#3#4{\XOut{#1}{#2}\DaggerA \XIn{#3}{#4}}
\def \invLT#1#2#3{{#1}\kern.2mm{`.}\kern.2mm\XOut{#2}{#3}}
\def \invRT#1#2#3{\XIn{#1}{#2}\kern.2mm{`.}\kern.2mm{#3}}
\def \negL #1 . #2 #3 {\NegL{#1}{#2}{#3}}
\def \negR #1 #2 . #3 {\NegR{#1}{#2}{#3}}
\def \negLT#1#2#3{{#1}\kern.2mm{`.}\kern.2mm\XOut{#2}{#3}}
\def \negRT#1#2#3{\XIn{#1}{#2}\kern.2mm{`.}\kern.2mm{#3}}
\def \pairT#1#2#3#4#5{\Group<{#1}\widehat{#2},{#3}\widehat{#4}>\kern.2mm{`.}\kern.2mm{#5}}
\def \orLT#1#2#3#4#5{{#1}\kern.2mm{`.}\kern.2mm (\widehat{#2}{#3}\kern.2mm{+}\kern.2mm\widehat{#4}{#5}) }
\def \inLT#1#2#3#4{{#1}(\widehat{#2}\kern.2mm{+}\kern.2mm{#3})\kern.2mm{`.}\kern.2mm{#4}}
\def \inRT#1#2#3#4{{#1}\,({#2}\kern.2mm{+}\kern.2mm\widehat{#3})\kern.2mm{`.}\kern.2mm{#4}}
\def \orRPT#1#2#3#4{{#1}\,(\widehat{#2}\kern.3mm{\mid}\kern.3mm\widehat{#3})\kern.2mm{`.}\kern.2mm{#4}}

\def \newconRT#1#2#3#4{\Group<{#1},{#2}>{#3}\kern.2mm{`.}\kern.2mm{#4}}
\def \newconLT#1#2#3#4#5{#1\kern.2mm{`.}\kern.2mm\Group<{#2}\widehat{#3},{#4}\widehat{#5}>}

\def \CBV{\Sc{cbv}}
\def \CBN{\Sc{cbn}}

\def \Weak{\mysl{W}}

\omtrue
\def \InvLterm#1#2#3{{\termcolour \if@om \omfalse\invLT{#1}{#2}{#3}\omtrue
	\else	\skp(\invLT{#1}{#2}{#3})\fi}}
\def \InvRterm#1#2#3{{\termcolour \if@om \omfalse\invRT{#1}{#2}{#3}\omtrue
	\else	\skp(\invRT{#1}{#2}{#3})\fi}}
\def \NegL#1#2#3{{\termcolour \if@om \omfalse\negLT{#1}{#2}{#3}\omtrue
	\else	\skp(\negLT{#1}{#2}{#3})\fi}}
\def \NegR#1#2#3{{\termcolour \if@om \omfalse\negRT{#1}{#2}{#3}\omtrue
	\else	\skp(\negRT{#1}{#2}{#3})\fi}}
\def \NewConR#1#2#3#4{{\termcolour \if@om \omfalse\newconRT{#1}{#2}{#3}{#4}\omtrue	\else	(\newconRT{#1}{#2}{#3}{#4})\fi}}
\def \NewConL#1#2#3#4#5{{\termcolour \if@om \omfalse\newconLT{#1}{#2}{#3}{#4}{#5}\omtrue	\else	(\newconLT{#1}{#2}{#3}{#4}{#5})\fi}}

\def \nonconnectable{\mbox{\tiny $\Box$}}
\def \nonconnectable{\diamond}
\def \newconR < {\@ifnextchar{o}{\newconRR < }{\newconRL < }}
\def \newconRL < #1 , #2 > #3 . #4 {\NewConR{\widehat{#1}}{\nonconnectable}{#3}{#4}}
\def \newconRR < #1 , #2 > #3 . #4 {\NewConR{\nonconnectable}{\widehat{#2}}{#3}{#4}}

\def \newconL #1 . < #2 #3 , #4 #5 > {\NewConL{#1}{#2}{#3}{#4}{#5}}

\def \longharpoon#1{\psset{unit=.7\point,linewidth=0.4\point,linecolor=ltermcol}%
\harpoonlength#1\addtolength{\harpoonlength}{-1\point}%
\psline{cc-cc}(-1,1.25)(-\harpoonlength,1.25)%
\pscurve{cc-cc}(-0.7,1.25)(-1.5,1.44)(-2,1.8)(-2.5,2.3)(-3,3.1)
\vbox to 2\point{}%
}

\def \Vec#1{\setbox155=\hbox{$#1$}%
\leavevmode\copy155\raise\ht155\hbox{\longharpoon{\wd155}}}

\def \longarrow#1{\psset{unit=1\point,linewidth=0.35\point}%
\psline{cc-cc}(0,0)(#1,0)%
\hspace*{#1}%
\pscurve{cc-cc}(0,0)(-1.1,.4)(-2.5,2.2)%
\pscurve{cc-cc}(0,0)(-1.1,-.4)(-2.5,-2.2)}

\def \Rel#1#2{\setbox155=\hbox{\scriptsize $#1$}\setbox156=\hbox{$#2$}%
\mathrel{\raise5.5pt\copy155\kern-1\wd155%
\kern-.5\wd155\kern-.5\wd156\raise-2pt\copy156\kern-.5\wd156\kern.5\wd155}}

\def \PiOverline#1{%
 \ifmycolour\psset{unit=1\point,linewidth=0.4\point,linecolor=picolour}%
 \else\psset{unit=1\point,linewidth=0.35\point,linecolor=black}
 \fi
 \setbox155=\hbox{$#1$}\leavevmode
 \raise\ht155\hbox{\psline{c-c}(.05\wd155,1.25)(.95\wd155,1.25)}\box155%
}
\def \Underline#1{%
 \ifmycolour\psset{unit=1\point,linewidth=0.4\point,linecolor=yellow}%
 \else\psset{unit=1\point,linewidth=0.35\point,linecolor=black}
 \fi
 \setbox156=\hbox{${#1}$}\leavevmode
 \raise-\dp156\hbox{\psline{c-c}(.05\wd156,-0.75)(.95\wd156,-0.75)}\box156 }

\def \ul#1{\Underline{\mbox{\scriptsize $#1$}}}

\def \CH#1{\kern1\point\lefttop{\mbox{\Purple $#1$}}\righttop\kern1\point}
\def \LtoLMMT#1{\kern1\point\lefttop{#1}\righttop\kern1\point}
\def \CH#1{\raise1.5\point\hbox{\footnotesize ${\mid}\kern-1.5\point[$}%
#1
\raise1.5\point\hbox{\footnotesize $]\kern-1.5\point{\mid}_{`l}$}}

\def \CutGraph{\@ifnextchar[{\CutGraphN}{\CutGraphN[\GCut]}}

\def \CutGraphN[#1]#2#3#4#5{%
\begin{psmatrix}[rowsep=1ex,colsep=.05ex]
&&&[name=cutroot]&\\[3ex]
&&&[name=cut]{#1}&\\
&[name=P]#2&&&&[name=Q]#5& \\
&&[name=a]#3&&[name=x]#4\\&
\end{psmatrix}
 \ncline{cutroot}{cut}
\ncline{cut}{P}
\ncline{cut}{a}
\ncline{cut}{x}
\ncline{cut}{Q}
}

\def \V#1{\setbox71=\hbox{#1}\setbox72=\hbox{\v{}}%
\copy71\kern-.5\wd71\kern-.3\wd72\raise\ht71\hbox{\raise.2em\hbox{\lower\ht72\hbox{\copy72}}}\kern-.7\wd72\kern.5\wd71}

\setbox138=\hbox{$=$}
\setbox139=\hbox{{\scriptsize $\Delta$}}%
\setbox149=\hbox{\raise-1.5\point\copy138\kern-.5\wd138\kern-.5\wd139\raise4\point\hbox{\copy139}\kern-.5\wd139\kern.5\wd138}
\def \ByDef {\mathrel{\copy149}}

\def \Arr{{\Rightarrow}}

\def \Cut{\@ifnextchar{\bgroup}{\CutTwo}{{\mysl{\Red cut}}}}

\def \Exp	{\@ifnextchar{\bgroup}{\ExpT}{{\mysl{exp}}}}
\def \Imp	{\@ifnextchar{\bgroup}{\MedT}{{\mysl{imp}}}}
\def \Med	{\@ifnextchar{\bgroup}{\MedT}{{\mysl{med}}}}
\def \Cap	{\@ifnextchar{\bgroup}{\CapT}{{\mysl{cap}}}}

\def \Ins {\mysl{exp-imp}}

\def \actR	{{\Yellow \DaggerR}\mysl{a}}

\def \deactR	{{\Yellow \DaggerR}\mysl{d}}

\def \renR	{{\termcolour {\DaggerR}\mysl{ren}}}

\def \gcR {{\DaggerR}\mysl{gc}}

\def \Ri {{\DaggerR}\mysl{cap}}
\def \Rii {{\DaggerR}\mysl{exp}}
\def \Riii {{\DaggerR}\mysl{imp-outs}}
\def \Riv {{\DaggerR}\mysl{imp-ins}}
\def \Rv {{\DaggerR}\mysl{cut}}

\def \actL	{\mysl{a}{\Yellow \DaggerL}}

\def \deactL	{\mysl{d}{\Yellow \DaggerL}}

\def \renL	{{\termcolour \mysl{ren}{\DaggerL}}}

\def \gcL {\mysl{gc}{\DaggerL}}

\def \Li {\mysl{cap}{\DaggerL}}
\def \Lii {\mysl{exp-outs}{\DaggerL}}
\def \Liii {\mysl{exp-ins}{\DaggerL}}
\def \Liv {\mysl{imp}{\DaggerL}}
\def \Lv {\mysl{cut}{\DaggerL}}

\def \CutL#1#2{\@ifnextchar{\bgroup}{\CutLeft{#1}{#2}}{\widehat{#1} \DaggerL \widehat{#2}}}
\def \CutR#1#2{\@ifnextchar{\bgroup}{\CutRight{#1}{#2}}{\widehat{#1} \DaggerR \widehat{#2}}}
\def \CutTwo#1#2{\@ifnextchar{\bgroup}{\CutT{#1}{#2}}{\widehat{#1} \Dagger \widehat{#2}}}
\def \ColCut#1#2{\@ifnextchar{\bgroup}{\ColCutT{#1}{#2}}{\widehat{#1} \Dagger \widehat{#2}}}
\def \Cutdl#1#2#3#4{\XOut{#1}{#2} \Dagger\\ \tab \XIn{#3}{#4}}
\def \CutDl#1#2#3#4{\XOut{#1}{#2} \Dagger\\ \tab\tab \XIn{#3}{#4}}

\def \notin {\,{\not\in}\,}

\def \TrTwo#1#2{\Sem{\ltermcolour #1}{\termcolour #2}^{\sclmmt}}
\def \Tran#1{\@ifnextchar{\bgroup}{\TrTwo{#1}}{\TrTwo{#1}{}}}

\def \CapRule#1#2#3#4#5{ 
 (\Cap): &
 \Inf	{ \PDer{ \Cap{#1}{#2} }{#4,\stat{#1}{#3} }{\stat{#2}{#3}, {#5} } }
}

\def \DagRule#1#2#3#4#5#6#7{ 
 (\Cut):&
 \Inf	{ \PDer{#1}{#6}{ \stat{#2}{#5}, {#7} } \quad \PDer{#4}{{#6}, \stat{#3}{#5} }{#7} }
	{ \PDer{ \Cut{#1}{#2}{#3}{#4} }{#6}{#7} }
}

\def \MedRule#1#2#3#4#5#6#7#8#9{%
 (\Imp):&
 \Inf	{ \PDer{#1}{#8}{\stat{#2}{#6}, #9} \quad \PDer{#5}{ #8,\stat{#4}{#7} }{#9} }
	{ \PDer{ \Med{#1}{#2}{#3}{#4}{#5} }{#8, \stat{#3}{{#6}\arr{#7}} }{#9}}
}

\def \ImpRule#1#2#3#4#5#6#7#8#9{%
 (\Imp):&
 \Inf	{ \PDer{#1}{#8}{\stat{#2}{#6}, #9} \quad \PDer{#5}{ #8,\stat{#4}{#7} }{#9} }
	{ \PDer{ \Med{#1}{#2}{#3}{#4}{#5} }{#8, \stat{#3}{{#6}\arr{#7}} }{#9}}
}

\def \ExpRule#1#2#3#4#5#6#7#8{ 
 (\Exp):&
 \Inf	{ \PDer{#2}{ {#7},\stat{#1}{#5} }{\stat{#3}{#6}, #8} }
	{ \PDer{ \Exp{#1}{#2}{#3}{#4} }{#7}{ \stat{#4}{#5\arr{#6}}, #8} }
}

\newdimen \GeneWidth \newdimen \GeneW \newdimen \HGeneW
\newdimen \GeneHeight \newdimen \GeneH \newdimen \HGeneH
\newcount \PictW \newcount \HPictW
\newcount \PictH \newcount \HPictH
\newcount \BoxH \newcount \HBoxH
\newcount \BoxW \newcount \HBoxW
\newcount \HOffset \newcount \VOffset

\def \StoupFrameR#1{\setbox159=\hbox{$#1$}%
\PictW\wd159\divide\PictW\unitlength\advance\PictW3%
\PictH\ht159`mltiply\PictH18\divide\PictH10\divide\PictH\unitlength
\,\raise-4\unitlength\hbox{\PSFrame(\PictW,\PictH)}%
\kern1\point\box159\,,}

\def \StoupFrameL#1{\setbox159=\hbox{$#1$}%
\PictW\wd159\divide\PictW\unitlength\advance\PictW3%
\PictH\ht159`mltiply\PictH18\divide\PictH10\divide\PictH\unitlength
,\,\raise-4\unitlength\hbox{\PSFrame(\PictW,\PictH)}%
\kern1\point\box159\,}

\def \Gene#1{\unitlength.09\point\setbox140=\hbox{$#1$}%
\GeneHeight\ht140\advance\GeneHeight55\unitlength%
\GeneWidth\wd140\advance\GeneWidth50\unitlength%
\PictW\GeneWidth\divide\PictW\unitlength%
\PictH\GeneHeight\divide\PictH\unitlength%
\ifdim\wd140>10\point
\begin{picture}(\PictW,\PictH)%
 \put(0,0)	{\makebox(\PictW,\PictH){\box140}}%
 \put(0,0)	{\PSFrame(\PictW,\PictH)}%
\end{picture}
\else \box140
\fi}

\def \InvGene#1{\unitlength.09\point%
\setbox140=\hbox{$#1$}%
\GeneHeight\ht140\advance\GeneHeight60\unitlength
\GeneWidth\wd140\advance\GeneWidth50\unitlength
\PictW\GeneWidth\divide\PictW\unitlength
\PictH\GeneHeight\divide\PictH\unitlength
\begin{picture}(\PictW,\PictH)
 \put(0,0)	{\makebox(\PictW,\PictH){\box140}}
 \put(0,0)	{\PSFrame(0,0)}
\end{picture}}

\newcommand{\DiagCap}[2]{\unitlength.09\point%
\begin{picture}(360,120)(-180,-55)
	\put(-180,0)	{\Yellow \vector(1,0){140}}
	\put(-110,50)	{\makebox(0,0){\small $\termcolour #1$}}
	\put(-40,-30)	{\PSBox(80,80)}
	\put(110,50)	{\makebox(0,0){\small $\termcolour #2$}}
	\put(40,0)	{\Yellow	\vector(1,0){140}}
 \end{picture}}
\def \Diagcaps<#1,#2>{\DiagCap{#1}{#2}}

\newcommand{\DiagExp}[4]{\unitlength.09\point%
\setbox143=\hbox{{\termcolour \small $#2$}}%
\GeneWidth\wd143
	\advance\GeneWidth320\unitlength
	\BoxW\GeneWidth \divide\BoxW\unitlength \HBoxW\BoxW \divide\HBoxW2%
	\advance\GeneWidth200\unitlength \PictW\GeneWidth \divide\PictW\unitlength \HPictW\PictW \divide\HPictW2%
\GeneHeight\ht143 \advance\GeneHeight 80\unitlength
	\BoxH\GeneHeight \divide\BoxH\unitlength \HBoxH\BoxH \divide\HBoxH2%
	\advance\GeneHeight70\unitlength
\PictH\GeneHeight \divide\PictH\unitlength \HPictH\PictH \divide\HPictH2%
\HOffset\wd143 \divide\HOffset\unitlength \divide\HOffset2%
\VOffset\HPictW \divide\VOffset\unitlength \divide\VOffset2 \advance\VOffset20%
 \def \ContentExp {\begin{picture}(\PictW,\PictH)(-\HPictW,-\HPictH)
 \put(-50,0){%
 \VOffset\ht143\divide\VOffset\unitlength\divide\VOffset2\advance\VOffset50%
 \makebox(0,0){\Gene{\begin{picture}(\BoxW,\BoxH)(-\HBoxW,20)
		\put(-\HOffset,\VOffset)	{
		\put(-160,-25) {\Yellow \vector(1,0){150}}
		\put(-85,22)	{\makebox(0,0){\termcolour \small $\widehat{#1}$}}
	}%
	\put(0,\VOffset)	{\makebox(0,0){\box143}}
		\put(\HOffset,\VOffset) {
		\put(10,-25)	{\Yellow \vector(1,0){150}}
		\put(75,20)	{\makebox(0,0){\termcolour \small $\widehat{#3}$}}
	}
 \end{picture}}}}
 \HOffset\HBoxW\advance\HOffset-40
 \put(\HOffset,-15)	{\Yellow \vector(1,0){130}}
 \advance\HOffset65
 \put(\HOffset,35) {\makebox(0,0){\termcolour \small $#4$}}
 \end{picture}}
\if@om \omfalse 
\ContentExp 
\omtrue \else \Gene{\ContentExp} \fi
}

\newdimen \MedW \newdimen \MedH
\newcommand{\DiagMed}[5]{\unitlength.09\point%
\setbox144=\hbox{{\termcolour \small $#1$}}%
\setbox145=\hbox{{\termcolour \small $#5$}}%
\setbox146=\hbox{\termcolour \small $#3$}%
\MedW \wd144 \advance \MedW \wd145 
\advance \MedW 480\unitlength
\ifdim \ht144 > \ht145 \MedH \ht144 \else \MedH \ht145 \fi \advance \MedH 70\unitlength
\BoxW\MedW\divide\BoxW\unitlength\HBoxW\BoxW\divide\HBoxW2%
\BoxH\MedH\divide\BoxH\unitlength\HBoxH\BoxH\divide\HBoxH2%
\GeneWidth \MedW \advance \GeneWidth 240\unitlength
\GeneHeight \MedH \advance \GeneHeight 70\unitlength
\PictW \GeneWidth \divide \PictW \unitlength \HPictW \PictW \divide \HPictW 2%
\PictH \GeneHeight \divide \PictH \unitlength \HPictH \PictH \divide \HPictH 2%
\def \ContentMed{\begin{picture}(\PictW,\PictH)(-\HPictW,-\HPictH)
 \put(75,0){%
	\makebox(0,0){%
	\VOffset \HPictH \advance \VOffset -3%
	\HOffset \MedW 
		\advance \HOffset \wd144
		\advance \HOffset -\wd145 
	\divide \HOffset \unitlength \divide \HOffset 2
 \VOffset \HBoxH \advance \VOffset 0
 \Gene{\begin{picture}(\BoxW,\BoxH)(-\HOffset,-\VOffset)
	\put(-220,0) {\makebox(0,0)[r]{\box144}}
	\put(-180,-20) {\Yellow \vector(1,0){120}}
	\put(-130,30)	{\makebox(0,0){\termcolour \small $\widehat{#2}$}}
	\put(0,0)	{\makebox(0,0){${\Yellow [~]}$}}
	\put(120,30)	{\makebox(0,0){\termcolour \small $\widehat{#4}$}}
	\put(60,-20)	{\Yellow \vector(1,0){120}}
	\put(220,0)	{\makebox(0,0)[l]{\box145}}
	\end{picture}}}}
 \HOffset\HBoxW\advance\HOffset120
 \put(-\HOffset,-15) {\Yellow \vector(1,0){140}}
 \advance\HOffset-70
 \put(-\HOffset,35) {\makebox(0,0){\termcolour $#3$}}
 \end{picture}}
\ContentMed 
}
\def \Diagmed #1 #2 [#3] #4 #5 {\DiagMed{#1}{#2}{#3}{#4}{#5}}
 
\newcommand{\DiagDMed}[5]{\unitlength.09\point%
\setbox144=\hbox{{\termcolour $#1$}}\setbox145=\hbox{{\termcolour $#5$}}\setbox146=\hbox{\termcolour $#3$}%
\ifdim \wd144 > \wd145 \MedW \wd144 \else \MedW \wd145 \fi 
\advance \MedW 300\unitlength
\MedH \ht144 \advance \MedH \ht145 \advance \MedH 100\unitlength
\BoxW\MedW\divide\BoxW\unitlength\HBoxW\BoxW\divide\HBoxW2%
\BoxH\MedH\divide\BoxH\unitlength\HBoxH\BoxH\divide\HBoxH2%
\GeneWidth \MedW \advance \GeneWidth 240\unitlength
\GeneHeight \MedH \advance \GeneHeight 70\unitlength
\PictW \GeneWidth \divide \PictW \unitlength \HPictW \PictW \divide \HPictW 2%
\PictH \GeneHeight \divide \PictH \unitlength \HPictH \PictH \divide \HPictH 2%
\def \ContentDMed{\begin{picture}(\PictW,\PictH)(-\HPictW,-\HPictH)%
 \put(80,0){%
	\makebox(0,0){%
	\VOffset \HPictH \divide \VOffset 2%
	\ifdim \wd144 > \wd145 \HOffset \wd144
	\else \HOffset \wd145 \fi
	\divide \HOffset \unitlength \advance \HOffset 50%
 \Gene{\begin{picture}(\BoxW,\BoxH)(-\HOffset,-\HBoxH)
	\put(0,\VOffset){%
	\put(0,-20) {\makebox(0,0)[r]{\box144}}
	\put(40,-30) {\Yellow \vector(1,0){120}}
	\put(95,18)	{\makebox(0,0){\termcolour $\widehat{#2}$}}
	}
	\put(-\HOffset,-\VOffset){%
	\put(80,0)	{\makebox(0,0){${\Yellow [~]}$}}
	\put(175,18)	{\makebox(0,0){\termcolour $\widehat{#4}$}}
	\put(125,-30)	{\Yellow \vector(1,0){120}}
	\put(250,0)	{\makebox(0,0)[l]{\box145}}
	}
	\end{picture}}}}%
 \HOffset\HBoxW\advance\HOffset120%
 \put(-\HOffset,-20) {\Yellow \vector(1,0){140}}%
 \advance\HOffset-70%
 \put(-\HOffset,30) {\makebox(0,0){\termcolour $#3$}}%
 \end{picture}}%
 \ContentDMed 
}

\newcommand{\DiagCutA}[5]{\unitlength.09\point%
\setbox141=\hbox{\small $#1$}%
\setbox142=\hbox{\small $#4$}%
\GeneWidth \wd141 \advance \GeneWidth \wd142 \advance \GeneWidth 400\unitlength
\ifdim \ht141 > \ht142 \GeneHeight \ht141 \else \GeneHeight \ht142 \fi \advance \GeneHeight 30\unitlength
\PictW \GeneWidth \divide \PictW \unitlength \HPictW \PictW \divide \HPictW 2%
\PictH \GeneHeight \divide \PictH \unitlength \HPictH \PictH \divide \HPictH 2 \advance \PictH 20%
\HOffset \GeneWidth
\ifdim \wd142 > \wd141 \advance \HOffset -\wd142 \advance \HOffset \wd141
\else \advance \HOffset \wd141 \advance \HOffset -\wd142 \fi
\divide \HOffset \unitlength \divide \HOffset 2
\VOffset \HPictH \advance \VOffset 0
\def \ContentCut{ \begin{picture}(\PictW,\PictH)(-\HOffset,-\VOffset)
	\put(-178,0)	{\makebox(0,0)[r]{\termcolour \box141}}
	\put(-140,-20)	{\Yellow \vector(1,0){160}}
	\put(-70,32)	{\makebox(0,0){\termcolour $\widehat{#2}$}}
	\put(0,-100)	{\makebox(0,0){\mbox{\Yellow \footnotesize \sc #5}}}
	\put(70,32)	{\makebox(0,0){\termcolour $\widehat{#3}$}}
	\put(0,-20)	{\Yellow \line(1,0){140}}
	\put(180,0)	{\makebox(0,0)[l]{\termcolour \box142}}
 \end{picture}}
\Gene{ \ContentCut }
}

\newcount\tempcount
\newcommand{\XNat}[1]{\tempcount#1\raise-4\point\hbox{\DiagCap{x}{`a}}\loop{\ifnum\tempcount>0%
\,[f]\,\raise-4\point\hbox{\DiagCap{x}{`a}}%
\advance\tempcount by-1 }\repeat}

\def \without{\mathord{\setminus}}
\def \typeofargs#1#2{{\It {typeof}}\ #1\ #2}
\def \typeof{\@ifnextchar\bgroup{\typeofargs}{\It {typeof}}}

\def \ftsc#1{\textsc{\scriptsize #1}}
\def \tinysc#1{\textsc{\tiny #1}}
\def \ftrm#1{\textrm{\scriptsize #1}}

\def \unifyContextsargs#1#2{\setbox31=\hbox{$#1$}\setbox32=\hbox{$#2$}
\ifdim\wd31<12\point\ifdim\wd32<12\point {\It {unifyCont}}\ #1\ #2
\else {\It {unifyCont}}\ #1\ (#2) \fi
\else \ifdim\wd32<12\point {\It {unifyCont}\ (#1)\ #2}
\else {\It {unifyCont}}\ (#1)\ (#2) \fi \fi }

\def \unifyContexts{\@ifnextchar\bgroup{\unifyContextsargs}{\It {unifyCont}}}

\def \by {\mathbin{\mapsto}}

\def \GCut {\textsf{ Cut}}

\def\Unr#1{\textsl{Unrv}~#1}
\def \Unravel{\@ifnextchar\bgroup{\Unr}{\textsl{Unrv}}}

\def \succ {\mysl{succ}}

\def \LC	{\mbox{$\Red `l$-calculus}}

\def \DBs {\textsf{\small DB}}
\def \DBf#1#2#3{\DBs(#1,#2,#3)}
\def \DB {\@ifnextchar{\bgroup}{\DBf}{\DBs}}

\def \Pred[#1]{~[\skp#1\skp]}

\omtrue
\def \expT#1#2#3#4{{\typecolour \XOut{\XIn{#1}{#2}}{#3}\cdot #4}}
\def \medT#1#2#3#4#5{{\typecolour \XOut{#1}{#2}\,[#3]\,\XIn{#4}{#5}}}

\def \expT#1#2#3#4{\XOut{\XIn{#1}{#2}}{#3}\kern.2mm{`.}\kern.2mm#4}

\def \medT#1#2#3#4#5{\setbox111=\hbox{$#1$}\setbox112=\hbox{$#5$}%
\ifdim\wd111<10pt%
	\ifdim\wd112<10pt%
	\XOut{\box111}{#2}\,[#3]\,\XIn{#4}{\box112}%
	\else
	\XOut{\box111}{#2}~[#3]~\XIn{#4}{\box112}
	\fi
\else
	\XOut{\box111}{#2}~[#3]~\XIn{#4}{\box112}
\fi}

\def \cutV#1#2#3#4{\XOut{#1}{#2} \DaggerV \XIn{#3}{#4}}
\def \cutN#1#2#3#4{\XOut{#1}{#2} \DaggerN \XIn{#3}{#4}}
\def \cutleft#1#2#3#4{\XOut{#1}{#2} \DaggerL \XIn{#3}{#4}}
\def \cutright#1#2#3#4{\XOut{#1}{#2} \DaggerR \XIn{#3}{#4}}
\def \cutA#1#2#3#4{\XOut{#1}{#2}\DaggerA \XIn{#3}{#4}}
\def \negLT#1#2#3{{#1}\kern.2mm{`.}\kern.2mm\XOut{#2}{#3}}
\def \negRT#1#2#3{\XIn{#1}{#2}\kern.2mm{`.}\kern.2mm{#3}}

\def \orLT#1#2#3#4#5{{#1}\kern.2mm{`.}\kern.2mm\Group<\widehat{#2}{#3}\kern.3mm{\mid}\kern.3mm\widehat{#4}{#5}>}
\def \orRLT#1#2#3#4{{#1}\Group<\widehat{#2}\kern.2mm{\mid}\kern.2mm{#3}>\kern.2mm{`.}\kern.2mm{#4}}
\def \orRRT#1#2#3#4{{#1}\Group<{#2}\kern.2mm{\mid}\kern.2mm\widehat{#3}>\kern.2mm{`.}\kern.2mm{#4}}
\def \orRPT#1#2#3#4{{#1}\Group<\widehat{#2}\kern.2mm{\mid}\kern.2mm\widehat{#3}>\kern.2mm{`.}\kern.2mm{#4}}

\def \Caps(#1,#2){{\termcolour \slangle#1{`.}#2\srangle}}
\def \CapT#1#2{\Caps(#1,#2)}
\def \ExpT#1#2#3#4{{\termcolour \if@om \omfalse\expT{#1}{#2}{#3}{#4}\omtrue
	\else (\expT{#1}{#2}{#3}{#4}) \fi}}
\def \Exps#1#2#3.#4{{\termcolour \if@om \omfalse\expT{#1}{#2}{#3}{#4}\omtrue
	\else (\expT{#1}{#2}{#3}{#4})\fi}}
\def \MedT#1#2#3#4#5{{\termcolour \if@om \omfalse\medT{#1}{#2}{#3}{#4}{#5}\omtrue
	\else (\medT{#1}{#2}{#3}{#4}{#5})\fi}}
\def \Meddl#1#2#3#4#5{\if@om \omfalse\medTdl{#1}{#2}{#3}{#4}{#5}\omtrue
	\else (\medTdl{#1}{#2}{#3}{#4}{#5})\fi}
\def \Meds#1#2[#3]#4#5{{\termcolour \if@om \omfalse\medT{#1}{#2}{#3}{#4}{#5}\omtrue \else (\medT{#1}{#2}{#3}{#4}{#5})\fi}}
\def \CutT#1#2#3#4{{\termcolour \if@om \omfalse\cutT{#1}{#2}{#3}{#4}\omtrue
	\else	(\cutT{#1}{#2}{#3}{#4})\fi}}
\def \ColCutT#1#2#3#4{{\termcolour \if@om \omfalse\ColcutT{#1}{#2}{#3}{#4}\omtrue
	\else	(\ColcutT{#1}{#2}{#3}{#4})\fi}}
\def \CutLog#1#2#3#4{{\termcolour \if@om \omfalse\cutlog{#1}{#2}{#3}{#4}\omtrue
	\else	(\cutlog{#1}{#2}{#3}{#4})\fi}}
\def \CutLeft#1#2#3#4{{\termcolour \if@om \omfalse\cutleft{#1}{#2}{#3}{#4}\omtrue \else	(\cutleft{#1}{#2}{#3}{#4})\fi}}
\def \CutRight#1#2#3#4{{\termcolour \if@om \omfalse\cutright{#1}{#2}{#3}{#4}\omtrue \else	(\cutright{#1}{#2}{#3}{#4})\fi}}
\def \CutV#1#2#3#4{{\termcolour \if@om \omfalse\cutV{#1}{#2}{#3}{#4}\omtrue
	\else	(\cutV{#1}{#2}{#3}{#4} )\fi}}
\def \CutN#1#2#3#4{{\termcolour \if@om \omfalse\cutN{#1}{#2}{#3}{#4}\omtrue
	\else	(\cutN{#1}{#2}{#3}{#4} )\fi}}
\def \CutA#1#2#3#4{{\termcolour \if@om \omfalse\cutA{#1}{#2}{#3}{#4}\omtrue
	\else	(\cutA{#1}{#2}{#3}{#4})\fi}}

\def \NegL#1#2#3{{\termcolour \if@om \omfalse\negLT{#1}{#2}{#3}\omtrue
	\else	(\negLT{#1}{#2}{#3})\fi}}
\def \NegR#1#2#3{{\termcolour \if@om \omfalse\negRT{#1}{#2}{#3}\omtrue
	\else	(\negRT{#1}{#2}{#3})\fi}}
\def \AndL#1#2#3#4{{\termcolour \if@om \omfalse \andLPT{#1}{#2}{#3}{#4} \omtrue \else \skp (\andLPT{#1}{#2}{#3}{#4}) \skp \fi}}
\def \AndLL#1#2#3#4{{\termcolour \if@om \omfalse \andLLT{#1}{#2}{#3}{#4} \omtrue \else \skp (\andLLT{#1}{#2}{#3}{#4}) \skp\fi}}
\def \AndLR#1#2#3#4{{\termcolour \if@om \omfalse\andLRT{#1}{#2}{#3}{#4}\omtrue	\else	(\andLRT{#1}{#2}{#3}{#4})\fi}}
\def \AndLP#1#2#3#4{{\termcolour \if@om \omfalse\andLPT{#1}{#2}{#3}{#4}\omtrue	\else	(\andLPT{#1}{#2}{#3}{#4})\fi}}
\def \AndR#1#2#3#4#5{{\termcolour \if@om \omfalse\andRT{#1}{#2}{#3}{#4}{#5}\omtrue	\else	(\andRT{#1}{#2}{#3}{#4}{#5})\fi}}
\def \OrL#1#2#3#4#5{{\termcolour \if@om \omfalse\orLT{#1}{#2}{#3}{#4}{#5}\omtrue	\else	(\orLT{#1}{#2}{#3}{#4}{#5})\fi}}
\def \OrRL#1#2#3#4{{\termcolour \if@om \omfalse\orRLT{#1}{#2}{#3}{#4}\omtrue	\else	(\orRLT{#1}{#2}{#3}{#4})\fi}}
\def \OrRR#1#2#3#4{{\termcolour \if@om \omfalse\orRRT{#1}{#2}{#3}{#4}\omtrue	\else	(\orRRT{#1}{#2}{#3}{#4})\fi}}
\def \OrRP#1#2#3#4{{\termcolour \if@om \omfalse\orRPT{#1}{#2}{#3}{#4}\omtrue	\else	(\orRPT{#1}{#2}{#3}{#4})\fi}}

\def \NewConR#1#2#3#4{{\termcolour \if@om \omfalse\newconRT{#1}{#2}{#3}{#4}\omtrue	\else	(\newconRT{#1}{#2}{#3}{#4})\fi}}
\def \NewConL#1#2#3#4#5{{\termcolour \if@om \omfalse\newconLT{#1}{#2}{#3}{#4}{#5}\omtrue	\else	(\newconLT{#1}{#2}{#3}{#4}{#5})\fi}}

\def \caps<#1,#2>{\Cap{#1}{#2}}
\def \med #1 #2 [#3] #4 #5 {\Med{#1}{#2}{#3}{#4}{#5}}
\def \imp #1 #2 [#3] #4 #5 {\Med{#1}{#2}{#3}{#4}{#5}}
\def \impdl #1 #2 [#3] #4 #5 {\Meddl{#1}{#2}{#3}{#4}{#5}}
\def \exp #1 #2 #3 . #4 {\Exp{#1}{#2}{#3}{#4}}
\def \expl #1 #2 #3 #4 . #5 {{#1}:\Exp{#2}{#3}{#4}{#5}}
\def \cut #1 #2 + #3 #4 {\Cut{#1}{#2}{#3}{#4}}
\def \cutdl #1 #2 + #3 #4 {\Cutdl{#1}{#2}{#3}{#4}}
\def \cutDl #1 #2 + #3 #4 {\CutDl{#1}{#2}{#3}{#4}}
\def \cutLog #1 #2 + #3 #4 {\CutLog{#1}{#2}{#3}{#4}}
\def \cutL #1 #2 + #3 #4 {\CutL{#1}{#2}{#3}{#4}}
\def \cutR #1 #2 + #3 #4 {\CutR{#1}{#2}{#3}{#4}}
\def \colcut #1 #2 + #3 #4 {\ColCut{#1}{#2}{#3}{#4}}
\def \colcutR #1 #2 + #3 #4 {\ColCut{#1}{#2}{#3}{#4}}
\def \colcutL #1 #2 + #3 #4 {\ColCut{#1}{#2}{#3}{#4}}
\def \negL #1 . #2 #3 {\NegL{#1}{#2}{#3}}
\def \negR #1 #2 . #3 {\NegR{#1}{#2}{#3}}

\def \andLL #1 . < #2 , #3 > #4 {\AndLL{#1}{#2}{\nonconnectable}{#4}}
\def \andLP #1 . < #2 , #3 > #4 {\AndLP{#1}{#2}{#3}{#4}}
\def \andLR #1 . < #2 , #3 > #4 {\AndLR{#1}{\nonconnectable}{#3}{#4}}
\def \andLLaux #1 . < #2 , {\@ifnextchar{o}{\andLL #1 . < #2 , }{\andLP #1 . < #2 , }}
\def \andL #1 . < {\@ifnextchar{o}{\andLR {#1} . < }{\andLLaux {#1} . < }}
\def \andLL#1#2#3#4{{#1}\kern.2mm{`.}\kern.2mm\Group<\widehat{#2},{#3}>{#4}}
\def \andLR#1#2#3#4{{#1}\kern.2mm{`.}\kern.2mm\Group<{#2},\widehat{#3}>{#4}}
\def \andLLT#1#2#3#4{{#1}\kern.2mm{`.}\kern.2mm\Group<\widehat{#2},{#3}>{#4}}
\def \andLRT#1#2#3#4{{#1}\kern.2mm{`.}\kern.2mm\Group<{#2},\widehat{#3}>{#4}}
\def \andLPT#1#2#3#4{{#1}\kern.2mm{`.}\kern.2mm\Group<\widehat{#2},\widehat{#3}>{#4}}
\def \andRT#1#2#3#4#5{\Group<{#1}\widehat{#2},{#3}\widehat{#4}>\kern.2mm{`.}\kern.2mm{#5}}
\def \andRrule{\Conj\textit{R}}
\def \andRlog < #1 #2 , #3 #4 > . #5 {\AndR{#1}{#2}{#3}{#4}{#5}}
\def \andR{\@ifnextchar<{\andRlog}{\andRrule}}
\def \invL #1 . #2 #3 {\InvLterm{#1}{#2}{#3}}
\def \invR #1 #2 . #3 {\InvRterm{#1}{#2}{#3}}

\def \orR #1 < {\@ifnextchar{o}{\orRR {#1} < }{\orRLaux {#1} < }}
\def \orRLaux #1 < #2 | {\@ifnextchar{o}{\orRL #1 < #2 | }{\orRP #1 < #2 | }}
\def \orRL #1 < #2 | #3 > . #4 {\OrRL{#1}{#2}{\nonconnectable}{#4}}
\def \orRR #1 < #2 | #3 > . #4 {\OrRR{#1}{\nonconnectable}{#3}{#4}}
\def \orRP #1 < #2 | #3 > . #4 {\OrRP{#1}{#2}{#3}{#4}}
\def \projLT#1#2#3#4{{#1}\kern.2mm{`.}\kern.2mm\Group<\widehat{#2},{#3}>{#4}}
\def \projLR #1 . < #2 , #3 > #4 {\AndL{#1}{#2}{#3}{#4} }
\def \projL #1 . < #2 , #3 > #4 {\AndLL{#1}{#2}{\nonconnectable}{#4}}
\def \projLaux #1 . < #2 , {\@ifnextchar{o}{\projL {#1} . < {#2} , }{\projLR {#1} . < {#2} , }}
\def \projR #1 . < #2 , #3 > #4 {\AndLR{#1}{\nonconnectable}{#3}{#4}}
\def \proj #1 . < {\@ifnextchar{o}{\projR {#1} . < }{\projLaux {#1} . < }}
\def \pair < #1 #2 , #3 #4 > . #5 {\AndR{#1}{#2}{#3}{#4}{#5}}
\def \inP #1 < {\@ifnextchar{o}{\inR {#1} < }{\inLaux {#1} < }}
\def \inLaux #1 < #2 | {\@ifnextchar{o}{\inL #1 < #2 | }{\inRP #1 < #2 | }}
\def \inL #1 < #2 | #3 > . #4 {\OrRL{#1}{#2}{\nonconnectable}{#4}}
\def \inR #1 < #2 | #3 > . #4 {\OrRR{#1}{\nonconnectable}{#3}{#4}}
\def \inRP #1 < #2 | #3 > . #4 {\OrRP{#1}{#2}{#3}{#4}}
\def \choice #1 . < #2 #3 | #4 #5 > {\OrL{#1}{#2}{#3}{#4}{#5}}

\def \QED{%
 \unskip
 \hfill 
 \begingroup
 \unitlength\point
 \linethickness{.5\point}%
 \framebox(7,7){}%
 \endgroup
 \kern 10\point
}

 \newenvironment{AuxCases}[1]{\left \{ \def\arraystretch{1.1} \begin{array}{#1} }
 { \end{array} \right . }

\def \Natural{
{\sf I\kern-1\point N}}

\def \mybullet{\vspace*{2mm} ~ \kern-5mm $\bullet$}

\def \import{\textsl{import}}

\makeatother

\def \ellips(#1,#2){\put (0,0)	{\psellipse[unit=\unitlength, linewidth=1.25pt](#1,#2)}}

\newskip \myitemindent \myitemindent0mm

\def \myitem[#1]{\setbox55=\hbox{#1} \item[{#1}] \leavevmode $
\kern-\wd55 \kern5.25mm \kern-\myitemindent 
 \begin{array}[t]{@{}lclclclclcl} \kern-3mm
\kern\wd55 \kern\myitemindent}

\def\renewfont#1#2{\font#1=#2\relax}

\def \SemThree#1#2#3{%
	\setbox136=\hbox{{\scriptsize ${#2}$}}%
	\setbox137=\hbox{#3}%
	\semEl{#1}_{\kern-1\point \copy136}
	\ifdim\wd136>1pt\kern-\wd136\raise5\point\copy137
	\ifdim\wd136>\wd137 \kern-\wd137\kern\wd136\fi
	\fi}
\def \SemL#1#2{\SemThree{{#1}}{#2}{\mbox{\scriptsize $\lambda$}}}

\def\Red{}
\def \redX{\red}
\def\Stat#1#2{#1{:}#2}
\def\TurnX{\Turn}

\newcommand{\Long}[1]{} \newcommand{\Short}[1]{#1}
 \def\clearpage{}
 \def \correct{~ \\[-22\point]}

\begin{document}

\bibliographystyle {plain}

\title {Reduction in {\X} does not agree with {I}ntersection and {U}nion {T}ypes}
\Short{\subtitle {Extended abstract} }
\author{Steffen van Bakel} 
\address{Department of Computing, Imperial College London, \\ 180 Queen's Gate, London SW7 2BZ, UK} 
\email{\tt svb@doc.ic.ac.uk}
\date{}
\maketitle

 \begin{abstract}
This paper defines intersection and union type assignment for the calculus $\X$, a substitution free language that enjoys the Curry-Howard correspondence with respect to Gentzen's sequent calculus for classical logic.
We show that this notion is 
closed for subject-expansion, and show that it needs to be restricted to satisfy subject-reduction as well, making it unsuitable to define a semantics.
 \end{abstract}


\section*{Introduction}
This paper will present a notion of intersection and union type assignment for the (untyped) calculus $\X$, that was first defined in \cite{Lengrand'03} and later extensively studied in \cite{vBLL'05}.
$\X$ is based on the sequent calculus \cite{Gentzen'35}, in contrast to the {\LC} \cite{Barendregt'84} which is related to natural deduction (see also \cite{Gentzen'35}); in $\X$, duality is ubiquitous, as for example call-by-name is dual to call-by-value (see also \cite{WadlerDual}), and as intersection will be shown to be dual to union in this paper. 
The advantage of using the sequent approach here is that now we can explore the duality of intersection and union fully, through which we can study and explain various anomalies of union type assignment \cite{Pierce91b,Barbanera-Dezani-Liguoro-IaC'95} and quantification \cite{Harper-Lillibridge'91,Milner-et.al'97}.

The type system defined here initially will be shown to be the natural one, in that intersection and union play their expected roles for \emph{witness expansion} (also called \emph{completeness}).
However, we show that \emph{witness reduction} (also called \emph{soundness}, the converse of completeness) no longer holds, and will reason that this is caused by the fact that both intersection and union lack a logical foundation: the obtained system is not Curry-Howard \cite{Hindley'84}, {\em i.e.}$\!$ accompanying syntax for the intersection and union type constructors is missing.
This problem also appears in other contexts, such as that of {\ML} with side-effects \cite{Harper-Lillibridge'91,Wright'95,Milner-et.al'97}, and that of using intersection and union types in an operational setting \cite{Davies-Pfenning'01,Dunfield-Pfenning'00}.
As here, also there the cause of the problem is that the type-assignment rules are not fully logical, making the context calls (which form part of the reduction in $\X$) unsafe.
As, in part, already has been observed in \cite{Herbelin'05} in the context of the calculus $\lmmt$ of \cite{Curien-Herbelin'00}, the problem is that the added rules are not logical.

The advantage of studying this problem in the context of the highly symmetric sequent calculi will be made clear: intersection and union are truly dual for these calculi, and the at the time surprising loss of soundness for the system with intersection and union types in \cite{Pierce91b,Barbanera-Dezani-Liguoro-IaC'95} becomes now natural and inevitable.
Also, we will show that it is not union alone that causes problems, but that the problem is much more profound: although both intersection and union might be seem to be related to the (logical) \emph{and} and \emph{or}, the fact that they are both \emph{not} logical destroys the soundness, both for a system based on intersection as for a system based on union.
This also explains why, for {\ML} with side-effects, quantification is no longer sound \cite{Harper-Lillibridge'91,Milner-et.al'97}: also the $(\forall I)$ and $(\forall E)$ rules of {\ML} are not logical. 

Intersection and union have been studied in the context of classical sequents in \cite{Bakel-ITRS'04,DGL-ITRS'04,Herbelin'05,DGL-LPAR'05,DGL-CDR'08}, and all these systems suffer from the same kind of problem with respect to reduction.
In this paper we will improve on those results by \Long{showing all proofs in detail, and }presenting two \Long{more expressive }systems that address the problems successfully.

\Long{
~\\

\noindent
The origin of $\X$ lies within the quest for a language designed to give a Curry-Howard-de Bruijn correspondence to the sequent calculus for Classical Logic.
The sequent calculus, introduced by Gentzen in \cite{Gentzen'35}, is a logical system in which the rules only introduce connectives (but on both sides of a sequent), in contrast to natural deduction which uses introduction and elimination rules.
The only way to eliminate a connective is to eliminate the whole formula in which it appears, via an application of the $(\Cut)$-rule.
A number of variants exist for Gentzen's calculus for classical logic {\LK};
the one we will consider in this paper, and that lies at the basis of the calculus \X, is the calculus known as $G_3$ \cite{Kleene'52}.

$G_3$ can be defined by:
{\def \stat#1#2 {#2}\def \Cap {\Ax} \def\Exp {\arrR}\def\Med {\arrL}
\def \PDer#1#2#3 { \Der {#2} {#3} }
 \[ \begin {array} {rlcrl}
\CapRule {y} {`a} {A} { \Gamma} { \Delta}
&\quad&
\DagRule {M} {`a} {x} {N} {A} { \Gamma} { \Delta}
\\[5mm]
\MedRule {M} {`a} {y} {x} {N} {A} {B} { \Gamma} { \Delta }
&\quad&
\ExpRule {x} {M} {`a} {`b} {A} {B} { \Gamma} { \Delta }
 \end {array} \]}%
It allows sequents of the form $\Der { A_1,\ldots,A_n }{ B_1,\ldots,B_m }$, where $A_1,\ldots,A_n$ is to be understood as $A_1 {\wedge}\ldots{\wedge} A_n$ and $B_1,\ldots,B_m$ is to be understood as $B_1 {\vee}\ldots{\vee} B_m$, and has implicit weakening and contraction.

Starting from different approaches in that area \cite{Curien-Herbelin'00,Urban'00}, in \cite{Lengrand'03} the calculus $\X$ was introduced, and shown to be equivalent to the $\lmmt$-calculus in terms of expressivity.
Using this correspondence, a strong normalization result is shown for $\lmmt$.
In fact, \cite{Lengrand'03} did not study any property of untyped $\X$, but focused only on its type aspects in connection with the sequent calculus.
\Comment{%
$\X$ is defined as a substitution-free calculus that is well equipped to describe the behaviour of functional programming languages at a very low level of granularity (see \cite{Lengrand'03,vBLL'05}).

As far as the Curry-Howard isomorphism is concerned, $\X$ stands out in that it is the first calculus to achieve that in full for a classical logic.
For example, in $\lmmt$, all provable propositions can be inhabited, but not all terms correspond to proofs, and in $`l`m$, not all proof contractions can be represented, since there reduction is confluent.

When studying $\X$ as an untyped language, soon the (unexpected) expressivity properties surfaced: it became apparent that $\X$ provides an excellent general purpose machine, very well suited to encode various calculi (for details, see \cite{vBLL'05}).
Amongst the calculi studied in that paper, the Calculus of Explicit Substitutions $\Lx$ \cite{Bloo-Rose'95} stands out.
In fact, a `subatomic' level was reached by decomposing explicit substitutions into smaller components.
Even more, the calculus $\X$ is actually symmetric~\cite{Barbanera-Berardi'96}; the `\emph{cut}', represented by $\Cut{P}{`a} {x}{Q}$ represents, in a sense, the parameter passing of explicit substitution of $P$ for $x$ in $Q$, but also that of calling the context $Q$ via $`a$ in $P$.
}

Perhaps the main feature of $\X$ is that it constitutes a \emph{variable} and \emph{substitution}-free method of computation.
Rather than having variables like $x$ representing places where terms can be inserted, in $\X$ the symbol $x$ represents a \emph{socket}, to which a term can be \emph{attached} via a plug $`a$. 
In fact, then main computational operation in $\X$ is \emph{renaming} in that reduction in $\X$ boils down to renaming of sockets and plugs, in addition to the restructuring of proofs.
\Comment{%
The definition of reduction on  $\X$ shows nicely how the interaction between sockets and plugs subtly and gently percolates through the terms. 

Although the origin of {\X} is a logic, and one could expect it to be close to \LC, it is in fact specified as a \emph{conditional term rewriting system}; the only non-standard aspect is that it treats \emph{three} different classes of variables (for plugs, sockets, and nets).
}
\vspace{3mm}

\noindent
In this paper we will treat $\X$ as a pure, untyped calculus, and ignore its logical origin in that we define a notion of sequent-style intersection type assignment on $\X$; intersection types are notorious for lacking a solid background in logic \cite{Hindley'84}.
We will see that, in view of the special nature of $\X$ as an input-output calculus, to achieve a natural notion of type assignment that is closed for expansion, we will also need to add union types.

The system we define in this paper is a conservative extension of the Intersection Type Assignment System for the {\LC} (see also \cite{Barendregt-Coppo-Dezani'83,Bakel-TCS'92,Bakel-TCS'95,Bakel-NDJFL'04}), in that lambda terms typeable in that system translate to $\X$-nets, while preserving the type.
As was the case for systems with intersection types for the {\LC} \cite{Barendregt-Coppo-Dezani'83,Bakel-TCS'95}, in order to get a notion of type assignment that is closed for $\eta$-reduction, we would need to introduce a $\leq$-relation on types which is contra-variant in the arrow; for simplicity, this is not part of the present system.
The system presented here is also a natural extension of the system considered in \cite{vBLL'05}, i.e.\ the basic implicative system for Classical Logic, but extended with intersection and union types and the type constants $\top$ and $\bottom$.
The main results of this paper are that this notion is closed for expansion, but needs to be restricted to become closed for reduction.

 \emptyline
This paper is constructed as follows.
Section \ref {sect:Xsyntax} presents the syntax and reduction system of the calculus $\X$.
In Section \ref {typing for X} we define the basic system of context assignment for $\X$; in Section \ref{X versus LC} we will embed the {\LC} into $\X$, and, in Section \ref{int for LC}, present an Intersection System for the \LC.
Then, in Section \ref {intun for X}, we will define a notion of type assignment on $\X$ that uses intersection and union types
 and show that type assignment in the strict system is preserved by the translation of the {\LC} into $\X$.

This is followed in Section \ref{WE} where we show that the system of Section~\ref{intun for X} is closed for witness expansion; as for witness-reduction, in Section \ref{problems} we will see that in quite a number of cases this collapses.
We will partially overcome these shortcomings in Section 
\ref{CBN} where we define a number of subsystems, restrictions of the system defined in Section \ref{intun for X}, that are either closed for Call-By-Name or Call-By-Value reduction.
Being restrictions, for these systems now witness-expansion collapses.
We draw our conclusions in Section \ref {Conclusions}. 
}

This paper corrects \cite{Bakel-ITRS'04}; the system now types all normal forms, and the subject-reduction problem is caught.


 \def\Peirce{
 \Inf	[\Exp]
	{ \Inf [\Imp]
		{ \Inf [\Exp]
			{ \Inf [\Cap]
				{ \DerX { \caps<y,`e> }{ \stat{y}{A} }{ \stat{`e}{A}, \stat{`h}{B} } }
			}
			{ \DerX { \Exp{y}{ \caps<y,`e> }{`h}{`a} }{ }{ \stat{`a}{A\arr B}, \stat{`e}{A} } }
		  \quad
		  \Inf	[\Cap]
			{ \DerX { \caps<w,`e> }{ \stat{w}{A} }{ \stat{`e}{A} } }
		}
		{ \DerX { \Imp{ \Exp{y}{ \caps<y,`e> }{`h}{`a} }{`a}{z}{w}{ \caps<w,`e> } }{ \stat{z}{(A \arr B)\arr A} }{ \stat{`e}{A} } }
	}
	{ \DerX {
 \Exp{z}{ \Imp{ \Exp{y}{ \caps<y,`e> }{`h}{`a} }{ `a }{ z }{ w }{
	  \caps<w,`e> } }{`e}{`g} }{ }{ \stat{`g}{((A \arr B)\arr A)\arr A} }
	}
}

\section{The calculus $\X$} \label{sect:Xsyntax}

In this section we will give the definition of the $\X$-calculus which has been proven to be a fine-grained implementation model for various well-known calculi \cite{vBLL'05}, like the $`l$-calculus \cite{Barendregt'84}, $\Lx$ \cite{Bloo-Rose'95}, $\lmu$ \cite{Parigot'92} and $\lmmt$ \cite{Herbelin'05}.
As discussed in the introduction, the calculus {\X} is inspired by the sequent calculus; the system we will consider in this section has only implication, and no structural rules.
$\X$ features two separate categories of `connectors', \emph {plugs} and \emph {sockets}, that act as input and output channels, respectively\Long{, and is defined without any notion of substitution or application}.
 \begin {definition}[Syntax]
The nets of the $\X$-calculus are defined by the following syntax, where $x,y$ range over the infinite set of {\em sockets}, $`a, `b$ over the infinite set of {\em plugs}.
\Comment{
 \[ \begin{array}{rrl@{\dquad}l}
P,Q &::=& \caps<x,`a> & \emph{capsule}
\\ & \mid & 
\exp y P `b . `a & \emph{export}
\\ & \mid & 
 \imp P `b [y] x Q & \emph{import}
\\ & \mid & 
\cut P `a + x Q & \emph{cut}
 \end{array} \]
}
 \[ \begin{array}{rcccccccc}
P,Q &::=& \caps<x,`a> & \mid & \exp y P `b . `a & \mid &  \imp P `b [y] x Q & \mid & \cut P `a + x Q 
\\
&& \emph{capsule} & & \emph{export} & & \emph{import} & & \emph{cut}
 \end{array} \]
 \end {definition}
The $\hat{`.}$ symbolises that the socket or plug underneath is bound in the net. 
The notion of bound and free connector is defined as usual, and we will identify nets that only differ in the names of bound connectors, as usual.
\Long{

 \begin {definition}\label{free and bound}
The \emph{free sockets} and \emph{free plugs} in a net are:

\[
 \begin {array}{llll}
 \fs (\caps<x,`a>) & = & \{x\} & \\
 \fs ( \exp x P `a . `b ) & = & \fs (P) \without \{x\} & \\
 \fs ( \med P `a [y] x Q ) & = & \fs (P) \Union \{y\} \Union (\fs (Q)\without \{x\}) & \\
 \fs ( \cut P `a + x Q ) & = & \fs (P) \Union (\fs (Q) \without \{x\}) &
\\[2mm]
 \end {array}
\
 \begin {array}{llll}
\fp ( \caps<x,`a> ) & = & \{`a\} &
\\
\fp ( \exp x P `a . `b ) & = & (\fp (P) \without \{`a\}) \Union \{`b\} &
\\
\fp ( \med P `a [y] x Q ) & = & (\fp (P) \without \{`a\}) \Union \fp (Q) &
\\
\fp ( \cut P `a + x Q ) & = & (\fp (P) \without \{`a\}) \Union \fp (Q) &
\\[2mm]
 \end {array}
 \]
}
We write $\fs (P)$ for the set of free sockets of $P$, and $\fp (P)$ for the set of free plugs of $P$;
a socket $x$ or plug $`a$ occurring in $P$ which is not free is called \emph{bound}\Long{, and we write $x`:\BS{P}$ and $`a`:\BP{P}$}.
We will write $x \in \fs (P,Q)$ for $x \in \fs (P) \And x \in \fs (Q)$.
\Long{ \end {definition}}
We adopt Barendregt's convention\Long{ (called convention on variables by Barendregt, but here it will be a convention on names),} in that free and bound names will be different.

The calculus, defined by the reduction rules below, explains in detail how cuts are propagated through nets to be eventually evaluated at the level of capsules, where the renaming takes place.
Reduction is defined by specifying both the interaction between well-connected basic syntactic structures, and how to deal with propagating active nodes to points in the net where they can interact. 

It is important to know when a connector is introduced\Long{, i.e.\ is connectable, i.e.\ is exposed and unique}; this will play a crucial role in the reduction rules. 
Informally, a net $P$, containing a socket $x$, introduces $x$ if $P$ is constructed from sub-nets which do not contain $x$ as free socket, so $x$ only occurs at the ``top level.'' 
\Long{
This means that $P$ is either an import with a middle connector $[x]$ or a capsule with left part $x$. 
Similarly, a net introduces a plug $`a$ if it is an export that ``creates'' $`a$ or a capsule with right part $`a$. 
}

 \begin {definition}[Introduction]
\begin {description}
 \item [$P$ introduces $x$]
Either $P = \imp Q `b [x] y R $ with $x \notin \FS{Q,R}$, or $P = \caps<x,`a>$.
 \item [$P$ introduces $`a$]
Either $P = \exp x Q `b . `a $ and $`a \notin \FP{Q}$, or $P = \caps<x,`a>$.
 \end {description}
 \end {definition}
 
The principal reduction rules are:

 \begin {definition}[Logical rules]
Let $`a$ and $x$ be introduced in, respectively, the left- and right-hand side of the main cuts below.
 \[ \begin {array}{rrcll}
 (\Cap): &
\cut \caps<y,`a> `a + x \caps<x,`b> &\redX& \caps<y,`b>
\\
 (\Exp): &
\cut { \exp y P `b . `a } `a + x \caps<x,`g> &\redX&
 \exp y P `b . `g 
\\
 (\Imp): &
\cut \caps<y,`a> `a + x { \imp Q `b [x] z R } &\redX&
 \imp Q `b [y] z R 
\\
 (\Ins): &
\cut { \exp y P `b . `a } `a + x { \imp Q `g [x] z R } 
 &\redX&
\left \{
\begin{array}{l}
 \cut Q `g + y { \cut P `b + z R } 
\\
 \cut { \cut Q `g + y P } `b + z R 
\end{array} \right. 
 \end {array} \]
 \end {definition}

\Long{
The first three logical rules above specify a renaming procedure, whereas the last rule specifies the basic computational step: it links the export of a function, available on the plug $`a$, to an adjacent import via the socket $x$.
The effect of the reduction will be that the exported function is placed in-between the two sub-terms of the import, acting as interface.
Notice that two cuts are created in the result, that can be grouped in two ways; these alternatives do not necessarily share all normal forms (reduction is non-confluent, so normal forms are not unique).

In $\X$ there are in fact two kinds of reduction, the one above, and the one which  defines  how to reduce a cut when one of its sub-nets does not introduce a connector mentioned in the cut. 
This will involve moving the cut inwards, towards a position where the connector \emph{is} introduced.
In case both connectors are not introduced, the above rules cannot be applied, this search can start in either direction, indicated by the tilting of the dagger, via the \emph{activation} of the cut.
}
\Short{If these rules cannot be applied, cuts need to be \emph{activated}:}

 \begin {definition}[Active cuts]
The syntax is extended with two {\em flagged} or {\em active} cuts:
 \[ P ::= \ldots \mid \cutL P_1 `a + x P_2 \mid \cutR P_1 `a + x P_2 \]

We define two cut-activation rules.
\[ \begin {array}{llcll}
 (\actL): & \cut P `a + x Q &\redX& \cutL P `a + x Q 
& \textit{if $P$ does not introduce $`a$}
\\
 (\actR): & \cut P `a + x Q &\redX& \cutR P `a + x Q 
&\textit{if $Q$ does not introduce $x$}
 \end {array} \] 
 \end {definition}

The next rules define how to move an activated dagger inwards.

 \begin {definition}[Propagation rules]
{\bf Left propagation}:
\[ \kern-1mm \begin {array}{rrcll}
 (\deactL): & \cutL \caps<y,`a> `a + x P &\redX& \cut \caps<y,`a> `a + x P 
\\
 (\Li): & \cutL \caps<y,`b> `a + x P &\redX& \caps<y,`b>, & `b \not= `a 
\\
 (\Lii): & \cutL { \exp y Q `b . `a } `a + x P &\redX&
\cut { \exp y { \cutL Q `a + x P } `b . `g } `g + x P , & `g\textit{ fresh} 
\\
 (\Liii): & \cutL { \exp y Q `b . `g } `a + x P &\redX&
\exp y { \cutL Q `a + x P } `b . `g , & `g \not= `a
\\
 (\Liv): & \cutL { \imp Q `b [z] y R } `a + x P &\redX&
\imp { \cutL Q `a + x P } `b [z] y { \cutL R `a + x P } 
\\
(\Lv): & \cutL { \cut Q `b + y R } `a + x P &\redX&
\cut { \cutL Q `a + x P } `b + y { \cutL R `a + x P } 
 \end {array} \]

\noindent {\bf Right propagation}:
\[ \begin {array}{rlcll}
 (\deactR): & \cutR P `a + x \caps<x,`b> &\redX& \cut  P `a + x \caps<x,`b> 
\\
 (\Ri): & \cutR P `a + x \caps<y,`b> &\redX& \caps<y,`b>, & y \not= x
\\
 (\Rii): & \cutR P `a + x { \exp y Q `b . `g } &\redX& \exp y { \cutR P `a + x Q } `b . `g 
\\
 (\Riii): & \cutR P `a + x { \imp Q `b [x] y R } &\redX&
\cut P `a + z { \imp { \cutR P `a + x Q } `b [z] y { \cutR P `a + x R } },
 & z\textit{ fresh} 
\\
 (\Riv): & \cutR P `a + x { \imp Q `b [z] y R } &\redX& 
\imp { \cutR P `a + x Q } `b [z] y { \cutR P `a + x R }, & z \not= x 
\\
 (\Rv): & \cutR P `a + x { \cut Q `b + y R } &\redX&
\cut { \cutR P `a + x Q } `b + y { \cutR P `a + x R } 
 \end {array} \]

We write $\redX$ for the (reflexive, transitive, compatible) reduction relation generated by the logical, propagation and activation rules.
 \end {definition}

\Comment{\emph{Summarising}, reduction brings all cuts down to logical cuts where both connectors are single and introduced, or elimination cuts that are cutting towards a capsule that does not contain the relevant connector, as in $\cutR P `a + x \caps<z,`b> $ or $\cutL \caps<z,`b> `a + x P $; performing the elimination cuts, via $(\Ri)$ or $(\Li)$, will then remove the term $P$.}

In \cite{Bakel-Lescanne'08}, two sub-reduction systems were introduced which explicitly favour one kind of activation whenever the above critical pair occurs:

 \begin {definition}[Call By Name and Call By Value]
We define Call By Name ({\CBN}) and Call By Value ({\CBV}) reduction by:
 \begin {itemize} \itemsep0pt
 \item
If a cut can be activated in two ways, {\CBV} only allows to activate it via $(\actL)$; we write $P \redCBV Q$ in that case.
This is obtained by replacing rule $(\actR)$ with:
 \[ \def\arraystretch{1} \begin {array}{llcll}
 (\actR): & \cut P `a + x Q & \redCBV & \cutR P `a + x Q , &
 \textit{if $P$ introduces $`a$ and $Q$ does not introduce $x$}.
\end {array} \]

 \item
{\CBN} can only activate such a cut via $(\actR)$; like above, we write $P \redCBN Q$.
Likewise, we can reformulate this as the reduction system obtained by replacing rule $(\actL)$ with:
 \[ \def\arraystretch{1} \begin {array}{llcll}
 (\actL): & \cut P `a + x Q & \redCBN & \cutL P `a + x Q , &
 \textit{if $P$ does not introduce $`a$ and $Q$ introduces $x$}.
\end {array} \]

 \item
As in \cite{Lengrand'03}, we split the two variants of $(\Ins)$ over the two notions of reduction:
 \[ \begin{array}{rcl}
\cut { \exp y P `b . `a } `a + x { \imp Q `g [x] z R } & \redCBV & \cut Q `g + y { \cut P `b + z R } 
\\
\cut { \exp y P `b . `a } `a + x { \imp Q `g [x] z R } & \redCBN & \cut{ \cut Q `g + y P } `b + z R
 \end{array} \]

 \end {itemize}
\Long{This way, we obtain two notions of reduction that are clearly confluent: all rules are left-linear and non-overlapping.}

 \end {definition}

Notice that the full reduction relation $\redX$ is not confluent; this comes in fact from the critical pair that activates a cut $\cut P `a + x Q $ in two ways. 
In fact, assuming $`a$ does not occur in $P$ and $x$ does not occur in $Q$, then $\cut P `a + x Q $ reduces to both $P$ and $Q$.
The first reduction takes place in $\CBV$, the second in $\CBN$.

\Long{
In \cite{Bakel-Lescanne'08,vBR'06} some basic properties are shown, which essentially show that the calculus is well-behaved, as well as the relation between $\X$ and a number of other calculi.
These results motivate the formulation of admissible rules:

 \begin {lemma}[Garbage Collection and Renaming \cite{vBR'06}]\label{renaming}
 \[ \begin {array}{rrcl@{\quad}l}
(\gcL): & \cutL P `a + x Q &\redX& P &\textrm{if }`a \not \in \FP{P}
\\ 
(\gcR): & \cutR P `a + x Q &\redX& Q &\textrm{if }x \not \in \FS{Q}
\\
(\renL): & \cut P `d + z \caps<z,`a> &\redX& P[`a/`d]
\\ 
(\renR): & \cut \caps<z,`a> `a + x + P &\redX& P[z/x]
\end {array} \]
 \end {lemma}
}

 \section{Typing for $\X$: from $G_3$ to \X} \label {types section} \label {typing for X}

$\X$ offers a natural presentation of the classical propositional calculus with implication, and can be seen as a variant of the $G_3$ system for {\LK} \cite{Kleene'52}.

\Long{We first define types and contexts.}

 \begin {definition}[Types and Contexts]\label{types}

 \begin {enumerate}

 \firstitem
The set of types is defined by the grammar \\
 $ \begin {array}{rcl}
A,B & ::= & \tvar \mid A \arrow B,
 \end {array} $
where $\tvar$ is a basic type of which there are infinitely many.

 \item
A {\em context of sockets} $`G$ is a mapping from sockets to types, denoted as a 
finite set of {\em statements} $\stat{x}{A}$, such that the {\em subject} of the statements ($x$) are distinct.
We write $`G_1,`G_2$ to mean the union of $`G_1$ and $`G_2$, provided $`G_1$ and $`G_2$ are compatible (if $`G_1$ contains $\stat{x}{A_1}$ and $`G_2$ contains $\stat{x}{A_2}$ then $A_1 = A_2$), and write $`G, \stat{x}{A}$ for $`G, \{\stat{x}{A}\}$.

 \item
Contexts of {\em plugs} $`D$ are defined in a similar way.

 \end {enumerate}

 \end {definition}

\Long{
The notion of type assignment on $\X$ that we present in this section is $G_3$, the basic implicative system for Classical Logic (Gentzen system LK) with implicit contraction and weakening as described above.
The Curry-Howard property is easily achieved by erasing all term-information.
When building witnesses for proofs, propositions receive names; those that appear in the left part of a sequent receive names like $x, y, z$, etc, and those that appear in the right part of a sequent receive names like $`a, `b, `g$, etc.
When in applying a rule a formula disappears from the sequent, the corresponding connector will get bound in the net that is constructed, and when a formula gets created, a different connector will be associated to it.
}

 \begin {definition}[Typing for $\X$]\label{Typing for X}

 \begin {enumerate}

 \firstitem
{\em Type judgements} are expressed via the ternary relation \\ $\DerX {P}{`G}{`D}$, where $`G$ is a context of {\em sockets} and $`D$ is a context of {\em plugs}, and $P$ is a net.
We say that $P$ is the {\em witness} of this judgement.

 \item
{\em Context assignment for} $\X$ is defined by the following rules:
 \[ \def\arraystretch{2.75}
 \begin {array}{rlcrl}
\\[-15mm]
 \CapRule{y}{`a}{A}{ `G}{ `D} 
&\quad&
 \ImpRule{P}{`a}{y}{x}{Q}{A}{B}{ `G}{ `D } 
\\
 \ExpRule{x}{P}{`a}{`b}{A}{B}{`G}{`D}
&\quad&
 \DagRule{P}{`a}{x}{Q}{A}{ `G}{ `D}
 \end {array}
 \]
We write $\derX  P : `G |- `D $ if there exists a derivation that has this judgement in the bottom line, and write $\derX \D :: P : `G |- `D $ if we want to name that derivation.

 \end {enumerate}
 \end {definition}

Notice that $`G$ and $`D$ carry the types of the free connectors in $P$, as unordered sets.
There is no notion of type for $P$ itself, instead the derivable statement shows how $P$ is connectable.

\Long{
 \begin{example}[A proof of Peirce's Law] \label{peirce example}
The following is a proof for Peirce's Law in {\LK}:
\[ 
{
 \def \Exp{\ArrR}
 \def \Imp{\ArrL}
 \def \Cap{\Ax}
 \def\arr{\Arrow}
 \def \stat#1#2{#2}
 \def \DerX#1#2#3{ \Der {#2}{#3} }
 \begin {array}{ccc}
\Peirce
 \end {array}
}
 \]
Inhabiting this proof in $\X$ gives the derivation:
\[ \Peirce \]
 \end{example}
}

The soundness result of simple type assignment with respect to reduction is stated as usual:

 \begin {theorem}[Witness reduction \cite{Bakel-Lescanne'08}] 
\label{Witness reduction}
If $\DerX {P}{`G}{`D}$, and $P \redX Q$, then $\DerX {Q}{`G}{`D}$.
 \end {theorem}




 \section{The relation with the Lambda Calculus} \label{X versus LC}

The remainder of this paper will be dedicated to the definition of a notion of intersection type assignment on $\X$.
The definition will be such that it will be a natural extension of a system with intersection types for the \LC; we will start by briefly summarising the latter.
We assume the reader to be familiar with the {\LC} \cite{Barendregt'84}; we just recall the definition of lambda terms and $\beta$-contraction.

\Long{We will write $\n$ for $\{1, \ldots, n\}$, where $n \geq 0$.}

\Long{
 \begin{definition} [Lambda terms and $\beta$-contraction \cite{Barendregt'84}]
\label {lambda terms}

 \begin{enumerate}

 \firstitem \label {L-terms}
The set $\Lambda$ of {\em lambda terms} is defined by the syntax:
 \[ \begin{array}{rcl}
M,N &::=& x \mid `l x . M \mid MN
 \end{array} \]

 \item \label {red}
The reduction relation $\bred$ is defined as the contextual
(i.e.~compatible \cite{Barendregt'84}) closure of the rule:
 \[ \begin{array}{rcl}
( `l x . M ) N &\bred& M [ N /x ]
 \end{array} \]

The relation $\dbred$ is the reflexive and transitive closure of $\bred$, and the $\eqb$ is the equivalence relation generated by $\dbred$.

 \end{enumerate}

 \end{definition}
}
We can define the direct encoding of the {\LC} into $\X$:

 \begin{definition} [\cite{vBLL'05}] \label{lc to x}
The interpretation of lambda terms into terms of $\X$ via the plug $`a$, $\SemL M {`a}$, is defined by:
 \[ \begin{array}{rcll}
\SemL{x}{`a} &=& \caps<x,`a> \\
\SemL{`lx.M}{`a} &=& \exp x \SemL{M}{`b} `b . `a \\
\SemL{MN}{`a} &=& \cut \SemL{M}{`g} `g + x { \med \SemL{N}{`b} `b [x] y \caps<y,`a> }, & \textit{with $x$ fresh}
 \end{array} \]
\Long{
We can even represent substitution explicitly (so interpret $\Lx$ \cite{Bloo-Rose'95}), by adding
\[ \begin{array}{rcll}
\SemL{\clo{M}{x}{N} }{`a} &=& \cutR \SemL{N}{`g} `g + x \SemL{M}{`a} , & `g\textit{ fresh}
\end{array}\]
}
 \end{definition}
Notice that every sub-term of $\SemL{M}{`a}$ has exactly one free plug, which corresponds to the name of hole of the present context in which $M$ appears, i.e.~its continuation.

\Long{
 \begin{example} \label {example delta i}
In Figure~\ref{delta i to i} we show how the circuit $\SemL{ (`l x.xx)(`l y.y) }{`a}$ can be reduced, using the rules above.
\begin{figure}[t]
 \[ \begin{array}{lcl}
\SemL{ (`l x.xx)(`l y.y) }{`a} 
&=& 
\\
\cut \SemL{`l x.xx}{`b} `b + z { \med \SemL{`l y.y}{`g} `g [z] v \caps<v,`a> } 
&=& 
\\
\cut { \exp x \SemL{xx}{`e} `e . `b } `b + z {\med \SemL{`l y.y}{`g} `g [z] v \caps<v,`a>  } 
& \red & (\Ins) 
\\
\cut \SemL{`l y.y}{`g} `g + x { \cut \SemL{xx}{`e} `e + v \caps<v,`a>  } 
& = & 
\\
\cut \SemL{`l y.y}{`g} `g + x { \cut { \cut \caps<x,`g> `g + z { \med \caps<x,`d> `d [x] w \caps<w,`e> } } `e + v \caps<v,`a> } 
& \red & (\renL), (\renR), (\actR) 
\\
\cutR \SemL{`l y.y}{`g} `g + x { \med \caps<x,`d> `d [x] w \caps<w,`a>  }
& \red & \\
\multicolumn{3}{l}{
\cut \SemL{`l y.y}{`g} `g + z { \med { 
	\cutR \SemL{`l y.y}{`g} `g + x \caps<x,`d> 
 } `d [z] w {
	\cutR \SemL{`l y.y}{`g} `g + x \caps<w,`a> 
} } }
\\ & \red & 
(\deactR), (\Rii), (\renR)
\\
\cut \SemL{`l y.y}{`g} `g + z { \med \SemL{`l y.y}{`d} `d [z] w \caps<w,`a> }
& = & 
\\
\cut { \exp y \caps<y,`s> `s . `g } `g + z { \med \SemL{`l y.y}{`d} `d [z] w \caps<w,`a> }
& \red & (\Ins) 
\\
\cut \SemL{`l y.y}{`d} `d + y { \cut \caps<y,`s> `s + w \caps<w,`a>  }
& \red & (\Cap), (\renR) 
\\
\SemL{`l y.y}{`a} 
\\
 \end{array} \]
\caption{ \label{delta i to i} Reducing $\SemL{ (`l x.xx)(`l y.y) }{`a}$ to $\SemL{ `l y.y }{`a}$.}
\end{figure}
 \end{example}
}

As shown in \cite{vBLL'05}, the notion of Curry type assignment for the \LC, $\derL `G |- M : A $, is strongly linked to the one defined for $\X$.

 \begin{definition}[Curry type assignment for \LC]

The type assignment rules for the Curry type assignment system for the {\LC} are:
 \[ \begin{array}{ccccc}
 (\Ax):
	\Inf	{ \derL `G,x{:}A |- x : A }
 &\quad&
 (\arrI):
	\Inf	{ \derL `G, x{:}A |- M : B }
		{ \derL `G |- `lx . M : A\arr B }
 &\quad&
 (\arrE):
	\Inf	{ \derL `G |- M : A \arr B \quad \derL `G |- N : A }
		{ \derL `G |- MN : B }
 \end{array} \]

 \end{definition}

In \cite{vBLL'05}, the following relation is shown between {\LC} and $\X$:

 \begin{theorem} [\cite{vBLL'05}]
 \begin{enumerate}
 \firstitem If $M \bred N$, then $\SemL{M}{`a} \red \SemL{N}{`a}$.
 \item If $M \redCBN N$, then $\SemL{M}{`a} \redCBN \SemL{N}{`a}$.
 \item If $M \redCBV N$, then $\SemL{M}{`a} \redCBV \SemL{N}{`a}$.
 \item If $\derL `G |- M : A $, then $\derX  \SemL{M}{`a} : `G |- `a{:}A $.
 \end{enumerate}
 \end{theorem}

 \section{Intersection Type Assignment for the Lambda Calculus} \label{int for LC}

The notion of intersection type assignment for $\X$ as defined in the next section is a conservative extension of the Intersection Type Assignment System of \cite {Barendregt-Coppo-Dezani'83}, in that we can translate lambda terms typeable in that system to $\X$ circuits while preserving types.
In this section, we will briefly discuss that system; we will modify it slightly, since we do not want to model extensionality.

The type assignment system presented here is based on the BCD-system defined by H.~Barendregt, M.~Coppo and M.~Dezani-Cianca\-glini in \cite{Barendregt-Coppo-Dezani'83}, in turn based on the system as presented in \cite{Coppo-Dezani-Venneri'81}.
The BCD-system treats the two type constructors `$\arrow$' and `$\inter$' the same, allowing, in particular, intersection to occur at the right of arrow types.
It also introduced a partial order relation `$\seq$' on types, adds the type assignment rule $(\seqr)$, and introduced a more general form of the rules concerning intersection.
We will deviate here slightly from that system.

 \begin{definition} [Intersection types, statements, and contexts] \label {intersection types}

 \begin {enumerate}

 \firstitem \label {Types}
Let $\Phi$ be a countable (infinite) set of type-variables, ranged over by $\tvar$.
$\Types$, the set of {\em intersection types}, ranged over by $A, B, \ldots$%
\footnote{In \cite{Barendregt-Coppo-Dezani'83}, Greek characters are used to represent types, and $`w$ is used for $\Top$; we use Greek characters for \emph{plugs}.}, is defined through:
\Short{$}\Long{\[} \begin {array}{rcl}
 A, B &::=& `v \mid \Top \mid (A \arr B) \mid (A \inter B)
 \end {array}. 
\Short{$}\Long{\]}
$\Top$ is pronounced ``\emph{top}''.

 \item A {\em statement} is an expression of the form $M : A $, with $M \in \Lambda$, and $A \in \Types$.
 $M$ is the {\em subject} and $A$ the {\em predicate} of $M : A $.

 \item A {\em context} $`G$ is a partial mapping from term variables to intersection types, and we write $x{:}A \ele `G$ if $`G\,x = A$, {i.e.} if $A$ is the type stored for $x$ in $`G$.
We will write $x \not\in `G$ if $`G$ is not defined on $x$, and $`G \Except x$ when we remove $x$ from the domain of $`G$.

\Comment{
 \item
For contexts $`G_1, \ldots, `G_n$, the context $`G_1 \inter \dots \inter `G_n$ is defined by: 
$`G_1 \inter \dots \inter `G_n \, x = A_1 \inter \dots \inter A_m$ if and only if $\{ A_1, \ldots, A_m \} = \{ A \mid \Exists \iotn [ `G_i\,x = A \}$ is the set of all statements about $x$ that occur in $`G_1 \bigcup \dots \bigcup `G_n$.
}

 \item
We write $`G \inter x{:}A$ for the context $`G \inter \{x{:}A\}$, i.e., the context defined by:
 \[ \begin{array}{rcll}
`G \inter x{:}A & = & `G \union \{x{:}A\}, & \textrm{if $x \not \in `G$} 
\\ & = & 
`G \Except x \Union \{x{:}A\inter B\}, & \textrm{if }x{:}B \in `G
 \end{array} \]
We will often write $`G, x{:}A$ for $`G \inter x{:}A$ when $x \not \in `G$.
In the notation of types, as usual, right-most outer-most brackets will be omitted.

 \end{enumerate}

 \end{definition}

We will consider a pre-order (i.e.\ reflexive and transitive relation) on types which takes into account the idem-potence, commutativity and associativity of the intersection type constructor, and defines $\top$ to be the maximal element.

\begin{definition} [Relations on types] \label {seq definition}
\Long{ \begin{enumerate}

 \firstitem \label {klgi} }
On $\Types$, the \Long{type inclusion }%
relation $\seq$ is defined as the smallest pre-order such that:
\Long{
 \[ \begin{array}{rcl}
A \seq \Top \dquad A \inter B \seq A,B \dquad C \seq A \And C \seq B \Then$ $C \seq A \inter B  \end {array} \]
}
\Short{ $ 
A \seq \Top, ~ A \inter B \seq A , ~ A \inter B \seq B$, and $ C \seq A \And C \seq B \Then$ $C \seq A \inter B $. 
}
 
\Long{\item}
The relation $\equ$ is defined by:
\Short{$A \seq B \seq A \Then A \equ B $ and $ A \equ C~\And~B \equ D \Then A \arr B \equ C \arr D.$}%
\Long{\[ \begin{array}{rcl} 
 A \seq B \seq A &\Then& A \equ B 
\\
 A \equ C~\And~B \equ D &\Then& A \arr B \equ C \arr D.
 \end{array} \] } 

\Long{ \end {enumerate} }

$\Types$ will be considered modulo $\equ$; then $\seq$ becomes a partial order.
 \end {definition}

We need to point out that the $\seq$ relation as defined in \cite{Barendregt-Coppo-Dezani'83} is slightly different.
It also contains the cases
$ ( A \arr B )\inter( A \arr C ) \seq A \arr (B \inter C) $,
$ C \seq A~\And~B \seq D \Then A \arr B \seq C \arr D $, and
 $  \Top \seq \Top\arr \Top$.
These were mainly added to obtain a system closed for $\eta$-reduction (see also \cite{Bakel-TCS'95}), which is not an issue in this paper.
\Long{Also, we could add these cases to Definition~\ref{context assignment}, but would then have to add a type assignment rule dealing explicitly with $\seq$.}

It is easy to show that both $(A \inter B) \inter C \sim B \inter (A \inter C)$ and $A \inter B \sim B \inter A$, so the type constructor $\inter$ is associative and commutative, and we will write $A \inter B \inter C$ rather than $(A \inter B) \inter C$.
We will write $\n$ for the set $\{1,\ldots,n\}$, and often write $\AoI{n}$ for $A_1 \inter \dots \inter A_n$, and consider $\Top$ to be the empty intersection: $\Top = \AoI{0}$.

 \begin {definition} \label {intersection type assignment} \label {BCD}
\emph{Type assignment} \Long{and \emph{derivations} are }%
\Short{is }defined by the following natural deduction system.
 \[ \begin{array}{c}
 \begin {array}{rlcrlcrl}
( \Ax): &
 \Inf	{ \derI `G,x{:}A |- x : A }
&\quad&
(\intI): &
 \Inf	[n \geq 0]
 	{ \derI `G |- M : A_j 
	 \quad
	 (\forall j \ele \n) 
	}
	{ \derI `G |- M : \AoI{n} }
&\quad&
(\intE): &
 \Inf	[j \ele \n]
	{ \derI `G |- M : \AoI{n} }
	{ \derI `G |- M : A_j }
\end{array}
\\[4mm]
\begin {array}{rlcrlcrl}
(\arrI): &
 \Inf	{ \derI  `G,x{:}A |- M : B }
	{ \derI  `G |- `lx.M : {A \arr B} }
&\quad&
(\arrE): &
 \Inf	{ \derI  `G |- M : {A\arr B} 
	  \quad
	  \derI  `G |- N : A 
	}
	{ \derI  `G |- MN : B }
\end{array}
\end{array}\]

\Long{We will write $\derI `G |- M : A $ for statements that are derived using these rules.}

 \end {definition}

Again, notice that the original definition contained also the rule $(\seqr)$, 
\Long{
\[
(\seqr): ~
 \Inf	[ A \seq B ]
 	{ \derI  `G |- M : A }
 	{ \derI  `G |- M : B }
\]
}%
added to be able to express contra-variance of the (original) $\seq$-relation over arrow types.
The system as set up here does not need this rule.

\Long{
It is easy to check that the type assignment system is closed under $\eqbr$.
First, that the system is closed under subject reduction can be illustrated by the following `Cut and Paste' proof:
Suppose that $\derI `G |- (`lx.M)N : B $.
Assume that $(\arrE)$ is applied last, so there exists $ A$ such that
 \[
\derI `G |- `l x.M : A\arrow B  \Rm{ and } \derI `G |- N : A .
 \]
Notice that $(\arrI)$ is the last step performed for the first result\Long{ and $(\intI)$ should be the last step for the latter}%
, so
 \[
 \derI `G,x{:}A |- M : B \Rm{ and } \derI `G |- N : A .
 \]
Then a derivation for $\derI `G |- M[N/x] : B $ can be obtained by replacing, in the derivation for $\derI `G,x{:}A |- M : D $, the sub-derivation $\derI `G,x{:}A |- x : A $ by the derivation for $\derI `G |- N : A $.

The problem to solve in a proof for closure under $ `b$-equality is then that of $ `b$-expansion:
 \[
\Rm{if }\derI `G |- M[N/x] : C , \Rm{ then }\derI `G |- ( `l x.M)N : C .
 \]
Assume that the term-variable $x$ occurs in $M$ and the term $N$ is a sub-term of $M[N/x]$, so $N$ is typed in the derivation for $\D \dcol \derI `G |- M[N/x] : C $, probably with several different types $\Dotn$.
A derivation for
$
 \derI `G,x{:}{\DoIn} |- M : C 
$
can be obtained by replacing, in $\D$, all derivations for $\derI `G |- N : D_i $ by the derivation for
 $
 \derI x{:}{\DoIn} |- x : D_i .
 $
Then, using $(\intI)$, we have $\derI `G |- N : {\DoIn} $, and, using $(\arrI)$, $\derI `G |-  `lx.M : \DoIn \arrow C $.
Then, using $(\arrE)$, the redex can be typed.

When the term-variable $x$ does not occur in $M$, the term $N$ is a not a sub-term of $M[N/x]$ and $\derI `G |- M[N/x] : C $ stands for $\derI `G |- M : C $.
In this case, the type $ \Top$ is used: since $x$ does not occur in $M$, by weakening $x{:} \Top$ can be assumed to appear in $`G$, and rule $(\arrI)$ gives $\derI `G |-  `lx.M : \Top\arr C $.
By $(\intI)$, $\derI `G |- N : {\Top} $, so, using $(\arrE)$, the redex can be typed.

}




\section{Intersection and Union Context Assignment for $\X$} \label {intun for X}

The notion of intersection context assignment on $\X$ that we will present in this section is a natural extension of the system considered in \cite{vBLL'05}, i.e.\ the basic implicative system for Classical Logic, but extended with intersection and union types and the type constants $\Top$ and $\Bottom$.
The system we present here is a correction of the system presented in \cite{Bakel-ITRS'04}.

\Long{
Note that, as with intersection type assignment for the {\LC}, the relation between the system we presented above and (classical) logic is lost.
In fact, although the rules for dealing with intersection and union types are similar to the rules for the logical connectors `\emph{and}' ($\And$) and `\emph{or}' ($\Or$), intersection and union have \emph{no} logical content: read as logical rules, $(\intR)$ and $(\unL)$ are only applicable when all sub-formulae are shown by proofs with \emph{identical structure} - they have the same witness \cite{Hindley'84}.

The system presented in this section is set up to be concise, and closed for reduction and expansion.
As discussed in the previous section, the intersection type constructor is needed for expansion in the {\LC}; in $\X$ we need intersection during expansion when, for example, considering the reduction step $(\Riii)$:
 \[  
\cutR P `a + x {\imp Q `g [x] y R } \red \cut P `a + v { \imp { \cutR P `a + x Q } `g [v] y { \cutR P `a + x R } } 
 \] 
Assume we have a derivation for the right-hand side, constructed like this (notice that, by Barendregt's convention, we can assume that $y$ is not free in $P$):
 \[ \kern-1cm
 \Inf	[\Cut]
	{ \InfBox{\D_1}
		{ \derX P : `G |-  `a{:}A\arr B,`D }
	  \kern-5mm
\raise.5\RuleH\hbox to 9.5cm{
	  \Inf	[\Imp]
		{ \Inf	[\Cut]
			{ \InfBox{\D_2}
				{ \derX P : `G |- `a{:}C,`D }
			  \quad
			  \InfBox{\D_3}
				{ \derX Q : `G,x{:}C |- `g{:}A,`D }
			}
			{ \derX \cutR P `a + x Q : `G |- `g{:}A,`D }
\raise2.5\RuleH\hbox to 25mm {\kern-22mm
		  \Inf	[\Cut]
			{ \InfBox{\D_4}
				{ \derX P : `G |- `a{:}D,`D }
			  \quad 
			  \InfBox{\D_5}
				{ \derX R : `G,y{:}B,x{:}D |- `D }
			}
			{ \derX \cutR P `a + x R : `G,y{:}B |- `D }
}
\multiput(-10,0)(0,7){7}{.}
		}
		{ \derX \imp { \cutR P `a + x Q } `g [v] y { \cutR P `a + x R } : `G,v{:}A\arr B |- `D }
}
\multiput(-20,0)(0,7){2}{.}
	}
	{ \derX \cut P `a + v { \imp { \cutR P `a + x Q } `g [v] y { \cutR P `a + x R } } : `G |- `D }
 \]
Then, in order to construct a derivation for the left-hand side, we need to combine the derivations for $P$; since this combines a number of output types on $`a$, as above for the \LC, the natural tool is intersection:
 \[ \kern-3cm
\Inf	
	{ 
	  \Inf	[\intR]
		{ \InfBox{\D_1}
			{ \derX P : `G |- `a{:}A\arrow B,`D }
		  \quad
		  \InfBox{\D_2}
			{ \derX P : `G |- `a{:}C,`D }
		  \quad
		  \InfBox{\D_4}
			{ \derX P : `G |- `a{:}D,`D }
		}
		{ \derX P : `G |- `a{:}(A\arr B)\inter C\inter D,`D }
\raise3.5\RuleH\hbox to 2cm {\kern-3.5cm
	  \Inf	[\Imp]
		{ \InfBox{\D_3}
			{ \derX Q : `G,x{:}C |- `g{:}A,`D }
		  \quad 
		  \InfBox{\D_5}
			{ \derX R : `G,y{:}B,x{:}D |- `D }
		}
		{ \derX \imp Q `g [x] y R : `G,x{:}(A\arr B)\inter C\inter D |- `D }
}
\multiput(-20,-1)(0,7){9}{.}
	}
	{ \derX \cutR P `a + x { \imp Q `g [x] y R } : `G |- `D }
 \]
Notice that we also naturally build an intersection type for the socket $x$.

Similarly, we need union, for example, when dealing with rule $(\Lii)$: 
\[
\cutL { \exp y P `e . `a } `a + x Q \red \cut { \exp y { \cutL P `a + x Q } `e . `g } `g + x Q \] 
with $`g$ fresh.
Assume we have a derivation for the right-hand side, shaped like
 \[ 
\Inf	[\Cut]
	{ 
	  \Inf	[\arrR]
		{ \Inf	[\Cut]
			{ \InfBox{\D_1}
				{ \derX P : `G,y{:}A |- `e{:}B,`a{:}C,`D } 
			  \dquad
			  \InfBox{\D_2}
				{ \derX Q : `G,x{:}C |- `D }
			}
			{ \derX \cutL P `a + x Q : `G,y{:}A |- `e{:}B,`D }
		}
		{ \derX \exp y { \cutL P `a + x Q } `e . `g : `G |- `g{:}A\arr B,`D }
	  \quad
	  \InfBox{\D_3}
		{ \derX Q : `G,x{:}A\arr B |- `D }
	}
	{ \derX \cut { \exp y { \cutL P `a + x Q } `e . `g } `g + x Q : `G |- `D }
 \]
When trying to build a derivation for the left-hand side, we need a way to gather the derivations for $Q$ using a new type constructor, dual to intersection; the natural choice is union:
 \[
\Inf	[\Cut]
	{ \Inf	[\arrR]
		{ \InfBox{\D_1}
			{ \derX P : `G,y{:}A |- `e{:}B,`a{:}C,`D } 
		}
		{ \derX \exp y P `e . `a : `G |- `a{:}C\union (A\arr B),`D } 
	  \quad
	  \Inf	[\unL]
		{ \InfBox{\D_2}
			{ \derX Q : `G,x{:}C |- `D }
		  \quad
		  \InfBox{\D_3}
			{ \derX Q : `G,x{:}A\arr B |- `D }
		}
		{ \derX Q : `G,x{:}C\union (A\arr B) |- `D }
	}
	{ \derX \cutL {\exp y P `e . `a } `a + x Q : `G |- `D }
 \]
Also here the building of the union type for $x$ (as in the logical rule for $\Or$) is mirrored by the automatic construction of a union type for the plug $`a$.

Using union and intersection is also supported by the normal view for statements like $`G~\Turn~`D$, in that the formulae in the context $`G$ are all necessary for the result, and not all the formulae in $`D$ necessarily follow from $`G$; in other words, the formulae in the context $`G$ are connected through the logical `and', whereas those in $`D$ are connected through the logical `or'.
And although `intersection' is not `and', and `union' is not `or', the link between these concepts is strong enough to justify our choice.

\Comment{
Initially a system was set up that allowed \emph{only} intersection types for sockets, and \emph{only} union types for plugs, but this soon proved to be too restrictive.
As can be seen by the examples above, intersection types are sometimes needed on plugs, as union types can be needed on sockets.
This corresponds in the {\LC} to intersection types being used to deal with multiple occurrences of variables with different types, but also to model that terms can have multiple types within a fixed context.

Essentially, that initial choice still stands: intersection types for sockets, and union types for plugs.
However, a union type like $A \union B$ for sockets is allowed, but only if both $A$ and $B$ can be justified (see rule $(\unL)$); similarly, an intersection type like $A \inter B$ for plugs is only allowed if both $A$ and $B$ can be justified (see rule $(\intR)$; this corresponds to rule $(\intI)$ in $\TurnI$).
}
}

\emptyline
The following definition of types is a natural extension of the notion of types of the previous section, by adding union as a type constructor.

 \begin{definition} [Intersection and Union Types, Contexts] \label {interunion types}

 \begin{enumerate}

 \firstitem \label {IntUn Types}
The set $\Types$ of {\em intersection-union types}, ranged over by $A, B, \ldots$ is defined by:
\Short{$}\Long{\[}
 \begin{array}{rcl}
 \Types &::=& \tvar \mid \Top \mid \Bottom \mid (\Types\arr \Types) \mid (\Types \inter \Types) \mid (\Types \union \Types)
 \end{array}
\Short{$.}\Long{\]}

The set $\Tproper$ is the set of {\em proper types}, defined by:
\Short{$}\Long{\[}
 \begin{array}{rcl}
\Tproper &::=& \tvar \mid (\Types \arr \Types)
 \end{array}
\Short{$.}\Long{\]}
\Comment{
the set $\IntTypes$ is the set of {\em intersection types}, defined by:
\Short{$}\Long{\[}
 \begin{array}{rcl}
\IntTypes &::=& \Tproper \mid \Top \mid (\IntTypes \inter \IntTypes)
 \end{array}
\Short{,$}\Long{\]}
and $\UnTypes$ is the set of {\em union types}, defined by:
\Short{$}\Long{\[}
 \begin{array}{rcl}
\UnTypes &::=& \Tproper \mid \Bottom \mid (\UnTypes \union \UnTypes)
 \end{array}
\Short{.$}\Long{\]}
}

 \item A \emph{context} $`G$ of sockets ($`D$ of plugs) is a partial mapping from sockets (plugs) to types in $\Types$, represented as a set of statements with only distinct connectors as subjects.
We write $x \ele `G$ ($x \ele `D$) if $x$ ($`a$) gets assigned a type by $`G$ ($`D$).


 \end{enumerate}

 \end{definition}

\Long{
We will omit unnecessary brackets in types; the type constructors `$\inter$' and `$\union$' will bind more strongly than `$\arr$', so $A \inter B\arr C\inter D$ stands for $((A \inter B)\arr (C \inter D))$, $A\arr B \inter C\arr D$ stands for $(A\arr ((B \inter C)\arr D))$, and $(A\arr B )\inter C\arr D$ stands for $(((A\arr B )\inter C)\arr D)$.
We will sometimes write the omitable brackets for readability.
}

We will consider a pre-order on types which takes into account the idempotence, commutativity and associativity of the intersection and union type constructors, and defines $\Top$ to be the maximal element, and $\Bottom$ to be the minimal.

\begin{definition} [Relations on types] 
\label {interunion leq definition}
\Long{
 \begin{enumerate}

 \firstitem 
}
The relation $\seqr$ is defined as the least pre-order on $\Types$ such that:
\Long{
 \[ \begin{array}{ccccc}
A \inter B \seq A,B &\quad& A \seq \Top &\quad& A \seq B \And A \seq C \Then A \seq B \inter C \\
A,B \seq A \union B && \Bottom \seq A && A \seq C \And B \seq C \Then A \union B \seq C 
 \end {array} \]
We \Long{will }write $A \seqI B$ for the relation on types defined by the first line, and $A \seqU B$ for the one defined by the second.

}
\Short{
 $A \seq A$, $A \inter B \seq A$, $A \inter B \seq B$, $C \seq A \And C \seq B \Then C \seq A \inter B$, $A \seq \Top$, and $A \seq A\union B$, $B \seq A\union B$, $A \seq C \And $ $ B \seq C \Then A \union B \seq C$, and $\Bottom \seq A$. 
}

\Long{ Notice that, since the relation defined here is a natural extension of that in Definition~\ref{seq definition}, we are free to use the same symbol.

  \item  
}
The equivalence relation $\equ$ on types is defined as before\Short{ .}%
\Long{, by: 
 \[ \begin{array}{rcl}
A \seq B \seq A &\Then& A \equ B
\\
 A \equ C~\And~B \equ D &\Then& A \arr B \equ C \arr D
 \end{array} \]

\end{enumerate}
 
}
 
We will consider types modulo $\equr$; then $\leq$ becomes a partial order.
 \end{definition}
Notice that we can show that $A \union (B\inter C) \seq (A \union B) \inter (A \union C)$, but cannot show the converse.

Remark, as mentioned above, that the relation is \emph{not} defined over arrow types, as in the system of \cite{Barendregt-Coppo-Dezani'83}.
More pointedly, we do not consider the type $A\arr (C\int (C\arr D))$ smaller than $(A\arr C)\int (A\arr C\arr D)$; the system would not be closed for the relation.

 \begin{definition}
 \begin{enumerate}

 \firstitem
For contexts of sockets $`G_1, \ldots, `G_n$, the context $`G_1 \inter \dots \inter `G_n$ is defined by:
\Short{\\ $}\Long{\[}
x{:}A_1 \inter \dots \inter A_m \in `G_1 \inter \dots \inter `G_n
\Short{$}\Long{\]}
if and only if $\{ x{:}A_1, \ldots, x{:}A_m \}$ is the set of all statements about $x$ that occur in $\bigcup_{n}`G_i$, where $\bigcup$ is set-union.
\Long{We write $`G \inter x{:}A$ for the context of sockets $`G \inter \{x{:}A\}$
\Comment{, i.e., the context defined by:

 \[ \begin{array}{rcll}
 `G \inter x{:}A & = & `G \cup \{x{:}A\}, & \textrm{if }x \not\in `G \\
	& = & `G \Except x \cup \{x{:}A\inter B\}, & \textrm{if }x{:}B \ele `G
 \end{array} \]
A}
; as before, we will write $`G,x{:}A$ for $`G \inter x{:}A$ when $x \not \in `G$, and 
 
}
\Short{
The notations $`G \inter x{:}A$ and $`G,x{:}A$ are defined as above; 
}
we will write $\GamoIn$ for $`G_1 \inter \dots \inter `G_n$.

 \item
For contexts of plugs, $`D_1, \ldots,`D_n$, the context $`D_1 \union \dots \union `D_n$%
\Short{
and the notions $`a{:}A \union `D $ and $`a{:}A, `D $ are 
}
\Long{ is }%
defined \Short{similarly.}%
\Long{
by:
\[ `a{:}A_1 \union \dots \union A_m \in `D_1 \union \dots \union `D_n\] if and only if $\{ `a{:}A_1, \ldots, `a{:}A_m \}$ is the set of all statements about $`a$ that occur in $\bigcup_{n}`D_i$.
We write $`a{:}A \union `D$ for the context of sockets $\{ `a{:}A \} \union `D$\Comment{, i.e., the context defined by:
 \[ \begin{array}{rcll}
`a{:}A \union `D & = & \{ `a{:}A \} \cup `D, & \textrm{if }`a \not \in `D \\
	& = & \{`a{:}A\union B\} \cup `D \Except `a, & \textrm{if }`a{:}B \ele `D
 \end{array} \]
A}, and write $`a{:}A,`D$ for $`a{:}A \union `D$ when $`a \not \in `D$, and write $\DeloUn$ for $`D_1 \union \dots \union `D_n$.
}

\Long{\item
As above, $`G_1,`G_2$ stands for the context of sockets that essentially represents their union, but assuming $`G_1$ and $`G_2$ are \emph{compatible}: if $x{:}A_1 \ele `G_1$, and $x{:}A_2\ele `G_2$, then $A_1 = A_2$.
We will call this union a \emph{compatible union}, written $\compunion$, and write $\compunion_{n}`G_i$ for $`G_1, \ldots, `G_n$.
Similar for $`D_1,`D_2$ and $\compunion_{n}`D_i$. 
}
 \end{enumerate}

 \end{definition}

\Long{As before, w}\Short{W}%
e will write $\AoIn$ for $\AoInfull$ (with each $A_i$ in $\UnTypes$), and $\Top$ (\emph{top}) for the empty intersection type, as well as $\AoUn$ for $\AoUnfull$ ($A_i$ in $\IntTypes$), and $\Bottom$ (\emph{bottom}) for the empty union.

We will now define a notion of intersection-union context assignment for $\X$.

 \begin{definition}[Intersection and Union Typing for $\X$] \label{intersection and union tas}\label{context assignment}

\Long{
 \begin{enumerate}
  \firstitem
\emph{Intersection type judgements} are expressed via a ternary relation $\derX P
: `G |- `D $, where $`G$ is a context of {\em sockets} and $`D$ is a context of {\em plugs}, and $P$ is a net.
We say that $P$ is the {\em witness} of this judgement.

 \item  
}
{\em Intersection and union context assignment for} $\X$ is defined by the following sequent style calculus:
\[ \def\arraystretch{2.5}
\begin{array}{rlcrl}
 (\Ax): &
\Inf	{ \derX \caps<y,`a> : `G,\stat{y}{A} |- `a{:}A,`D }
&&
 (\Cut): &
\Inf	
	{ \derX P : `G |- `a{:}A,`D \quad \derX Q : `G, \stat{x}{A} |- `D }
	{ \derX \cut P `a + x Q : `G |- `D }
\\
 (\arrL):&
\Inf	{ \derX P : `G |- `a{:}A,`D \quad \derX Q : `G,\stat{x}{B} |- `D }
	{ \derX \imp P `a [y] x Q : `G \inter \stat{y}{A\arr B} |- `D }
&&
 (\arrR): &
\Inf	{ \derX P : `G,\stat{x}{A} |- `a{:}B,`D }
	{ \derX \exp x P `a . `b : `G |- \stat{`b}{A\arr B \union `D} }
\\
(\intR): &
\Inf	[n \geq 0]
	{ \derX P : `G |- `a{:}A_j,`D
	 \quad (\forall j \ele \n)
	}
	{ \derX P : `G |- `a{:}\AoIn,`D }
&&
(\unL): &
\Inf	[n \geq 0]
	{ \derX P : `G,x{:}A_j |- `D
	 \quad (\forall j \ele \n)
	}
	{ \derX P : `G, x{:}\AoU{n} |- `D }
\\
(\intE): &
 \Inf	[j \ele \n]
	{ \derX P : `G |- `a{:}\AoI{n},`D }
	{ \derX P : `G |- `a{:}A_j,`D }
&&
(\unE): &
 \Inf	[j \ele \n]
	{ \derX P : `G,x{:}\AoU{n} |- `D }
	{ \derX P : `G,x{:}A_js |- `D }
\\
 \end{array} \]
\Comment{
 \[ \begin{array}{rlcrlcrl}
(\seqr): &
\Inf	[`G' \seq `G,`D \seq `D']
	{ \derX P : `G |- `D }
	{ \derX P : `G' |- `D' }
&\quad&
 \end{array} \]
}
NB: rule $(\Cut)$ is also used for the activated cuts.

\Long{
 \item
We will write $\derX P : `G |- `D $ if there exists a derivation that has this judgement in the bottom line, and write $\D \dcol \derX P : `G |- `D $ if we want to name that derivation.

 \end{enumerate}
 
}

 \end{definition}
As will be argued below, this notion of type assignment is too liberal to obtain preservation of types under conversion.
\Long{As argued above, t}\Short{T}%
he system is constructed to satisfy preservation of types under expansion (see Theorem \ref{witness expansion}), but we will see that it is not closed for reduction (Section \ref{problems}).
We will partly recover from this in 
\Long{sections \ref{CBN} and \ref{CBV}}%
\Short{Section \ref{CBN}}%
, where we define restrictions of the system above that satisfy preservation of types under, respectively, {\CBN} and {\CBV} reduction.
However, a natural consequence of these restrictions made is that the systems no longer will be closed for expansion.
\Long{
\Comment{
\emptyline 
The most important thing to notice is that, via the rules, \emph{intersection types} are built of contexts of sockets, and \emph{union types} are built for contexts of plugs, so all contexts are pure.
However, a union type can appear in a contexts of sockets, but only via rule $(\unL)$; similarly, an intersection type can appear in a contexts of plugs, but only via rule $(\intR)$.

Notice that the last two rules $(\Bottom)$ and $(\Top)$ are almost special cases of rules $(\unL)$ and $(\intR)$, but for the fact that the intersection (or union) over zero sets would give the empty set; since weakening should be an admissible rule, we need to state the zero cases as separate rules.
This could be avoided by forcing the rules to agree on the contexts, as in, for example,
 \[ \begin{array}{rlcrl}
(\unL)': & \Inf	[n \geq 0]
		{ \derX P : `G,x{:}A_1 |- `D
		 \quad \ldots \quad
		 \derX P : `G,x{:}A_n |- `D
		}
		{ \derX P : {`G,x{:}\AoUn} |- `D }
 \end{array} \]
 \[ \begin{array}{rlcrl}
(\intR)': & \Inf	[n \geq 0]
		{ \derX P : `G |- `a{:}A_1,`D
		 \quad \ldots \quad
		 \derX P : `G |- `a{:}A_n,`D
		}
		{ \derX P : `G |- {`a{:}\AoIn,`D} }
 \end{array} \]
We have chosen for the present multiplicative-style variant that allows for the combination of contexts since this is convenient when construction derivations where it avoids additional weakening steps.
In order to not have derivations littered with applications of the \emph{Weakening} rule, we allow rules to \emph{combine} the contexts of the subterms involved; this does not exclude the normal approach, since the contexts can be equal, whereby the alternative rules become derivable.
Whenever possible, we will treat the rules $(\Bottom)$ and $(\Top)$ as empty cases of $(\unL)$ and $(\intR)$, respectively, and will use a short-hand notation for $(\intR)$ and $(\unL)$:
 \[ \begin{array}{rlcrl}
(\unL): &
\Inf	[n \geq 0]
	{ \derX P : `G_i,x{:}A_i |- `D_i
	 \quad 
	 (\forall i \ele \n)
	}
	{ \derX P : {\GamoIn, x{:}\AoUn} |- {\DeloUn} }
&\quad&
(\intR): &
\Inf	[n \geq 0]
	{ \derX P : `G_i |- `a{:}A_i,`D_i
	 \quad 
	 (\forall i \ele \n)
	}
	{ \derX P : {\GamoIn} |- {`a{:}\AoIn,\DeloUn} }
 \end{array} \]
so will allow the counter be zero as well in these rules.

As mentioned above, the system has been set up such that intersections occur on the left, and unions on the right, i.e.~each $`G$ ($`D$) is `pure', so only contains types in $\UnTypes$ ($\IntTypes$).
However, this property is obviously violated by rule $(\unL)$ (rule $(\intR)$); but, as with the intersection introduction rule $(\intI)$ in the strict system of \cite{Bakel-TCS'92,Bakel-TCS'95,Bakel-NDJFL'04}, the pureness restriction put on the other rules prohibit that types constructed via those rules can be used arbitrarily.
In fact, if a derivation ends like
\[
\Inf	[\unL]
	{ \derX P : `G_i, x{:}A_i |- `D_i
	 \quad 
	 (\forall i \ele \n)
	}
	{ \derX P : {\GamoIn,x{:}\AoUn} |- {\DeloUn} }
\]
then only the rules $(\Cut)$, $(\arrL)$ and $(\arrR)$ can be used to construct a larger derivation with this.
In the first case, the type $\AoUn$ will disappear altogether; in the other, it will become part of an arrow type, which itself is either part of a union type in the context for plugs, or part of an intersection type in the context of sockets.
In both cases, if the original contexts were pure, then the resulting contexts will be as well, so the extension of contexts by allowing union types on the left and intersection types on the right is only `temporary'.
To highlight this, consider the following:

 \begin{remark}

 \begin{itemize}
 \firstitem Since in rule $(\Ax)$ the type $A$ is proper, even the contexts $`G \inter \stat{y}{A}$ and $`a{:}A \union `D$ are pure.

 \item In rule $(\Cut)$, $A \ele \Types$, so could be either an intersection or a union.
The restriction that $`G_1, `G_2,`D_1$, and $`D_2$ are pure, however, forces this $A$ to be the \emph{only} type that disturbs the purity.
The purity is disturbed possibly in the premisses (but not in the conclusion) on the left or on the right, but not on both; also, if $A \ele \Tproper$, all contexts are pure.

 \item In rule $(\arrL)$, the type $A$ might be an intersection, or $B$ a union; notice that the contexts in the conclusion are pure.

 \item In rule $(\arrR)$, the type $A$ might be a union, or $B$ an intersection; notice that the contexts in the conclusion are pure.

 \item In rules $(\intR)$ and $(\unL)$, the introduced type is the only one disturbing the purity of the contexts.

 \end{itemize}

So, if an \emph{impure} type is introduced, this can happen only one-at-a-time.
Also, any next step in the derivation will have to remove the impure type, either via $(\Cut)$, $(\arrL)$, or $(\arrR)$.

 \end{remark}
}

 \begin{definition} [Proper derivations]
We use the auxiliary notion $\derPure P : `G |- `D $ for derivations that do not end with either rule $(\unL)$, $(\Bottom)$, $(\intR)$ or $(\Top)$.
 \end{definition}

 \begin{lemma}[Weakening]
The following rule is admissible:
 \[ \begin{array}{rlcrl}
(\Weak): &
\Inf	{ \derX P : `G |- `D }
	{ \derX P : `G \inter x{:}A |- {`a{:}B \union `D} }
 \end{array} \]
 \end{lemma}

Notice that the rules
 \[ \begin{array}{rlcrl}
(\intL): &
 \Inf	{ \derX P : `G,x{:}A |- `D }
	{ \derX P : `G,x{:}A \inter B |- `D }
&\dquad&
(\unR): &
 \Inf	
	{ \derX P : `G |- `a{:}A,`D }
	{ \derX P : `G |- `a{:}A \union B,`D }
 \end{array} \]
are admissible by the Weakening result.

We can even show the following (standard) result:

 \begin{lemma}[Thinning] \label{thinning}
The following rules are admissible:
 \[ \begin{array}{rlcrl}
(\textit{Tl}): &
\Inf	[x \notele \FS{P} ]
	{ \derX P : `G |- `D }
	{ \derX P : `G \Except x |- `D }
&\quad&
(\textit{Tr}): &
\Inf	[`a \notele \FP{P} ]
	{ \derX P : `G |- `D }
	{ \derX P : `G |- `D \Except `a }
 \end{array} \]
 \end{lemma}
Applying these repeatedly gives that the following rule is admissible:
 \[ \begin{array}{rlcrl}
(\textit{T}): &
\Inf	{ \derX P : `G |- `D }
	{ \derX P : \{x{:}A \ele `G \mid x \ele \FS{P} \} |- \{`a{:}A \ele `D \mid `a \ele \FP{P} \} }
 \end{array} \]

As usual, we will now formulate a generation lemma, that links the syntax of nets to assignable statements; the property holds for all notions of context assignment defined in this paper.

 \begin{lemma}[Proper generation Lemma] \label{pure generation lemma} \label{pure Gen lemma}

 \begin{enumerate}

 \firstitem If $\derPure \caps<x,`a> : `G |- `D $, then there exists $A,B$ such that $x{:}A \in `G$, $`a{:}B \ele `D$, and $A \seq B$.

 \item If $\derPure \exp x P `a . `b : `G |- `D $, then there exists $`b{:}C \ele `D$ with $C \ele \Types$ and $A, B$ such that $\derX P : `G' ,x{:}A |- `a{:}B,`D' $ and $A\arr B \seqs C$.

 \item If $\derPure \imp { P } `a [x] y { Q } : `G |- `D $, then there exists $x{:}C \ele `G$ with $C \ele \Types$, and $A, B$ such that $\derX P : `G |- `a{:}A,`D $, $\derPure Q : `G,x{:}B |- `D $ and $C \seq A\arr B$.

 \item \label{generation cut}
 If $\derPure \cut P `a + x Q : `G |- `D $, then there exists $A$ such that either $\derPure P : `G |- `a{:}A,`D $ and $\derX Q : `G,x{:}A |- `D $ or $\derX P : `G |- `a{:}A,`D $ and $\derPure Q : `G,x{:}A |- `D $ or $\derPure P : `G |- `a{:}A,`D $ and $\derPure Q : `G,x{:}A |- `D $.

 \end{enumerate}
 \end{lemma}

 \begin{proof}
Straightforward. \QED
 \end{proof}

In fact, by Lemma~\ref{pure generation lemma}\ref{generation cut}, the derivation for a cut is shaped like either:
 \[
 \Inf	[\Cut]
	{ \Inf	[\intR]
		{ \InfBox{ \derPure P : `G |- `a{:}C_i,`D }
		  \quad (\forall i \ele \n)
		}
		{ \derX P : `G |- `a{:}\CoI{n},`D }
	  \dquad
	  \InfBox{ \derPure Q : `G,x{:}\CoI{n} |- `D }
	}
	{ \derPure \cut P `a + x Q : `G |- `D }
 \]
 \[
 \Inf	[\Cut]
	{ \InfBox{\derPure P : `G  |- `a{:}\CoU{n},`D } 
	  \dquad
	  \Inf	[\unL]
		{ \InfBox{ \derPure Q : `G,x{:}C_i |- `D }
		  \quad (\forall i \ele \n)
		}
		{ \derX Q : `G,x{:}\CoU{n} |- `D }
	}
	{ \derPure \cut P `a + x Q : `G |- `D }
 \]
 \[
 \Inf	[\Cut),(A\ele \Tproper]
	{ \InfBox{ \derPure P : `G  |- `a{:}C,`D }
	  \dquad
	  \InfBox{ \derPure Q : `G,x{:}C |- `D }
	}
	{ \derPure \cut P `a + x Q : `G |- `D }
 \]

\Comment{
 \begin{lemma}[Generation Lemma] \label{Gen Lemma}
If $\derX P : `G |- `D $, then either
 \begin{enumerate}
 \item $\derPure P : `G |- `D $, or
 \item the derivation ends with $(\unL)$ and there exists $y,n$ such that, for every $i \ele \n$ there exist $A_i, `G_i,`D_i$ such that $\derPure P : `G_i, y{:}A_i |- `D_i $, and $`G = \inter_{n}`G_i, y{:}\AoUn$, and $`D = \union_{n}`D_i$, or
 \item the derivation ends with $(\intR)$ and there exists $`a,n$ such that, for every $i \ele \n$ there exist $A_i, `G_i,`D_i$ such that $\derPure P : `G_i |- `a{:}A_i,`D_i $, and $`G = \inter_{n}`G_i$, and $`D = `a{:}\AoIn, \union_{n}`D_i$.
 \end{enumerate}
 \end{lemma}

 \begin{proof}
Easy. \QED
 \end{proof}
}

 \begin{example}
The following results are pure derivations; notice that, since $\Bottom$ and $\Top$ are proper types, they can be used on the right and on the left, respectively, without disturbing the purity of contexts.
 \[
 \Inf	[\arrR]
	{ \Inf	
		{ \Inf	[\arrL]
			{ \Inf	[\Bottom]
				{ \derX \caps<x,`a> : {x{:}\Bottom} |- {`a{:}\Bottom} }
			 \Inf	
				{ \derPure \caps<z,`b> : z{:}A |- `b{:}A }
			}
			{ \derPure \imp \caps<x,`a> `a [y] z \caps<z,`b> : x{:}\Bottom, y{:}\Bottom\arr A |- `b{:}A }
		}
		{ \derPure \exp x {\imp \caps<x,`a> `a [y] z \caps<z,`b> } `b . `g : y{:}\Bottom\arr A |- {`g{:}\Bottom\arr A} }
	}
	{ \derPure \exp y {\exp x {\imp \caps<x,`a> `a [y] z \caps<z,`b> } `b . `g } `g . `d  : {} |- {`d{:}(\Bottom\arr A)\arr \Bottom\arr A} }
\] 
\[
\Inf	[\arrR]
	{ \Inf	[\arrR]
		{ \Inf	
			{ \derPure \caps<x,`g> : {x{:}A, y{:}\Top} |- `g{:}A }
		}
		{ \derPure \exp y \caps<x,`g> `g . `b : x{:}A |- `b{:}\Top\arr A }
	}
	{ \derPure \exp x { \exp y \caps<x,`g> `g . `b } `b . `a : {\emptyset} |- `a{:}A\arr \Top\arr A }
 \]
Notice that these nets correspond to $\SemL{`l yx.yx}{`d}$ and $\SemL{`l xy.x}{`a}$, respectively.
 \end{example}

We now show that our notion of context assignment is closed for \emph{renaming cuts}:

 \begin{lemma} \label{closed for right renaming}
 \begin{enumerate}
 \firstitem \label{activated part}
If $\derX \cutL P `a + x \caps<x,`b> : `G |- `D $, then $\derX P[`b/`a] : `G |- `D $.
 \item \label{unactivated part}
If $\derX \cut P `a + x \caps<x,`b> : `G |- `D $, then $\derX P[`b/`a] : `G |- `D $.
 \end{enumerate}
 \end{lemma}

 \begin{proof}

Simultaneously by coinduction on reduction paths, where we focus on the structure of $P$.

 \begin{enumerate}
 \item 

We only show two interesting parts.

 \begin{widedescription}

\Comment{
 \item[$P= \caps<y,`a>$] Then $\cutL \caps<y,`a> `a + x \caps<x,`b> \red \cut  \caps<y,`a> `a + x \caps<x,`b> $.
Immediate, since active and inactive cuts are typed with the same rules.

 \item[$P= \caps<y,`g>$, with $`g \not= `a$]
Then $ \cutL \caps<y,`g> `a + x \caps<x,`b> \red \caps<y,`b> $.

 \begin{wideitemize}

\item[\ref{intersection case}] \correct 
 \[ 
 \Inf	
	{ \Inf	[\intI]
		{ \Inf	
			{ \derPure \caps<y,`g> : `G\inter y{:}A |- `a{:}C_i,`g{:}A\union `D }
		  \quad (\forall i \ele \n)
		}
		{ \derX \caps<y,`g> : `G\inter y{:}A |- `a{:}\CoI{n},`g{:}A\union `D }
	  \quad
	  \Inf{ \derX \caps<x,`b> : `G,x{:}{\CoIn} |- `b{:}C_j\union`D }
	}
	{ \derPure \cutL \caps<y,`g> `a + z \caps<x,`b> : `G\inter y{:}A |- `g{:}A\,\union\,`b{:}C_j\,\union\, `D }
 \]
 \[
\Inf	{ \derPure \caps<y,`b> : `G\inter y{:}A |- `g{:}A\,\union\,`b{:}C_j\,\union\, `D }
 \]

\item[\ref{union case}] \correct
 \[
\Inf	
	{ \Inf	
		{ \derXN \caps<y,`g> : `G\inter y{:}A |- `a{:}\CoUn,`g{:}A\union `D }
	  \quad
	  \Inf	[\unL]
		{ \Inf{ \derXN \caps<x,`b> : `G,x{:}C_i |- `b{:}C_i\union `D }
		  \quad (\forall i \ele \n)
		}
		{ \derXN \caps<x,`b> : `G,x{:}{\CoUn} |- `b{:}\CoUn\, \union `D }
	}
	{ \derXN \cutL \caps<y,`g> `a + x \caps<x,`b> : `G\inter y{:}A |- `g{:}A\,\union\,`b{:}\CoUn\,\union `D }
 \]
 \[
\Inf	
	{ \derXN \caps<y,`g> : `G\inter y{:}A |- `g{:}A\,\union\, `g{:}\CoUn\,\union `D }
 \]

\item[\ref{pure case}] \correct
 \[
\Inf	
	{ \Inf	
		{ \derXN \caps<y,`g> : `G\inter y{:}A |- `a{:}C,`g{:}A\union `D }
	  \quad
	  \Inf	{ \derXN \caps<x,`b> : `G,x{:}C |- `b{:}C\union `D }
	}
	{ \derXN \cutL \caps<y,`g> `a + x \caps<x,`b> : `G\inter y{:}A |- `g{:}A\,\union\,`b{:}C\,\union `D }
 \]
 \[
\Inf	
	{ \derXN \caps<y,`g> : `G\inter y{:}A |- `g{:}A\,\union\,`b{:}C\,\union `D }
 \]

 \end{wideitemize}

}

\item[$P = \exp y Q `d . `a $]
Then $\cutL {\exp y Q `d . `a } `a + x \caps<x,`b> \red \cut { \exp x {\cutL Q `a + x \caps<x,`b> } `d . `g } `g + x \caps<x,`b> $. 

 \begin{wideitemize}

\item[\ref{intersection case}] 
Suppose the derivation for $\exp y Q `d . `a $ ends with $(\intR)$:
 \[ \kern-5mm
 \Inf	
	{
	 \Inf	[\intR]
		{ \Inf	
			{ \InfBox{\D^i_1}
				{ \derX Q : `G,y{:}A_i |- `d{:}B_i,`a{:}A_i\arr B_i,`D } 
			}
			{ \derPure \exp y Q `d . `a : `G |- `a{:}A_i\arr B_i,`D } 
		  (\forall i \ele \n)
		}
		{ \derX \exp y Q `d . `a : `G |- `a{:}\int{n}(A_i\arr B_i),`D } 
	  ~
	  \Inf	{ \derPure \caps<x,`b> : `G,x{:}\int{n}(A_i\arr B_i) |- `b{:}A_j\arr B_j\,\union`D }
	}
	{ \derPure \cutL {\exp y Q `d . `a } `a + x \caps<x,`b> : `G |- `b{:}A_j\arr B_j\,\union`D }
 \]
By induction, we have $\derX Q[`g/`a] : `G,y{:}A_i |- `d{:}B_i,`g{:}A_i\arr B_i,`D $, for all $i \ele \n$, and we can construct:
 \[  \kern30mm
 \Inf	
	{
\multiput(10,-2)(0,7){2}{.}
\raise.5\RuleH\hbox to 4cm{\kern-35mm  
	 \Inf	[\intR]
		{ \Inf	
			{ \InfBox{\D^i_1}
				{ \derX Q[`g/`a] : `G,y{:}A_i |- `d{:}B_i,`g{:}A_i\arr B_i,`D } 
			}
			{ \derPure \exp y \,Q[`g/`a]\, `d . `g : `G |- `g{:}A_i\arr B_i,`D } 
		  (\forall i \ele \n)
		}
		{ \derX \exp y \,Q[`g/`a]\, `d . `g : `G |- `g{:}\int{n}(A_i\arr B_i),`D } 
}
	  \Inf	{ \derPure \caps<x,`b> : `G,x{:}\int{n}(A_i\arr B_i) |- `b{:}A_j\arr B_j\,\union`D }
	}
	{ \derPure \cut {\exp y \,Q[`g/`a]\, `d . `g } `g + x \caps<x,`b> : `G |- `b{:}A_j\arr B_j\,\union`D }
 \]
(Notice that, if $`a$ does not occur free in $Q$, we can add $`g{:}A_i\arr B_i$ by weakening.)

The proof finishes by part \ref{unactivated part} with $\derX  (\exp y \,Q[`g/`a]\, `d . `g )[`b/`g] : `G |- `b{:}A_j\arr B_j\,\union`D $ which we can abbreviate by $ \derX (\exp y Q `d . `a )[`b/`a] : `G |- `b{:}A_j\arr B_j\,\union`D $.

The proof is similar for the case that the right-hand side finished with $(\unL)$, or when both are pure.

 \end{wideitemize}

\item[$P = \cutL P_1 `g + y P_2 $]
In order to run the outermost cut in $ \cutL {\cutL P_1 `g + y P_2 } `a + x \caps<x,`b> $, first we need to propagate the innermost first.
For example, let $P_1 = \exp z P'_1 `r . `s $ (assuming $`s \not=`a$), then 
 \[ \begin{array}{ll}
\cutL { \cutL { \exp z P'_1 `r . `s } `g + y P_2 } `a + x \caps<x,`b> & \red 
\\
\cutL { \exp z { \cutL P'_1 `g + y P_2 } `r . `s } `a + x \caps<x,`b> & \red 
\\
\exp z { \cutL{ \cutL P'_1 `g + y P_2 } `a + x \caps<x,`b> } `r . `s 
 \end{array} \]
and, by induction $ \derX \exp z { \cutL P'_1 `g + y P_2 }[`b/`a]\, `r . `s : `G |- `D $, so also $ \derX ( \exp z { \cutL P'_1 `g + y P_2 } `r . `s ) [`b/`a] : `G |- `D $.
Now, since $ ( \cutL { \exp z P'_1 `r . `s } `g + y P_2 ) [`b/`a] \red ( \exp z { \cutL P'_1 `g + y P_2 } `r . `s ) [`b/`a] $, by Theorem~\ref{witness expansion}, we also have $ \derX ( \cutL { \exp z P'_1 `r . `s } `g + y P_2 ) [`b/`a] : `G |- `D $.

\newpage
All other cases are dealt with similarly.

 \end{widedescription}

 \item
We distinguish the following cases (the first two consider when $`a$ is introduced in $P$):

 \begin{widedescription}
%
%
 \item[$P = \caps<y,`a>$] Then $\cut \caps<y,`a> `a + x \caps<x,`b> \red \caps<y,`b>$.
Easy, considering the three possibilities for the derivation for the cut.
\Comment{
We consider the three cases:

\begin{wideitemize}
\item [\ref{intersection case}] \correct
\[  
 \Inf	
	{ \Inf	[\intR]
		{ \Inf	{ \derPure \caps<y,`a> : `G\inter y{:}C_i |- `a{:}C_i,`D }
		  \quad 
		  (\forall i \ele \n)
		}
		{ \derX \caps<y,`a> : `G\inter y{:}\CoI{n} |- `a{:}\CoI{n},`D }
	  \Inf	{ \derPure \caps<x,`b> : `G,x{:}\CoI{n} |- `b{:}C_j\union `D }
	}
	{ \derPure \cut \caps<y,`a> `a + x \caps<x,`b> : `G\inter y{:}\CoI{n} |- `b{:}C_j\union `D }
 \]
Notice that $C_i \ele \Tproper$ for $\iotn$; by rule $(\Ax)$, we can derive:
 \[ 
\Inf	{ \derPure \caps<y,`b>  : `G\inter y{:}\CoI{n} |- `b{:}C_j\union `D }
 \]

\item[\ref{union case}] \correct
\[  
 \Inf	
	{ \Inf	{ \derPure \caps<y,`a> : `G\inter y{:}C_j |- `a{:}\CoU{n},`D }
	  \Inf	[\unL]
		{ \Inf	{ \derPure \caps<x,`b> : `G,x{:}C_i |- `b{:}C_i\union `D }
		  \quad 
		  (\forall i \ele \n)
		}
		{ \derX \caps<x,`b> : `G,x{:}\CoU{n} |- `b{:}\CoU{n}\union `D }
	}
	{ \derPure \cut \caps<y,`a> `a + x \caps<x,`b> : `G\inter y{:}C_j |- `b{:}\CoU{n}\union `D }
 \]
Notice that $C_i \ele \Tproper$ for $\iotn$; by rule $(\Ax)$, we can derive:
 \[ 
\Inf	{ \derPure \caps<y,`b>  : `G\inter y{:}C_j |- `b{:}\CoU{n}\union `D }
 \]

\item [\ref{pure case}] \correct
Easy, \Comment{
\[  
 \Inf	
	{ \Inf	{ \derPure \caps<y,`a> : `G\inter y{:}A |- `a{:}A,`D }
	  \quad
	  \Inf	{ \derPure \caps<x,`b> : `G,x{:}A |- `b{:}A\union `D }
	}
	{ \derPure \cut \caps<y,`a> `a + x \caps<x,`b> : `G\inter y{:}A |- `b{:}A\union `D }
 \]
Then $A \ele \Tproper$, and we can derive:
 \[
\Inf	{ \derPure \caps<y,`b> : `G\inter y{:}A |- `b{:}A\union `D }
 \]
}

\end{wideitemize}

}


 \item[$P =  \exp y P' `g . `a $, with $`a$ not free in $P'$] 
 Then $ \cut { \exp y P' `g . `a } `a + x \caps<x,`b> \red \exp y P' `g . `b $; we can assume that $`a$ does not occur in $`D$, but cannot assume that for $`b$.
By Lemma~\ref{pure generation lemma}\ref{generation cut}, we have three cases: 

\begin{wideitemize}

 \item[\ref{intersection case}] \correct
 \[ 
\Inf	
	{ \Inf	[\intR]
		{ \Inf	
			{ \InfBox{\D_i}
				{ \derX P' : `G,y{:}A_i |- `g{:}B_i,`D } 
			}
			{ \derPure \exp y P' `g . `a : `G |- `a{:}A_i\arr B_i,`D }
		  \quad (\forall i \ele \n)
		}
		{ \derX \exp y P' `g . `a : `G |- `a{:}\int{i}(A_i\arr B_i),`D }
	  \Inf	
		{ \derPure \caps<x,`b> : `G,x{:}\int{i}(A_i\arr B_i) |- `b{:}A_j\arr B_j\,\union `D }
	}
	{ \derPure \cut { \exp y P' `g . `a } `a + x \caps<x,`b> : `G |- `b{:}A_j\arr B_j\,\union `D }
 \]
For the right-hand side we can construct:
 \[ 
 \Inf	
	{ \InfBox{\D_j}
		{ \derX P' : `G,y{:}A_j |- `g{:}B_j,`D } 
	}
	{ \derPure \exp y P' `g . `b : `G |- `b{:}(A_j\arr B_j)\,\union `D }
 \]

 \item[\ref{union case}] Let $C = \CoU{n}$.
 \[ \kern-4.7cm 
\Inf	
	{
	  \Inf	[\Weak]
	  	{ \Inf	
			{ \InfBox{\D}
				{ \derX P' : `G,y{:}A |- `g{:}B,`D } 
			}
			{ \derX \exp y P' `g . `a : `G |- `a{:}A\arr B,`D }
		}
		{ \derX \exp y P' `g . `a : `G |- `a{:}(A\arr B)\union C,`D }
\raise1.4\RuleH\hbox to 5cm{\kern-17mm  
	  \Inf	[\unL]
		{ \Inf	{(\forall i \ele \n)}
			{ \derPure \caps<x,`g>  : `G,x{:}C_i |- `b{:}C_i\union `D }
		  \Inf	{\derPure \caps<x,`g>  : `G,x{:}A\arr B |- `b{:}A\arr B\union `D } 
		}
		{ \derX \caps<x,`g> : `G,x{:}(A\arr B)\union C |- `b{:}(A\arr B)\union C\, \union `D }
}
\multiput(-10,-1)(0,7){4}{.}
	}
	{ \derPure \cut {\exp y P' `g . `a } `a + x \caps<x,`b> : `G |-  `b{:}(A\arr B)\union C\, \union `D }
 \]
Notice, since $`a \not\in \fp(P')$, $`a{:}C$ is added by weakening.
Then for the right-hand side we can construct:
 \[ 
\Inf	[\Weak]
	{ \Inf	
		{ \InfBox{\D}
			{ \derX P' : `G,y{:}A |- `g{:}B,`D } 
		}
		{ \derPure \exp y P' `g . `b : `G |- `b{:}(A\arr B)\union `D }
	}
	{ \derPure \exp y P' `g . `b : `G |- `b{:}(A\arr B)\union C\, \union `D }
 \]

\item[\ref{pure case}] \correct
\[ \Inf	
	{ \Inf	
		{ \InfBox{\D}
			{ \derX P' : `G,y{:}A |- `g{:}B,`D } 
		}
		{ \derPure \exp y P' `g . `a : `G |- `a{:}A\arr B,`D }
	  \Inf	
		{ \derPure \caps<x,`b> : `G,x{:}A\arr B |- `b{:}A\arr B\union `D }
	}
	{ \derPure \cut { \exp y P' `g . `a } `a + x \caps<x,`b> : `G |- `b{:}A\arr B\union `D }
 \]
Notice that then
 \[ 
\Inf	[\arrR]
	{ \InfBox{\D}
		{ \derX P' : `G,y{:}A |- `g{:}B,`D } 
	}
	{ \derX \exp y P' `g . `b : `G |- `b{:}A\arr B\union `D }
 \]

\end{wideitemize}


\item[$P$ does not introduce $`a$]
Then $\cut P `a + x \caps<x,`b> \red \cutL P `a + x \caps<x,`b> $, and the proof follows by part \ref{activated part}

\end{widedescription}

\end{enumerate}

 \end{proof}

Similarly (dually), we can prove:

 \begin{lemma} \label{closed for left renaming}
 \begin{enumerate}
 \item 
If $\derX \cutR \caps<y,`a> `a + x Q : `G |- `D $, then $\derX Q[y/x] : `G[y/x] |- `D $.
 \item 
If $\derX \cut \caps<y,`a> `a + x Q : `G |- `D $, then $\derX Q[y/x] : `G[y/x] |- `D $.
 \end{enumerate}
 \end{lemma}

We will now show that we can simulate intersection and union elimination.

 \begin{lemma} \label{intersection elimination}
 \begin{enumerate}
 \item
If $\derX P : `G |- `a{:}A\inter B,`D $ then $\derX P : `G |- `a{:}A,`D $ and $\derX P : `G |- `a{:}B,`D $.
 \item
If $\derX P : `G,x{:}A\union B |- `D $ then $\derX P : `G,x{:}A |- `D $ and $\derX P : `G,x{:}B |- `D $.
 \end{enumerate}
 \end{lemma}
 \begin{proof}
Assume, without loss of generality, that $A,B \ele \Tproper$.

 \begin{enumerate}

 \item
Let $`b$ be fresh. Notice that we can construct: 
\[
\Inf	{ \derX P : `G |- `a{:}A\inter B,`D
	  \quad
	  \Inf { \derX \caps<x,`b> : x{:}A\inter B |- `b{:}A }
	}
	{  \derX \cut P `a + x \caps<x,`b> : `G |- `b{:}A,`D }
\]
Then, by Lemma~\ref{closed for right renaming}, we have $ \derX P[`b/`a] : `G |- `b{:}A,`D $, so also $ \derX P : `G |- `a{:}A,`D $.

 \item
Let $y$ be fresh. Notice that we can construct: 
\[
\Inf	{ \Inf { \derX \caps<y,`a> : y{:}A |- `a{:}A\union B }
	  \quad
	  \derX Q : `G,x{:}A\union B |- `D
	}
	{  \derX \cut \caps<y,`a> `a + x Q : `G,y{:}A |- `D }
\]
Then, by Lemma~\ref{closed for left renaming}, we have $ \derX Q [y/x] : `G,y{:}A |- `D $, so also $ \derX Q : `G,x{:}A |- `D $.

 \end{enumerate}
\end{proof}

\def \Reduct{\cut { \exp y \caps<y,`m> `m . `g } `g + x { \cut {
\cut \caps<x,`g> `g + z { \imp \caps<x,`d> `d [z] w \caps<w,`e> }
} `e + v \caps<v,`a> } }


\Comment{
 \begin{example} \label{example delta i typed}
We will show how the net $\SemL{(`l x.xx)(`l y.y)}{`a}$
can be typed, using the rules above, and type some of the reducts as well.
Notice that
 \[ \begin{array}{rcl}
\SemL{(`l x.xx)(`l y.y)}{`a} &=&
\cut	{ \exp	x { \cut \caps<x,`g> `g + z { \imp \caps<x,`d> `d [z] w \caps<w,`e> } } `e . `b } `b + z { \imp { \exp y \caps<y,`m> `m . `g } `g [z] v \caps<v,`a> }
 \end{array} \]
We start with $ \derX \SemL{`l x.xx}{`b} : {\emptyset} |- {`b{:}((C\arr C)\arr C\arr C) \inter (C\arr C)\arr C\arr C} $ (notice that this net is exactly $\exp x { \cut \caps<x,`g> `g + z { \imp \caps<x,`d> `d [z] w \caps<w,`e> } } `e . `b $).
 \[
\Inf	[\arrR]
	{ \Inf	[\Cut]
		{ \Inf	
			{ \derX \caps<x,`g> : x{:}(C\arr C)\arr C\arr C |- `g{:}(C\arr C)\arr C\arr C }
		 \quad
\raise1.25\RuleH\hbox to 3cm{\kern -5cm
		 \Inf [\arrL]
			{ \Inf	
				{ \derX \caps<x,`d> : x{:}C\arr C |- `d{:}C\arr C }
			 \quad
			 \Inf	
				{ \derX \caps<w,`e> : w{:}C\arr C |- `e{:}C\arr C }
			}
			{ \derX \imp \caps<x,`d> `d [z] w \caps<w,`e> : x{:}C\arr C,z{:}(C\arr C)\arr C\arr C |- `e{:}C\arr C }
}
\multiput(-40,-1)(0,7){4}{.}
		}
		{ \derX \cut \caps<x,`g> `g + z {\imp \caps<x,`d> `d [z] w \caps<w,`e> } : x{:}((C\arr C)\arr C\arr C) \inter (C\arr C) |- `e{:}C\arr C }
	}
	{ \derX \exp x { \cut \caps<x,`g> `g + z { \imp \caps<x,`d> `d [z] w \caps<w,`e> } } `e . `b : {\emptyset} |- { `b{:}((C\arr C)\arr C\arr C) \inter (C\arr C)\rightarrow C\arr C} }
 \]
Now we derive $ \derX \imp { \exp y \caps<y,`m> `m . `g } `g [z] v \caps<v,`a> : z{:}((C\arr C)\arr C\arr C) \inter (C\arr C)\rightarrow C\arr C |- {\emptyset} $.
 \[ \kern1cm
\Inf	[\arrL]
	{
\multiput(60,-1)(0,7){2}{.}
\raise.75\RuleH\hbox to 7cm{\kern -2cm
	 \Inf	[\intR]
		{ \Inf	[\arrR]
			{ \Inf	
				{ \derX \caps<y,`m> : y{:}C\arr C |- `m{:}C\arr C }
			}
			{ \derX \exp y \caps<y,`m> `m . `g : {\emptyset} |- `g{:}(C\arr C)\arr C\arr C }
		 \Inf	[\arrR]
			{ \Inf	
				{ \derX \caps<y,`m> : y{:}C |- `m{:}C }
			}
			{ \derX \exp y \caps<y,`m> `m . `g : {\emptyset} |- `g{:}C\arr C }
		}
		{ \kern-2cm \derX \exp y \caps<y,`m> `m . `g : {\emptyset} |- `g{:}((C\arr C)\arr C\arr C) \inter (C\arr C) } 
}
	 \quad
	 \Inf	
		{\derX \caps<v,`a> : v{:}C\arr C |- `a{:}C\arr C }
	}
	{ \derX \imp { \exp y \caps<y,`m> `m . `g } `g [z] v \caps<v,`a> : z{:}((C\arr C)\arr C\arr C) \inter (C\arr C)\rightarrow C\arr C |- `a{:}C\arr C }
 \]

The result $ \derX \SemL{(`l x.xx)(`l y.y)}{`a} : {\emptyset} |- `a{:}C\arr C $ now follows by applying the rule $(\Cut)$.
Notice that, although the cut is typed using a proper type, in the second sub-derivation the plug $`g$ carries a true intersection type; this is necessary to be able to express that the net that will be connected as a mediator to the socket $z$ has the appropriate type.

To derive $ \derX {\Reduct} : {\emptyset} |- `a{:}C\arr C $, we first derive:
\Comment{
 \[
\Inf	[\intR]
	{ \Inf	
		{ \Inf	
			{ \derX \caps<y,`m> : y{:}C\arr C |- `m{:}C\arr C } }
		{ \derX \exp y \caps<y,`m> `m . `g : {\emptyset} |- `g{:}(C\arr C)\arr C\arr C }
	 \Inf	
		{ \Inf	
			{ \derX \caps<y,`m> : y{:}C |- `m{:}C }
		}
		{ \derX \exp y \caps<y,`m> `m . `g : {\emptyset} |- `g{:}C\arr C }
	}
	{ \derX \exp y \caps<y,`m> `m . `g : {\emptyset} |- `g{:}((C\arr C)\arr C\arr C) \inter C\arr C }
 \]
}
{
 \[
\Inf	
	{ \Inf	
		{ \Inf	
			{ \derX \caps<x,`g> : x{:}(C\arr C)\arr C\arr C |- `g{:}(C\arr C)\arr C\arr C }
\multiput(40,-1)(0,7){4}{\makebox(0,0){.}}
\raise1.25\RuleH\hbox to 3cm{\kern -4cm
		 \Inf [\arrL]
			{ \Inf	
				{ \derX \caps<x,`d> : x{:}C\arr C |- `d{:}C\arr C }
			 \Inf	
				{ \derX \caps<w,`e> : w{:}C\arr C |- `e{:}C\arr C }
			}
			{ \derX \imp \caps<x,`d> `d [z] w \caps<w,`e> : x{:}C\arr C,z{:}(C\arr C)\arr C\arr C |- `e{:}C\arr C }
}
		}
		{ \derX \cut \caps<x,`g> `g + z { \imp \caps<x,`d> `d [z] w \caps<w,`e> } : x{:}((C\arr C)\arr C\arr C) \inter (C\arr C) |- `e{:}C\arr C }
\multiput(40,-1)(0,7){4}{\makebox(0,0){.}}
\raise1.15\RuleH\hbox to 3cm{\kern -1.5cm
	 \Inf	{\derX \caps<v,`a> : v{:}C\arr C |- `a{:}C\arr C }
}
	}
	{ \derX \cut { \cut \caps<x,`g> `g + z { \imp \caps<x,`d> `d [z] w \caps<w,`e> } } `e + v \caps<v,`a> : x{:}((C\arr C)\arr C\arr C) \inter (C\arr C) |- `a{:}C\arr C }
 \]
}

Notice that these derivations are (naturally) composed out of subderivations of the first construction.
Again, applying rule $(\Cut)$ to $\D$, we get the result.

Now the latter net $\cut { \cut \caps<x,`g> `g + z { \imp \caps<x,`d> `d [z] w \caps<w,`e> } } `e + v \caps<v,`a> $ reduces in two (renaming) steps to the net $\imp \caps<x,`d> `d [x] w \caps<w,`a> $, and we obtain:
{
 \[
\Inf	
	{ \Inf	[\intR]
		{ \Inf	
			{ \Inf	
				{ \derX \caps<y,`m> : y{:}C |- `m{:}C }
			}
			{ \derX \exp y \caps<y,`m> `m . `g : {\emptyset} |- `g{:}C\arr C }
		 \kern-2mm
		 \Inf	
			{ \Inf	
				{ \derX \caps<y,`m> : y{:}C\arr C |- `m{:}C\arr C }
			}
			{ \derX \exp y \caps<y,`m> `m . `g : {\emptyset} |- {`g{:}(C\arr C)\arr C\arr C} }
		}
		{ \derX \exp y \caps<y,`m> `m . `g : {\emptyset} |- {`g{:}(C\arr C)\inter ((C\arr C)\arr C\arr C)} }
\multiput(40,-1)(0,7){10}{.}
\raise3.25\RuleH\hbox to 3cm{\kern -7cm
	 \Inf	
		{ \Inf	
			{ \derX \caps<x,`d> : x{:}C\arr C |- `d{:}C\arr C }
		 \Inf	
			{ \derX \caps<w,`a> : w{:}C\arr C |- `a{:}C\arr C }
		}
		{ \derX \imp \caps<x,`d> `d [x] w \caps<w,`a> : x{:}(C\arr C) \inter ((C\arr C)\arr C\arr C) |- `a{:}C\arr C }
}
	}
	{ \derX \cut { \exp y \caps<y,`m> `m . `g } `g + x { \imp \caps<x,`d> `d [x] w \caps<w,`a> } : {\emptyset} |- `a{:}C\arr C }
 \]
}

This now reduces to
 \[
\cut { \exp y \caps<y,`m> `m . `g } `g + z { \imp {
\cutR { \exp y \caps<y,`m> `m . `g } `g + x \caps<x,`d>
 } `d [z] w {
\cutR { \exp y \caps<y,`m> `m . `g } `g + x \caps<w,`a>
} }
\]
which is typeable as follows (to save space, we abbreviate $\exp y \caps<y,`m> `m . `g $ by $\SemL{`l y.y}{`g}$; notice the role of the type $\Top$ here, and how the derivation for the intersection type for $\SemL{`l y.y}{`g}$ from the derivation above distributes)
{
 \[ \kern-5mm
\Inf	
	{
\multiput(20,-1)(0,7){13}{.}
\raise4.25\RuleH\hbox to 1.5cm{
	 \Inf	
		{ \Inf	{ \derX \caps<y,`m> : y{:}C\arr C |- `m{:}C\arr C } }
		{ \derX \SemL{`l y.y}{`g} : {\emptyset} |- `g{:}(C\arr C)\arr C\arr C }
}
	 \Inf	
		{ \Inf	
			{ \Inf	
				{ \Inf	{ \derX \caps<y,`m> : y{:}C |- `m{:}C } }
				{\derX \SemL{`l y.y}{`g} : {\emptyset} |- `g{:}C\arr C }
			 \Inf { \derX \caps<x,`d> : x{:}C\arr C |- `d{:}C\arr C }
			}
			{ \derX \cutR \SemL{`l y.y}{`g} `g + x \caps<x,`d> : {\emptyset} |- `d{:}C\arr C }
\raise2.25\RuleH\hbox to 2cm{\kern-5cm
		 \Inf	
			{ \Inf	{ \derX \SemL{`l y.y}{`g} : {\emptyset} |- {`g{:}\Top} }
			 \Inf	{ \derX \caps<w,`a> : {w{:}C\arr C, x{:}\Top} |- `a{:}C\arr C }
			}
			{ \derX \cutR \SemL{`l y.y}{`g} `g + x \caps<w,`a> : w{:}C\arr C |- `a{: C\arr C } }
}
\multiput(0,-1)(0,7){7}{.}
		}
		{ \derX \imp { \cutR \SemL{`l y.y}{`g} `g + x \caps<x,`d> } `d [z] w { \cutR \SemL{`l y.y}{`g} `g + x \caps<w,`a> } : z{:}(C\arr C)\arr C\arr C |- `a{:}C\arr C }
	}
	{ \derX \cut \SemL{`l y.y}{`g} `g + z { \imp { \cutR \SemL{`l y.y}{`g} `g + x \caps<x,`d> } `d [z] w { \cutR \SemL{`l y.y}{`g} `g + x \caps<w,`a> } } : {\emptyset} |- `a{:}C\arr C }
 \]
}

This reduces to $\cut { \exp y \caps<y,`m> `m . `g } `g + z { \imp { \exp y \caps<y,`m> `m . `d } `d [z] w \caps<w,`a> } $, which we can type as follows:
{
 \[
\Inf	
	{ \Inf	
		{ \Inf	{ \derX \caps<y,`m> : y{:}C\arr C |- {`m{:}C\arr C} } }
		{ \derX \exp y \caps<y,`m> `m . `g : {\emptyset} |- {`g{:}(C\arr C)\arr C\arr C} }	 \Inf	
			{ \Inf	
				{ \Inf	{ \derX \caps<y,`m> : y{:}C |- `m{:}C } }
			{\derX \exp y \caps<y,`m> `m . `d : {\emptyset} |- {`d{:}C\arr C} }
		 \Inf	{ \derX \caps<w,`a> : w{:}C\arr C |- {`a{:}C\arr C} }
		}
		{ \derX \imp { \exp y \caps<y,`m> `m . `d } `d [z] w \caps<w,`a> : z{:}(C\arr C)\arr C\arr C |- {`a{:}C\arr C} }
	}
	{ \derX \cut { \exp y \caps<y,`m> `m . `g } `g + z { \imp { \exp y \caps<y,`m> `m . `d } `d [z] w \caps<w,`a> } : {\emptyset} |- {`a{:}C\arr C} }
 \]
}

This net now reduces to $\cut { \exp y \caps<y,`m> `m . `d } `d + y { \cut \caps<y,`m> `m + w \caps<w,`a> } $, which is typeable as follows:
 \[
\Inf	
	{ \Inf	
		{ \Inf	{ \derX \caps<y,`m> : y{:}C |- `m{:}C } }
		{\derX \exp y \caps<y,`m> `m . `d : {\emptyset} |- {`d{:}C\arr C} }
	 \Inf	
		{ \Inf	{ \derX \caps<y,`m> : y{:}C\arr C |- {`m{:}C\arr C} }
		  \Inf	{ \derX \caps<w,`a> : w{:}C\arr C |- {`a{:}C\arr C} }
		}
		{ \derX \cut \caps<y,`m> `m + w \caps<w,`a> : y{:}C\arr C |- {`a{:}C\arr C} }
	}
	{ \derX \cut { \exp y \caps<y,`m> `m . `d } `d + y { \cut \caps<y,`m> `m + w \caps<w,`a> } : {\emptyset} |- {`a{:}C\arr C} }
 \]

This net now reduces to $\exp y \caps<y,`m> `m . `a $, which is typeable by:
 \[
 \Inf	
	{ \Inf	{ \derX \caps<y,`m> : y{:}C |- `m{:}C } }
	{\derX \exp y \caps<y,`m> `m . `a : {\emptyset} |- {`a{:}C\arr C} }
 \]
which concludes this example.
 \end{example}
Part of the operations performed on derivations in this example will reappear in the proofs below (Theorem~\ref{witness reduction for CBN} and \ref{witness reduction for CBV}).

}

}

We can \Long{now }show that typeability is preserved by $\SemL{`.}{`a}$:

\begin{theorem}
If $\derI `G |- M : A $, then $\derX \SemL{M}{`a} : `G |- `a{:}A $.
\end{theorem}

\Long{
 \begin{proof}
By induction on the structure of derivations in $\TurnL$.
 \begin{description}

 \item [$\Ax$]
Then $M\equiv x$, and $`G = `G', x{:}A$, so $x \notin `G'$; let $A = \AoIn$.
We can construct
 \[ \begin{array}{c}
\Inf	[\intR]
	{ \Inf	
		{ \derX \SemL{x}{`a} : `G',x{:}A_i |- `a{:}A_i }
	  \quad (\forall i \ele \n)
	}
	{ \derX \SemL{x}{`a} : `G ',x{:}\AoI{n} |- `a{:}\AoI{n} }
 \end{array} \]

 \item [$\arrI$]
Then $M \equiv `lx.N$, $A = C\arr D$, and $\derI `G,x{:}C |- N : D $. Then $\D
\dcol \derX \SemL{N}{`b} : `G,x{:}C |- `b{:}D $ exists by induction, and we can construct:

 \[ \begin{array}{c}
\Inf[\arrR]
	{ \InfDerX  {\D} :: \SemL{N}{`b} : `G,x{:}C |- `b{:}D }
	{ \derX \exp x \SemL{N}{`b} `b . `a : `G |- {`a{:}C\arr D} }
 \end{array} \]
Notice that $\exp x \SemL{N}{`b} `b . `a = \SemL{`lx.N}{`a}$.
 \item [$\arrE$]
Then 
$M \equiv M_1M_2$, and there exists $B$ such that both $\derI `G |- M_1 : B\arr A $ and $\derI `G |- M_2 : B $.
By induction, both
$\D_1 \dcol \derX \SemL{M_1}{`g} : `G |- {`g{:}B\arr A} $, and
$\D_2 \dcol \derX \SemL{M_2}{`b} : `G |- `b{:}B $ exist, and we can construct:
 \[ \begin{array}{c}
 \Inf	[\Cut]
	{ \InfDerX  \D_1 :: \SemL{M_1}{`g} : `G |- {`g{:}B\arr A}
	\quad \quad
	 \Inf	[\arrL]
		{ \InfDerX  \D_2 :: \SemL{M_2}{`b} : `G |- `b{:}B
		\quad
		 \Inf	
			{ \derX \caps<y,`a> : y{:}A |- `a{:}A }
		}
		{ \derX \imp \SemL{M_2}{`b} `b [x] y \caps<y,`a> : `G,x{:}B\arr A |- `a{:}A }
	}
	{ \derX {\cut \SemL{M_1}{`g} `g + x { \imp \SemL{M_2}{`b} `b [x] y \caps<y,`a> } } : `G |- `a{:}A }
 \end{array} \]
Notice that $\SemL{M_1M_2}{`a} = \cut \SemL{M_1}{`g} `g + x { \imp \SemL{M_2}{`b} `b [x] y \caps<y,`a> } $, and that, by construction, $x,y \not\in `G$.

 \item[$\intI$]
Then $A = B\inter C$ (wlog, assume that $B,C$ are proper), and we have $\derI `G |- M : B $ and $\derI `G |- M : C $.
By induction, we have both $\derX \SemL{M}{`a} : `G |- `a{:}B $ and $\derX \SemL{M}{`a} : `G |- `a{:}C $, so, by rule $(\intR)$, also $\derX \SemL{M}{`a} : `G |- `a{:}B\inter C $.

 \item[$\intE$]
Then, wlog, we have $\derI `G |- M : A\inter B $.
By induction, $\derX \SemL{M}{`a} : `G |- `a{:}A\inter B $, so, by Lemma~\ref{intersection elimination}, also $\derX \SemL{M}{`a} : `G |- `a{:}A $.

\item [$\Top$]
Notice that $\derX \SemL{P}{`a} : `G |- {`a{:}\Top} $ by rule $(\Top)$.

 \end{description}
 \end{proof}

}


{ \def \DBwnameAndRule#1[#2]#3{#3}


\Long{\section{Preservance of types under expansion}}
\Short{\section{Witness expansion and reduction}}
 \label{WE}

One of the main properties of the intersection type assignment system is the perseverance of types under both subject reduction and subject expansion.
We will show the same results for our system for $\X$, but with restrictions.
We are able to show the witness expansion result for the notion of context assignment of Definition \ref{context assignment}, but for witness reduction, we will have to limit that notion.

\Long{In this section, we will show that witness expansion follows relatively easily.}

 \begin{theorem}[Witness expansion] \label{witness expansion}
Let $P \red Q$: if $\derPure Q : `G |- `D $ then $\derPure P : `G |- `D $.
 \end{theorem}

\Long{
 \begin{proof}

By induction on the definition of $\red$, where we focus on the rules: the proof consists of giving, for each rule $\emph{Lhs} \red \emph{Rhs}$, the 'minimal' derivation for the right-hand side, and that, using the restrictions that poses, we can type the left-hand side.
We will focus on the last derivation rule applied; whenever this is either rule $(\Top)$ or $(\Bottom)$, the result is easy.
If the derivation ends with 
 \[ \Inf	[\intR]
	{ \InfBox{\D_i}{ \derX \emph{Rhs} : `G |- `a{:}A_i,`D } \quad (\forall i \ele \n) }
	{ \derX \emph{Rhs} : `G |- `a{:}\AoI{n},`D }
 \]
each derivation $\D_i$ is pure.  Assuming we have solved the pure case, we obtain derivations $\D'_i$ for the right-hand side, and can then derive:
 \[ \Inf	[\intR]
	{ \InfBox{\D'_i}{ \derX \emph{Lhs} : `G |- `a{:}A_i,`D } \quad (\forall i \ele \n) }
	{ \derX \emph{Lhs} : `G |- `a{:}\AoI{n},`D }
 \]
A similar argument holds for a derivation ending in $(\unL)$.
Therefore, in the remainder of the proof, we focus on the pure case: the derivation for the right-hand side is pure.

\Short{We only show the interesting cases. }%
Notice the use of derivation rules $(\unL)$ in the case for $(\Lii)$ shown, and $(\intR)$ in the one for $(\Riii)$.

 \leftmarginii 0cm
 \leftmarginiii 5mm
 \begin{description} \itemsep5pt
 \item[Logical rules]

 \begin{description}


\Long{

 \item[\Ax] $\cut \caps<y,`a> `a + x \caps<x,`b> \red \caps<y,`b>$.
We can assume that $`a$ not in $`D$, and $x$ not in $`G$.
We have
 \[ 
 \Inf	{ \derPure \caps<y,`b> : `G\inter y{:}A |- `b{:}A\union `D } 
 \]
and we can construct
 \[  
 \Inf	
	{ \Inf	{ \derPure \caps<y,`a> : `G\inter y{:}A |- `a{:}A,`D }
	  \quad
	  \Inf	{ \derPure \caps<x,`b> : `G,x{:}A |- `b{:}A\union `D }
	}
	{ \derPure \cut \caps<y,`a> `a + x \caps<x,`b> : `G\inter y{:}A |- `b{:}A\union `D }
 \]


 \item[\Exp] $ \cut { \exp y P `b . `a } `a + x \caps<x,`g> \red \exp y P `b . `g $, with $`a \not\in \FP{P}$; we can assume that $`a$ does not occur in $`D$, but cannot assume that for $`g$.
We have
 \[ 
\Inf	
	{ \InfBox{\D}
		{ \derX P : `G,y{:}A |- `b{:}B,`D } 
	}
	{ \derX \exp y P `b . `g : `G |- `g{:}A\arr B\union `D }
 \]
and we can construct
 \[  
 \Inf	
	{ \Inf	
		{ \InfBox{\D}
			{ \derX P : `G,y{:}A |- `b{:}B,`D } 
		}
		{ \derX \exp y P `b . `a : `G |- `a{:}A\arr B,`D }
\quad
	  \Inf	{ \derPure \caps<x,`g> : `G,x{:}A\arr B |- `g{:}A\arr B\union `D }
	}
	{ \derPure \cut { \exp y P `b . `a } `a + x \caps<x,`g> : `G |- `g{:}A\arr B\union `D }
 \]


 \item[\Imp] $ \cut \caps<y,`a> `a + x { \imp P `b [x] z Q } \red \imp P `b [y] y Q $, with $x \not\in \FS{P,Q}$; we can assume that $x$ does not occur in $`G$, but cannot assume that for $y$. 
We have
 \[  
 \Inf	
	{ \InfBox{\D_2}
		{ \derPure P : `G |- `b{:}A,`D } 
	  \dquad 
	  \InfBox{\D_3}
		{ \derPure Q : `G,z{:}B |- `D } 
	}
	{ \derPure \imp P `b [y] z Q : `G\inter y{:}A\arr B |- `D }
 \]
and we can construct
 \[ 
\Inf	
	{ \Inf	
		{ \derPure \caps<y,`a> : `G\inter y{:}A\arr B |- `a{:}A\arr B ,`D }
	  \Inf	
		{ \InfBox{\D_2}
			{ \derPure P : `G |- `b{:}A,`D } 
		  \dquad
		  \InfBox{\D_3}
			{ \derPure Q : `G,z{:}B |- `D } 
		}
		{ \derPure \imp P `b [x] z Q : `G, x{:}A\arr B |- `D }
	}
	{ \derPure \cut \caps<y,`a> `a + x { \imp P `b [x] z Q } : `G\inter y{:}A\arr B |- `D }
 \]

}


 \item[\Ins] $ \begin{array}{rcl}
 \cut { \exp y P `b . `a } `a + x { \med Q `g [x] z R }
\red \left \{ \begin{array}{l}
	\cut { \cut Q `g + y P } `b + z R  \\
	\cut Q `g + y {\cut P `b + z R } 
\end{array} \right.
\end{array} $, with $`a \not\in \FP{R}$, $x \not\in \FS{Q,R}$.
Notice that 
 \[
\Inf	
	{ \Inf	
		{ \InfBox{\D_2}
			{\derX Q : `G |- `g{:}A,`D } 
		  \quad
		  \InfBox{\D_1}
			{ \derX P : `G,y{:}A |- `b{:}B,`D }
		}
		{ \derPure \cut Q `g + y P : `G |- `b{:}B,`D }
	  \quad
	  \InfBox{\D_3}
		{\derX R : `G,z{:}B |- `D }
	}
	{ \derPure \cut { \cut Q `g + y P } `b + z R : `G |- `D }
 \] 
and we can construct
 \[ 
\Inf	
	{ \Inf	
		{ \InfBox{\D_1}
			{\derX P : `G,y{:}A |- `b{:}B,`D } 
		}
		{ \derPure \exp y P `b . `a : `G |- `a{:}A\arr B,`D }
	  \quad
	  \Inf	
		{ \InfBox{\D_2}
			{\derX Q : `G |- `g{:}A,`D } 
		  \quad
		  \InfBox{\D_3}
			{\derX R : `G,z{:}B |- `D }
		}
		{ \derPure \imp Q `g [x] z R : `G,x{:}A\arr B |- `D }
	}
	{ \derPure \cut {\exp y P `b . `a } `a + x { \med Q `g [x] z R } : `G |- `D }
 \] 
Notice that it is irrelevant if $A$ or $B$ is an intersection or a union.
Starting from 
 \[
 \Inf	
	{ \InfBox{\D_2}
		{\derX Q : `G |- `g{:}A,`D } 
	  \quad
	  \Inf	
		{ \InfBox{\D_1}
			{ \derX P : `G,y{:}A |- `b{:}B,`D }
		  \quad
		  \InfBox{\D_3}
			{\derX R : `G,z{:}B |- `D }
		}
		{ \derPure \cut P `b + z R : `G,y{:}A |- `D }
	}
	{ \derPure \cut Q `g + y { \cut P `b + z R } : `G |- `D }
 \]
we can construct the same derivation.

 \end{description}


\Long{

 \item[Activating the cuts]
Immediate, since activated cuts are typed with the same rule as deactivated cuts.

}


 \item[Left propagation]

 \begin{description}

\Long{

 \item[\deactL] $\cutL \caps<y,`a> `a + x P \red \cut \caps<y,`a> `a + x P $.
Immediate, since activated cuts are typed with the same rule as deactivated cuts.

}


 \item[\Li]
$ \cutL \caps<y,`b> `a + x P \red \caps<y,`b> $, with $`b \not= `a$.
We have
 \[
\Inf	
	{ \derX \caps<y,`b> : `G\inter y{:}A |- `b{:}A\union `D }
 \]
and we can construct
 \[  
\Inf	
	{ \Inf	
		{ \derX \caps<y,`b> : `G\inter y{:}A |- `a{:}\Bottom,`b{:}A\union `D }
	  \quad
	  \Inf[\Bottom]{ \derX P : `G,x{:}{\Bottom} |- `D }
	}
	{ \derX \cutL \caps<y,`b> `a + x P : `G\inter y{:}A |- `b{:}A\union `D }
 \]


 \item[\Lii] $ \cutL { \exp y P `b . `a } `a + x Q \red \cut { \exp y { \cutL P `a + x Q } `b . `g } `g + x Q $, with $`g$ fresh.
Notice that $`a$ might not be introduced in $ \exp y P `b . `a $, so might appear inside $P$; also, $y,`b \not \in \FC{Q}$.
We have
 \[ 
\Inf	
	{ 
	  \Inf	
		{ \Inf	
			{ \InfBox{\D_1}
				{ \derX P : `G,y{:}A |- `b{:}B,`a{:}\CoU{n},`D } 
			  \dquad
			  \Inf	[\unL]
				{ \InfBox{\D^i_2}	
					{ \derX Q : `G,x{:}C_i |- `D }
				  \quad (\forall i \ele \n)
				}
				{ \derX Q : `G,x{:}\CoU{n} |- `D }
			}
			{ \derPure \cutL P `a + x Q : `G,y{:}A |- `b{:}B,`D }
		}
		{ \derPure \exp y { \cutL P `a + x Q } `b . `g : `G |- `g{:}A\arr B,`D }
	  \kern-2cm
	  \InfBox{\D_3}
		{ \derPure Q : `G,x{:}A\arr B |- `D }
	}
	{ \derPure \cut { \exp y { \cutL P `a + x Q } `b . `g } `g + x Q : `G |- `D }
 \]
and we can construct
 \[  
\Inf	
	{ \Inf	
		{ \InfBox{\D_1}
			{ \derX P : `G,y{:}A |- `b{:}B,`a{:}\CoU{n},`D } 
		}
		{ \derPure \exp y P `b . `a : `G |- `a{:}\CoU{n}\union (A\arr B),`D } 
	  \Inf	[\unL]
		{ \InfBox{\D^i_2}
			{ \derPure Q : `G,x{:}C_i |- `D }
		  ~ (\forall i \ele \n)
		  \quad
		  \InfBox{\D_3}
			{ \derPure Q : `G,x{:}A\arr B |- `D }
		}
		{ \derPure Q : `G,x{:}\CoU{n}\union (A\arr B) |- `D }
	}
	{ \derPure \cutL {\exp y P `b . `a } `a + x Q : `G |- `D }
 \]


\Long{
 \item[\Liii] $ \cutL { \exp y P `b . `g } `a + x Q \red \exp y { \cutL P `a + x Q } `b . `g $, with $`g \not = `a$.
We have
 \[
 \Inf	{ \Inf	{ \InfBox{\D_1}
			{ \derX P : `G,y{:}A |- `a{:}C,`b{:}B,`D } 
		  \dquad
		  \InfBox{\D_2}
			{ \derPure Q : `G,x{:}C |- `D }
		}
		{ \derPure \cutL P `a + x Q : `G,y{:}A |- `b{:}B,`D }
	}
	{ \derPure \exp y { \cutL P `a + x Q } `b . `g : `G |- `g{:}A\arr B,`D } 
 \]
and we can construct
 \[  
\Inf	{ \Inf	{ \InfBox{\D_1}
			{ \derX P : `G,y{:}A |- `a{:}C,`b{:}B,`D } 
		}
		{ \derPure \exp y P `b . `g : `G |- `a{:}C,`g{:}A\arr B\union `D }
	  \quad
	  \InfBox{\D_2}
		{ \derPure Q : `G,x{:}C |- `D }
	}
	{ \derPure \cutL { \exp y P `b . `g } `a + x Q : `G |- `g{:}A\arr B,`D } 
 \]
}


\item[\Liv] $ \cutL { \imp P `b [z] y Q } `a + z R \red \imp { \cutL P `a + x R } `b [z] y { \cutL Q `a + x R } $.
We have
 \[ \kern 2cm
 \Inf	{ 
\multiput(20,0)(0,7){4}{.}
\raise1.25\RuleH\hbox to 2cm {\kern-3cm
	  \Inf	{ \InfBox{\D_1}
			{ \derX P : `G |- `a{:}\CoU{n},`b{:}A,`D } 
		  \Inf	[\unL]
			{ \InfBox{\D^i_3}
				{ \derX R : `G,x{:}C_i |- `D } 
			  ~(\forall i \ele \n)
			}
			{ \derX R : `G,x{:}\CoU{n} |- `D } 
		}
		{ \derPure \cutL P `a + x R : `G |- `b{:}A,`D \hspace*{15mm} } 
}
	  \quad
	  \Inf	{ \InfBox{\D_2}
			{ \derX Q : `G,y{:}B |- `a{:}\un{m}{C_{n{+}j}},`D } 
		  \Inf	[\unL]
			{ \InfBox{\D^{n{+}j}_3}
				{ \derX R : `G,x{:}C_{n{+}j} |- `D } 
			  ~(\forall j \ele \m)
			}
			{ \derX R : `G,x{:}\un{m}{C_{n{+}j}} |- `D } 
		}
		{ \derPure \cutL Q `a + x R : `G,y{:}B |- `D }
	}
	{ \derPure \imp { \cutL P `a + x R } `b [z] y { \cutL Q `a + x R } : `G\inter z{:}A\arr B |- `D }
 \]
and we can construct
 \[  
\Inf	{ \Inf	{ \InfBox{\D_1}
			{ \derX P : `G |- `a{:}\CoU{n},`b{:}A,`D }
		  \quad
		  \InfBox{\D_2}
			{ \derX Q : `G,y{:}B |- `a{:}\un{m}{C_{n{+}j}},`D }
		}
		{ \derPure \imp P `b [z] y Q : `G\inter z{:}A\arr B |- `a{:}\CoU{n+m},`D }
	  \quad
	  \Inf	[\unL]
		{ \InfBox{\D^i_3}
			{ \derX R : `G,x{:}C_i |- `D } 
		  ~(\forall i \ele \underline{n{+}m})
		}
		{ \derX R : `G,x{:}\CoU{n+m} |- `D } 
	}
	{ \derPure \cutL { \imp P `b [z] y Q } `a + x R : `G\inter z{:}A\arr B |- `D }
 \]


\Long{
 \item[\Lv] $ \cutL { \cut P `b + y Q } `a + x R \red \cut { \cutL P `a + z R } `b + y { \cutL Q `a + x R } $.
We have
 \[ \kern 2cm
 \Inf	{ 
\multiput(20,0)(0,7){6}{.}
\raise2\RuleH\hbox to 1.5cm {\kern-3cm
	  \Inf	{ \InfBox{\D_1}
			{ \derX P : `G |- `a{:}\CoU{n},`b{:}A,`D } 
		  \quad
		  \Inf	[\unL]
			{ \InfBox{\D^i_3}
				{ \derX R : `G,x{:}C_i |- `D } 
			  ~(\forall i \ele \n)
			}
			{ \derX R : `G,x{:}\CoU{n} |- `D } 
		}
		{ \derPure \cutL P `a + x R : `G |- `b{:}A,`D } 
}
	  \quad
	  \Inf	{ \InfBox{\D_2}
			{ \derX Q : `G,y{:}A |- `a{:}\un{m}{C_{n{+}j}},`D } 
		  \quad
		  \Inf	[\unL]
			{ \InfBox{\D^{n{+}j}_3}
				{ \derX R : `G,x{:}C_{n{+}j} |- `D } 
			  ~(\forall j \ele \m)
			}
			{ \derX R : `G,x{:}\un{m}{C_{n{+}j}} |- `D } 
		}
		{ \derPure \cutL Q `a + x R : `G,y{:}A |- `D }
	}
	{ \derPure \cut { \cutL P `a + x R } `b + y { \cutL Q `a + x R } : `G |- `D }
 \]
and we can construct
 \[  
\Inf	{ \Inf	{ \InfBox{\D_1}
			{ \derX P : `G |- `a{:}\CoU{n},`b{:}A,`D }
		  \quad
		  \InfBox{\D_2}
			{ \derX Q : `G,y{:}A |- `a{:}\un{m}{C_{n{+}j}},`D }
		}
		{ \derPure \cut P `b + y Q : `G |- `a{:}\CoU{n},`D }
	  \quad
	  \Inf	[\unL]
		{ \InfBox{\D^i_3}
			{ \derX R : `G,x{:}C_i |- `D } 
		  ~(\forall i \ele \underline{n{+}m})
		}
		{ \derX R : `G,x{:}\CoU{n} |- `D } 
	}
	{ \derPure \cutL { \cut P `b + y Q } `a + x R : `G |- `D }
 \]
}

 \end{description}

 \item[Right propagation]


 \begin{description}

\Long{

 \item[\deactR] $ \cutR P `a + x \caps<x,`b> \red \cut P `a + x \caps<x,`b> $.
Immediate, since active cuts and inactive cuts are typed with the same rule.

}


 \item[\Ri] $ \cutR P `a + x \caps<y,`b> \red \caps<y,`b>,\ y \not= x $. 
We have
 \[
\Inf	{ \derPure \caps<y,`b> : `G\inter y{:}B |- `b{:}B\union `D }
 \]
and we can construct
 \[  
\Inf	{ \Inf	[\Top]
		{ \derPure P : `G\inter y{:}B |- `a{:}\Top,`D }
	  \quad
	  \Inf	{ \derPure \caps<y,`b> : `G,y{:}B,x{:}{\Top} |- `b{:}B\union `D }
	}
	{ \derPure \cutR P `a + x \caps<y,`b> : `G\inter y{:}B |- `b{:}B\union `D }
 \]


\Long{
 \item[\Rii] $ \cutR P `a + x { \exp y Q `b . `g } \red \exp y { \cutR P `a + x Q } `b . `g $.
We have
 \[
\Inf	{ \Inf	{ \InfBox{\D_1}
			{ \derPure P : `G |- `a{:}C,`D }
		  \dquad 
		  \InfBox{\D_2}
			{ \derX Q : `G,x{:}C,y{:}A |- `b{:}B,`D }
		}
		{ \derPure \cutR P `a + x Q : `G,y{:}A |- `b{:}B,`D }
	}
	{ \derPure \exp y { \cutR P `a + x Q } `b . `g : `G |- `g{:}A\arr B\union `D }
 \]
and we can construct
 \[  
\Inf	{ \InfBox{\D_1}
		{ \derPure P : `G |- `a{:}C,`D }
	  \quad
	  \Inf	{ \InfBox{\D_2}
			{ \derX Q : `G,x{:}C,y{:}A |- `b{:}B,`D }
		}
		{ \derPure \exp y Q `b . `g : `G,x{:}C |- `g{:}A\arr B\union `D }
	}
	{ \derPure \cutR P `a + x { \exp y Q `b . `g } : `G |- `g{:}A\arr B\union `D }
 \]
}


 \item[\Riii] $ \cutR P `a + x {\imp Q `b [x] y R } \red \cut P `a + v { \imp { \cutR P `a + x Q } `b [v] y { \cutR P `a + x R } } $, with $v$ fresh.
We have
 \[ \kern-3cm
 \Inf	
	{ \InfBox{\D_1}
		{ \derX P : `G |-  `a{:}A\arr B,`D }
	  \kern-3cm
	  \Inf	
		{ \Inf	
			{ \Inf	[\intR]
				{ \InfBox{\D^i_1}
					{ \derX P : `G |- `a{:}C_i,`D }
				  ~(\forall i \ele \n)
				}
				{ \derX P : `G |- `a{:}\CoI{n},`D }
			  ~
			  \InfBox{\D_2}
				{ \derX Q : `G,x{:}\CoI{n} |- `b{:}A,`D }
			}
			{ \derX \cutR P `a + x Q : `G |- `b{:}A,`D }
\raise1.25\RuleH\hbox to 35mm {\kern-4.5cm
		  \Inf	
			{ \Inf	[\intR]
				{ \InfBox{\D^{n{+}j}_1}
					{ \derX P : `G |- `a{:}C_{n{+}j},`D }
				  ~(\forall j \ele \m)
				}
				{ \derX P : `G |- `a{:}\int{m}{C_{n{+}j}},`D }
			  ~ 
			  \InfBox{\D_3}
				{ \derX R : `G,y{:}B,x{:}\int{m}{C_{n{+}j}} |- `D }
			}
			{ \hspace*{2cm} \derX \cutR P `a + x R : `G,y{:}B |- `D }
}
\multiput(-20,0)(0,7){4}{.}
		}
		{ \derX \imp { \cutR P `a + x Q } `b [v] y { \cutR P `a + x R } : `G,v{:}A\arr B |- `D }
	}
	{ \derPure \cut P `a + v { \imp { \cutR P `a + x Q } `b [v] y { \cutR P `a + x R } } : `G |- `D }
 \]
and we can construct
 \[  
 \kern-4cm
\Inf	
	{ 
	  \Inf	[\intR]
		{ \InfBox{\D^i_1}
			{ \derX P : `G |- `a{:}C_i,`D }
		  ~(\forall i \ele \underline{n{+}m})
		\quad
		\InfBox{\D_1}
			{ \derPure P : `G |- `a{:}A\arr B,`D }
		}
		{ \derPure P : `G |- `a{:}\CoI{n{+}m}\inter (A\arr B),`D }
\raise2\RuleH\hbox to 2cm {\kern-4cm
	  \Inf	
		{ \InfBox{\D_2}
			{ \derX Q : `G,x{:}\CoI{n} |- `b{:}A,`D }
		  \dquad 
		  \InfBox{\D_3}
			{ \derX R : `G,y{:}B,x{:}\int{m}{C_{n{+}j}} |- `D }
		}
		{ \derPure \imp Q `b [x] y R : `G,x{:}\CoI{n{+}m}\inter (A\arr B) |- `D }
}
\multiput(-20,0)(0,7){6}{.}
	}
	{ \derPure \cutR P `a + x { \imp Q `b [x] y R } : `G |- `D }
 \]


 \item[\Riv] $ \cutR P `a + x {\imp Q `b [z] y R } \red \imp { \cutR P `a + x Q } `b [z] y { \cutR P `a + z R },\ x \not= z $.  
We have
 \[ \kern-3.5cm
 \Inf	
	{ \Inf	
		{ \Inf	[\intR]
			{ \InfBox{\D^i_1}
				{ \derX P : `G |- `a{:}C_i,`D }
			  ~(\forall i \ele \n)
			}
			{ \derX P : `G |- `a{:}\CoI{n},`D }
		  ~
		  \InfBox{\D_2}
			{ \derX Q : `G,x{:}\CoI{n} |- `b{:}A,`D }
		}
		{ \derX \cutR P `a + x Q : `G |- `b{:}A,`D }
\raise1.25\RuleH\hbox to 3cm {\kern-4cm
	  \Inf	
		{ \Inf	[\intR]
			{ \InfBox{\D^{n{+}j}_1}
				{ \derX P : `G |- `a{:}C_{n{+}j},`D }
			  ~(\forall j \ele \m)
			}
			{ \derX P : `G |- `a{:}\int{m}{C_{n{+}j}},`D }
		  ~ 
		  \InfBox{\D_3}
			{ \derX R : `G,y{:}B,x{:}\int{m}{C_{n{+}j}} |- `D }
		}
		{ \kern2cm \derX \cutR P `a + x R : `G,y{:}B |- `D }
}
\multiput(-20,0)(0,7){4}{.}
	}
	{ \derX \imp { \cutR P `a + x Q } `b [z] y { \cutR P `a + x R } : `G,z{:}A\arr B |- `D }
 \]
and we can construct
 \[  
 \kern-1cm
\Inf	
	{ \Inf	[\intR]
		{ \InfBox{\D^i_1}
			{ \derX P : `G |- `a{:}C_i,`D }
		  ~(\forall i \ele \underline{n{+}m})
		}
		{ \derPure P : `G |- `a{:}\CoI{n{+}m},`D }
	  \quad
	  \Inf	
		{ \InfBox{\D_2}
			{ \derX Q : `G,x{:}\CoI{n} |- `b{:}A,`D }
		  \dquad 
		  \InfBox{\D_3}
			{ \derX R : `G,y{:}B,x{:}\int{m}{C_{n{+}j}} |- `D }
		}
		{ \derPure \imp Q `b [x] y R : `G,x{:}\CoI{n{+}m} |- `D }
	}
	{ \derPure \cutR P `a + x { \imp Q `b [z] y R } : `G,z{:}A\arr B |- `D }
 \]


\Long{
 \item[\Rv] $ \cutR P `a + x { \cut Q `b + y R } \red \cut { \cutR P `a + x Q } `b + y { \cutR P `a + z R } $.
We have
 \[ \kern-4cm
 \Inf	
	{ \Inf	
		{ \Inf	[\intR]
			{ \InfBox{\D^i_1}
				{ \derX P : `G |- `a{:}C_i,`D }
			  ~(\forall i \ele \n)
			}
			{ \derX P : `G |- `a{:}\CoI{n},`D }
		  ~
		  \InfBox{\D_2}
			{ \derX Q : `G,x{:}\CoI{n} |- `b{:}A,`D }
		}
		{ \derX \cutR P `a + x Q : `G |- `b{:}A,`D }
\raise1.25\RuleH\hbox to 3cm {\kern-4cm
	  \Inf	
		{ \Inf	[\intR]
			{ \InfBox{\D^{n{+}j}_1}
				{ \derX P : `G |- `a{:}C_{n{+}j},`D }
			  ~(\forall j \ele \m)
			}
			{ \derX P : `G |- `a{:}\int{m}{C_{n{+}j}},`D }
		  ~ 
		  \InfBox{\D_3}
			{ \derX R : `G,y{:}B,x{:}\int{m}{C_{n{+}j}} |- `D }
		}
		{ \kern2cm \derX \cutR P `a + x R : `G,y{:}B |- `D }
}
\multiput(-20,0)(0,7){4}{.}
	}
	{ \derX \cut { \cutR P `a + x Q } `b + y { \cutR P `a + x R } : `G |- `D }
 \]
and we can construct
 \[  
 \kern-1cm
\Inf	
	{ \Inf	[\intR]
		{ \InfBox{\D^i_1}
			{ \derX P : `G |- `a{:}C_i,`D }
		  ~(\forall i \ele \underline{n{+}m})
		}
		{ \derPure P : `G |- `a{:}\CoI{n{+}m},`D }
	  \quad
	  \Inf	
		{ \InfBox{\D_2}
			{ \derX Q : `G,x{:}\CoI{n} |- `b{:}A,`D }
		  \dquad 
		  \InfBox{\D_3}
			{ \derX R : `G,y{:}B,x{:}\int{m}{C_{n{+}j}} |- `D }
		}
		{ \derPure \cut Q `b + y R : `G,x{:}\CoI{n{+}m} |- `D }
	}
	{ \derPure \cutR P `a + x { \cut Q `b + y R } : `G |- `D }
 \]
}

 \end{description}

 \end{description}

 \end{proof}

}



\Long{\section{The problem with witness reduction}}
 \label{problems}

As in the system of \cite{Barbanera-Dezani-Liguoro-IaC'95} defined for the {\LC}, we suffer loss of the subject reduction property (here called witness reduction).
This problem also appears in other contexts, such as that of {\ML} with side-effects \cite{Harper-Lillibridge'91,Wright'95,Milner-et.al'97}, and that of using intersection and union types in an operational setting \cite{Davies-Pfenning'01,Dunfield-Pfenning'00}, and has also been observed (giving little detail) in \cite{Herbelin'05}.
The advantage of studying this problem in the context of sequent calculi is clearly shown by the examples in this section.
%
These examples will lead to the definition of \Em{two} restrictions on the notions of type assignment, $\TurnN$ and $\TurnV$, that we will show to be closed for reduction for, respectively, call-by-name and call-by-value reduction.

As in \cite{Barbanera-Dezani-Liguoro-IaC'95}, for {\X}, using the (unrestricted) notion of type assignment we gave above, we can show that subject reduction does not hold in general.

\Comment{
 \begin{example}[First counterexample]
Let
 \[ \begin{array}{rcl}
`G &=& x{:}(A\arr A\arr C)\int (B\arr B\arr C), y{:}D\arr A\un B, z{:}D
\\
`D &=& `a{:}C
 \end{array} \]

Without loss of generality, let
 \[ \begin{array}{rcl}
\SemL{xt}{`e} &=& \imp \caps<t,`s> `s [x] e \caps<e,`e> 
\\
\SemL{xtt}{`b} &=& \cut \SemL{xt}{`e} `e + c { \imp \caps<t,`m> `m [c] d \caps<d,`b> } 
 \end{array} \]
We can derive:
 \[
\D^A_1 =
 \Inf	
	{ 
	  \Inf	
		{ \Inf	{ \derX \caps<t,`s> : t{:}A |- `s{:}A }
		  \Inf	{ \derX \caps<e,`e> : e{:}A\arr C |- `e{:}A\arr C }
		}
		{ \derX \SemL{xt}{`e} 
		  : x{:}A\arr A\arr C , t{:}A |- `e{:}A\arr C }
	  \Inf	
		{ \Inf	{ \derX \caps<t,`m> : t{:}A |- `m{:}A }
		  \Inf	{ \derX \caps<d,`b> : d{:}C |- `b{:}C }
		}
		{ \derX \imp \caps<t,`m> `m [c] d \caps<d,`b>  : t{:}A, c{:}A\arr C |- `b{:}C }
	}
	{ \derX \SemL{xtt}{`b} : x{:}A\arr A\arr C ,  t{:}A |- `b{:}C }
 \]
 and, likewise, $\D^B_1 \dcol \derX \SemL{xtt}{`b} : x{:}B\arr B\arr C ,  t{:}B |- `b{:}C $.
Take 
 \[
\D_1 =
  \Inf	[\unL]
	{ \InfDerX \D^A_1 :: \SemL{xtt}{`b} : x{:}A\arr A\arr C , t{:}A |- `b{:}C
	  \quad \quad
	  \InfDerX \D^B_1 :: \SemL{xtt}{`b} : x{:}B\arr B\arr C ,  t{:}B |- `b{:}C
	}
	{ \derX \SemL{xtt}{`b} : x{:}(A\arr A\arr C)\int (B\arr B\arr C ) , t{:}A\un B |- `b{:}C }
 \]

Let
 \[ \begin{array}{rcr}
\SemL{yz}{`d} &=& 
\cut \caps<y,`m> `m + s { \imp \caps<z,`r> `r [s] u \caps<u,`d> } 
 \end{array} \]
We can derive (where $E = {D\arr A\un B}$):
 \[ 
\D_2 =
\Inf	
 { \Inf	
	{ \derX \caps<y,`m> : y{:}E |- `m{:}E }
   \Inf	
	{ \Inf	{ \derX \caps<z,`r>  : z{:}D |- `r{:}D }
	  \Inf	[\unL]
		{ \Inf { \derX \caps<u,`d>  : u{:}A |- `d{:}A }
		  \Inf { \derX \caps<u,`d>  : u{:}B |- `d{:}B }
		}
		{ \derX \caps<u,`d>  : u{:}A\un B |- `d{:}A\un B }
	}
	{ \derX \imp \caps<z,`r> `r [s] u \caps<u,`d> : s{:}D\arr A\un B, z{:}D |- `d{:}A\un B }
 }
 { \derX 
	\SemL{yz}{`d} : y{:}D\arr A\un B,z{:}D |- `d{:}A\un B }
 \]
(Notice that in both derivations we can weaken the left context to $`G$.)

Combining these gives:
 \[
\Inf	
	{ \InfDerX \D_2 :: \SemL{yz} {`d} : `G |- `d{:}A\un B,`a{:}C
	  \dquad
	  \Inf	[\unL]
		{ \InfDerX \D^A_1 :: \SemL{xtt}{`b} : `G , t{:}A |- `a{:}C
		  \dquad
		  \InfDerX \D^B_1 :: \SemL{xtt}{`b} : `G , t{:}B |- `a{:}C }
		{ \derX \SemL{xtt}{`b} : `G , t{:}A\un B |- `a{:}C }
	}
	{ \derX \cut \SemL{yz}{`d}\, `d + t \SemL{xtt}{`a} : `G |- `a{:}C }
\]
and the \CBN-reduction runs as follows:
 \[ \begin{array}{lcl}
\cut \SemL{yz}{`d}\, `d + t \SemL{xtt}{`a} 
&\ByDef& \\
\cut \SemL{yz}{`d}\, `d + t { \cut \SemL{xt}{`e} `e + c { \imp \caps<t,`m> `m [c] d \caps<d,`a> } }
& \red & (\actR), (\Rv) \\ 
\cut { \cutR \SemL{yz}{`d}\, `d + t  \SemL{xt}{`e}  } `e + c {  \cutR \SemL{yz}{`d}\, `d + t  { \imp \caps<t,`m> `m [c] d \caps<d,`a> } }
 \end{array} \]
The last net above is not typeable using the same contexts.


First of all, note that we certainly cannot create a derivation with $`G$ which uses $(\unL)$ as last step for the import $\imp \caps<t,`m> `m [c] d \caps<d,`a> $: this would give 
 \[ \Inf	[\unL]
	{  \Inf	{ \Inf	{ \derX \caps<t,`m>  : t{:}A |- `m{:}A }
		  \Inf	{ \derX \caps<d,`a>  : d{:}C |- `a{:}C }
		}
		{ \derX \imp \caps<t,`m> `m [c] d \caps<d,`a>  : `G,t{:}A,c{:}A\arr C |- `a{:}C }
		  \quad
	  \Inf	{ \Inf	{ \derX \caps<t,`m>  : t{:}B |- `m{:}B }
		  \Inf	{ \derX \caps<d,`a>  : d{:}C |- `a{:}C }
		}
		{ \derX \imp \caps<t,`m> `m [c] d \caps<d,`a>  : `G,t{:}B,c{:}B\arr C |- `a{:}C }
	}
	{ \derX \imp \caps<t,`m> `m [c] d \caps<d,`a>  : `G,t{:}A\un B,c{:}(A\arr C)\int (B\arr C) |- `a{:}C }
\]
We now would need to construct a derivation for 
$ \derX \SemL{xt}{`e} : `G |- `e{:}(A\arr C)\int (B\arr C) $.
This is impossible; we can at most achieve, using $(\intR)$, 
 \[ \Inf	[\intR]
	{ \Inf	{ \Inf	{ \derX \caps<t,`s> : t{:}A |- `s{:}A } 
		  \Inf	{ \derX \caps<e,`e> : e{:}A\arr C |- `e{:}A\arr C  }
		}
		{ \derX \SemL{xt}{`e} : x{:}A\arr A\arr C,t{:}A |- `e{:}A\arr C }
	  \Inf	{ \Inf	{ \derX \caps<t,`s> : t{:}B |- `s{:}B } 
		  \Inf	{ \derX \caps<e,`e> : e{:}B\arr C |- `e{:}B\arr C  }
		}
		{ \derX \SemL{xt}{`e} : x{:}B\arr B\arr C,t{:}B |- `e{:}B\arr C }
	}
	{ \derX \SemL{xt}{`e} : x{:}(A\arr A\arr C)\int (B\arr B\arr C),t{:}A\int B |- `e{:}(A\arr C)\int (B\arr C) }
\]
which has an intersection rather than the required union for $t$, or, using $(\unL)$:
 \[ \Inf	[\unL]
	{ \Inf	{ \Inf	{ \derX \caps<t,`s> : t{:}A |- `s{:}A } 
		  \Inf	{ \derX \caps<e,`e> : e{:}A\arr C |- `e{:}A\arr C  }
		}
		{ \derX \SemL{xt}{`e} : x{:}A\arr A\arr C,t{:}A |- `e{:}A\arr C }
	  \Inf	{ \Inf	{ \derX \caps<t,`s> : t{:}B |- `s{:}B } 
		  \Inf	{ \derX \caps<e,`e> : e{:}B\arr C |- `e{:}B\arr C  }
		}
		{ \derX \SemL{xt}{`e} : x{:}B\arr B\arr C,t{:}B |- `e{:}B\arr C }
	}
	{ \derX \cutR \SemL{yz}{`d}\, `d + t  \SemL{xt}{`e} : x{:}(A\arr A\arr C)\int (B\arr B\arr C),t{:}A\un B |- `e{:}(A\arr C)\un (B\arr C) }
\]
which has a union for $`e$.

So the derivation would need to have the following shape:
 \[ \kern-2.5cm
\Inf	{ \Inf	{ \InfDerX \D_2 :: \SemL{yz}{`d} : `G |- `d{:}A\un B 
		  \quad
		  \InfDerX ? :: \SemL{xt}{`e} : `G,t{:}A\un B |- `e{:}A\un B\arr C 
		}
		{ \derX \cutR \SemL{yz}{`d}\, `d + t  \SemL{xt}{`e} : `G |- `e{:}A\un B\arr C }
\raise3.25\RuleH\hbox to 2cm{\kern-7cm
	  \Inf	{ \InfDerX \D_2 :: \SemL{yz}{`d} : `G |- `d{:}A\un B 
		  \quad
		  \Inf	{  \InfDerX {} :: \caps<t,`m>  : t{:}A\un B |- `m{:}A\un B 
			  \quad
			  \Inf	{ \derX \caps<d,`a>  : d{:}C |- `a{:}C }
			}
			{ \derX \imp \caps<t,`m> `m [c] d \caps<d,`a>  : `G,t{:}A\un B,c{:}A\un B\arr C |- `a{:}C }
		}
		{ \derX \cutR \SemL{yz}{`d}\, `d + t  { \imp \caps<t,`m> `m [c] d \caps<d,`a> }  : `G,c{:}A\un B\arr C |- `a{:}C }
}
\multiput(-10,0)(0,7){10}{.}
	}
	{ \derX \cut P `a + x Q 
	 : `G |- `a{:}C }
\]
and we need a derivation for 
 \[ 
\InfDerX ? :: \SemL{xt}{`e} : `G , t{:}A\un B |- `e{:}A\un B\arr C
\]
We can only derive $t{:}A\un B$ on the left via rule $(\unL)$; let 
 \[ \D^A_3 = 
 \Inf	{ \Inf	{ \derX \caps<t,`s> : t{:}A |- `s{:}A }
		  \Inf	{ \derX \caps<e,`e> : e{:}A\arr C |- `e{:}A\arr C }
		}
		{ \derX \SemL{xt}{`e} : `G , t{:}A |- `e{:}A\arr C }
\]
and $\D^B_3$ similar, then 
 \[ \D_3 = 
 \Inf	[\unL]
	{ \InfDerX \D^A_3 :: \SemL{xt}{`e} : `G , t{:}A |- `e{:}A\arr C  
	  \dquad
	  \InfDerX \D^B_3 :: \SemL{xt}{`e} : `G , t{:}B |- `e{:}B\arr C  
	}
	{ \derX \SemL{xt}{`e} : `G , t{:}A\un B |- `e{:}(A\arr C)\un (B\arr C) }
\]
So, as above, we end up with the statement $`e{:}(A\arr C)\un (B\arr C)$, and not with $A\un B\arr C$.

\end{example}

 \begin{example}
Notice that the right-activation made above 
is not compulsory, and in fact is only so in the {\CBN} reduction strategy.
If we active the cut in the other direction, we get
 $ \begin{array}{rcl}
\cut \SemL{yz}{`d}\, `d + t \SemL{xtt}{`a} &\redCBV& \imp \caps<z,`r> `r [y] t {\imp \caps<t,`s> `s [x] e { \imp \caps<t,`m> `m [e] d \caps<d,`a> } }
 \end{array} $.
 \[ \begin{array}{lcl}
\cut \SemL{yz}{`d}\, `d + t \SemL{xtt}{`a}
& = \\
\cut { \cut \caps<y,`m> `m + s { \imp \caps<z,`r> `r [s] u \caps<u,`d> } } `d + t \SemL{xtt}{`a}
& \red & (\actL), (\Lv), (\Liv) \\
\multicolumn{3}{l}{
\cut { \cutL \caps<y,`m> `d + t \SemL{xtt}{`a} } `m + s { \imp {\cutL \caps<z,`r> `d + t \SemL{xtt}{`a} } `r [s] u {\cutL \caps<u,`d> `d + t \SemL{xtt}{`a} } } 
} \\
& \red & (\Li)\,(2\times), (\deactL), (\renR), (`a) \\
\cut \caps<y,`m> `m + s { \imp \caps<z,`r> `r [s] t \SemL{xtt}{`a} } 
& \red & (\Imp) \\
\imp \caps<z,`r> `r [y] t \SemL{xtt}{`a} 
 \end{array} \]
This last net \textit{is} typeable using the same contexts.
 \[ 
  \Inf	{ \Inf	{ \derX \caps<z,`r> : z{:}D |- `r{:}D }
	  \quad
	  \InfDerX \D_1 ::  \SemL{xtt}{`a} : x{:}(A\arr A\arr C)\int (B\arr B\arr C), z{:}D, t{:}A\un B |- `a{:}C 
	}
	{ \derX \imp \caps<z,`r> `r [y] t \SemL{xtt}{`a} : x{:}(A\arr A\arr C)\int (B\arr B\arr C), z{:}D, y{:}D\arr A\un B |- `a{:}C }
\]
even when we reduce $\SemL{xtt}{`a}$ to its normal form 
$ \imp \caps<t,`s> `s [x] e { \imp \caps<t,`m> `m [e] d \caps<d,`b> } $.
 \[ 
  \Inf	{ \Inf	{ \derX \caps<z,`r> : z{:}D |- `r{:}D }
	  \Inf  {
\multiput(10,0)(0,7){10}{.}
\raise3.5\RuleH\hbox to 1cm{\kern -3cm
		  \Inf	{ \Inf	{ \derX \caps<t,`s> : t{:}A |- `s{:}A }
			  \Inf	{ \Inf	{ \derX \caps<t,`m> : t{:}A |- `m{:}A }
				  \Inf	{ \derX \caps<d,`a> : d{:}C |- `a{:}C }
				}
				{ \derX \imp \caps<t,`m> `m [e] d \caps<d,`a> : t{:}A, e{:}A\arr C |- `a{:}C }
			}
			{ \derX \imp \caps<t,`s> `s [x] e { \imp \caps<t,`m> `m [e] d \caps<d,`a> } : x{:}A\arr A\arr C, t{:}A |- `a{:}C }
}
		  \Inf	{ \Inf	{ \derX \caps<t,`s> : t{:}B |- `s{:}B }
			  \Inf	{ \Inf	{ \derX \caps<t,`m> : t{:}B |- `m{:}B }
				  \Inf	{ \derX \caps<d,`a> : d{:}C |- `a{:}C }
				}
				{ \derX \imp \caps<t,`m> `m [e] d \caps<d,`a> : t{:}B, e{:}B\arr C |- `a{:}C }
			}
			{ \derX \imp \caps<t,`s> `s [x] e { \imp \caps<t,`m> `m [e] d \caps<d,`a> } : x{:}B\arr B\arr C, t{:}B |- `a{:}C }
		}
		{ \derX \imp \caps<t,`s> `s [x] e { \imp \caps<t,`m> `m [e] d \caps<d,`a> } : `G\Except y, t{:}A\un B |- `a{:}C } 
	}
	{ \derX \imp \caps<z,`r> `r [y] t {\imp \caps<t,`s> `s [x] e { \imp \caps<t,`m> `m [e] d \caps<d,`a> } } : `G |- `a{:}C }
\]
 \end{example}
}

 \begin{example}[%
First counterexample] \label{second counterexample}
Take 
 \[ \begin{array}[t]{rcl}
\Long{\cut \SemL{xx}{`a} `a + y \SemL{yy}{`b} & \redCBV &
\cut {\imp \caps<x,`g> `g [x] v \caps<v,`a> } `a + y {\imp \caps<y,`d> `d [y] w \caps<w,`b> } }
\Short{\cut {\imp \caps<x,`g> `g [x] v \caps<v,`a> } `a + y {\imp \caps<y,`d> `d [y] w \caps<w,`b> } & \redCBV & \imp \caps<x,`g> `g [x] v { \imp \caps<v,`d> `d [v] w \caps<w,`b> } }
 \end{array} \]
 
We can type the \Long{last }\Short{first }net as follows:
%
%
 \[ \kern-3mm
\Long{
\Inf	
	{ 
\raise2.25\RuleH\hbox to 6cm{\kern -1.5cm
	\Inf	[\intR]
		{ \Inf	{ \Inf	{ \derX \caps<x,`g>  : x{:}A |- `g{:}A }
			  \Inf	{ \derX \caps<v,`a>  : v{:}C |- `a{:}C }
			}
			{ \derX \imp \caps<x,`g> `g [x] v \caps<v,`a>  : { x{:}A\int (A\arr C)  } |- `a{:}C }
		  \Inf	{ \Inf { \derX \caps<x,`g>  : x{:}A |- `g{:}A }
			  \Inf	{ \derX \caps<v,`a>  : { v{:}C\arr D  } |- { `a{:}C\arr D } }
			}
			{ \derX \imp \caps<x,`g> `g [x] v \caps<v,`a>  : { x{:}A\int (A\arr C\arr D) } |- { `a{:}C\arr D } }
		}
		{ \derX \imp \caps<x,`g> `g [x] v \caps<v,`a>  : { x{:}A\int (A\arr C)\int (A\arr C\arr D) } |- { `a{:}C\int (C\arr D) } }
}
\multiput(-100,0)(0,7){7}{\makebox(0,0){.}}
	  \Inf	{ \Inf	{ \derX \caps<y,`d> : y{:}C |- `d{:}C }
		  \quad
		  \Inf	{ \derX \caps<w,`b> : w{:}D |- `b{:}D }
		}
		{ \derX \imp \caps<y,`d> `d [y] w \caps<w,`b> : { y{:}C\int (C\arr D) } |- `b{:}D }
	}
	{ \derX \cut {\imp \caps<x,`g> `g [x] v \caps<v,`a> } `a + y {\imp \caps<y,`d> `d [y] w \caps<w,`b> } : { x{:}A\int (A\arr C)\int (A\arr C\arr D) } |- `b{:}D }
}
\Short{ \derX \cut {\imp \caps<x,`g> `g [x] v \caps<v,`a> } `a + y {\imp \caps<y,`d> `d [y] w \caps<w,`b> } : { x{:}A\int (A\arr C)\int (A\arr C\arr D) } |- `b{:}D }
\]

\Long{
Now this cut in this net can be activated in \emph{two} directions, both to the left as to the right.
The (\CBV-)reduction 
 \[\begin{array}{rcl}
\cut {\imp \caps<x,`g> `g [x] v \caps<v,`a> } `a + y {\imp \caps<y,`d> `d [y] w \caps<w,`b> } & \redCBV & \imp \caps<x,`g> `g [x] v { \imp \caps<v,`d> `d [v] w \caps<w,`b> } 
\end{array}
\] 
now proceeds:
 \[ \begin{array}[t]{lcl}
\cut {\imp \caps<x,`g> `g [x] v \caps<v,`a> } `a + y {\imp \caps<y,`d> `d [y] w \caps<w,`b> } 
& \red & (\actL) \\
\cutL {\imp \caps<x,`g> `g [x] v \caps<v,`a> } `a + y {\imp \caps<y,`d> `d [y] w \caps<w,`b> } 
& \red & (\Liv) \\
\imp {\cutL \caps<x,`g> `a + y {\imp \caps<y,`d> `d [y] w \caps<w,`b> } } `g [x] v { \cutL \caps<v,`a> `a + y {\imp \caps<y,`d> `d [y] w \caps<w,`b> } } 
& \red & (\Li) \\
\imp \caps<x,`g> `g [x] v { \cutL \caps<v,`a> `a + y {\imp \caps<y,`d> `d [y] w \caps<w,`b> } } 
& \red & 
\\ \multicolumn{3}{r}{
(\deactL), (\actR), (\Riv), (\deactR), (\Cap), (\Ri), (\Med) 
} \\
\imp \caps<x,`g> `g [x] v { \imp \caps<v,`d> `d [v] w \caps<w,`b> } 
 \end{array} \]

Subject reduction collapses in the $(\Liv)$ step (notice that the next step, $(\Li)$, does not change this fact): we cannot derive 
 \[ 
\derX \imp \caps<x,`g> `g [x] v { \cutL \caps<v,`a> `a + y {\imp \caps<y,`d> `d [y] w \caps<w,`b> } }  : { x{:}A\int (A\arr C)\int (A\arr C\arr D) } |- `b{:}D 
 \]
}

\Short{We cannot use these contexts to type the right-hand side net.}
\Long{
We can only derive this result via a derivation with last step either $(\arrL)$, $(\intR)$ or $(\unL)$.
Since neither the type associated to $x$ is a union or that to $`b$ an intersection, the only possibility is:
 \[ \kern-5mm
  \Inf	[\arrL]
 	{ \Inf	{ \derX \caps<x,`g>  : x{:}A |- `g{:}A }
 	  \quad
 	  \Inf	
 		{ \Inf	{ \derX \caps<y,`d>  : y{:}C |- `d{:}C }
 		  \Inf	{ \derX \caps<w,`b>  : w{:}D |- `b{:}D }
 		}
 		{ \derX \imp \caps<y,`d> `d [y] w \caps<w,`b>  : y{:}C\int (C\arr D) |- `b{:}D }
 	}
 	{ \derX \imp \caps<x,`g> `g [x] v { \imp \caps<y,`d> `d [v] w \caps<w,`b> }  : x{:}A\int (A\arr C\int (C\arr D) ) |- `b{:}D }
 \]
Notice that the last rule forces the construction of the type $A\arr (C\int (C\arr D))$, which should be amongst the types $A$, $A\arr C$ or $A\arr C\arr D$; this is not the case.

As can be expected from the witness-expansion theorem, we can derive

 \[ \kern5mm
 \Inf{
\multiput(10,0)(0,7){7}{.}
\raise2.25\RuleH\hbox to 6cm{\kern -10mm 
 	\Inf{\Inf{\derX \caps<x,`g> : x{:}A |- x{:}A } 
		  \Inf[\intR]
 			{\Inf{ \derX \caps<v,`a>  : v{:}C |- `a{:}C }
 			 \Inf{ \derX \caps<v,`a>  : v{:}C\arr D |- `a{:}C\arr D }
 			}
			{ \derX \caps<v,`a>  : v{:}C\int (C\arr D) |- `a{:}C\int (C\arr D) }
		}
		{ \derX \imp \caps<x,`g> `g [x] v \caps<v,`a>  : x{:}A\int (A\arr C\int (C\arr D) ) |- `a{:}C\int (C\arr D)  }
}
	  \quad
	  \Inf{\Inf { \derX \caps<y,`d>  : y{:}C |- `d{:}C }
 		   \Inf { \derX \caps<w,`b>  : w{:}D |- `b{:}D }
		}
		{ \derX \imp \caps<y,`d> `d [y] w \caps<w,`b>    : y{:}C\int (C\arr D) |- `b{:}D }
	}
 	{ \derX \cut {\imp \caps<x,`g> `g [x] v \caps<v,`a> } `a + y {\imp \caps<y,`d> `d [y] w \caps<w,`b> }   : x{:}A\int (A\arr C\int (C\arr D) ) |- `b{:}D }
 \]
}
The {\CBN} reduction
\Long{%
 \[\begin{array}{rcl}
\Long{\cut  \SemL{xx}{`a} `a + y \SemL{yy}{`b} }
\Short{\cut {\imp \caps<x,`g> `g [x] v \caps<v,`a> } `a + y {\imp \caps<y,`d> `d [y] w \caps<w,`b> } }
& \redCBN & \imp \caps<x,`g> `g [x] v { \imp {\imp \caps<x,`g> `g [x] u \caps<u,`d> } `d [v] w \caps<w,`b> } 
\end{array}
\] }
however, does not pose any problems\Short{.}\Long{
:
 \[ \begin{array}[t]{lcl}
\cut  \SemL{xx}{`a} `a + y \SemL{yy}{`b} 
& \red & (\actR) 
, (\Riii) \hspace*{16mm} \\
\cut \SemL{xx}{`a} `a + z { \imp {\cutR \SemL{xx}{`a} `a + y \caps<y,`d> } `d [z] w {\cutR \SemL{xx}{`a} `a + y \caps<w,`b> } }
 \end{array} \]
In this step, the derivation for $ \SemL{xx}{`a} $ splits into $ \derX \SemL{xx}{`a}  : x{:}A\int (A\arr C\int (C\arr D) ) |- `a{:}C $ and $ \derX \SemL{xx}{`a}  : x{:}A\int (A\arr C\int (C\arr D) ) |- `a{:}C\arr D $; since these do not exist as subderivations, we need $(\intE)$ here.
We then continue:
 \[ \begin{array}[t]{lcl}
\cut \SemL{xx}{`a} `a + z { \imp {\cutR \SemL{xx}{`a} `a + y \caps<y,`d> } `d [z] w {\cutR \SemL{xx}{`a} `a + y \caps<w,`b> } }
& \red & (\deactR), (\actL), (\Ri) \\
\cut \SemL{xx}{`a} `a + z { \imp {\cutL \SemL{xx}{`a} `a + y \caps<y,`d> } `d [z] w \caps<w,`b> }
& \red & (\renR) \\
\cut \SemL{xx}{`a} `a + z { \imp \SemL{xx}{`d} `d [z] w \caps<w,`b> } 
& \red & (\actL) 
, (\Imp) \\
\cutL {\imp \caps<x,`g> `g [x] v \caps<v,`a> } `a + z { \imp \SemL{xx}{`d} `d [z] w \caps<w,`b> } 
& \red & (\Liv), (\Li), (\deactL), (\Med) \\
\imp \caps<x,`g> `g [x] v { \imp \SemL{xx}{`d} `d [v] w \caps<w,`b> } 
& \red & (\Imp) \\
\imp \caps<x,`g> `g [x] v { \imp {\imp \caps<x,`g> `g [x] u \caps<u,`d> } `d [v] w \caps<w,`b> }
 \end {array} \]
Notice that the last net is in normal form (and corresponds to $\SemL{xx(xx)}{`b}$), and that we can derive:
 %
 \[
  \Inf	
	{ \Inf	{ \derX \caps<x,`g>  : y{:}A |- `a{:}A }
	  \quad
	  \Inf	
		{ \Inf	{ \Inf	{ \derX \caps<x,`g>  : x{:}A |- `g{:}A }
			  \Inf	{ \derX \caps<u,`d>  : u{:}C |- `d{:}C }
			}
			{ \derX \imp \caps<x,`g> `g [x] u \caps<u,`d>  : x{:}A\int (A\arr C) |- `d{:}C }
		  \Inf	{ \derX \caps<w,`b>  : w{:}D |- `b{:}D }
		}
		{ \derX \imp {\imp \caps<x,`g> `g [x] u \caps<u,`d> } `d [v] w \caps<w,`b> : x{:}A\int (A\arr C), v{:}C\arr D |- `b{:}D }
	}
	{ \derX \imp \caps<x,`g> `g [x] v { \imp {\imp \caps<x,`g> `g [x] u \caps<u,`d> } `d [v] w \caps<w,`b> }  : x{:}A\int (A\arr C)\int (A\arr C\arr D) |- `b{:}D }
 \]
}

 \end{example}

 \begin{example}[%
Second counterexample] \label{third counterexample}
Similarly, we can derive:
 \[ \kern-2.25cm
\Long{
\Inf	
	{  \Inf	
		{ \Inf	
			{ \derX \caps<x,`d> : x{:}A |- `d{:}A,`b{:}B } 
		}
		{ \derX \exp x \caps<x,`d> `b . `d : {} |- `d{:}A\un (A\arr B) }
\raise1.25\RuleH\hbox to 85mm{\kern -17mm
	  \Inf	[\unL]
		{ \Inf	
			{ \Inf	
				{ \derX \caps<z,`a> : z{:}A,v{:}C |- `a{:}A } }
			{ \derX \exp v \caps<z,`a> `a . `g  : z{:}A |- `g{:}C\arr A }
		  \quad
		  \Inf	
			{ \Inf	
				{ \derX \caps<z,`a> : z{:}A\arr B, v{:}C |- `a{:}A\arr B } }
			{ \derX \exp v \caps<z,`a> `a . `g  : z{:}A\arr B |- `g{:}C\arr A\arr B }
		}
		{ \derX \exp v \caps<z,`a> `a . `g  : z{:}A\un (A\arr B) |-  `g{:}(C\arr A)\un (C\arr A\arr B) }
}
\multiput(-40,0)(0,7){4}{.}
	}
	{ \derX \cut { \exp x \caps<x,`d> `b . `d } `d + z {\exp v \caps<z,`a> `a . `g } : {} |- `g{:}(C\arr A)\un (C\arr A\arr B) }
}
\Short{ \derX \cut { \exp x \caps<x,`d> `b . `d } `d + z {\exp v \caps<z,`a> `a . `g } : {} |- `g{:}(C\arr A)\un (C\arr A\arr B) }
\]
\Long{This net has no obvious counterpart in the {\LC}.}

This net reduces  $ \begin{array}{lcl}
\cut {\exp x \caps<x,`d> `b . `d } `d + z {\exp v \caps<z,`a> `a . `g } 
& \redCBN & \exp v { \exp x \caps<x,`a> `b . `a } `a . `g 
 \end{array} $%
\Short{; we cannot derive the same type for the latter term.}
\Long{:
 \[ \begin{array}{lcl}
\cut {\exp x \caps<x,`d> `b . `d } `d + z {\exp v \caps<z,`a> `a . `g } 
& \red & (\actR) 
, (\Rii) \\
\exp v { \cutR {\exp x \caps<x,`d> `b . `d } `d + z \caps<z,`a> } `a  . `g 
& \red & (\deactR), (\actL) \\
\exp v { \cutL {\exp x \caps<x,`d> `b . `d } `d + z \caps<z,`a> } `a . `g 
& \red & (\Lii) \\§
\exp v { \cut {\exp x { \cutL \caps<x,`d> `d + z \caps<z,`a> } `b . `e } `e + z \caps<z,`a> } `a . `g 
& \red & (\deactL), (\Cap), (\Exp) \\
\exp v { \exp x \caps<x,`a> `b . `a } `a . `g 
 \end{array} \]

Here the subject reduction property collapses in the $(\Rii)$ step.
We cannot type the net $\exp v { \cutR {\exp x \caps<x,`d> `b . `d } `d + z \caps<z,`a> } `a  . `g $ with the contexts above: the only way we can manage is:
 \[
 \Inf	
	{ \Inf	
		{ \Inf	
			{ \Inf	
				{ \derX \caps<x,`d> : v{:}C,x{:}A |- `d{:}A,`b{:}B }
			}
			{ \derX \exp x \caps<x,`d> `b . `d : v{:}C |- `d{:}A\un (A\arr B) }
		  \quad
		  \Inf	[\unL]
		  	{ \Inf	{ \derX \caps<z,`a> : z{:}A |- `a{:}A } 
			  \Inf	{ \derX \caps<z,`a> : z{:}A\arr B |- `a{:}A\arr B } 
			}
			{ \derX \caps<z,`a> : z{:}A\un (A\arr B) |- `a{:}A\un (A\arr B) }
		}
		{ \derX \cutR {\exp x \caps<x,`d> `b . `d } `d + z \caps<z,`a> : v{:}C |- `a{:}A\un (A\arr B) }
	}
	{ \derX \exp v { \cutR {\exp x \caps<x,`d> `b . `d } `d + z \caps<z,`a> } `a  . `g 
 : {} |- `g{:}C\arr A\un (A\arr B) }	
 \]

}
The {\CBV} reduction, on the other hand,  
\Long{
 \[ \begin{array}{lcl}
\cut {\exp x \caps<x,`d> `b . `d } `d + z {\exp v \caps<z,`a> `a . `g } 
& \redCBV & 
\exp v { \exp x { \exp y \caps<x,`e> `e . `g } `b . `a } `a . `g 
 \end{array} \]
}
does not create a problem\Short{.}%
\Long{:

 %
 \[ \begin{array}{lcl}
\cut {\exp x \caps<x,`d> `b . `d } `d + z {\exp v \caps<z,`a> `a . `g } 
& \red & (\actL)
,
(\Lii) \\
\cut {\exp x { \cutL \caps<x,`d> `d + z {\exp v \caps<z,`a> `a . `g } } `b . `s } `s + z {\exp v \caps<z,`a> `a . `g } 
& \red & (\deactL), (\actR), (\Rii), (\deactR), (\Cap) \\
\cut {\exp x { \exp v \caps<x,`a> `a . `g } `b . `s } `s + z {\exp v \caps<z,`a> `a . `g } 
& \red & (\actR) 
,
(\Rii), (`a)  \\
\exp v { \cutR {\exp x { \exp y \caps<x,`e> `e . `g } `b . `s } `s + z \caps<z,`a> } `a . `g 
& \red & (\deactR), (\Exp) \\
\exp v { \exp x { \exp y \caps<x,`e> `e . `g } `b . `a } `a . `g 
 \end{array} \]
and we can type the last net as follows:

 \[ 
\Inf	
	{ \Inf	
		{ \Inf	
			{ \Inf	
				{ \derX \caps<x,`e> : v{:}C,x{:}A,y{:}C |- `e{:}A,`b{:}B }
			}
			{ \derX \exp y \caps<x,`e> `e . `g  : v{:}C,x{:}A |- `g{:}C\arr A,`b{:}B } 
		}
		{ \derX \exp x { \exp y \caps<x,`e> `e . `g } `b . `a : v{:}C |- `a{:}A\arr B,`g{:}C\arr A }
	}
	{ \derX \exp v { \exp x { \exp y \caps<x,`e> `e . `g } `b . `a } `a . `g : {} |- `g{:}(C\arr A)\un (C\arr A\arr B) }
\]
So {\CBV}-reduction poses no problem in this case.
}
 \end{example}

\Long{
These \Short{two }\Long{three }examples show that it is not possible to obtain a general witness-reduction result; we can (partially) recover from that, as we will show in the next section.
}

\Long{
In the remainder of this section we will show eight cases in which the proof of the witness-reduction property collapses\Comment{, that directly correspond to the noted problem cases}.
Each of these cases involve a propagation rule: the first four are cases in which left-propagation goes wrong, the last four when this happens when right-propagating.

 \begin{example}[\Lii] 
$ \cutL { \exp y P `b . `a } `a + x Q \red \cut { \exp y { \cutL P `a + x Q } `b . `g } `g + x Q $, with $`g$ fresh.
Notice that $`a$ is not introduced in $ \exp y P `b . `a $, so might appear inside $P$; also, $y,`b \not \in \FC{Q}$.
%
Suppose the derivation is shaped like: 
\Comment{
 \[ 
 \Inf	
	{ \Inf	[\intR]
		{ \Inf	
			{ \InfBox{\D^i_1}
				{ \derX P : `G,y{:}A_i |- `b{:}B_i,`a{:}C_i,`D } 
			}
			{ \derPure \exp y P `b . `a : `G |- `a{:}(A_i\arr B_i)\union C_i,`D } 
		  ~ (\forall i \ele \n)
		}
		{ \derX \exp y P `b . `a : `G |- `a{:}\int{n}({(A_i\arr B_i)\union C_i}),`D } 
	  \quad
	  \InfBox{\D_2}
		{ \derPure Q : `G,x{:}\int{n}({(A_i\arr B_i)\union C_i}) |- `D }
	}
	{ \derPure \cutL {\exp y P `b . `a } `a + x Q : `G |- `D }
 \]
Notice that then $\D_2$ derives a union type for $x$, which is in conflict with the fact that the derivation is pure.
So, in fact, $C_i = A_i\arr B_i$, and we get 
}
 \[ 
 \Inf	
	{
	 \Inf	[\intR]
		{ \Inf	
			{ \InfBox{\D^i_1}
				{ \derX P : `G,y{:}A_i |- `b{:}B_i,`a{:}A_i\arr B_i,`D } 
			}
			{ \derPure \exp y P `b . `a : `G |- `a{:}A_i\arr B_i,`D } 
		  ~ (\forall i \ele \n)
		}
		{ \derX \exp y P `b . `a : `G |- `a{:}\int{n}(A_i\arr B_i),`D } 
	  \quad
	  \InfBox{\D_2}
		{ \derPure Q : `G,x{:}\int{n}(A_i\arr B_i) |- `D }
	}
	{ \derPure \cutL {\exp y P `b . `a } `a + x Q : `G |- `D }
 \]

Now, for the right-hand side, we can at most derive:
 \[
 \Inf	[?]
 	{ \Inf	
		{ \Inf	
			{ \Inf	[\intR]
				{ \InfBox{\D^1_j}
					{ \derX P : `G,y{:}A_j |- `b{:}B_j,`a{:}A_j\arr B_j,`D } 
				  \quad (\forall j \ele \n)
				}
				{ \derX P : `G,y{:}\int{n}{A_i} |- `b{:}\un{n}{B_i},`a{:}\int{n}(A_i\arr B_i),`D }
		  \quad
			  \InfBox{\D_2}
				{ \derPure Q : `G,x{:}\int{n}(A_i\arr B_i) |- `D }
			}
			{ \derPure \cutL P `a + x Q  : `G,\int{n}{A_i} |- `b{:}\un{n}{B_i},`D } 
		}
		{ \derX \exp y { \cutL P `a + x Q } `b . `g : `G |- `g{:}\int{n}{A_i}\arrow\un{n}{ B_i},`D }
	  \kern-2cm 
	  \InfBox{\D_2}
		{ \derPure Q : `G,x{:}\int{n}(A_i\arr B_i) |- `D }
	}
	{ \derPure \cut { \exp y {\cutL P `a + x Q } `b . `g } `g + x Q : `G |- `D }
\]
\Comment{
If $`a$ is introduced, we can derive:
 \[
 \Inf	
 	{ \Inf [\intR]
		{ \Inf	
			{ \Inf	
				{ \Inf	[\Weak]
					{ \InfBox{\D^1_j}
						{ \derX P : `G,y{:}A_j |- `b{:}B_j,`D } 
					}
					{ \derX P : `G,y{:}A_j |- `b{:}B_j,`a{:}\int{n}(A_i\arr B_i),`D } 
				  \quad
					  \InfBox{\D_2}
						{ \derPure Q : `G,x{:}\int{n}(A_i\arr B_i) |- `D }
				}
				{ \derPure \cutL P `a + x Q  : `G,A_j |- `b{:}B_j,`D } 
			}
			{ \derX \exp y { \cutL P `a + x Q } `b . `g : `G |- `g{:}A_j\arrow B_j,`D }
			\kern -3cm (\forall j \ele \n)
		}
		{ \derX \exp y { \cutL P `a + x Q } `b . `g : `G |- `g{:}\int{n}(A_i\arrow B_i),`D }
	  \quad
	  \InfBox{\D_2}
		{ \derPure Q : `G,x{:}\int{n}(A_i\arr B_i) |- `D }
	}
	{ \derPure \cut { \exp y {\cutL P `a + x Q } `b . `g } `g + x Q : `G |- `D }
\]
for every $j \ele \n$.
}


Notice that $\int{n}{A_i}\arrow\un{n}{ B_i} \not= \int{n}(A_i\arrow B_i)$, so we can't close the $(?)$ gap.

\end{example}

\begin{example}[\Liii] $ \cutL { \exp y P `b . `g } `a + x Q \red \exp y { \cutL P `a + x Q } `b . `g $, with $`g \not = `a$.
%
Suppose the derivation is shaped like: 
 \[ 
\Inf	
	{ \Inf	[\intR]
		{ \Inf	
			{ \InfBox{\D^i_1}{ \derX P : `G ,y{:}A_i |- `b{:}B_i,`a{:}C_i,`D } }
			{ \derPure \exp y P `b . `g : `G |- `g{:}A_i\arr B_i,`a{:}C_i,`D }
		  \quad (\forall i \ele \n)
		}
		{ \derX \exp y P `b . `g : `G |- `g{:}{\un{n}( A_i\arr B_i )},`a{:}\CoI{n},`D }
	  \quad
	  \InfBox{\D_2}{ \derPure Q : `G, x{:}\CoI{n} |- `D }
	}
	{ \derPure \cutL {\exp y P `b . `g } `a + x Q : `G |- `g{:}\un{n}(A_i\arr B_i),`D }
 \]
Based on the sub-derivations, the only thing we can derive for the right-hand side is: 
 \[ \Inf	
	{ \Inf	
		{ \Inf	[\intR]
			{ \InfBox{\D^i_1}{ \derPure P : `G ,y{:}A_i |- `b{:}B_i,`a{:}C_i,`D } 
			  \quad (\forall i \ele \n)
			}
			{ \derX P : `G ,y{:}\AoI{n} |- `b{:}\BoU{n},`a{:}\CoI{n},`D }
		  \quad
		  \InfBox{\D_2}{ \derPure Q : `G, x{:}\CoI{n} |- `D }
		}
		{ \derPure \cutL P `a + x Q : `G ,y{:}\AoI{n} |- `b{:}\BoU{n},`D }
	}
	{ \derPure \exp y { \cutL P `a + x Q } `b . `g : `G |- `g{:}\AoI{n}\arr\BoU{n},`D }
\]
Notice that $\un{n}( A_i\arr B_i ) \not= \AoI{n}\arr\BoU{n}$.

\end{example}

\begin{example}[\Liv] $ \cutL { \imp P `b [z] y Q } `a + z P \red \imp { \cutL P `a + x R } `b [z] y { \cutL Q `a + x R } $.
%
Suppose the derivation is shaped like: 
 \[ \kern-1cm
 \Inf	{ \Inf	[\intR]
		{ \Inf	{ \InfBox{\D^i_1}
				{ \derX P : `G |- `a{:}C_i,`b{:}A_i,`D }
			  \quad
			  \InfBox{\D^i_2}
				{ \derX Q : `G,y{:}B_i |- `a{:}C_i,`D }
			}
			{ \derPure \imp P `b [z] y Q : `G\inter z{:}A_i\arr B_i |- `a{:}C_i,`D }
		(\forall i \ele \n)
		}
		{ \derX \imp P `b [z] y Q : `G\inter\, z{:}\int{n}{(A_i\arr B_i)} |- `a{:}\CoI{n},`D }
	  \quad
	  \InfBox{\D_3}
		{ \derPure R : `G,x{:}\CoI{n} |- `D }
	}
	{ \derPure \cutL { \imp P `b [z] y Q } `a + x R : `G\inter\, z{:}\int{n}{(A_i\arr B_i)} |- `D }
 \]
Again, we can only derive:
 \[ \kern-3cm
 \Inf	{ \Inf	{ \Inf	[\intR]
			{ \InfBox{\D^i_1}
				{ \derPure P : `G |- `a{:}C_i,`b{:}A_i,`D } 
			  ~(\forall i \ele \n)
			}
			{ \derX P : `G |- `a{:}\CoI{n},`b{:}\AoU{n},`D } 
		  ~
		  \InfBox{\D_3}
			{ \derPure R : `G,x{:}\CoI{n} |- `D } 
		}
		{ \derPure \cutL P `a + x R : `G |- `b{:}\AoU{n},`D } 
\raise1.25\RuleH\hbox to 3cm {\kern-3.5cm
	  \Inf	{ \Inf	[\intR]
			{ \InfBox{\D^i_2}
				{ \derPure Q : `G,y{:}B_i |- `a{:}C_i,`D } 
			  ~(\forall i \ele \n)
			}
			{ \derX Q : `G,y{:}\BoI{n} |- `a{:}\CoI{n},`D } 
		  ~
		  \InfBox{\D_3}
			{ \derPure R : `G,x{:}\CoI{n} |- `D } 
		}
		{ \hspace*{15mm}\derPure \cutL Q `a + x R : `G,y{:}\BoI{n} |- `D }
}
 \multiput(-10,0)(0,7){4}{.}
	}
	{ \derPure \imp { \cutL P `a + x R } `b [z] y { \cutL Q `a + x R } : `G,z{:}\AoU{n}\arrow \BoI{n} |- `D }
 \]
Notice that $\AoU{n}\arr \BoI{n} \not= \int{n}{(A_i\arr B_i)}$.

 \end{example}

\begin{example}[\Lv] $ \cutL { \cut P `b + y Q } `a + x R \red \cut { \cutL P `a + z P } `b + y { \cutL Q `a + x R } $.
%
Suppose the derivation is shaped like: 
 \[ 
 \Inf	{ \Inf	[\intR]
		{ \Inf	{ \InfBox{\D^i_1}
				{ \derX P : `G |- `a{:}C_i,`b{:}A_i,`D }
			  \quad
			  \InfBox{\D^i_2}
				{ \derX Q : `G,y{:}A_i |- `a{:}C_i,`D }
			}
			{ \derPure \cut P `b + y Q : `G |- `a{:}C_i,`D }
		  (\forall i \ele \n)
		}
		{ \derX \cut P `b + y Q : `G |- `a{:}\CoI{n},`D }
	  \quad
	  \InfBox{\D_3}
		{ \derPure R : `G,x{:}\CoI{n} |- `D }
	}
	{ \derPure \cutL { \cut P `b + y Q } `a + x R : `G |- `D }
 \]
 \[ \kern-2.5cm
 \Inf	[?]
	{ \Inf	{ \Inf	[\intR]
			{ \InfBox{\D^i_1}
				{ \derPure P : `G |- `a{:}C_i,`b{:}A_i,`D } 
			  ~(\forall i \ele \n)
			}
			{ \derX P : `G |- `a{:}\CoI{n},`b{:}\AoU{n},`D } 
		  ~
		  \InfBox{\D_3}
			{ \derPure R : `G,x{:}\CoI{n} |- `D } 
		}
		{ \derPure \cutL P `a + x R : `G |- `b{:}\AoU{n},`D } 
\raise1.25\RuleH\hbox to 2.5cm {\kern-3.75cm
	  \Inf	{ \Inf	[\intR]
			{ \InfBox{\D^i_2}
				{ \derPure Q : `G,y{:}A_i |- `a{:}C_i,`D } 
			  ~(\forall i \ele \n) 
			}
			{ \derX Q : `G,y{:}\AoI{n} |- `a{:}\CoI{n},`D } 
		  ~
		  \InfBox{\D_3}
			{ \derPure R : `G,x{:}\CoI{n} |- `D } 
		}
		{ \hspace*{2cm} \derPure \cutL Q `a + x R : `G,y{:}\AoI{n} |- `D }
}
 \multiput(-10,0)(0,7){4}{.}
	}
	{ \derPure \cut { \cutL P `a + x R } `b + y { \cutL Q `a + x R } : `G |- `D }
 \]
Notice that $\AoU{n} \not= \AoI{n}$, so we can't close the $(?)$ gap.

 \end{example}

\begin{example}[\Rii] $ \cutR P `a + x { \exp y Q `b . `g } \red \exp y { \cutR P `a + x Q } `b . `g $.
%
Suppose the derivation is shaped like: 
 \[
\Inf	{
	 \InfBox{\D_1}
		{ \derPure P : `G |- `a{:}\CoU{n},`D }
	  \quad
	  \Inf	[\unL]
		{  \Inf	{ \InfBox{\D^i_2}
				{ \derX Q : `G,x{:}C_i,y{:}A_i |- `b{:}B_i,`D }
			}
			{ \derPure \exp y Q `b . `g : `G,x{:}C_i |- `g{:}A_i\arr B_i\union `D }
		  ~(\forall i \ele \n)
		}
		{ \derX \exp y Q `b . `g : `G,x{:}\CoU{n} |- `g{:}\un{n}(A_i\arr B_i)\union `D }
	}
	{ \derPure \cutR P `a + x { \exp y Q `b . `g } : `G |- `g{:}\un{n}(A_i\arr B_i)\union `D }
 \]
We can now only derive:
 \[
\Inf	{ \Inf	{ \InfBox{\D_1}
			{ \derPure P : `G |- `a{:}\CoU{n},`D }
		  \quad 
		  \Inf	[\unL]
			{ 
			\InfBox{\D^1_2}
				{ \derX Q : `G,x{:}C_i,y{:}A_i |- `b{:}B_i,`D }
			  ~~(\forall i \ele \n)
			}
			{ \derX Q : `G,x{:}\CoU{n},y{:}\AoI{n} |- `b{:}\BoU{n},`D }
		}
		{ \derPure \cutR P `a + x Q : `G,y{:}\AoI{n} |- `b{:}\BoU{n},`D }
	}
	{ \derPure \exp y { \cutR P `a + x Q } `b . `g : `G |- `g{:}\AoI{n}\arr \BoU{n}\union `D }
 \]
Notice that $\un{n}{(A_i\arr B_i)} \not= \AoI{n}\arr \BoU{n}$.

 \end{example}

\begin{example}[\Riii] $ \cutR P `a + x {\imp Q `b [x] y R } \red \cut P `a + v { \imp { \cutR P `a + x Q } `b [v] y { \cutR P `a + x R } } $, with $v$ fresh.
%
Suppose the derivation is shaped like: 
 \[ 
\Inf	
	{ \InfBox{\D_1}
		{ \derPure P : `G |- `a{:}\un{n}({(A_i\arr B_i)\inter C_i}),`D }
	  \Inf	[\unL]
		{ \Inf	
			{ \InfBox{\D^i_2}
				{ \derX Q : `G,x{:}C_i |- `b{:}A_i,`D }
			  \quad 
			  \InfBox{\D^i_3}
				{ \derX R : `G,x{:}C_i,y{:}B_i |- `D }
			}
			{ \derPure \imp Q `b [x] y R : `G,x{:}(A_i\arr B_i)\inter C_i |- `D }
		  (\forall i \ele \n)
		}
		{ \derX \imp Q `b [x] y R : `G,x{:}\un{n}({(A_i\arr B_i)\inter C_i}) |- `D }
	}
	{ \derPure \cutR P `a + x { \imp Q `b [x] y R } : `G |- `D }
 \]
Now, since $\D_1$ is pure, the intersections in the type for $`a$ are not really there, so $C_i = A_i\arrow B_i$, and we get:
 \[ 
\Inf	
	{ \InfBox{\D_1}
		{ \derPure P : `G |- `a{:}\un{n}(A_i\arr B_i),`D }
	  \kern-1cm
	  \Inf	[\unL]
		{ \Inf	
			{ \InfBox{\D^i_2}
				{ \derX Q : `G,x{:}A_i\arr B_i |- `b{:}A_i,`D }
			  \quad 
			  \InfBox{\D^i_3}
				{ \derX R : `G,x{:}A_i\arr B_i,y{:}B_i |- `D }
			}
			{ \derPure \imp Q `b [x] y R : `G,x{:}(A_i\arr B_i) |- `D }
		  (\forall i \ele \n)
		}
		{ \derX \imp Q `b [x] y R : `G,x{:}\un{n}(A_i\arr B_i) |- `D }
	}
	{ \derPure \cutR P `a + x { \imp Q `b [x] y R } : `G |- `D }
 \]
Now we can at best construct:
 \[ \kern-10mm
 \Inf	[?]
	{ \InfBox{\D_1}
		{ \derX P : `G |-  `a{:}\un{n}{(A_i\arr B_i)},`D }
	  ~
	  \Inf	
		{%
\multiput(10,0)(0,7){10}{.}
\raise3\RuleH\hbox to 2cm {\kern-5cm
		  \Inf	
			{ \InfBox{\D_1}
				{ \derX P : `G |-  `a{:}\un{n}{(A_i\arr B_i)},`D }
			  \Inf	[\unL]
				{ \InfBox{\D^i_2}
					{ \derX Q : `G,x{:}A_i\arr B_i |- `b{:}A_i,`D }
				  ~(\forall i \ele \n)
				}
				{ \derX Q : `G,x{:}\un{n}{(A_i\arr B_i)} |- `b{:}\un{n}{A_i},`D }
			}
			{ \derX \cutR P `a + x Q : `G |- `b{:}\un{n}{A_i},`D }
}
\kern-8mm
		  \Inf	
			{ \InfBox{\D_1}
				{ \derX P : `G |-  `a{:}\un{n}{(A_i\arr B_i)},`D }
\raise.9\RuleH\hbox to 3cm {\kern-3cm
			  \Inf	[\unL]
				{ \InfBox{\D^i_3}
					{ \derX R : `G,y{:}B_i,x{:}A_i\arr B_i |- `b{:}A_i,`D }
				  ~(\forall i \ele \n)
				}
				{ \derX R : `G,y{:}\int{n}{B_i},x{:}\un{n}{(A_i\arr B_i)} |- `b{:}\un{n}{A_i},`D }
}
\multiput(-10,0)(0,7){3}{.}
			}
			{ \derX \cutR P `a + x R : `G,y{:}\int{n}{B_i} |- `D }
		}
		{ \derX \imp { \cutR P `a + x Q } `b [v] y { \cutR P `a + x R } : `G,v{:}\un{n}{A_i}\arr \int{n}{B_i} |- `D }
	}
	{ \derPure \cut P `a + v { \imp { \cutR P `a + x Q } `b [v] y { \cutR P `a + x R } } : `G |- `D }
 \]
Notice that $\un{n}(A_i\arr B_i) \not= \un{n}{A_i}\arr \int{n}{B_i}$, so we cannot fill the $(?)$ gap.

 \end{example}

\begin{example}[\Riv] $ \cutR P `a + x {\imp Q `b [z] y R } \red \imp { \cutR P `a + x Q } `b [z] y { \cutR P `a + z R },\ x \not= z $.  
%
Let $C = \CoU{n}$.
Suppose the derivation is shaped like: 
 \[
\Inf	{ \InfBox{\D_1}
		{ \derX P : `G |- `a{:}C,`D }
	\kern-1cm
	  \Inf	[\unL]
		{ \Inf	{ \InfBox{\D^i_2}
				{ \derX Q : `G,x{:}C_i |- `b{:}A_i,`D }
			  \dquad
			  \InfBox{\D^i_3}
				{ \derX R : `G,x{:}C_i,y{:}B_i |- `D }
			}
			{ \derX \imp Q `b [z] y R : `G\inter\, z{:}A_i\arr B_i,x{:}C_i |- `D }
		  ~(\forall i \ele \n)
		}
		{ \derX \imp Q `b [z] y R : `G\inter\, z{:}\int{n}(A_i\arr B_i),x{:}C |- `D }
	}
	{ \derX \cutR P `a + x { \imp Q `b [z] y R } : `G\inter\, z{:}\int{n}(A_i\arr B_i) |- `D }
 \]
Then the only derivation we can construct is:
 \[ 
\Inf	{ \Inf	{ \InfBox{\D_1}
			{ \derX P : `G |- `a{:}C,`D }
		  ~
		  \Inf	[\unL]
			{ \InfBox{\D^i_2}
				{ \derX Q : `G,x{:}C_i |- `b{:}A_i,`D }
			}
			{ \derX Q : `G,x{:}C |- `b{:}\AoU{n},`D }
		}
		{ \derX \cutR P `a + x Q : `G |- `b{:}\AoU{n},`D }
	  \Inf	{ \InfBox{\D_1}
			{ \derX P : `G |- `a{:}C,`D }
		  ~
		  \Inf	[\unL]
			{ \InfBox{\D^i_3}
				{ \derX R : `G,x{:}C_i,y{:}B_i |- `D }
			}
			{ \derX R : `G,x{:}C,y{:}\int{n}{B_i} |- `D }
		}
		{ \derX \cutR P `a + z R : `G ,y{:}\int{n}{B_i} |- `D }
	}
	{ \derX \imp { \cutR P `a + x Q } `b [z] y { \cutR P `a + z R } : `G\inter\, z{:}\AoU{n}\arr \int{n}{B_i} |- `D }
 \]
Notice that $\int{n}(A_i\arr B_i) \not= \AoU{n}\arr \int{n}{B_i}$.
 \end{example}

\begin{example}[\Rv] $ \cutR P `a + x { \cut Q `b + y R } \red \cut { \cutR P `a + x Q } `b + y { \cutR P `a + z R } $.
%
Let $C = \CoU{n}$.
Suppose the derivation is shaped like: 
 \[
\Inf	{ \InfBox{\D_1}
		{ \derX P : `G |- `a{:}C,`D }
\kern-1cm
	  \Inf	[\unL]
		{ \Inf	{ \InfBox{\D^i_2}
				{ \derX Q : `G,x{:}C_i |- `b{:}A_i,`D }
			  \dquad
			  \InfBox{\D^i_3}
				{ \derX R : `G,x{:}C_i,y{:}A_i |- `D }
			}
			{ \derX \cut Q `b + y R : `G,x{:}C_i |- `D }
		  ~(\forall i \ele \n)
		}
		{ \derX \cut Q `b + y R : `G,x{:}C |- `D }
	}
	{ \derX \cutR P `a + x { \cut Q `b + y R } : `G |- `D }
 \]
 \[ 
 \Inf	[?]
	{ \Inf	{ \InfBox{\D_1}
			{ \derX P : `G |- `a{:}C,`D }
		  ~
		  \Inf	[\unL]
			{ \InfBox{\D^i_2}
				{ \derX Q : `G,x{:}C_i |- `b{:}A_i,`D }
			}
			{ \derX Q : `G,x{:}C |- `b{:}\AoU{n},`D }
		}
		{ \derX \cutR P `a + x Q : `G |- `b{:}\AoU{n},`D }
	  \Inf	{ \InfBox{\D_1}
			{ \derX P : `G |- `a{:}C,`D }
		  ~
		  \Inf	[\unL]
			{ \InfBox{\D^i_3}
				{ \derX R : `G,x{:}C_i,y{:}A_i |- `D }
			}
			{ \derX R : `G,x{:}C,y{:}\AoI{n} |- `D }
		}
		{ \derX \cutR P `a + z R : `G ,y{:}\AoI{n} |- `D }
	}
	{ \derX \cut { \cutR P `a + x Q } `b + y { \cutR P `a + z R } : `G |- `D }
 \]
Notice that $\AoU{n} \not= \AoI{n}$, so we cannot fill the $(?)$ gap.

 \end{example}

}
So, this notion of type assignment is too liberal to obtain preservation of types under conversion: analysing the problems above, we can summarise them by: both
 \begin{flatenumerate}
 \item \label{no right into union}
right-propagation into union, and
 \item \label{no left into intersection}
left-propagation into intersection
 \end{flatenumerate}
break the witness-reduction property.

\Long{
In this paper, the solution to the above problem is to `flatten' the above context assignment rules in Section~\ref{CBN logical} by defining {\em two} separate systems, $\derXRN P : `G |- `D $ and $\derXRV P : `G |- `D $, as in Definition \ref{CBN logical type assignment} and 
\ref{CBV logical type assignment}, that treat either intersection or union more as their logical counterpart.
The first system will in fact give a notion of context assignment for the {\CBN} restriction which already solves problem \ref{no left into intersection}, and solves problem \ref{no right into union} by restricting the union type on plugs to those obtained by weakening, and not allowing union types on the right of the arrow type constructor.
After that, we will define two restrictions in Section~\ref{CBN}, $\derXN P : `G |- `D $ and $\derXV P : `G |- `D $, that either limit intersections to nets that introduce plugs, or union to nets that introduce sockets.
 }

\Long{ \include{WRrestricted}}
}
{ \def \DBwnameAndRule#1[#2]#3{#3}


\Long{ \section{A system with preservance of types under {\CBN} reduction} \label{CBN} }
\Short{ \section{Systems with preservance of types under {\CBN} or {\CBV} reduction} \label{CBN} }

In this section, we will try and retrieve the witness-reduction property using a 
restriction of the system
proposed in the previous section.
The approach we choose here is, in fact, partially inspired by \cite{Barbanera-Dezani-Liguoro-IaC'95}, where union types can only be assigned to values.
The solutions we present here are, however, very different: we do not need to limit the structure of types, and, for {\CBN}, limit union types to \emph{names}, i.e.~nets that introduce sockets.
For {\CBV}, we limit intersection types to \emph{values}, i.e.~nets that introduce plugs; this is reminiscent of the limitation in {\ML} of quantification of types to terms that are values \cite{Harper-Lillibridge'91,Milner-et.al'97}, and is used also in \cite{Summers'08}.

We define $\derXN P : `G |- `D $ as in Definition~\ref{intersection and union tas}, where we change the applicability of rule $(\unL)$, and add a rule to treat left-activated cuts:


 \begin{definition}\label{CBN type assignment}

\Long{ \begin{enumerate}

 \firstitem }
The context assignment rules for $\TurnN$ are\Short{ (we only show the changed rules)}:
 \[ \def\arraystretch{2.75} 
 \begin{array}{rl}
 \\[-12mm] 
\Long{
 (\Ax): &
 \Inf	
	{ \derXN \caps<y,`b> : `G ,y{:}A |- `b{:}A,`D }
 \\
}
 (\Cut):&
 \Inf	[\emph{for inactive and right-activated cuts}]
	{ \derXN P : `G |- `a{:}A,`D \quad \derXN Q : `G, \stat{x}{A} |- `D }
	{ \derXN \cut P `a + x Q : `G |- `D }
 \\
(\DaggerL):&
 \Inf	[ A\emph{ not an intersection type},\textit{$x$ introduced}]
	{ \derXN P : `G |- `a{:}A,`D \quad \derXN Q : `G, \stat{x}{A} |- `D }
	{ \derXN \cutL P `a + x Q : `G |- `D }
 \\
\Long{
 (\arrL):&
 \Inf	{ \derXN P : `G |- `a{:}A,`D \quad \derXN Q : `G,\stat{x}{B} |- `D }
	{ \derXN \imp P `a [y] x Q : `G \inter \stat{y}{A\arr B} |- `D }
 \dquad
 (\arrR):
 \Inf	{ \derXN P : `G,\stat{x}{A} |- `a{:}B,`D }
	{ \derXN \exp x P `a . `b : `G |- \stat{`b}{A\arr B}\union `D }
 \\
(\unL): &
 \Inf	[n \geq 0,x\emph{ introduced in }P]
	{ \derXN P : `G,x{:}A_i |- `D 
		\quad (\forall i \ele \n) }
	{ \derXN P : `G, x{:}\AoU{n} |- `D }
\\
(\intR): &
 \Inf	[n \geq 0]
	{ \derXN P : `G |- `a{:}A_i,`D \quad (\forall i \ele \n) }
	{ \derXN P : `G |- `a{:}\AoIn,`D }
}
\Short{
(\unL): &
 \Inf	[n \geq 0,x\emph{ introduced in }P]
	{ \derXN P : `G,x{:}A_i |- `D 
	  \quad (\forall i\ele \n)
	  }
	{ \derXN P : `G, x{:}\AoU{n} |- `D }
}
 \end{array} \]

\Long{ \item
Pure derivations are those that end with either rule $(\Ax)$, $(\Cut)$, $(\DaggerL)$, $(\arrL)$, and $(\arrR)$.

 \end{enumerate}}

 \end{definition}

The {\CBN} reduction
\Long{%
 \[\begin{array}{rcl}
\Long{\cut  \SemL{xx}{`a} `a + y \SemL{yy}{`b} }
\Short{\cut {\imp \caps<x,`g> `g [x] v \caps<v,`a> } `a + y {\imp \caps<y,`d> `d [y] w \caps<w,`b> } }
& \redCBN & \imp \caps<x,`g> `g [x] v { \imp {\imp \caps<x,`g> `g [x] u \caps<u,`d> } `d [v] w \caps<w,`b> } 
\end{array}
\] }
however, does not pose any problems\Short{.}\Long{
:
 \[ \begin{array}[t]{lcl}
\cut  \SemL{xx}{`a} `a + y \SemL{yy}{`b} 
& \red & (\actR) 
, (\Riii) \\
\cut \SemL{xx}{`a} `a + z { \imp {\cutR \SemL{xx}{`a} `a + y \caps<y,`d> } `d [z] w {\cutR \SemL{xx}{`a} `a + y \caps<w,`b> } }
& \red & (\deactR), (\actL), (\Ri) \\
\cut \SemL{xx}{`a} `a + z { \imp {\cutL \SemL{xx}{`a} `a + y \caps<y,`d> } `d [z] w \caps<w,`b> }
& \red & (\renR) \\
\cut \SemL{xx}{`a} `a + z { \imp \SemL{xx}{`d} `d [z] w \caps<w,`b> } 
& \red & (\actL) 
, (\Imp) \\
\cutL {\imp \caps<x,`g> `g [x] v \caps<v,`a> } `a + z { \imp \SemL{xx}{`d} `d [z] w \caps<w,`b> } 
& \red & (\Liv), (\Li), (\deactL), (\Med) \\
\imp \caps<x,`g> `g [x] v { \imp \SemL{xx}{`d} `d [v] w \caps<w,`b> } 
& \red & (\Imp) \\
\imp \caps<x,`g> `g [x] v { \imp {\imp \caps<x,`g> `g [x] u \caps<u,`d> } `d [v] w \caps<w,`b> }
 \end {array} \]
}


\Long{ \begin{remark}
Before coming to the proof that the system is closed for reduction, w}%
\Short{W}e can easily verify that this notion of type assignment is \emph{not closed} for witness expansion.
This is clear from the fact that the side-condition of rule $(\unL)$ is not preserved by witness expansion: take $\cut {\exp y P `b . `a } `a + x \caps<x,`g> $ such that $`a$ is introduced, and $`g$ does not appear in $P$, then $`g$ is introduced in the net that is the result of contracting this cut, $\exp y P `b . `g $, but not in the net $\cut {\exp y P `b . `a } `a + x \caps<x,`g> $ itself.
\Long{ \end{remark} }

\Long{We will see that t}%
\Short{T}he addition of rule $(\DaggerL)$ solves the problem of `left propagation into intersection' in the context of {\CBN} reduction, and that the restriction on rule $(\unL)$ solves `right propagation into union'.

 \begin{theorem}[Witness reduction for $\TurnN$ wrt \CBN] \label{witness reduction for CBN}
If $\derXN P : `G |- `D $, and $P \redCBN Q$, then $\derXN Q : `G |- `D $.
 \end{theorem}

\Long{

 \begin{proof}
By induction on the definition of $\red$, where we focus on the rules: the proof consists of showing, for each rule, the 'minimal' derivation for the left-hand side, and that, using the restrictions that poses, we can type the right-hand side.
We only show the interesting cases.

In general, by Lemma~\ref{pure Gen lemma}, the derivation for the left-hand side is shaped like 

 \begin{enumerate}

 \item \label{intersection case} \correct
\[ \begin{array}[t]{c} 
\Inf	[\Cut]
	{ \Inf	[\intR]
		{ \InfBox{ \derPure P : `G |- `a{:}C_i,`D }
		  \quad (\forall i \ele \n)
		}
		{ \derXN P : `G |- `a{:}\CoI{n},`D }
	  \quad
	  \InfBox{ \derPure Q : `G, x{:}\CoI{n} |- `D }
	}
	{ \derPure \cut P `a + x Q : `G |- `D }
\end{array} \]
 
 \item \label{union case} \correct
\[ \begin{array}[t]{c} 
\Inf	[\Cut]
	{ \InfBox{ \derPure P : `G |- `a{:}\CoU{n},`D }
	  \quad
	  \Inf	[\unL]
		{ \InfBox{ \derPure Q : `G, x{:}C_i |- `D }
		  \quad (\forall i \ele \n)
		}
		{ \derXN Q : `G, x{:}\CoU{n} |- `D }
	}
	{ \derPure \cut \caps<y,`a> `a + x Q : `G |- `D }
\end{array} \]

 \item \label{pure case} \correct
\[ \begin{array}[t]{c} 
\Inf	[\Cut]
	{ \InfBox{ \derPure P : `G |- `a{:}C,`D }
	  \dquad
	  \InfBox{ \derPure Q : `G, x{:}C |- `D }
	}
	{ \derPure \cut P `a + x Q : `G |- `D }
\end{array} \]

\end{enumerate}

We will refer to these respectively as the \emph{intersection case} (case \ref{intersection case}), the \emph{union case} (case \ref{union case}) or the \emph{pure case} (case \ref{pure case}). 

 \begin{description}\itemsep5pt
\Comment{

 \item[Logical rules]

 \begin{description}


 \item[\Ax] $\cut \caps<y,`a> `a + x \caps<x,`b> \red \caps<y,`b>$.
We consider the three cases:

 \itemindent-.7cm 
 \begin{wideitemize}
 
\item[\ref{intersection case}] \correct
\[ 
 \Inf	
	{ \Inf	[\intR]
		{ \Inf	{ \derPure \caps<y,`a> : `G\inter y{:}C_i |- `a{:}C_i,`D }
		  \quad 
		  (\forall i \ele \n)
		}
		{ \derXN \caps<y,`a> : `G\inter y{:}\CoI{n} |- `a{:}\CoI{n},`D }
	  \Inf	{ \derPure \caps<x,`b> : `G, x{:}\CoI{n} |- `b{:}C_j\union `D }
	}
	{ \derPure \cut \caps<y,`a> `a + x \caps<x,`b> : `G\inter y{:}\CoI{n} |- `b{:}C_j\union `D }
 \]
Notice that $C_i \ele \Tstrict$ for $\iotn$; by rule $(\Ax)$, we can derive:
 \[ 
\Inf	{ \derPure \caps<y,`b> : `G\inter y{:}\CoI{n} |- `b{:}C_j\union `D }
 \]

\item[\ref{union case}] \correct
\[ 
 \Inf	
	{ \Inf	{ \derPure \caps<y,`a> : `G\inter y{:}C_j |- `a{:}\CoU{n},`D }
	  \Inf	[\unL]
		{ \Inf	{ \derPure \caps<x,`b> : `G, x{:}C_i |- `b{:}C_i\union `D }
		  \quad 
		  (\forall i \ele \n)
		}
		{ \derXN \caps<x,`b> : `G, x{:}\CoU{n} |- `b{:}\CoU{n}\union `D }
	}
	{ \derPure \cut \caps<y,`a> `a + x \caps<x,`b> : `G\inter y{:}C_j |- `b{:}\CoU{n}\union `D }
 \]
Notice that $\caps<x,`b>$ introduces $x$ and that $C_i \ele \Tstrict$ for $\iotn$; by rule $(\Ax)$, we can derive:
 \[ 
\Inf	{ \derPure \caps<y,`b> : `G\inter y{:}C_j |- `b{:}\CoU{n}\union `D }
 \]

\item[\ref{pure case}] \correct
\[ 
 \Inf	
	{ \Inf	{ \derPure \caps<y,`a> : `G\inter y{:}A |- `a{:}A,`D }
	  \quad
	  \Inf	{ \derPure \caps<x,`b> : `G, x{:}A |- `b{:}A\union `D }
	}
	{ \derPure \cut \caps<y,`a> `a + x \caps<x,`b> : `G\inter y{:}A |- `b{:}A\union `D }
 \]
Then $A \ele \Tstrict$, and we can derive:
 \[
\Inf	{ \derPure \caps<y,`b> : `G\inter y{:}A |- `b{:}A\union `D }
 \]

 \end{wideitemize}
 

 \item[\Exp] $ \cut { \exp y P `b . `a } `a + x \caps<x,`g> \red \exp y P `b . `g $, with $`a \not\in \FP{P}$; we can assume that $`a$ does not occur in $`D$, but cannot assume that for $`g$.
Again, we have three cases: 

 \begin{wideitemize}

\item[\ref{intersection case}] \correct
\[ 
 \Inf	
	{ \Inf	[\intR]
		{ \Inf	
			{ \InfBox{\D_i}
				{ \derXN P : `G,y{:}A_i |- `b{:}B_i,`D } 
			}
			{ \derPure \exp y P `b . `a : `G |- `a{:}A_i\arr B_i,`D }
		  \quad (\forall i \ele \n)
		}
		{ \derXN \exp y P `b . `a : `G |- `a{:}\int{i}(A_i\arr B_i),`D }
	  \Inf	
		{ \derPure \caps<x,`g> : `G,x{:}\int{i}(A_i\arr B_i) |- `g{:}A_j\arr B_j\,\union `D }
	}
	{ \derPure \cut { \exp y P `b . `a } `a + x \caps<x,`g> : `G |- `g{:}A_j\arr B_j\,\union `D }
 \]
For the right-hand side we can construct:
 \[ 
 \Inf	
	{ \InfBox{\D_j}
		{ \derXN P : `G,y{:}A_j |- `b{:}B_j,`D } 
	}
	{ \derPure \exp y P `b . `g : `G |- `g{:}(A_j\arr B_j)\,\union `D }
 \]

\item[\ref{union case}] \correct 

\[ \kern-4.9cm 
 \Inf	
	{
	  \Inf	[\Weak]
		{ \Inf	
			{ \InfBox{\D}
				{ \derXN P : `G,y{:}A |- `b{:}B,`D } 
			}
			{ \derXN \exp y P `b . `a : `G |- `a{:}A\arr B,`D }
		}
		{ \derXN \exp y P `b . `a : `G |- `a{:}(A\arr B)\union \CoU{n},`D }
\raise.75\RuleH\hbox to 5cm{\kern-13mm 
	  \Inf	[\unL]
		{ \Inf	{(\forall i \ele \n)}
			{ \derPure \caps<x,`g> : `G, x{:}C_i |- `g{:}C_i\union `D }
		  \Inf	{\derPure \caps<x,`g> : `G, x{:}A\arr B |- `g{:}A\arr B\union `D } 
		}
		{ \derXN \caps<x,`g> : `G, x{:}(A\arr B)\union \CoU{n} |- `g{:}(A\arr B)\union \CoU{n}\, \union `D }
}
	}
	{ \derPure \cut {\exp y P `b . `a } `a + x \caps<x,`g> : `G |- `g{:}(A\arr B)\union \CoU{n}\, \union `D }
 \]
Notice that $\caps<x,`g>$ introduces $x$ and that, since $`a \not\in \fp(P)$, $`a{:}\CoU{n}$ is added by weakening.
Then for the right-hand side we can construct:
 \[ 
\Inf	[\Weak]
	{ \Inf	
		{ \InfBox{\D}
			{ \derXN P : `G,y{:}A |- `b{:}B,`D } 
		}
		{ \derPure \exp y P `b . `g : `G |- `g{:}(A\arr B)\union `D }
	}
	{ \derPure \exp y P `b . `g : `G |- `g{:}(A\arr B)\union \CoU{n}\, \union `D }
 \]

\item[\ref{pure case}] \correct
\[ \Inf	
	{ \Inf	
		{ \InfBox{\D}
			{ \derXN P : `G,y{:}A |- `b{:}B,`D } 
		}
		{ \derPure \exp y P `b . `a : `G |- `a{:}A\arr B,`D }
	  \Inf	
		{ \derPure \caps<x,`g> : `G,x{:}A\arr B |- `g{:}A\arr B\union `D }
	}
	{ \derPure \cut { \exp y P `b . `a } `a + x \caps<x,`g> : `G |- `g{:}A\arr B\union `D }
 \]
Notice that then
 \[ 
\Inf	[\arrR]
	{ \InfBox{\D}
		{ \derXN P : `G,y{:}A |- `b{:}B,`D } 
	}
	{ \derXN \exp y P `b . `g : `G |- `g{:}A\arr B\union `D }
 \]

 \end{wideitemize}


 \item[\Imp] $ \cut \caps<y,`a> `a + x { \imp P `b [x] z Q } \red \imp P `b [y] y Q $, with $x \not\in \FS{P,Q}$; we can assume that $x$ does not occur in $`G$, but cannot assume that for $y$. 
The three cases are:

 \begin{wideitemize}

\item[\ref{intersection case}] Let $ C = \CoI{n}$.\\[-12pt]
 \[ \kern4.25cm
\Inf	
	{
\raise3.5\RuleH \hbox to 5cm{\kern -4.5cm
	  \Inf	[\intR]
		{ \Inf	
			{ \derPure \caps<y,`a> : `G\inter y{:}C_i |- `a{:}C_i,`D }
			(\forall i \ele \n)~
		  \Inf	
			{ \derPure \caps<y,`a> : `G\inter y{:}A\arr B |- `a{:}A\arr B,`D }
		}
		{ \derXN \caps<y,`a> : `G\inter y{:}(A\arr B)\inter C |- `a{:}(A\arr B)\inter C,`D }
}
	  \Inf	[\Weak]
		{ \Inf	
			{ \InfBox{\D_2}
				{ \derPure P : `G |- `b{:}A,`D } 
			  \quad
			  \InfBox{\D_3}
				{ \derPure Q : `G,z{:}B |- `D } 
			}
			{ \derPure \imp P `b [x] z Q : `G,x{:}A\arr B |- `D }
		}
		{\derPure \imp P `b [x] z Q : `G,x{:}(A\arr B)\inter C |- `D }
	}
	{ \derPure \cut \caps<y,`a> `a + x { \imp P `b [x] z Q } : `G\inter y{:}(A\arr B)\inter C |- `D }
 \]
Since $x \notele P,Q$, $x{:}C$ has been added by weakening; we can derive:
 \[
 \Inf	[\Weak]
	{ \Inf	
		{ \InfBox{\D_2}
			{ \derPure P : `G |- `b{:}A,`D } 
		  \dquad 
		  \InfBox{\D_3}
			{ \derPure Q : `G,z{:}B |- `D } 
		}
		{ \derPure \imp P `b [y] z Q : `G\inter y{:}A\arr B |- `D }
	}
	{ \derPure \imp P `b [y] z Q : `G\inter y{:}(A\arr B)\inter C |- `D }
 \]

\item[\ref{union case}] Let $(A_j\arr B_j)\union C = \un{n}( A_i\arr B_i )$. 
\[ 
 \Inf	
	{ \Inf	
		{ \derPure \caps<y,`a> : `G\inter y{:}A_j\arr B_j |- `a{:}(A_j\arr B_j)\union C ,`D }
	  \Inf	[\unL]
		{ \Inf	
			{ \InfBox{\D_2^i}
				{ \derXN P : `G |- `b{:}A_i,`D } 
			  \quad
			  \InfBox{\D_3^i}
				{ \derXN Q : `G, z{:}B_i |- `D } 
			}
			{ \derPure \imp P `b [x] z Q : `G , x{:}A_i\arr B_i |- `D }
		  ~ (\forall i \ele n)
		}
		{ \derXN \imp P `b [x] z Q : `G , x{:}\un{n}( A_i\arr B_i ) |- `D }
	}
	{ \derPure \cut \caps<y,`a> `a + x { \imp P `b [x] z Q } : `G\inter y{:}A_j\arr B_j |- `D }
 \]
Since $x$ does not appear free in $P$ or $Q$, it is introduced in $ \imp P `b [x] z Q $, so rule $(\unL)$ is applied correctly.
 \[
 \Inf	[\arrL]
	{ \InfBox{\D_2^j}
		{ \derXN P : `G |- `b{:}A_j,`D } 
	  \dquad 
	  \InfBox{\D_3^j}
		{ \derXN Q : `G, z{:}B_j |- `D } 
	}
	{ \derXN \imp P `b [y] z Q : `G\inter y{:}A_j\arr B_j |- `D }
 \] 

\item[\ref{pure case}] \correct
 \[ 
\Inf	
	{ \Inf	
		{ \derPure \caps<y,`a> : `G\inter y{:}A\arr B |- `a{:}A\arr B ,`D }
	  \Inf	[\arrL]
		{ \InfBox{\D_2}
			{ \derPure P : `G |- `b{:}A,`D } 
		  \quad
		  \InfBox{\D_3}
			{ \derPure Q : `G, z{:}B |- `D } 
		}
		{ \derPure \imp P `b [x] z Q : `G , x{:}A\arr B |- `D }
	}
	{ \derPure \cut \caps<y,`a> `a + x { \imp P `b [x] z Q } : `G\inter y{:}A\arr B |- `D }
 \]
 \[
 \Inf	[\arrL]
	{ \InfBox{\D_2}
		{ \derPure P : `G |- `b{:}A,`D } 
	  \dquad 
	  \InfBox{\D_3}
		{ \derPure Q : `G, z{:}B |- `D } 
	}
	{ \derPure \imp P `b [y] z Q : `G\inter y{:}A\arr B |- `D }
 \]

 \end{wideitemize}


 \item[\Ins] $ \begin{array}{rcl}
 \cut { \exp y P `b . `a } `a + x { \imp Q `g [x] z R }
\red \left \{ \begin{array}{l}
	\cut { \cut Q `g + y P } `b + z R \\
	\cut Q `g + y {\cut P `b + z R } 
\end{array} \right.
\end{array} $, with $`a \not\in \FP{R}, x \not\in \FS{Q,R}$.

 \begin{wideitemize}

\item[\ref{intersection case}]
Let $(A_j \arr B_j) \inter C = \int{n}{(A_i \arr B_i)}$.
 \[ 
 \Inf	
	{
	  \Inf	[\intR]
		{ \Inf	
			{ \InfBox{\D_1^i}
				{\derXN P : `G,y{:}A_i |- `b{:}B_i,`D } 
			}
			{ \derPure \exp y P `b . `a : `G |- `a{:}A_i\arr B_i,`D }
		  ~ (\forall i \ele \n)
		}
		{ \derXN \exp y P `b . `a : `G |- `a{:}\int{n}{(A_i \arr B_i)},`D }
	  \Inf	[\Weak]
		{ \Inf	
			{ \InfBox{\D_2}
				{\derXN Q : `G |- `g{:}A_j,`D } 
			  \quad
			  \InfBox{\D_3}
				{\derXN R : `G,z{:}B_j |- `D }
			}
			{ \derPure \imp Q `g [x] z R : `G,x{:}A_j\arr B_j |- `D }
		}
		{ \derPure \imp Q `g [x] z R : `G,x{:}(A_j\arr B_j)\inter C |- `D }
	}
	{ \derPure \cut {\exp y P `b . `a } `a + x { \imp Q `g [x] z R } : `G |- `D }
 \] 
Notice that, since $x$ is unique, $C$ is added by weakening.
We can construct
 \[
\Inf	
	{ \Inf	
		{ \InfBox{\D_2}
			{\derXN Q : `G |- `g{:}A_j,`D } 
		  \quad
		  \InfBox{\D_1^j}
			{ \derPure P : `G,y{:}A_j |- `b{:}B_j,`D }
		}
		{ \derPure \cut Q `g + y P : `G |- `b{:}B_j,`D }
	  \quad
	  \InfBox{\D_3}
		{\derXN R : `G,z{:}B_j |- `D }
	}
	{ \derPure \cut { \cut Q `g + y P } `b + z R : `G |- `D }
 \] 
and
 \[
 \Inf	
	{ \InfBox{\D_2}
		{\derXN Q : `G |- `g{:}A_j,`D } 
	  \quad
	  \Inf	
		{ \InfBox{\D_1^j}
			{ \derPure P : `G,y{:}A_j |- `b{:}B_j,`D }
		  \quad
		  \InfBox{\D_3}
			{\derXN R : `G,z{:}B_j |- `D }
		}
		{ \derXN \cut P `b + z R : `G,y{:}A_j |- `D }
	}
	{ \derXN \cut Q `g + y { \cut P `b + z R } : `G |- `D }
 \]

\item[\ref{union case}] Let $\un{n}( A_i\arr B_i ) = (A_j \arr B_j) \union C $
 \[ 
 \Inf	
	{ \Inf	[\Weak]
		{ \Inf	
			{ \InfBox{\D_1}
				{\derXN P : `G,y{:}A_j |- `b{:}B_j,`D } 
			}
			{ \derPure \exp y P `b . `a : `G |- `a{:}A_j \arr B_j,`D }
		}
		{ \derPure \exp y P `b . `a : `G |- `a{:}(A_j \arr B_j) \union C,`D }
	  \Inf	[\unL]
		{ \Inf	
			{ \InfBox{\D_2^i}
				{\derXN Q : `G |- `g{:}A_i,`D } 
			  \quad
			  \InfBox{\D_3^i}
				{\derXN R : `G,z{:}B_i |- `D }
			}
			{ \derPure \imp Q `g [x] z R : `G,x{:}A_i\arr B_i |- `D }
		  ~ (\forall i \ele \n)
		}
		{ \derXN \imp Q `g [x] z R : `G,x{:}\un{n}( A_i\arr B_i ) |- `D }
	}
	{ \derPure \cut {\exp y P `b . `a } `a + x { \imp Q `g [x] z R } : `G |- `D }
 \] 
Since $x$ does not appear free in $Q$ or $R$, it is introduced in $\imp Q `g [x] z R $, so rule $(\unL)$ is applied correctly.
Notice that $`a{:}C$ in $\D_1$ is added by weakening.
We can construct
 \[
\Inf	
	{ \Inf	
		{ \InfBox{\D_2^j}
			{\derXN Q : `G |- `g{:}A_j,`D } 
		  \dquad
		  \InfBox{\D_1}
			{ \derXN P : `G,y{:}A_j |- `b{:}B_j,`D }
		}
		{ \derPure \cut Q `g + y P : `G |- `b{:}B_j,`D }
	  \dquad
	  \InfBox{\D_3^j}
		{\derXN R : `G,z{:}B_j |- `D }
	}
	{ \derPure \cut { \cut Q `g + y P } `b + z R : `G |- `D }
 \] 
and
 \[
 \Inf	
	{ \InfBox{\D_2^j}
		{\derXN Q : `G |- `g{:}A_j,`D } 
	  \dquad
	  \Inf	
		{ \InfBox{\D_1}
			{ \derXN P : `G,y{:}A_j |- `b{:}B_j,`D }
		  \dquad
		  \InfBox{\D_3^j}
			{\derXN R : `G,z{:}B_j |- `D }
		}
		{ \derXN \cut P `b + z R : `G,y{:}A_j |- `D }
	}
	{ \derXN \cut Q `g + y { \cut P `b + z R } : `G |- `D }
 \]

\item[\ref{pure case}] \correct
 \[ 
\Inf	
	{ \Inf	
		{ \InfBox{\D_1}
			{\derXN P : `G,y{:}A |- `b{:}B,`D } 
		}
		{ \derPure \exp y P `b . `a : `G |- `a{:}A\arr B,`D }
	  \quad
	  \Inf	
		{ \InfBox{\D_2}
			{\derXN Q : `G |- `g{:}A,`D } 
		  \quad
		  \InfBox{\D_3}
			{\derXN R : `G,z{:}B |- `D }
		}
		{ \derPure \imp Q `g [x] z R : `G,x{:}A\arr B |- `D }
	}
	{ \derPure \cut {\exp y P `b . `a } `a + x { \imp Q `g [x] z R } : `G |- `D }
 \] 
Then we can construct
 \[
\Inf	
	{ \Inf	
		{ \InfBox{\D_2}
			{\derXN Q : `G |- `g{:}A,`D } 
		  \quad
		  \InfBox{\D_1}
			{ \derXN P : `G,y{:}A |- `b{:}B,`D }
		}
		{ \derPure \cut Q `g + y P : `G |- `b{:}B,`D }
	  \quad
	  \InfBox{\D_3}
		{\derXN R : `G,z{:}B |- `D }
	}
	{ \derPure \cut { \cut Q `g + y P } `b + z R : `G |- `D }
 \] 
and
 \[
 \Inf	
	{ \InfBox{\D_2}
		{\derXN Q : `G |- `g{:}A,`D } 
	  \quad
	  \Inf	
		{ \InfBox{\D_1}
			{ \derXN P : `G,y{:}A |- `b{:}B,`D }
		  \quad
		  \InfBox{\D_3}
			{\derXN R : `G,z{:}B |- `D }
		}
		{ \derPure \cut P `b + z R : `G,y{:}A |- `D }
	}
	{ \derPure \cut Q `g + y { \cut P `b + z R } : `G |- `D }
 \]

 \end{wideitemize}

 \end{description}

}

 \item[Activating the cuts]

 \begin{description}  
 \itemsep10pt

 \item[\actL] $\cut P `a + x Q \redX \cutL P `a + x Q $, if $P$ does not introduce $`a$, and $Q$ introduces $x$. 

 \itemindent-.7cm 

 \begin{wideitemize}

 \item[\ref{intersection case}] \correct
 \[ 
 \Inf	[\Cut]
	{ \Inf	[\intR]
		{ \InfBox{ \derXN P : `G |- `a{:}C_i,`D }
		  \quad (\forall i \ele \n)
		}
		{ \derXN P : `G |- `a{:}\CoI{n},`D }
	  \quad
	  \InfBox{ \derXN Q : `G, x{:}\CoI{n} |- `D }
	}
	{ \derXN \cut P `a + x Q : `G |- `D }
 \]
Since $x$ is introduced in $Q$, there is only \emph{one} occurrence of $x$, either in a capsule, or in an import; in both cases the assigned type $A$ is strict, so $A$ is not an intersection type.
Then $C_j = A$, for some $j \ele \n$, and the other $C_i$ are added by weakening.
So we have, by thinning:
 \[ 
 \InfBox{ \derXN Q : `G,x{:}A |- `D } \]
Notice that 
 \[ 
 \InfBox{ \derXN P : `G |- `a{:}A,`D } \]
is part of the derivation on the left-hand side, so, in particular, we can derive:
 \[ 
 \Inf	[\DaggerL]
	{\derXN P : `G |- `a{:}A,`D \quad \derXN Q : `G,x{:}A |- `D }
	{\derXN \cutL P `a + x Q : `G |- `D } 
 \]

 \item[\ref{union case}] \correct
 \[ 
 \Inf	[\Cut]
	{ \InfBox{ \derXN P : `G |- `a{:}\CoU{n},`D }
	  \quad
	  \Inf	[\unL]
		{ \InfBox{ \derXN Q : `G, x{:}C_i |- `D }
		  \quad (\forall i \ele \n)
		}
		{ \derXN Q : `G, x{:}\CoU{n} |- `D }
	}
	{ \derXN \cut P `a + x Q : `G |- `D }
 \]
Since $Q$ introduces $x$, rule $(\unL)$ is correctly applied.
Notice that $\CoU{n}$ is not an intersection type; we have:
 \[ 
 \Inf	[\DaggerL]
	{ \InfBox{ \derXN P : `G |- `a{:}\CoU{n},`D }
	  \quad
	  \Inf	[\unL]
		{ \InfBox{ \derXN Q : `G, x{:}C_i |- `D }
		  \quad (\forall i \ele \n)
		}
		{ \derXN Q : `G, x{:}\CoU{n} |- `D }
	}
	{ \derXN \cut P `a + x Q : `G |- `D }
 \]

 \item[\ref{pure case}] \correct
 \[ 
 \Inf	[\Cut]
	{ \InfBox{ \derXN P : `G |- `a{:}C,`D }
	  \dquad
	  \InfBox{ \derXN Q : `G, x{:}C |- `D }
	}
	{ \derXN \cut P `a + x Q : `G |- `D }
 \]
Since $C$ is not an intersection (and not a union), and we have:
 \[ 
 \Inf	[\DaggerL]
	{\derXN P : `G |- `a{:}C,`D \quad \derXN Q : `G,x{:}C |- `D }
	{\derXN \cutL P `a + x Q : `G |- `D } 
 \]

 \end{wideitemize}

\Comment{
 \item[\actR] $ \cut P `a + x Q \redX \cutR P `a + x Q $, if $Q$ does not introduce $x$. 
Immediate, since inactive cuts and right-activated cuts are typed with the same rule.
}

 \end{description}


 \item[Left propagation]

 \begin{description} 

\Comment{

 \item[\deactL] $\cutL \caps<y,`a> `a + x P \red \cut \caps<y,`a> `a + x P $.
Trivial, since the side-condition is dropped.

 \itemindent-.7cm 

 \item[\Li]
$ \cutL \caps<y,`b> `a + x P \red \caps<y,`b> $, with $`b \not= `a$.

 \begin{wideitemize}

\item[\ref{intersection case}]
Excluded by rule $(\DaggerL$).

\item[\ref{union case}] \correct
 \[
\Inf	[\DaggerL]
	{ \Inf	
		{ \derXN \caps<y,`b> : `G\inter y{:}A |- `a{:}\CoUn,`b{:}A\union `D }
	  \quad
	  \Inf	[\unL]
		{ \InfBox{ \derXN P : `G,x{:}C_i |- `D }
		  \quad (\forall i \ele \n)
		}
		{ \derXN P : `G,x{:}{\CoUn} |- `D }
	}
	{ \derXN \cutL \caps<y,`b> `a + x P : `G\inter y{:}A |- `b{:}A\union `D }
 \]
Notice that $P$ introduces $z$.
 \[
\Inf	
	{ \derXN \caps<y,`b> : `G\inter y{:}A |- `b{:}A\union `D }
 \]

\item[\ref{pure case}] \correct
 \[
\Inf	[\DaggerL]
	{ \Inf	
		{ \derXN \caps<y,`b> : `G\inter y{:}A |- `a{:}C,`b{:}A\union `D }
	  \quad
	  \InfBox{ \derXN P : `G,z{:}C |- `D }
	}
	{ \derXN \cutL \caps<y,`b> `a + z P : `G\inter y{:}A |- `b{:}A\union `D }
 \]
 \[
\Inf	
	{ \derXN \caps<y,`b> : `G\inter y{:}A |- `b{:}A\union `D }
 \]

 \end{wideitemize}

}

 \item[\Lii] $ \cutL { \exp y P `b . `a } `a + x Q \red \cut { \exp y { \cutL P `a + x Q } `b . `g } `g + x Q $, with $`g$ fresh.
Notice that $`a$ might appear inside $P$ -- wlog, we will assume that it does; also, $y,`b \not \in \FC{Q}$.

\Short{ \itemindent-.7cm }

 \begin{wideitemize}

\item[\ref{intersection case}]
Excluded by rule $(\DaggerL$).

\item[\ref{union case}] \correct
 \[ 
 \Inf	[\DaggerL]
	{ \Inf	[\arrR]
		{ \InfBox{\D_1}
			{ \derXN P : `G,y{:}A |- `b{:}B,`a{:}\CoU{n},`D } 
		}
		{ \derPure \exp y P `b . `a : `G |- `a{:}\CoU{n}\union (A\arr B),`D } 
	  \Inf	[\unL]
		{ \InfBox{\D^i_2}
			{ \derPure Q : `G,x{:}C_i |- `D }
		  ~(\forall i \ele \n)
		  \quad
		  \InfBox{\D_2}
			{ \derPure Q : `G,x{:}A\arr B |- `D }
		}
		{ \derXN Q : `G,x{:}\CoU{n}\union (A\arr B) |- `D }
	}
	{ \derPure \cutL {\exp y P `b . `a } `a + x Q : `G |- `D }
 \]
Notice that then $x$ is introduced in $Q$, although this plays no role here.
 \[ 
 \Inf	
	{ \Inf	[\arrR]
		{ \Inf	[\DaggerL]
			{ \InfBox{\D_1}
				{ \derPure P : `G,y{:}A |- `b{:}B,`a{:}\CoU{n},`D } 
			  \dquad
			  \Inf	[\unL]
				{ \InfBox{\D^i_2}
					{ \derPure Q : `G,x{:}C_i |- `D }
				  ~(\forall i \ele \n)
				}
				{ \derXN Q : `G,x{:}\CoU{n} |- `D }
			}
			{ \derPure \cutL P `a + x Q : `G,y{:}A |- `b{:}B,`D }
		}
		{ \derPure \exp y { \cutL P `a + x Q } `b . `g : `G |- `g{:}A\arr B,`D }
	  \kern-5mm
	  \InfBox{\D_2}
		{ \derPure Q : `G,x{:}A\arr B |- `D }
	}
	{ \derPure \cut { \exp y { \cutL P `a + x Q } `b . `g } `g + x Q : `G |- `D }
 \]

\item[\ref{pure case}] Then the $`a$ inside $P$ also has type $A\arr B$.
 \[
\Inf	[\DaggerL]
	{ \Inf	[\arrR]
		{ \InfBox{\D_1}
			{ \derXN P : `G,y{:}A |- `b{:}B,`a{:}A\arr B,`D } 
		}
		{ \derPure \exp y P `b . `a : `G |- `a{:}A\arr B,`D } 
	  \InfBox{\D_2}
		{ \derPure Q : `G,x{:}A\arr B |- `D }
	}
	{ \derPure \cutL {\exp y P `b . `a } `a + x Q : `G |- `D }
 \]
 \[ 
\Inf	
	{ \Inf	[\arrR]
		{ \Inf	[\DaggerL]
			{ \InfBox{\D_1}
				{ \derXN P : `G,y{:}A |- `b{:}B,`a{:}A\arr B,`D } 
			  \dquad
			  \InfBox{\D_2}	
				{ \derXN Q : `G,x{:}C |- `a{:}A\arr B }
			}
			{ \derPure \cutL P `a + x Q : `G,y{:}A |- `b{:}B,`D }
		}
		{ \derPure \exp y { \cutL P `a + x Q } `b . `g : `G |- `g{:}A\arr B,`D }
	  \quad
	  \InfBox{\D_2}
		{ \derPure Q : `G,x{:}A\arr B |- `D }
	}
	{ \derPure \cut { \exp y { \cutL P `a + x Q } `b . `g } `g + x Q : `G |- `D }
 \]

 \end{wideitemize}


\Comment{

 \item[\Liii] $ \cutL { \exp y P `b . `g } `a + x Q \red \exp y { \cutL P `a + x Q } `b . `g $, with $`g \not = `a$.

 \begin{wideitemize}

\item[\ref{intersection case}]
Excluded by rule $(\DaggerL$).

\item[\ref{union case}] \correct
 \[
\Inf	[\DaggerL]
	{ \Inf	
		{ \InfBox{\D_1}
			{ \derXN P : `G,y{:}A |- `b{:}B,`a{:}\CoU{n},`D } 
		}
		{ \derPure \exp y P `b . `g : `G |- `a{:}\CoU{n},`g{:}A\arr B\union `D } 
	  \quad
	  \Inf	[\unL]
		{ \InfBox{\D^i_2}
			{ \derPure Q : `G,x{:}C_i |- `D }
		  \quad (\forall i \ele \n)
		}
		{ \derXN Q : `G,x{:}\CoU{n} |- `D }
	}
	{ \derPure \cutL {\exp y P `b . `g } `a + x Q : `G |- `g{:}A\arr B\union `D }
 \]
 \[ 
 \Inf	
	{ \Inf	[\DaggerL]
		{ \InfBox{\D_1}
			{ \derPure P : `G,y{:}A |- `b{:}B,`a{:}\CoU{n},`D } 
		  \dquad
		  \Inf	[\unL]
			{ \InfBox{\D^i_2}
				{ \derPure Q : `G,x{:}C_i |- `D }
			  \quad (\forall i \ele \n)
			}
			{ \derXN Q : `G, x{:}\CoU{n} |- `D }
		}
		{ \derPure \cutL P `a + x Q : `G,y{:}A |- `b{:}B,`D }
	}
	{ \derPure \exp y { \cutL P `a + x Q } `b . `g : `G |- `g{:}A\arr B\union `D } 
 \]

\item[\ref{pure case}] \correct
 \[
\Inf	[\DaggerL]
	{ \Inf	{ \InfBox{\D_1}
			{ \derXN P : `G,y{:}A |- `a{:}C,`b{:}B,`D } 
		}
		{ \derPure \exp y P `b . `g : `G |- `a{:}C,`g{:}A\arr B\union `D }
	  \quad
	  \InfBox{\D_2}
		{ \derPure Q : `G,x{:}C |- `D }
	}
	{ \derPure \cutL { \exp y P `b . `g } `a + x Q : `G |- `g{:}A\arr B,`D } 
 \]
 \[
\Inf	{ \Inf	[\DaggerL]
		{ \InfBox{\D_1}
			{ \derXN P : `G,y{:}A |- `a{:}C,`b{:}B,`D } 
		  \dquad
		  \InfBox{\D_2}
			{ \derPure Q : `G,x{:}C |- `D }
		}
		{ \derPure \cutL P `a + x Q : `G,y{:}A |- `b{:}B,`D }
	}
	{ \derPure \exp y { \cutL P `a + x Q } `b . `g : `G |- `g{:}A\arr B,`D } 
 \]

 \end{wideitemize}

\item[\Liv] $ \cutL { \imp P `b [z] y Q } `a + z R \red \imp { \cutL P `a + x R } `b [z] y { \cutL Q `a + x R } $.

 \begin{wideitemize}

\item[\ref{intersection case}]
Excluded by rule $(\DaggerL$).

\item[\ref{union case}] \correct 
 \[
\Inf	[\DaggerL]
	{ \Inf	{ \InfBox{\D_1}
			{ \derXN P : `G |- `a{:}\CoU{n},`b{:}A,`D }
		  \quad
		  \InfBox{\D_2}
			{ \derXN Q : `G,y{:}B |- `a{:}\CoU{n},`D }
		}
		{ \derPure \imp P `b [z] y Q : `G\inter z{:}A\arr B |- `a{:}\CoU{n},`D }
	  \quad
	  \Inf	[\unL]
		{ \InfBox{\D^i_3}
			{ \derPure R : `G,x{:}C_i |- `D }
		  ~(\forall i \ele \n)
		}
		{ \derXN R : `G,x{:}\CoU{n} |- `D }
	}
	{ \derPure \cutL { \imp P `b [z] y Q } `a + x R : `G\inter z{:}A\arr B |- `D }
 \]
 \[ \kern 35mm
 \Inf	{
\multiput(10,-1)(0,7){4}{.}
 \raise1.25\RuleH\hbox to 2cm{\kern-4cm
	  \Inf	[\DaggerL]
		{ \InfBox{\D_1}
			{ \derPure P : `G |- `a{:}\CoU{n},`b{:}A,`D } 
		  \Inf	[\unL]
			{ \InfBox{\D^i_3}
				{ \derPure R : `G,x{:}C_i |- `D }
			  ~(\forall i \ele \n)
			}
			{ \derXN R : `G,x{:}\CoU{n} |- `D }
		}
		{ \derPure \cutL P `a + x R : `G |- `b{:}A,`D } 
}
	  \Inf	[\DaggerL]
		{ \InfBox{\D_2}
			{ \derPure Q : `G,y{:}B |- `a{:}\CoU{n},`D } 
		  \Inf	[\unL]
			{ \InfBox{\D^i_3}
				{ \derPure R : `G,x{:}C_i |- `D }
			  ~(\forall i \ele \n)
			}
			{ \derXN R : `G,x{:}\CoU{n} |- `D }
		}
		{ \derPure \cutL Q `a + x R : `G,y{:}B |- `D }
	}
	{ \derPure \imp { \cutL P `a + x R } `b [z] y { \cutL Q `a + x R } : `G\inter z{:}A\arr B |- `D }
 \]

\item[\ref{pure case}] \correct
 \[
\Inf	[\DaggerL]
	{ \Inf	{ \InfBox{\D_1}
			{ \derXN P : `G |- `a{:}C,`b{:}A,`D }
		  \quad
		  \InfBox{\D_2}
			{ \derXN Q : `G,y{:}B |- `a{:}C,`D }
		}
		{ \derPure \imp P `b [z] y Q : `G\inter z{:}A\arr B |- `a{:}C,`D }
	  \quad
	  \InfBox{\D_3}
		{ \derPure R : `G,x{:}C |- `D }
	}
	{ \derPure \cutL { \imp P `b [z] y Q } `a + x R : `G\inter z{:}A\arr B |- `D }
 \]
 \[ 
 \Inf	{ \Inf	[\DaggerL]
		{ \InfBox{\D_1}
			{ \derXN P : `G |- `a{:}C,`b{:}A,`D } 
		  \quad
		  \InfBox{\D_3}
			{ \derXN R : `G,x{:}C |- `D } 
		}
		{ \derPure \cutL P `a + x R : `G |- `b{:}A,`D } 
	  \quad
	  \Inf	[\DaggerL]
		{ \InfBox{\D_2}
			{ \derXN Q : `G,y{:}B |- `a{:}C,`D } 
		  \quad
		  \InfBox{\D_3}
			{ \derXN R : `G,x{:}C |- `D } 
		}
		{ \derPure \cutL Q `a + x R : `G,y{:}B |- `D }
	}
	{ \derPure \imp { \cutL P `a + x R } `b [z] y { \cutL Q `a + x R } : `G\inter z{:}A\arr B |- `D }
 \]

 \end{wideitemize}

 \item[\Lv] $ \cutL { \cut P `b + y Q } `a + x R \red \cut { \cutL P `a + z R } `b + y { \cutL Q `a + x R } $.

 \begin{wideitemize}

\item[\ref{intersection case}]
Excluded by rule $(\DaggerL$).

\item[\ref{union case}] \correct 
 \[ 
\Inf	[\DaggerL]
	{ \Inf	{ \InfBox{\D_1}
			{ \derXN P : `G |- `a{:}\CoU{n},`b{:}A,`D }
		  \quad
		  \InfBox{\D_2}
			{ \derXN Q : `G,y{:}A |- `a{:}\CoU{n},`D }
		}
		{ \derPure \cut P `b + y Q : `G |- `a{:}\CoU{n},`D }
	  \quad
	  \Inf	[\unL]
		{ \InfBox{\D^i_3}{ \derPure R : `G, x{:}C_i |- `D }
		  \quad (\forall i \ele \n)
		}
		{ \derXN R : `G, x{:}\CoU{n} |- `D }
	}
	{ \derPure \cutL { \cut P `b + y Q } `a + x R : `G |- `D }
 \]
 \[ \kern2.5cm \begin{array}[t]{c}
 \Inf	
	{
\multiput(10,-1)(0,7){4}{.}
\raise1.25\RuleH\hbox to 3.5cm{\kern-3cm
	  \Inf	[\DaggerL]
		{ \InfBox{\D_1}{ \derXN P : `G |- `a{:}\CoU{n},`b{:}A,`D }
		  \Inf	[\unL]
			{ \InfBox{\D^i_3}{ \derPure R : `G, x{:}C_i |- `D }
			  ~(\forall i \ele \n)
			}
			{ \derXN R : `G, x{:}\CoU{n} |- `D }
		}
		{ \derPure \cutL P `a + x R : `G |- `b{:}A,`D }
}
	  \Inf	[\DaggerL]
		{ \InfBox{\D_2}{ \derPure Q : `G,y{:}A |- `a{:}\CoU{n},`D }
		  \Inf	[\unL]
			{ \InfBox{\D^i_3}{ \derPure R : `G, x{:}C_i |- `D }
			  ~(\forall i \ele \n)
			}
			{ \derXN R : `G, x{:}\CoU{n} |- `D }
		}
		{ \derPure \cutL Q `a + x R : `G,y{:}A |- `D }
	}
	{ \derPure \cut { \cutL P `a + x R } `b + y { \cutL Q `a + x R } : `G |- `D }
 \end{array} \]

\item[\ref{pure case}] \correct
 \[
\Inf	[\DaggerL]
	{ \Inf	
		{ \InfBox{\D_1}
			{ \derXN P : `G |- `a{:}C,`b{:}B,`D }
		  \quad
		  \InfBox{\D_2}
			{ \derXN Q : `G,y{:}B |- `a{:}C,`D }
		}
		{ \derPure \cut P `b + y Q : `G |- `a{:}C,`D }
	  \quad
	  \InfBox{\D_3}
		{ \derPure R : `G,x{:}C |- `D }
	}
	{ \derPure \cutL { \cut P `b + y Q } `a + x R : `G |- `D }
\]
\[ 
\Inf	
	{ \Inf	[\DaggerL]
		{ \InfBox{\D_1}
{ \derXN P : `G |- `a{:}C,`b{:}B,`D }
		  \quad
		  \InfBox{\D_3}
			{ \derXN R : `G,x{:}C |- `D }
		}
		{ \derPure \cutL P `a + x R : `G |- `b{:}B,`D }
	  \Inf	[\DaggerL]
		{ \InfBox{\D_2}
			{ \derXN Q : `G,y{:}B |- `a{:}C,`D }
		  \quad
		  \InfBox{\D_3}
			{ \derXN R : `G,x{:}C |- `D }
		}
		{ \derPure \cutL Q `a + x R : `G,y{:}B |- `D }
	}
	{ \derPure \cut{ \cutL P `a + x R } `b + y { \cutL Q `a + x R } : `G |- `D }
 \]

 \end{wideitemize}
}

 \end{description}

 \item[Right propagation]


 \begin{description} 

\Comment{

 \item[\deactR] $ \cutR P `a + x \caps<x,`b> \red \cut P `a + x \caps<x,`b> $
trivial.

\itemindent-.7cm 


 \item[\Ri] $ \cutR P `a + x \caps<y,`b> \red \caps<y,`b>,\ y \not= x $. Easy.

 \begin{wideitemize}

\item[\ref{intersection case}] \correct
 \[
\Inf	{ \Inf	[\intR]
		{ \InfBox{\D_i}
			{ \derPure P : `G |- `a{:}A_i,`D }
		  \quad (\forall i \ele \n)		}
		{ \derXN P : `G |- `a{:}\AoI{n},`D }
	  \quad
	  \Inf	{ \derPure \caps<y,`b> : `G\inter y{:}B,x{:}\AoI{n} |- `b{:}B\union `D }
	}
	{ \derPure \cutR P `a + x \caps<y,`b> : `G\inter y{:}B |- `b{:}B\union `D }
 \]
 \[
\Inf	{ \derPure \caps<y,`b> : `G\inter y{:}B |- `b{:}B\union `D }
 \]

\item[\ref{union case}] \correct
 \[
\Inf	{ \InfBox{}{ \derPure P : `G |- `a{:}\AoU{n},`D }
	  \quad
	  \Inf	[\unL]
		{ \Inf	{ \derPure \caps<y,`b> : `G\inter y{:}B,x{:}A_1 |- `b{:}B\union `D }
		  ~\dots~
		  \Inf	{ \derPure \caps<y,`b> : `G\inter y{:}B,x{:}A_n |- `b{:}B\union `D }
		}
		{ \derXN \caps<y,`b> : `G\inter y{:}B,x{:}\AoU{n} |- `b{:}B\union `D }
	}
	{ \derPure \cutR P `a + x \caps<y,`b> : `G\inter y{:}B |- `b{:}B\union `D }
 \]
Then 
 \[
\Inf	{ \derPure \caps<y,`b> : `G\inter y{:}B |- `b{:}B\union `D }
 \]

\item[\ref{pure case}] \correct
 \[
\Inf	{ \InfBox{}
		{ \derPure P : `G\inter y{:}B |- `a{:}A,`D }
	  \quad
	  \Inf	{ \derPure \caps<y,`b> : `G,y{:}B,x{:}A |- `b{:}B\union `D }
	}
	{ \derPure \cutR P `a + x \caps<y,`b> : `G\inter y{:}B |- `b{:}B\union `D }
 \Quad
\Inf	{ \derPure \caps<y,`b> : `G\inter y{:}B |- `b{:}B\union `D }
 \]

 \end{wideitemize}


 \item[\Rii] $ \cutR P `a + x { \exp y Q `b . `g } \red \exp y { \cutR P `a + x Q } `b . `g $.
 
\Short{ \itemindent-.7cm }

 \begin{wideitemize}

\item[\ref{intersection case}] \correct
 \[
\Inf	{ \Inf	[\intR]
		{ \InfBox{\D^i_1}
			{ \derPure P : `G |- `a{:}C_i,`D }
		  \quad (\forall i \ele \n)
		}
		{ \derXN P : `G |- `a{:}\CoI{n},`D }
	  \quad
	  \Inf	{ \InfBox{\D_2}
			{ \derXN Q : `G,x{:}\CoI{n},y{:}A |- `b{:}B,`D }
		}
		{ \derPure \exp y Q `b . `g : `G,x{:}\CoI{n} |- `g{:}A\arr B\union `D }
	}
	{ \derPure \cutR P `a + x { \exp y Q `b . `g } : `G |- `g{:}A\arr B\union `D }
 \]
 \[
\Inf	{ \Inf	{ \Inf	[\intR]
			{ \InfBox{\D^i_1}
				{ \derPure P : `G |- `a{:}C_i,`D }
			  \quad (\forall i \ele \n)
			}
			{ \derXN P : `G |- `a{:}\CoI{n},`D }
		  \quad 
		  \InfBox{\D_2}
			{ \derXN Q : `G,x{:}\CoI{n},y{:}A |- `b{:}B,`D }
		}
		{ \derPure \cutR P `a + x Q : `G,y{:}A |- `b{:}B,`D }
	}
	{ \derPure \exp y { \cutR P `a + x Q } `b . `g : `G |- `g{:}A\arr B\union `D }
 \]

\ref {union case}
Impossible, since $ \exp y Q `b . `g $ does not introduce $x$.

\item[\ref{pure case}] \correct
 \[
\Inf	{ \InfBox{\D_1}
		{ \derPure P : `G |- `a{:}C,`D }
	  \quad
	  \Inf	{ \InfBox{\D_2}
			{ \derXN Q : `G,x{:}C,y{:}A |- `b{:}B,`D }
		}
		{ \derPure \exp y Q `b . `g : `G,x{:}C |- `g{:}A\arr B\union `D }
	}
	{ \derPure \cutR P `a + x { \exp y Q `b . `g } : `G |- `g{:}A\arr B\union `D }
 \]
 \[
\Inf	{ \Inf	{ \InfBox{\D_1}
			{ \derPure P : `G |- `a{:}C,`D }
		  \dquad 
		  \InfBox{\D_2}
			{ \derXN Q : `G,x{:}C,y{:}A |- `b{:}B,`D }
		}
		{ \derPure \cutR P `a + x Q : `G,y{:}A |- `b{:}B,`D }
	}
	{ \derPure \exp y { \cutR P `a + x Q } `b . `g : `G |- `g{:}A\arr B\union `D }
 \]

 \end{wideitemize}

}


 \item[$\Riii$] $ \cutR P `a + x {\imp Q `b [x] y R } \red \cut P `a + v { \imp {\cutR P `a + x Q } `b [v] y {\cutR P `a + x R } } $, with $v$ fresh.

 \begin{wideitemize}

\item[\ref{intersection case}] \correct 
 \[ \kern-25mm
\Inf	
	{
	  \Inf	[\intR]
		{ \InfBox{\D^i_1}
			{ \derPure P : `G |- `a{:}C_i,`D }
		  ~(\forall i \ele \n) 
		  \quad
		  \InfBox{\D_1}
			{ \derPure P : `G |- `a{:}A\arr B,`D }
		}
		{ \derXN P : `G |- `a{:}\CoI{n}\inter A\arr B,`D }
	  \quad
\raise1\RuleH\hbox to 4cm { \kern-20mm
	  \Inf	
		{ \InfBox{\D_2}
			{ \derXN Q : `G,x{:}\CoI{n} |- `b{:}A,`D }
		  \dquad 
		  \InfBox{\D_3}
			{ \derXN R : `G,y{:}B,x{:}\CoI{n} |- `D }
		}
		{ \hspace*{1cm} \derPure \imp Q `b [x] y R : `G,x{:}\CoI{n}\inter A\arr B |- `D }
}
\multiput(-10,-1)(0,7){2}{.}
	}
	{ \derPure \cutR P `a + x { \imp Q `b [x] y R } : `G |- `D }
 \]
 \[ \kern-25mm
 \Inf	
	{ \InfBox{\D_1}
		{ \derXN P : `G |- `a{:}A\arr B,`D }
	  \kern-2.75cm
	  \Inf	
		{
		  \Inf	
			{ \Inf	[\intR]
				{ \InfBox{\D^i_1}
					{ \derPure P : `G |- `a{:}C_i,`D }
				  ~(\forall i \ele \n) 
				}
				{ \derXN P : `G |- `a{:}\CoI{n},`D }
			  ~
			  \InfBox{\D_2}
				{ \derXN Q : `G,x{:}\CoI{n} |- `b{:}A,`D }
			}
			{ \derXN \cutR P `a + x Q : `G |- `b{:}A,`D }
\raise1.25\RuleH\hbox to 3cm {\kern-40mm 
		  \Inf	
			{ \Inf	[\intR]
				{ \InfBox{\D^i_1}
					{ \derPure P : `G |- `a{:}C_i,`D }
				  ~(\forall i \ele \n) 
				}
				{ \derXN P : `G |- `a{:}\CoI{n},`D }
			  ~ 
			  \InfBox{\D_3}
				{ \derXN R : `G,y{:}B,x{:}\CoI{n} |- `D }
			}
			{ \hspace*{2cm} \derXN \cutR P `a + x R : `G,y{:}B |- `D }
}
\multiput(-10,-1)(0,7){4}{.}
		}
		{ \derXN \imp { \cutR P `a + x Q } `b [v] y { \cutR P `a + x R } : `G,v{:}A\arr B |- `D }
	}
	{ \derPure \cut P `a + v { \imp { \cutR P `a + x Q } `b [v] y { \cutR P `a + x R } } : `G |- `D }
 \]


\item[\ref{union case}] \correct
 \[ \kern1cm
\Inf	
	{ \InfBox{\D_1}
		{ \derPure P : `G |- `a{:}\un{n}(A_i\arr B_i),`D }
	  \kern-1cm
	  \Inf	[\unL]
		{ \Inf	
			{ \InfBox{\D^i_2}
				{ \derX Q : `G,x{:}A_i\arr B_i |- `b{:}A_i,`D }
			  \quad 
			  \InfBox{\D^i_3}
				{ \derX R : `G,x{:}A_i\arr B_i,y{:}B_i |- `D }
			}
			{ \derPure \imp Q `b [x] y R : `G,x{:}A_i\arr B_i |- `D }
		  (\forall i \ele \n)
		}
		{ \derX \imp Q `b [x] y R : `G,x{:}\un{n}(A_i\arr B_i) |- `D }
	}
	{ \derPure \cutR P `a + x { \imp Q `b [x] y R } : `G |- `D }
 \]
The $x$ is introduced, so does not appear in either $Q$ or $R$, so, by the thinning lemma, we have $\derX Q : `G |- `b{:}A_i,`D $ and $\derX R : `G,y{:}B_i |- `b{:}A_i,`D $, so, by weakening, also $\derX Q : `G,x{:}{\Top} |- `b{:}A_i,`D $ and $\derX R : `G,x{:}{\Top},y{:}B_i |- `b{:}A_i,`D $.
We can then construct:
  \[ 
 \Inf	
	{ \InfBox{\D_1}
		{ \derX P : `G |-  `a{:}\un{n}(A_i\arr B_i),`D }
	  \quad
	  \Inf	[\unL]
		{ 
\hbox to 7cm{\kern-6cm
		  \Inf	
			{ 
			  \Inf	
				{ \Inf	[\Top]
					{ \derX P : `G |-  `a{:}\Top,`D }
				  \InfBox
					{\derX Q : `G,x{:}{\Top} |- `b{:}A_i,`D }
				}
				{ \derX \cutR P `a + x Q : `G |- `b{:}A_i,`D }
			  \Inf	
				{ \Inf	[\Top]
					{ \derX P : `G |-  `a{:}\Top,`D }
				  \InfBox
					{ \derX R : `G,x{:}{\Top},y{:}B_i |- `b{:}A_i,`D  }
				}
				{ \derX \cutR P `a + x R : `G,y{:}B_i |- `D }
			}
			{ \kern3.5cm\derX \imp { \cutR P `a + x Q } `b [v] y { \cutR P `a + x R } : `G,v{:}A_i\arr B_i |- `D }
}
		}
		{ \derX \imp { \cutR P `a + x Q } `b [v] y { \cutR P `a + x R } : `G,v{:}\un{n}(A_i\arr B_i) |- `D }
	}
	{ \derPure \cut P `a + v { \imp { \cutR P `a + x Q } `b [v] y { \cutR P `a + x R } } : `G |- `D }
 \]

\item[\ref{pure case}] 
Then also the $x$ inside $\imp Q `b [x] y R $ is typed with $A \arr B$.
 \[ \kern.5cm
\Inf	
	{ \InfBox{\D_1}
		{ \derPure P : `G |- `a{:}A\arr B,`D }
	  \quad
	  \Inf	
		{ \InfBox{\D_2}
			{ \derXN Q : `G,x{:}A\arr B |- `b{:}A,`D }
		  \dquad 
		  \InfBox{\D_3}
			{ \derXN R : `G,y{:}B,x{:}A\arr B |- `D }
		}
		{ \derPure \imp Q `b [x] y R : `G,x{:}A\arr B |- `D }
	}
	{ \derPure \cutR P `a + x { \imp Q `b [x] y R } : `G |- `D }
 \]
 \[ 
 \Inf	
	{ \InfBox{\D_1}
		{ \derXN P : `G |- `a{:}A\arr B,`D }
	  \quad
	  \Inf	
		{
\multiput(10,-1)(0,7){4}{.}
\raise1.25\RuleH\hbox to 2cm {\kern-45mm
		  \Inf	
			{ \InfBox{\D_1}
				{ \derXN P : `G |- `a{:}A\arr B,`D }
			  \quad
			  \InfBox{\D_2}
				{ \derXN Q : `G,x{:}A\arr B |- `b{:}A,`D }
			}
			{ \derXN \cutR P `a + x Q : `G |- `b{:}A,`D }
}
		  \Inf	
			{ \InfBox{\D_1}
				{ \derXN P : `G |- `a{:}A\arr B,`D }
			  \quad 
			  \InfBox{\D_3}
				{ \derXN R : `G,y{:}B,x{:}A\arr B |- `D }
			}
			{ \derXN \cutR P `a + x R : `G,y{:}B |- `D }
		}
		{ \derXN \imp { \cutR P `a + x Q } `b [v] y { \cutR P `a + x R } : `G,v{:}A\arr B |- `D }
	}
	{ \derPure \cut P `a + v { \imp { \cutR P `a + x Q } `b [v] y { \cutR P `a + x R } } : `G |- `D }
 \]

 \end{wideitemize}



 \item[\Riv] $ \cutR P `a + x {\imp Q `b [z] y R } \red \imp { \cutR P `a + x Q } `b [z] y { \cutR P `a + z R },\ x \not= z $. 
 
 \begin{wideitemize}

\item[\ref{intersection case}] \correct 
 \[
\Inf	{ \Inf	[\intR]
		{ \InfBox{\D^i_1}
			{ \derXN P : `G |- `a{:}C_i,`D }
		  \quad (\forall i \ele \n)
		}
		{ \derXN P : `G |- `a{:}\CoI{n},`D }
	  \Inf	{ \InfBox{\D_2}
			{ \derXN Q : `G,x{:}\CoI{n} |- `b{:}A,`D }
		  \quad
		  \InfBox{\D_3}
			{ \derXN R : `G,x{:}\CoI{n},y{:}B |- `D }
		}
		{ \derXN \imp Q `b [z] y P : `G\inter z{:}A\arr B,x{:}\CoI{n} |- `D }
	}
	{ \derXN \cutR P `a + x { \imp Q `b [z] y R } : `G\inter z{:}A\arr B |- `D }
 \]
 \[ \kern -3cm
\Inf	{ \Inf	{ \Inf	[\intR]
			{ \InfBox{\D^i_1}
				{ \derXN P : `G |- `a{:}C_i,`D }
			  ~ (\forall i \ele \n)
			}
			{ \derXN P : `G |- `a{:}\CoI{n},`D }
		  ~
		  \InfBox{\D_2}
			{ \derXN Q : `G,x{:}\CoI{n} |- `b{:}A,`D }
		}
		{ \derXN \cutR P `a + x Q : `G |- `b{:}A,`D }
\raise1.25\RuleH\hbox to 3cm {\kern-3cm
	  \Inf	{ \Inf	[\intR]
			{ \InfBox{\D^i_1}
				{ \derXN P : `G |- `a{:}C_i,`D }
			  ~ (\forall i \ele \n)
			}
			{ \derXN P : `G |- `a{:}\CoI{n},`D }
		  ~
		  \InfBox{\D_3}
			{ \derXN R : `G,x{:}\CoI{n},y{:}B |- `D }
		}
		{ \derXN \cutR P `a + z R : `G ,y{:}B |- `D }
}
\multiput(-10,-1)(0,7){4}{.}
	}
	{ \derXN \imp { \cutR P `a + x Q } `b [z] y { \cutR P `a + z R } : `G\inter z{:}A\arr B |- `D }
 \]

\item[\ref{union case}] 
Impossible, since $ \imp Q `b [z] y R $ does not introduce $x$.

\item[\ref{pure case}] \correct
 \[
\Inf	{ \InfBox{\D_1}
		{ \derXN P : `G |- `a{:}C,`D }
	  \quad
	  \Inf	{ \InfBox{\D_2}
			{ \derXN Q : `G,x{:}C |- `b{:}A,`D }
		  \dquad
		  \InfBox{\D_3}
			{ \derXN R : `G,x{:}C,y{:}B |- `D }
		}
		{ \derXN \imp Q `b [z] y R : `G\inter z{:}A\arr B,x{:}C |- `D }
	}
	{ \derXN \cutR P `a + x { \imp Q `b [z] y R } : `G\inter z{:}A\arr B |- `D }
 \]
 \[ 
\Inf	{ \Inf	{ \InfBox{\D_1}
			{ \derXN P : `G |- `a{:}C,`D }
		  \dquad
		  \InfBox{\D_2}
			{ \derXN Q : `G,x{:}C |- `b{:}A,`D }
		}
		{ \derXN \cutR P `a + x Q : `G |- `b{:}A,`D }
	  \quad
	  \Inf	{ \InfBox{\D_1}
			{ \derXN P : `G |- `a{:}C,`D }
		  \dquad
		  \InfBox{\D_3}
			{ \derXN R : `G,x{:}C,y{:}B |- `D }
		}
		{ \derXN \cutR P `a + z R : `G ,y{:}B |- `D }
	}
	{ \derXN \imp { \cutR P `a + x Q } `b [z] y { \cutR P `a + z R } : `G\inter z{:}A\arr B |- `D }
 \]

 \end{wideitemize}

\Comment{

 \item[\Rv] $ \cutR P `a + x { \cut Q `b + y R } \red \cut { \cutR P `a + x Q } `b + y { \cutR P `a + z R } $.
 
 \begin{wideitemize}

\item[\ref{intersection case}] \correct 
 \[
\Inf	{ \Inf	[\intR]
		{ \InfBox{\D^i_1}
			{ \derXN P : `G |- `a{:}C_i,`D }
		  \quad (\forall i \ele \n)
		}
		{ \derXN P : `G |- `a{:}\CoI{n},`D }
	  \Inf	{ \InfBox{\D_2}
			{ \derXN Q : `G,x{:}\CoI{n} |- `b{:}A,`D }
		  \quad
		  \InfBox{\D_3}
			{ \derXN R : `G,x{:}\CoI{n},y{:}A |- `D }
		}
		{ \derXN \cut Q `b + y R : `G,x{:}\CoI{n} |- `D }
	}
	{ \derXN \cutR P `a + x { \cut Q `b + y R } : `G |- `D }
 \]
 \[ \kern -3cm
\Inf	{ \Inf	{ \Inf	[\intR]
			{ \InfBox{\D^i_1}
				{ \derXN P : `G |- `a{:}C_i,`D }
			  ~ (\forall i \ele \n)
			}
			{ \derXN P : `G |- `a{:}\CoI{n},`D }
		  ~
		  \InfBox{\D_2}
			{ \derXN Q : `G,x{:}\CoI{n} |- `b{:}A,`D }
		}
		{ \derXN \cutR P `a + x Q : `G |- `b{:}A,`D }
\raise1.25\RuleH\hbox to 3cm {\kern-3cm
	  \Inf	{ \Inf	[\intR]
			{ \InfBox{\D^i_1}
				{ \derXN P : `G |- `a{:}C_i,`D }
			  ~ (\forall i \ele \n)
			}
			{ \derXN P : `G |- `a{:}\CoI{n},`D }
		  ~
		  \InfBox{\D_3}
			{ \derXN R : `G,x{:}\CoI{n},y{:}A |- `D }
		}
		{ \derXN \cutR P `a + z R : `G ,y{:}A |- `D }
}
\multiput(-10,-1)(0,7){4}{.}
	}
	{ \derXN \cut { \cutR P `a + x Q } `b + y { \cutR P `a + z R } : `G |- `D }
 \]

\item[\ref{union case}] 
Impossible, since $\cut Q `b + y R $ does not introduce $x$.

\item[\ref{pure case}] \correct
 \[
\Inf	{ \InfBox{\D_1}
		{ \derXN P : `G |- `a{:}C,`D }
	  \quad
	  \Inf	{ \InfBox{\D_2}
			{ \derXN Q : `G,x{:}C |- `b{:}A,`D }
		  \quad
		  \InfBox{\D_3}
			{ \derXN R : `G,x{:}C,y{:}A |- `D }
		}
		{ \derXN \cut Q `b + y R : `G,x{:}C |- `D }
	}
	{ \derXN \cutR P `a + x { \cut Q `b + y R } : `G |- `D }
 \]
 \[ 
\Inf	{ \Inf	{ \InfBox{\D_1}
			{ \derXN P : `G |- `a{:}C,`D }
		  \quad
		  \InfBox{\D_2}
			{ \derXN Q : `G,x{:}C |- `b{:}A,`D }
		}
		{ \derXN \cutR P `a + x Q : `G |- `b{:}A,`D }
	  \quad
	  \Inf	{ \InfBox{\D_1}
			{ \derXN P : `G |- `a{:}C,`D }
		  \quad
		  \InfBox{\D_3}
			{ \derXN R : `G,x{:}C,y{:}A |- `D }
		}
		{ \derXN \cutR P `a + z R : `G ,y{:}A |- `D }
	}
	{ \derXN \cut { \cutR P `a + x Q } `b + y { \cutR P `a + z R } : `G |- `D }
 \]

 \end{wideitemize}

}

 \end{description}

 \end{description}

 \end{proof}

}



\Long{ \section{A system with preservance of types under {\CBV} reduction} \label{CBV} }

We \Long{will now }\Short{also }define a notion of context assignment that will prove to be closed to reduction with respect to {\CBV} reduction.
Since the definition is in idea and concept entirely dual to the restriction for {\CBN} defined above, we will just focus on the differences.

 \begin{definition}\label{CBV type assignment}

\Long{
 \begin{enumerate}

 \item }
The context assignment rules for $\TurnV$ are\Short{ the same as those for $\TurnN$, except for}:
 \[ \kern-5mm \def\arraystretch{3} 
 \begin{array}{rl}
 \\[-18mm]  
\Long{
(\Ax): &
 \Inf	
	{ \derXV \caps<y,`b> : `G,\stat{y}{A} |- `b{:}A,`D }
 \\ }
(\Cut):&
 \Inf	[\emph{for inactive and left-activated cuts}]
	{ \derXV P : `G |- `a{:}A,`D \quad \derXV Q : `G, \stat{x}{A} |- `D }
	{ \derXV \cut P `a + x Q : `G |- `D }
 \\
(\DaggerR):&
 \Inf	[ A\emph{ not a union type},`a\textit{ introduced}]
	{ \derXV P : `G |- `a{:}A,`D \quad \derXV Q : `G, \stat{x}{A} |- `D }
	{ \derXV \cutR P `a + x Q : `G |- `D }
 \\
\Long{
(\arrL):&
 \Inf	{ \derXV P : `G |- `a{:}A,`D \quad \derXV Q : `G,\stat{x}{B} |- `D }
	{ \derXV \imp P `a [y] x Q : `G\inter \stat{y}{A\arr B} |- `D }
 \dquad
(\arrR):
 \Inf	{ \derXV P : `G,\stat{x}{A} |- `a{:}B\union `D }
	{ \derXV \exp x P `a . `b : `G |- \stat{`b}{A\arr B}\union `D }
 \\ 
(\unL): &
 \Inf	[n \geq 0]
	{ \derXV P : `G,x{:}A_i |- `D \quad  (\forall i \ele \n) }
	{ \derXV P : `G,x{:}{\AoUn} |- `D }
 \\
(\intR): &
 \Inf	[n \geq 0, `a\textit{ introduced in }P]
	{ \derXV P : `G |- `a{:}A_i,`D \quad (\forall i \ele \n) }
	{ \derXV P : `G |- `a{:}\AoIn,`D }
}
\Short{
(\intR): &
 \Inf	[n \geq 0, `a\textit{ introduced in }P]
	{ \derXV P : `G |- `a{:}A_i,`D \quad (\forall i \ele \n) }
	{ \derXV P : `G |- `a{:}\AoIn,`D }
}
 \end{array} \]

\Long{
 \item
Pure derivations are those that end with either rule $(\Ax)$, $(\Cut)$, $(\DaggerR)$, $(\arrL)$, and $(\arrR)$.

 \end{enumerate}
}

 \end{definition}

We can easily verify that this notion of type assignment is not closed for witness expansion.
This is clear from the fact that the side-condition of rule $(\intR)$ is not preserved
.

 \begin{theorem}[Witness reduction for $\TurnV$ wrt \CBV] \label{witness reduction for CBV}
If $\derXV P : `G |- `D $, and $P \redCBV Q$, then $\derXV Q : `G |- `D $.
 \end{theorem}

\Long{

 \begin{proof}
We only show the interesting cases.

 \begin{description} \itemsep5pt
\Long{
 \item[Logical rules]

 \begin{description}


 \item[\Ax] $\cut \caps<y,`a> `a + x \caps<x,`b> \red \caps<y,`b>$.
We consider the three cases:

 \itemindent-.7cm 

\begin{wideitemize}

\item[\ref{intersection case}] \correct
\[ \kern-12.5mm  
 \Inf	
	{ \Inf	[\intR]
		{ \Inf	{ \derPure \caps<y,`a> : `G\inter y{:}C_i |- `a{:}C_i,`D }
		  \quad 
		  (\forall i \ele \n)
		}
		{ \derXV \caps<y,`a> : `G\inter y{:}\CoI{n} |- `a{:}\CoI{n},`D }
	  \Inf	{ \derPure \caps<x,`b> : `G, x{:}\CoI{n} |- `b{:}C_j\union `D }
	}
	{ \derPure \cut \caps<y,`a> `a + x \caps<x,`b> : `G\inter y{:}\CoI{n} |- `b{:}C_j\union `D }
 \]
Notice that $\caps<y,`a>$ introduces $`a$ and that $C_i \ele \Tstrict$ for $\iotn$; by rule $(\Ax)$, we can derive:
 \[ 
\Inf	{ \derPure \caps<y,`b>  : `G\inter y{:}\CoI{n} |- `b{:}C_j\union `D }
 \]

\item[\ref{union case}] \correct
\[ \kern-12.5mm  
 \Inf	
	{ \Inf	{ \derPure \caps<y,`a> : `G\inter y{:}C_j |- `a{:}\CoU{n},`D }
	  \Inf	[\unL]
		{ \Inf	{ \derPure \caps<x,`b> : `G, x{:}C_i |- `b{:}C_i\union `D }
		  \quad 
		  (\forall i \ele \n)
		}
		{ \derXV \caps<x,`b> : `G, x{:}\CoU{n} |- `b{:}\CoU{n}\union `D }
	}
	{ \derPure \cut \caps<y,`a> `a + x \caps<x,`b> : `G\inter y{:}C_j |- `b{:}\CoU{n}\union `D }
 \]
Notice that $C_i \ele \Tstrict$ for $\iotn$; by rule $(\Ax)$, we can derive:
 \[ 
\Inf	{ \derPure \caps<y,`b>  : `G\inter y{:}C_j |- `b{:}\CoU{n}\union `D }
 \]

\item[\ref{pure case}] \correct
\[  
 \Inf	
	{ \Inf	{ \derPure \caps<y,`a> : `G\inter y{:}A |- `a{:}A,`D }
	  \quad
	  \Inf	{ \derPure \caps<x,`b> : `G, x{:}A |- `b{:}A\union `D }
	}
	{ \derPure \cut \caps<y,`a> `a + x \caps<x,`b> : `G\inter y{:}A |- `b{:}A\union `D }
 \]
Then $A \ele \Tstrict$, and we can derive:
 \[
\Inf	{ \derPure \caps<y,`b> : `G\inter y{:}A |- `b{:}A\union `D }
 \]

 \end{wideitemize}


 \item[\Exp] $ \cut { \exp y P `b . `a } `a + x \caps<x,`g> \red \exp y P `b . `g $, with $`a \not\in \FP{P}$; we can assume that $`a$ does not occur in $`D$, but cannot assume that for $`g$.
Again, we have three cases:

\begin{wideitemize}

\item[\ref{intersection case}] \correct
\[ 
 \Inf	
	{ \Inf	[\intR]
		{ \Inf	
			{ \InfBox{\D_i}
				{ \derXV P : `G,y{:}A_i |- `b{:}B_i,`D } 
			}
			{ \derPure \exp y P `b . `a : `G |- `a{:}A_i\arr B_i,`D }
		  \quad (\forall i \ele \n)
		}
		{ \derXV \exp y P `b . `a : `G |- `a{:}\int{i}(A_i\arr B_i),`D }
	  \Inf	
		{ \derPure \caps<x,`g> : `G,x{:}\int{i}(A_i\arr B_i) |- `g{:}A_j\arr B_j\,\union `D }
	}
	{ \derPure \cut { \exp y P `b . `a } `a + x \caps<x,`g> : `G |- `g{:}A_j\arr B_j\,\union `D }
 \]
Notice that $\exp y P `b . `a $ introduces $`a$. 
For the right-hand side we can construct:
 \[ 
 \Inf	
	{ \InfBox{\D_j}
		{ \derXV P : `G,y{:}A_j |- `b{:}B_j,`D } 
	}
	{ \derPure \exp y P `b . `g : `G |- `g{:}(A_j\arr B_j)\,\union `D }
 \]

\item[\ref{union case}] Let $C = \CoU{n}$.

\[ \kern-4.8cm 
 \Inf	
	{
	  \Inf	[\Weak]
	  	{ \Inf	
			{ \InfBox{\D}
				{ \derXV P : `G,y{:}A |- `b{:}B,`D } 
			}
			{ \derXV \exp y P `b . `a : `G |- `a{:}A\arr B,`D }
		}
		{ \derXV \exp y P `b . `a : `G |- `a{:}(A\arr B)\union C,`D }
\raise.75\RuleH\hbox to 5cm{\kern-13mm  
	  \Inf	[\unL]
		{ \Inf	{(\forall i \ele \n)}
			{ \derPure \caps<x,`g>  : `G, x{:}C_i |- `g{:}C_i\union `D }
		  \Inf	{\derPure \caps<x,`g>  : `G, x{:}A\arr B |- `g{:}A\arr B\union `D } 
		}
		{ \derXV \caps<x,`g> : `G, x{:}(A\arr B)\union C |- `g{:}(A\arr B)\union C\, \union `D }
}
	}
	{ \derPure \cut {\exp y P `b . `a } `a + x \caps<x,`g> : `G |-  `g{:}(A\arr B)\union C\, \union `D }
 \]
Notice, since $`a \not\in \fp(P)$, $`a{:}C$ is added by weakening.
Then for the right-hand side we can construct:
 \[ 
\Inf	[\Weak]
	{ \Inf	
		{ \InfBox{\D}
			{ \derXV P : `G,y{:}A |- `b{:}B,`D } 
		}
		{ \derPure \exp y P `b . `g : `G |- `g{:}(A\arr B)\union `D }
	}
	{ \derPure \exp y P `b . `g : `G |- `g{:}(A\arr B)\union C\, \union `D }
 \]

\item[\ref{pure case}] \correct
\[ \Inf	
	{ \Inf	
		{ \InfBox{\D}
			{ \derXV P : `G,y{:}A |- `b{:}B,`D } 
		}
		{ \derPure \exp y P `b . `a : `G |- `a{:}A\arr B,`D }
	  \Inf	
		{ \derPure \caps<x,`g> : `G,x{:}A\arr B |- `g{:}A\arr B\union `D }
	}
	{ \derPure \cut { \exp y P `b . `a } `a + x \caps<x,`g> : `G |- `g{:}A\arr B\union `D }
 \]
Notice that then
 \[ 
\Inf	[\arrR]
	{ \InfBox{\D}
		{ \derXV P : `G,y{:}A |- `b{:}B,`D } 
	}
	{ \derXV \exp y P `b . `g : `G |- `g{:}A\arr B\union `D }
 \]

 \end{wideitemize}


 \item[\Imp] $ \cut \caps<y,`a> `a + x { \imp P `b [x] z Q } \red \imp P `b [y] y Q $, with $x \not\in \FS{P,Q}$; we can assume that $x$ does not occur in $`G$, but cannot assume that for $y$. 
The three cases are:

\begin{wideitemize}

\item[\ref{intersection case}]  \correct
 \[ \kern4.25cm
\Inf	
	{
\raise3.5\RuleH \hbox to 5cm{\kern -4.5cm
	  \Inf	[\intR]
		{ \Inf	
			{ \derPure \caps<y,`a> : `G\inter y{:}C_i |- `a{:}C_i,`D }
			(\forall i \ele \n)
			~
		  \Inf	
			{ \derPure \caps<y,`a> : `G\inter y{:}A\arr B |- `a{:}A\arr B,`D }
		}
		{ \derXV \caps<y,`a> : `G\inter y{:}\CoI{n}\inter (A\arr B) |- `a{:}\CoI{n}\inter (A\arr B),`D }
}
	  \Inf	[\Weak]
		{ \Inf	
			{ \InfBox{\D_2}
				{ \derPure P : `G |- `b{:}A,`D } 
			  \quad
			  \InfBox{\D_3}
				{ \derPure Q : `G,z{:}B |- `D } 
			}
			{ \derPure \imp P `b [x] z Q : `G,x{:}A\arr B |- `D }
		}
		{\derPure \imp P `b [x] z Q : `G,x{:}\CoI{n}\inter (A\arr B) |- `D }
	}
	{ \derPure \cut \caps<y,`a> `a + x { \imp P `b [x] z Q } : `G\inter y{:}\CoI{n}\inter (A\arr B) |- `D }
 \]
Notice that $\caps<y,`a>$ introduces $`a$.
Since $x \notele P,Q$, $x{:}C$ has been added by weakening; we can derive:
 \[
 \Inf	[\Weak]
	{ \Inf	
		{ \InfBox{\D_2}
			{ \derPure P : `G |- `b{:}A,`D } 
		  \dquad 
		  \InfBox{\D_3}
			{ \derPure Q : `G,z{:}B |- `D } 
		}
		{ \derPure \imp P `b [y] z Q : `G\inter y{:}A\arr B |- `D }
	}
	{ \derPure \imp P `b [y] z Q : `G\inter y{:}\CoI{n}\inter (A\arr B) |- `D }
 \]

\item[\ref{union case}] Let $(A_j\arr B_j)\union C = \un{n}( A_i\arr B_i )$. 
\[ 
 \Inf	
	{ \Inf	
		{ \derPure \caps<y,`a> : `G\inter y{:}A_j\arr B_j |- `a{:}(A_j\arr B_j)\union C ,`D }
	  \Inf	[\unL]
		{ \Inf	
			{ \InfBox{\D_2^i}
				{ \derXV P : `G |- `b{:}A_i,`D } 
			  \quad
			  \InfBox{\D_3^i}
				{ \derXV Q : `G, z{:}B_i |- `D } 
			}
			{ \derPure \imp P `b [x] z Q : `G , x{:}A_i\arr B_i |- `D }
		  ~ (\forall i \ele n)
		}
		{ \derXV \imp P `b [x] z Q : `G , x{:}\un{n}( A_i\arr B_i ) |- `D }
	}
	{ \derPure \cut \caps<y,`a> `a + x { \imp P `b [x] z Q } : `G\inter y{:}A_j\arr B_j |- `D }
 \]
 \[
 \Inf	[\arrL]
	{ \InfBox{\D_2^j}
		{ \derXV P : `G |- `b{:}A_j,`D } 
	  \dquad 
	  \InfBox{\D_3^j}
		{ \derXV Q : `G, z{:}B_j |- `D } 
	}
	{ \derXV \imp P `b [y] z Q : `G\inter y{:}A_j\arr B_j |- `D }
 \] 

\item[\ref{pure case}] \correct
 \[ 
\Inf	
	{ \Inf	
		{ \derPure \caps<y,`a> : `G\inter y{:}A\arr B |- `a{:}A\arr B ,`D }
	  \Inf	[\arrL]
		{ \InfBox{\D_2}
			{ \derPure P : `G |- `b{:}A,`D } 
		  \quad
		  \InfBox{\D_3}
			{ \derPure Q : `G, z{:}B |- `D } 
		}
		{ \derPure \imp P `b [x] z Q : `G , x{:}A\arr B |- `D }
	}
	{ \derPure \cut \caps<y,`a> `a + x { \imp P `b [x] z Q } : `G\inter y{:}A\arr B |- `D }
 \]
 \[
 \Inf	[\arrL]
	{ \InfBox{\D_2}
		{ \derPure P : `G |- `b{:}A,`D } 
	  \dquad 
	  \InfBox{\D_3}
		{ \derPure Q : `G, z{:}B |- `D } 
	}
	{ \derPure \imp P `b [y] z Q : `G\inter y{:}A\arr B |- `D }
 \]

 \end{wideitemize}


 \item[\Ins] $ \begin{array}{rcl}
 \cut { \exp y P `b . `a } `a + x { \med Q `g [x] z R }
\red \left \{ \begin{array}{l}
	\cut { \cut Q `g + y P } `b + z R  \\
	\cut Q `g + y {\cut P `b + z R } 
\end{array} \right.
\end{array} $, with $`a \not\in \FP{R}, x \not\in \FS{Q,R}$.

\begin{wideitemize}

\item[\ref{intersection case}]
Let $(A_j \arr B_j) \inter C = \int{n}{(A_i \arr B_i)}$.
 \[ 
 \Inf	
	{
	  \Inf	[\intR]
		{ \Inf	
			{ \InfBox{\D_1^i}
				{\derXV P : `G,y{:}A_i |- `b{:}B_i,`D } 
			}
			{ \derPure \exp y P `b . `a : `G |- `a{:}A_i\arr B_i,`D }
		  ~ (\forall i \ele \n)
		}
		{ \derXV \exp y P `b . `a : `G |- `a{:}\int{n}{(A_i \arr B_i)},`D }
	  \Inf	[\Weak]
	  	{ \Inf	
			{ \InfBox{\D_2}
				{\derXV Q : `G |- `g{:}A_j,`D } 
			  \quad
			  \InfBox{\D_3}
				{\derXV R : `G,z{:}B_j |- `D }
			}
			{ \derPure \imp Q `g [x] z R : `G,x{:}A_j\arr B_j |- `D }
		}
		{ \derPure \imp Q `g [x] z R : `G,x{:}(A_j\arr B_j)\inter C |- `D }
	}
	{ \derPure \cut {\exp y P `b . `a } `a + x { \med Q `g [x] z R } : `G |- `D }
 \] 
Notice that $x$ is introduced.
Since $x$ is unique, $C$ is added by weakening.
We can construct
 \[
\Inf	
	{ \Inf	
		{ \InfBox{\D_2}
			{\derXV Q : `G |- `g{:}A_j,`D } 
		  \quad
		  \InfBox{\D_1^j}
			{ \derPure P : `G,y{:}A_j |- `b{:}B_j,`D }
		}
		{ \derPure \cut Q `g + y P : `G |- `b{:}B_j,`D }
	  \quad
	  \InfBox{\D_3}
		{\derXV R : `G,z{:}B_j |- `D }
	}
	{ \derPure \cut { \cut Q `g + y P } `b + z R : `G |- `D }
 \] 
and
 \[
 \Inf	
	{ \InfBox{\D_2}
		{\derXV Q : `G |- `g{:}A_j,`D } 
	  \quad
	  \Inf	
		{ \InfBox{\D_1^j}
			{ \derPure P : `G,y{:}A_j |- `b{:}B_j,`D }
		  \quad
		  \InfBox{\D_3}
			{\derXV R : `G,z{:}B_j |- `D }
		}
		{ \derXV \cut P `b + z R : `G,y{:}A_j |- `D }
	}
	{ \derXV \cut Q `g + y { \cut P `b + z R } : `G |- `D }
 \]

\item[\ref{union case}] Let $\un{n}( A_i\arr B_i ) = (A_j \arr B_j) \union C $
 \[ 
 \Inf	
	{ \Inf	[\Weak]
		{ \Inf	
			{ \InfBox{\D_1}
				{\derXV P : `G,y{:}A_j |- `b{:}B_j,`D } 
			}
			{ \derPure \exp y P `b . `a : `G |- `a{:}A_j \arr B_j,`D }
		}
		{ \derPure \exp y P `b . `a : `G |- `a{:}(A_j \arr B_j) \union C,`D }
	  \Inf	[\unL]
		{ \Inf	
			{ \InfBox{\D_2^i}
				{\derXV Q : `G |- `g{:}A_i,`D } 
			  \quad
			  \InfBox{\D_3^i}
				{\derXV R : `G,z{:}B_i |- `D }
			}
			{ \derPure \imp Q `g [x] z R : `G,x{:}A_i\arr B_i |- `D }
		  ~ (\forall i \ele \n)
		}
		{ \derXV \imp Q `g [x] z R : `G,x{:}\un{n}( A_i\arr B_i ) |- `D }
	}
	{ \derPure \cut {\exp y P `b . `a } `a + x { \med Q `g [x] z R } : `G |- `D }
 \] 
Notice that $`a{:}C$ in $\D_1$ is added by weakening.
We can construct
 \[
\Inf	
	{ \Inf	
		{ \InfBox{\D_2^j}
			{\derXV Q : `G |- `g{:}A_j,`D } 
		  \dquad
		  \InfBox{\D_1}
			{ \derXV P : `G,y{:}A_j |- `b{:}B_j,`D }
		}
		{ \derPure \cut Q `g + y P : `G |- `b{:}B_j,`D }
	  \dquad
	  \InfBox{\D_3^j}
		{\derXV R : `G,z{:}B_j |- `D }
	}
	{ \derPure \cut { \cut Q `g + y P } `b + z R : `G |- `D }
 \] 
and
 \[
 \Inf	
	{ \InfBox{\D_2^j}
		{\derXV Q : `G |- `g{:}A_j,`D } 
	  \dquad
	  \Inf	
		{ \InfBox{\D_1}
			{ \derXV P : `G,y{:}A_j |- `b{:}B_j,`D }
		  \dquad
		  \InfBox{\D_3^j}
			{\derXV R : `G,z{:}B_j |- `D }
		}
		{ \derXV \cut P `b + z R : `G,y{:}A_j |- `D }
	}
	{ \derXV \cut Q `g + y { \cut P `b + z R } : `G |- `D }
 \]

\item[\ref{pure case}] \correct
 \[ \kern-5mm
\Inf	
	{ \Inf	
		{ \InfBox{\D_1}
			{\derXV P : `G,y{:}A |- `b{:}B,`D } 
		}
		{ \derPure \exp y P `b . `a : `G |- `a{:}A\arr B,`D }
	  \quad
	  \Inf	
		{ \InfBox{\D_2}
			{\derXV Q : `G |- `g{:}A,`D } 
		  \quad
		  \InfBox{\D_3}
			{\derXV R : `G,z{:}B |- `D }
		}
		{ \derPure \imp Q `g [x] z R : `G,x{:}A\arr B |- `D }
	}
	{ \derPure \cut {\exp y P `b . `a } `a + x { \med Q `g [x] z R } : `G |- `D }
 \] 
Then we can construct
 \[
\Inf	
	{ \Inf	
		{ \InfBox{\D_2}
			{\derXV Q : `G |- `g{:}A,`D } 
		  \quad
		  \InfBox{\D_1}
			{ \derXV P : `G,y{:}A |- `b{:}B,`D }
		}
		{ \derPure \cut Q `g + y P : `G |- `b{:}B,`D }
	  \quad
	  \InfBox{\D_3}
		{\derXV R : `G,z{:}B |- `D }
	}
	{ \derPure \cut { \cut Q `g + y P } `b + z R : `G |- `D }
 \] 
and
 \[
 \Inf	
	{ \InfBox{\D_2}
		{\derXV Q : `G |- `g{:}A,`D } 
	  \quad
	  \Inf	
		{ \InfBox{\D_1}
			{ \derXV P : `G,y{:}A |- `b{:}B,`D }
		  \quad
		  \InfBox{\D_3}
			{\derXV R : `G,z{:}B |- `D }
		}
		{ \derPure \cut P `b + z R : `G,y{:}A |- `D }
	}
	{ \derPure \cut Q `g + y { \cut P `b + z R } : `G |- `D }
 \]

 \end{wideitemize}

 \end{description}

}

 \item[Activating the cuts]

 \begin{description}
 
 \item[\actL] $\cut P `a + x Q \red \cutL P `a + x Q $, if $`a$ not introduced in $P$.
Immediate, since left-activated cuts are typed with the same rule as inactive cuts.


 \item[\actR] $\cut P `a + x Q \red \cutR P `a + x Q $, if $`a$ introduced in $P$, and $x$ not introduced in $Q$.

\begin{wideitemize}

\item[\ref{intersection case}] \correct
\[
 \Inf	
	{  \Inf	[\intR]
		{ \InfBox{ \derXN P : `G |- `a{:}C_i,`D }
		  \quad (\forall i \ele \n)
		}
		{ \derXN P : `G |- `a{:}\CoI{n},`D }
	  \quad
	  \InfBox{ \derXN Q : `G, x{:}\CoI{n} |- `D }
	}
	{ \derXV \cut P `a + x Q : `G |- `D }
 \]
Notice that $`a$ is introduced in $P$ and that $\CoI{n}$ is not a union type, so we can construct: 
\[
 \Inf	
	{ \Inf	[\intR]
		{ \InfBox{ \derXN P : `G |- `a{:}C_i,`D }
		  \quad (\forall i \ele \n)
		}
		{ \derXN P : `G |- `a{:}\CoI{n},`D }
	  \quad
	  \InfBox{ \derXN Q : `G, x{:}\CoI{n} |- `D }
	}
	{ \derXV \cutR P `a + x Q : `G |- `D }
 \]

 \item[\ref{union case}] \correct
\[
 \Inf	
	{ \InfBox{\D_1}{ \derXV P : `G  |- `a{:}\CoU{n},`D }
	  \quad
	  \Inf	[\unL]
		{ \InfBox{\D^i_2}{ \derXV Q : `G, x{:}C_i |- `D }
		  \quad (\forall i \ele \n)
		}
		{ \derXV Q : `G, x{:}\CoU{n} |- `D }
	}
	{ \derXV \cut P `a + x Q : `G |- `D }
 \]
Since $`a$ is introduced in $P$, there is only \emph{one} occurrence of $`a$, either in a capsule, or in an export; in both cases the assigned type $A$ is strict, so $A$ is not an intersection type.
Then $C_j = A$, for some $j \ele \n$, and the other $C_i$ are added by weakening, so we have:
 \[ \InfBox{\D'_1}{ \derXV P : `G |- `a{:}C_j,`D } \]
so, in particular, we can derive:
 \[ 
 \Inf	[\DaggerR]
	{ \InfBox{\D'_1}{ \derXV P : `G |- `a{:}C_j,`D }
	  \quad 
	  \InfBox{\D^j_2}{ \derXV Q : `G,x{:}C_j |- `D } 
	}
	{ \derXV \cutR P `a + x Q : `G |- `D } 
 \]

 \item[\ref{pure case}] \correct
\[
 \Inf	
	{ \InfBox{\D_1}{ \derXV P : `G  |- `a{:}A,`D }
	  \quad
	  \InfBox{\D_2}{ \derXV Q : `G, x{:}A |- `D }
	}
	{ \derXV \cut P `a + x Q : `G |- `D }
 \]
Notice that $A$ is not an intersection type, so we have
 \[ 
 \Inf	[\DaggerR]
	{ \InfBox{\D'_1}{ \derXV P : `G |- `a{:}A,`D }
	  \quad 
	  \InfBox{\D^j_2}{ \derXV Q : `G,x{:}A |- `D } 
	}
	{ \derXV \cutR P `a + x Q : `G |- `D } 
 \]
 
 \end{wideitemize}

 \end{description}


 \item[Left propagation]

 \begin{description}

\Long{

 \item[\deactL] $\cutL \caps<y,`a> `a + z P \red \cut \caps<y,`a> `a + z P $.
Immediate, since left-activated cuts are typed with the same rule as deactivated cuts.

 \itemindent-.7cm 

 \item[\Li]
$ \cutL \caps<y,`b> `a + z P \red \caps<y,`b> $, with $`b \not= `a$.
Easy.

\Comment{
\begin{wideitemize}

\item[\ref{intersection case}] \correct
 \[ 
 \Inf	
	{ \Inf	[\intI]
		{ \Inf	
			{ \derPure \caps<y,`b> : `G\inter y{:}A |- `a{:}C_i,`b{:}A\union `D }
		  \quad (\forall i \ele \n)
		}
		{ \derXV \caps<y,`b> : `G\inter y{:}A |- `a{:}\CoI{n},`b{:}A\union `D }
	  \quad
	  \InfBox{ \derXV P : `G,x{:}{\CoIn} |- `D }
	}
	{ \derPure \cutL \caps<y,`b> `a + z P : `G\inter y{:}A |- `b{:}A\union `D }
 \]
 \[
\Inf	
	{ \derPure \caps<y,`b> : `G\inter y{:}A |- `b{:}A\union `D }
 \]

\item[\ref{union case}] \correct
 \[
\Inf	
	{ \Inf	
		{ \derXV \caps<y,`b> : `G\inter y{:}A |- `a{:}\CoUn,`b{:}A\union `D }
	  \quad
	  \Inf	[\unL]
		{ \InfBox{ \derXV P : `G,x{:}C_i |- `D }
		  \quad (\forall i \ele \n)
		}
		{ \derXV P : `G,x{:}{\CoUn} |- `D }
	}
	{ \derXV \cutL \caps<y,`b> `a + z P : `G\inter y{:}A |- `b{:}A\union `D }
 \]
 \[
\Inf	
	{ \derXV \caps<y,`b> : `G\inter y{:}A |- `b{:}A\union `D }
 \]

\item[\ref{pure case}] \correct
 \[
\Inf	
	{ \Inf	
		{ \derXV \caps<y,`b> : `G\inter y{:}A |- `a{:}C,`b{:}A\union `D }
	  \quad
	  \InfBox{ \derXV P : `G,x{:}C |- `D }
	}
	{ \derXV \cutL \caps<y,`b> `a + z P : `G\inter y{:}A |- `b{:}A\union `D }
 \]
 \[
\Inf	
	{ \derXV \caps<y,`b> : `G\inter y{:}A |- `b{:}A\union `D }
 \]

 \end{wideitemize}

}

}

 \item[$\Lii$] $ \cutL { \exp y P `b . `a } `a + x Q \red \cut { \exp y { \cutL P `a + x Q } `b . `g } `g + x Q $, with $`g$ fresh.
Notice that $y,`b \not \in \FC{Q}$.

\begin{wideitemize}

\item[\ref{intersection case}] \correct
 \[ 
 \Inf	
	{
	 \Inf	[\intR]
		{ \Inf	
			{ \InfBox{\D^i_1}
				{ \derX P : `G,y{:}A_i |- `b{:}B_i,`a{:}A_i\arr B_i,`D } 
			}
			{ \derPure \exp y P `b . `a : `G |- `a{:}A_i\arr B_i,`D } 
		  ~ (\forall i \ele \n)
		}
		{ \derX \exp y P `b . `a : `G |- `a{:}\int{n}(A_i\arr B_i),`D } 
	  \InfBox{\D_2}
		{ \derPure Q : `G,x{:}\int{n}(A_i\arr B_i) |- `D }
	}
	{ \derPure \cutL {\exp y P `b . `a } `a + x Q : `G |- `D }
 \]
Notice that then $`a$ is introduced, so, in particular, does not occur in $P$.
Then, by thinning and weakening we can derive $ \derX P : `G,y{:}A_i |- `b{:}B_i,`a{:}\Bottom,`D $, and we can construct:
\[
 \Inf	{ \Inf	[\intR]
 		{  \Inf	
			{ \InfBox{}
				{ \derX P : `G,y{:}A_i |- `b{:}B_i,`a{:}\Bottom,`D  } 
			  \quad
			  \Inf	[\Bottom]
				{ \derPure Q : `G,x{:}{\Bottom} |- `D }
			}	
			{ \derPure \cutL P `a + x Q  : `G,y{:}A_i |- `b{:}B_i,`D } 
		  \kern-5mm (\forall i \ele \n)
	 	}
		{ \derX \exp y { \cutL P `a + x Q } `b . `g : `G |- `g{:}\int{n}(A_i\arrow B_i),`D }
	  \InfBox{\D_2}
		{ \derPure Q : `G,x{:}\int{n}(A_i\arr B_i) |- `D }
	}
	{ \derPure \cut { \exp y {\cutL P `a + x Q } `b . `g } `g + x Q : `G |- `D }
\]

The union and pure cases are straightforward.

\Comment{
\item[\ref{union case}] \correct
 \[ 
 \Inf	
	{ \Inf	[\arrR]
		{ \InfBox{\D_1}
			{ \derXV P : `G,y{:}A |- `b{:}B,`a{:}\CoU{n},`D } 
		}
		{ \derPure \exp y P `b . `a : `G |- `a{:}\CoU{n}\union (A\arr B),`D } 
	  \Inf	[\unL]
		{ \InfBox{\D^i_2}
			{ \derPure Q : `G,x{:}C_i |- `D }
		  ~(\forall i \ele \n)
		  \quad
		  \InfBox{\D_2}
			{ \derPure Q : `G,x{:}A\arr B |- `D }
		}
		{ \derXV Q : `G,x{:}\CoU{n}\union (A\arr B) |- `D }
	}
	{ \derPure \cutL {\exp y P `b . `a } `a + x Q : `G |- `D }
 \]
 \[ 
 \Inf	
	{ \Inf	[\arrR]
		{ \Inf	
			{ \InfBox{\D_1}
				{ \derPure P : `G,y{:}A |- `b{:}B,`a{:}\CoU{n},`D } 
			  \dquad
			  \Inf	[\unL]
				{ \InfBox{\D^i_2}
					{ \derPure Q : `G,x{:}C_i |- `D }
				  ~(\forall i \ele \n)
				}
				{ \derXV Q : `G,x{:}\CoU{n} |- `D }
			}
			{ \derPure \cutL P `a + x Q : `G,y{:}A |- `b{:}B,`D }
		}
		{ \derPure \exp y { \cutL P `a + x Q } `b . `g : `G |- `g{:}A\arr B,`D }
	  \kern-5mm
	  \InfBox{\D_2}
		{ \derPure Q : `G,x{:}A\arr B |- `D }
	}
	{ \derPure \cut { \exp y { \cutL P `a + x Q } `b . `g } `g + x Q : `G |- `D }
 \]

\item[\ref{pure case}] Then the $`a$ inside $\D_1$ also has type $A\arr B$.
 \[
\Inf	
	{ \Inf	[\arrR]
		{ \InfBox{\D_1}
			{ \derXV P : `G,y{:}A |- `b{:}B,`a{:}A\arr B,`D } 
		}
		{ \derPure \exp y P `b . `a : `G |- `a{:}A\arr B,`D } 
	  \InfBox{\D_2}
		{ \derPure Q : `G,x{:}A\arr B |- `D }
	}
	{ \derPure \cutL {\exp y P `b . `a } `a + x Q : `G |- `D }
 \]
 \[ 
\Inf	
	{ \Inf	[\arrR]
		{ \Inf	
			{ \InfBox{\D_1}
				{ \derXV P : `G,y{:}A |- `b{:}B,`a{:}A\arr B,`D } 
			  \dquad
			  \InfBox{\D_2}	
				{ \derXV Q : `G,x{:}C |- `a{:}A\arr B }
			}
			{ \derPure \cutL P `a + x Q : `G,y{:}A |- `b{:}B,`D }
		}
		{ \derPure \exp y { \cutL P `a + x Q } `b . `g : `G |- `g{:}A\arr B,`D }
	  \quad
	  \InfBox{\D_2}
		{ \derPure Q : `G,x{:}A\arr B |- `D }
	}
	{ \derPure \cut { \exp y { \cutL P `a + x Q } `b . `g } `g + x Q : `G |- `D }
 \]

}

 \end{wideitemize}


\Long{

 \item[\Liii] $ \cutL { \exp y P `b . `g } `a + x Q \red \exp y { \cutL P `a + x Q } `b . `g $, with $`g \not = `a$.

\begin{wideitemize}

\item[\ref{intersection case}] 
Impossible, since $\exp y P `b . `g $ does not introduce $`a$.

\item[\ref{union case}] \correct
\[
\Inf	
	{ \Inf	
		{ \InfBox{\D_1}
			{ \derXV P : `G,y{:}A |- `b{:}B,`a{:}\CoU{n},`D } 
		}
		{ \derPure \exp y P `b . `g : `G |- `a{:}\CoU{n},`g{:}A\arr B\union `D } 
	  \quad
	  \Inf	[\unL]
		{ \InfBox{\D^i_2}
			{ \derPure Q : `G,x{:}C_i |- `D }
		  \quad (\forall i \ele \n)
		}
		{ \derXV Q : `G,x{:}\CoU{n} |- `D }
	}
	{ \derPure \cutL {\exp y P `b . `g } `a + x Q : `G |- `g{:}A\arr B\union `D }
 \]
 \[ \kern-1cm
 \Inf	
	{ \Inf	
		{ \InfBox{\D_1}
			{ \derPure P : `G,y{:}A |- `b{:}B,`a{:}\CoU{n},`D } 
		  \dquad
		  \Inf	[\unL]
			{ \InfBox{\D^i_2}
				{ \derPure Q : `G,x{:}C_i |- `D }
			  \quad (\forall i \ele \n)
			}
			{ \derXV Q : `G, x{:}\CoU{n} |- `D }
		}
		{ \derPure \cutL P `a + x Q : `G,y{:}A |- `b{:}B,`D }
	}
	{ \derPure \exp y { \cutL P `a + x Q } `b . `g : `G |- `g{:}A\arr B\union `D } 
 \]

\item[\ref{pure case}] \correct
 \[
\Inf	{ \Inf	{ \InfBox{\D_1}
			{ \derXV P : `G,y{:}A |- `a{:}C,`b{:}B,`D } 
		}
		{ \derPure \exp y P `b . `g : `G |- `a{:}C,`g{:}A\arr B\union `D }
	  \quad
	  \InfBox{\D_2}
		{ \derPure Q : `G,x{:}C |- `D }
	}
	{ \derPure \cutL { \exp y P `b . `g } `a + x Q : `G |- `g{:}A\arr B,`D } 
 \]
 \[
\Inf	{ \Inf	{ \InfBox{\D_1}
			{ \derXV P : `G,y{:}A |- `a{:}C,`b{:}B,`D } 
		  \dquad
		  \InfBox{\D_2}
			{ \derPure Q : `G,x{:}C |- `D }
		}
		{ \derPure \cutL P `a + x Q : `G,y{:}A |- `b{:}B,`D }
	}
	{ \derPure \exp y { \cutL P `a + x Q } `b . `g : `G |- `g{:}A\arr B,`D } 
 \]

 \end{wideitemize}

\item[\Liv] $ \cutL { \imp P `b [z] y Q } `a + z R \red \imp { \cutL P `a + x R } `b [z] y { \cutL Q `a + x R } $.

\begin{wideitemize}

\item[\ref{intersection case}]
Impossible, since $ \imp P `b [z] y Q $ does not introduce $`a$.

\item[\ref{union case}] Let $C = \CoU{n}$.
 \[
\Inf	{ \Inf	{ \InfBox{\D_1}
			{ \derXV P : `G |- `a{:}C,`b{:}A,`D }
		  \quad
		  \InfBox{\D_2}
			{ \derXV Q : `G,y{:}B |- `a{:}C,`D }
		}
		{ \derPure \imp P `b [z] y Q : `G\inter z{:}A\arr B |- `a{:}C,`D }
	  \quad
	  \Inf	[\unL]
		{ \InfBox{\D^i_3}
			{ \derPure R : `G,x{:}C_i |- `D }
		  ~(\forall i \ele \n)
		}
		{ \derXV R : `G,x{:}C |- `D }
	}
	{ \derPure \cutL { \imp P `b [z] y Q } `a + x R : `G\inter z{:}A\arr B |- `D }
 \]
 \[ \kern 35mm
 \Inf	{
 \raise2.25\RuleH\hbox to 2cm{\kern-4cm
	  \Inf	{ \InfBox{\D_1}
			{ \derPure P : `G |- `a{:}C,`b{:}A,`D } 
		  \quad
		  \Inf	[\unL]
			{ \InfBox{\D^i_3}
				{ \derPure R : `G,x{:}C_i |- `D }
			  ~(\forall i \ele \n)
			}
			{ \derXV R : `G,x{:}C |- `D }
		}
		{ \derPure \cutL P `a + x R : `G |- `b{:}A,`D } 
}
	  \Inf	{ \InfBox{\D_2}
			{ \derPure Q : `G,y{:}B |- `a{:}C,`D } 
		  \quad
		  \Inf	[\unL]
			{ \InfBox{\D^i_3}
				{ \derPure R : `G,x{:}C_i |- `D }
			  ~(\forall i \ele \n)
			}
			{ \derXV R : `G,x{:}C |- `D }
		}
		{ \derPure \cutL Q `a + x R : `G,y{:}B |- `D }
	}
	{ \derPure \imp { \cutL P `a + x R } `b [z] y { \cutL Q `a + x R } : `G\inter z{:}A\arr B |- `D }
 \]

\item[\ref{pure case}] \correct
 \[
\Inf	{ \Inf	{ \InfBox{\D_1}
			{ \derXV P : `G |- `a{:}C,`b{:}A,`D }
		  \quad
		  \InfBox{\D_2}
			{ \derXV Q : `G,y{:}B |- `a{:}C,`D }
		}
		{ \derPure \imp P `b [z] y Q : `G\inter z{:}A\arr B |- `a{:}C,`D }
	  \quad
	  \InfBox{\D_3}
		{ \derPure R : `G,x{:}C |- `D }
	}
	{ \derPure \cutL { \imp P `b [z] y Q } `a + x R : `G\inter z{:}A\arr B |- `D }
 \]
 \[ 
 \Inf	{ \Inf	{ \InfBox{\D_1}
			{ \derXV P : `G |- `a{:}C,`b{:}A,`D } 
		  \quad
		  \InfBox{\D_3}
			{ \derXV R : `G,x{:}C |- `D } 
		}
		{ \derPure \cutL P `a + x R : `G |- `b{:}A,`D } 
	  \quad
	  \Inf	{ \InfBox{\D_2}
			{ \derXV Q : `G,y{:}B |- `a{:}C,`D } 
		  \quad
		  \InfBox{\D_3}
			{ \derXV R : `G,x{:}C |- `D } 
		}
		{ \derPure \cutL Q `a + x R : `G,y{:}B |- `D }
	}
	{ \derPure \imp { \cutL P `a + x R } `b [z] y { \cutL Q `a + x R } : `G\inter z{:}A\arr B |- `D }
 \]

 \end{wideitemize}

 \item[\Lv] $ \cutL { \cut P `b + y Q } `a + x R \red \cut { \cutL P `a + z R } `b + y { \cutL Q `a + x R } $.

\begin{wideitemize}

\item[\ref{intersection case}]
Impossible, since $ \cut P `b + y Q $ does not introduce $`a$.

\item[\ref{union case}] Let $C = \CoU{n}$.
 \[ 
\Inf	
	{ \Inf	{ \InfBox{\D_1}
			{ \derXV P : `G |- `a{:}C,`b{:}A,`D }
		  \quad
		  \InfBox{\D_2}
			{ \derXV Q : `G,y{:}A |- `a{:}C,`D }
		}
		{ \derPure \cut P `b + y Q : `G |- `a{:}C,`D }
	  \quad
	  \Inf	[\unL]
		{ \InfBox{\D^i_3}{ \derPure R : `G, x{:}C_i |- `D }
		  ~(\forall i \ele \n)
		}
		{ \derXV R : `G, x{:}C |- `D }
	}
	{ \derPure \cutL { \cut P `b + y Q } `a + x R : `G |- `D }
 \]
 \[ \kern2.5cm 
 \Inf	{
\raise2.25\RuleH\hbox to 3.5cm{\kern-2.5cm
	  \Inf	
		{ \InfBox{\D_1}{ \derXV P : `G |- `a{:}C,`b{:}A,`D }
		  \Inf	[\unL]
			{ \InfBox{\D^i_3}{ \derPure R : `G, x{:}C_i |- `D }
			  ~(\forall i \ele \n)
			}
			{ \derXV R : `G, x{:}C |- `D }
		}
		{ \derPure \cutL P  `a + x R : `G |- `b{:}A,`D }
}
	  \Inf	
		{  \InfBox{\D_2}{ \derPure Q : `G,y{:}A |- `a{:}C,`D }
		  \Inf	[\unL]
			{ \InfBox{\D^i_3}{ \derPure R : `G, x{:}C_i |- `D }
			  ~(\forall i \ele \n)
			}
			{ \derXV R : `G, x{:}C |- `D }
		}
		{ \derPure \cutL Q `a + x R : `G,y{:}A |- `D }
	}
	{ \derPure \cut { \cutL P  `a + x R } `b + y { \cutL Q `a + x R } : `G |- `D }
\]

\item[\ref{pure case}] \correct
 \[
\Inf	
	{ \Inf	
		{ \InfBox{\D_1}
			{ \derXV P : `G |- `a{:}C,`b{:}B,`D }
		  \quad
		  \InfBox{\D_2}
			{ \derXV Q : `G,y{:}B |- `a{:}C,`D }
		}
		{ \derPure \cut P `b + y Q : `G |- `a{:}C,`D }
	  \quad
	  \InfBox{\D_3}
		{ \derPure R : `G,x{:}C |- `D }
	}
	{ \derPure \cutL { \cut P `b + y Q } `a + x R : `G |- `D }
\]
\[ 
\Inf	{ \Inf	{ \InfBox{\D_1}
{ \derXV P : `G |- `a{:}C,`b{:}B,`D }
		  \quad
		  \InfBox{\D_3}
			{ \derXV R : `G,x{:}C |- `D }
		}
		{ \derPure \cutL P `a + x R : `G |- `b{:}B,`D }
	  \Inf	{ \InfBox{\D_2}
			{ \derXV Q : `G,y{:}B |- `a{:}C,`D }
		  \quad
		  \InfBox{\D_3}
			{ \derXV R : `G,x{:}C |- `D }
		}
		{ \derPure \cutL Q `a + x R : `G,y{:}B |- `D }
	}
	{ \derPure \cut{ \cutL P `a + x R } `b + y { \cutL Q `a + x R } : `G |- `D }
 \]

 \end{wideitemize}

}

 \end{description}

 \item[Right propagation]


 \begin{description}

\Long{

 \item[\deactR] $ \cutR P `a + x \caps<x,`b> \red \cut P `a + x \caps<x,`b> $
trivial.

 \itemindent-.7cm 


 \item[\Ri] $ \cutR P `a + x \caps<y,`b> \red \caps<y,`b>,\ y \not= x $. Easy.

\Comment{

\begin{wideitemize}

\item[\ref{intersection case}] \correct
 \[
\Inf	{ \Inf	[\intR]
		{ \InfBox{\D_i}
			{ \derPure P : `G |- `a{:}A_i,`D }
		  \quad (\forall i \ele \n)		}
		{ \derXV P : `G |- `a{:}\AoI{n},`D }
	  \quad
	  \Inf	{ \derPure \caps<y,`b> : `G\inter y{:}B,x{:}\AoI{n} |- `b{:}B\union `D }
	}
	{ \derPure \cutR P `a + x \caps<y,`b> : `G\inter y{:}B |- `b{:}B\union `D }
 \]
 \[
\Inf	{ \derPure \caps<y,`b> : `G\inter y{:}B |- `b{:}B\union `D }
 \]

\item[\ref{union case}] \correct
 \[
\Inf	{ \InfBox{}{ \derPure P : `G |- `a{:}\AoU{n},`D }
	  \quad
	  \Inf	[\unL]
		{ \Inf	{ \derPure \caps<y,`b> : `G\inter y{:}B,x{:}A_1 |- `b{:}B\union `D }
		  ~\dots~
		  \Inf	{ \derPure \caps<y,`b> : `G\inter y{:}B,x{:}A_n |- `b{:}B\union `D }
		}
		{ \derXV \caps<y,`b> : `G\inter y{:}B,x{:}\AoU{n} |- `b{:}B\union `D }
	}
	{ \derPure \cutR P `a + x \caps<y,`b> : `G\inter y{:}B |- `b{:}B\union `D }
 \]
 \[
\Inf	{ \derPure \caps<y,`b> : `G\inter y{:}B |- `b{:}B\union `D }
 \]

\item[\ref{pure case}] \correct
 \[
\Inf	{ \InfBox{}
		{ \derPure P : `G\inter y{:}B |- `a{:}A,`D }
	  \quad
	  \Inf	{ \derPure \caps<y,`b> : `G,y{:}B,x{:}A |- `b{:}B\union `D }
	}
	{ \derPure \cutR P `a + x \caps<y,`b> : `G\inter y{:}B |- `b{:}B\union `D }
 \Quad
\Inf	{ \derPure \caps<y,`b> : `G\inter y{:}B |- `b{:}B\union `D }
 \]

 \end{wideitemize}

}


 \item[\Rii] $ \cutR P `a + x { \exp y Q `b . `g } \red \exp y { \cutR P `a + x Q } `b . `g $.
 
\begin{wideitemize}

\item[\ref{intersection case}] \correct
 \[
\Inf	{ \Inf	[\intR]
		{ \InfBox{\D^i_1}
			{ \derPure P : `G |- `a{:}C_i,`D }
		  \quad (\forall i \ele \n)
		}
		{ \derXV P : `G |- `a{:}\CoI{n},`D }
	  \quad
	  \Inf	{ \InfBox{\D_2}
			{ \derXV Q : `G,x{:}\CoI{n},y{:}A |- `b{:}B,`D }
		}
		{ \derPure \exp y Q `b . `g : `G,x{:}\CoI{n} |- `g{:}A\arr B\union `D }
	}
	{ \derPure \cutR P `a + x { \exp y Q `b . `g } : `G |- `g{:}A\arr B\union `D }
 \]
Then $P$ introduces $`a$.
 \[
\Inf	{ \Inf	{ \Inf	[\intR]
			{ \InfBox{\D^i_1}
				{ \derPure P : `G |- `a{:}C_i,`D }
			  \quad (\forall i \ele \n)
			}
			{ \derXV P : `G |- `a{:}\CoI{n},`D }
		  \quad 
		  \InfBox{\D_2}
			{ \derXV Q : `G,x{:}\CoI{n},y{:}A |- `b{:}B,`D }
		}
		{ \derPure \cutR P `a + x Q : `G,y{:}A |- `b{:}B,`D }
	}
	{ \derPure \exp y { \cutR P `a + x Q } `b . `g : `G |- `g{:}A\arr B\union `D }
 \]

\ref {union case}
Impossible; excluded by $(\DaggerR)$.

\item[\ref{pure case}] \correct
 \[
\Inf	{ \InfBox{\D_1}
		{ \derPure P : `G |- `a{:}C,`D }
	  \quad
	  \Inf	{ \InfBox{\D_2}
			{ \derXV Q : `G,x{:}C,y{:}A |- `b{:}B,`D }
		}
		{ \derPure \exp y Q `b . `g : `G,x{:}C |- `g{:}A\arr B\union `D }
	}
	{ \derPure \cutR P `a + x { \exp y Q `b . `g } : `G |- `g{:}A\arr B\union `D }
 \]
 \[
\Inf	{ \Inf	{ \InfBox{\D_1}
			{ \derPure P : `G |- `a{:}C,`D }
		  \dquad 
		  \InfBox{\D_2}
			{ \derXV Q : `G,x{:}C,y{:}A |- `b{:}B,`D }
		}
		{ \derPure \cutR P `a + x Q : `G,y{:}A |- `b{:}B,`D }
	}
	{ \derPure \exp y { \cutR P `a + x Q } `b . `g : `G |- `g{:}A\arr B\union `D }
 \]

 \end{wideitemize}

}


 \item[\Riii] $ \cutR P `a + x {\imp Q `b [x] y R } \red \cut P `a + v { \imp { \cutR P `a + x Q } `b [v] y { \cutR P `a + x R } } $, with $v$ fresh.

\begin{wideitemize}

\item[\ref{intersection case}] Let $C = \CoI{n}$.
 \[ \kern-25mm
\Inf	
	{
	  \Inf	[\intR]
		{ \InfBox{\D^i_1}
			{ \derPure P : `G |- `a{:}C_i,`D }
		  ~(\forall i \ele \n) 
		  \quad
		  \InfBox{\D_1}
			{ \derPure P : `G |- `a{:}A\arr B,`D }
		}
		{ \derXV P : `G |- `a{:}A\arr B\inter C,`D }
	  \quad
\raise.7\RuleH\hbox to 4cm { \kern-1.5cm
	  \Inf	
		{ \InfBox{\D_2}
			{ \derXV Q : `G,x{:}C |- `b{:}A,`D }
		  \dquad 
		  \InfBox{\D_3}
			{ \derXV R : `G,y{:}B,x{:}C |- `D }
		}
		{ \derPure \imp Q `b [x] y R : `G,x{:}A\arr B\inter C |- `D }
}
\multiput(-10,-1)(0,7){2}{.}
	}
	{ \derPure \cutR P `a + x { \imp Q `b [x] y R } : `G |- `D }
 \]
Then $P$ introduces $`a$.
 \[ \kern-3cm
 \Inf	
	{ \InfBox{\D_1}
		{ \derXV P : `G |-  `a{:}A\arr B,`D }
	  \kern-2.5cm
	  \Inf	
		{
		  \Inf	
			{ \Inf	[\intR]
				{ \InfBox{\D^i_1}
					{ \derPure P : `G |- `a{:}C_i,`D }
				  ~(\forall i \ele \n) 
				}
				{ \derXV P : `G |- `a{:}C,`D }
			  ~
			  \InfBox{\D_2}
				{ \derXV Q : `G,x{:}C |- `b{:}A,`D }
			}
			{ \derXV \cutR P `a + x Q : `G |- `b{:}A,`D }
\raise1.25\RuleH\hbox to 3cm {\kern-3cm 
		  \Inf	
			{ \Inf	[\intR]
				{ \InfBox{\D^i_1}
					{ \derPure P : `G |- `a{:}C_i,`D }
				  ~(\forall i \ele \n) 
				}
				{ \derXV P : `G |- `a{:}C,`D }
			  ~ 
			  \InfBox{\D_3}
				{ \derXV R : `G,y{:}B,x{:}C |- `D }
			}
			{ \derXV \cutR P `a + x R : `G,y{:}B |- `D }
}
\multiput(-10,-1)(0,7){4}{.}
		}
		{ \derXV \imp { \cutR P `a + x Q } `b [v] y { \cutR P `a + x R } : `G,v{:}A\arr B |- `D }
	}
	{ \derPure \cut P `a + v { \imp { \cutR P `a + x Q } `b [v] y { \cutR P `a + x R } } : `G |- `D }
 \]


\item[\ref{union case}] 
Impossible; excluded by $(\DaggerR)$.

\item[\ref{pure case}] Easy.
\Comment{

Then also the $x$ inside $\imp Q `b [x] y R $ is typed with $A \arr B$.
 \[ \kern.5cm
\Inf	
	{ \InfBox{\D_1}
		{ \derPure P : `G |- `a{:}A\arr B,`D }
	  \quad
	  \Inf	[\Imp]
		{ \InfBox{\D_2}
			{ \derXV Q : `G,x{:}A\arr B |- `b{:}A,`D }
		  \dquad 
		  \InfBox{\D_3}
			{ \derXV R : `G,y{:}B,x{:}A\arr B |- `D }
		}
		{ \derPure \imp Q `b [x] y R : `G,x{:}A\arr B |- `D }
	}
	{ \derPure \cutR P `a + x { \imp Q `b [x] y R } : `G |- `D }
 \]
  \[ 
 \Inf	
	{ \InfBox{\D_1}
		{ \derXV P : `G |-  `a{:}A\arr B,`D }
	  \quad
	  \Inf	
		{
\multiput(10,-1)(0,7){4}{.}
\raise1.25\RuleH\hbox to 2cm {\kern-4cm
		  \Inf	
			{ \InfBox{\D_1}
				{ \derXV P : `G |- `a{:}A\arr B,`D }
			  \quad
			  \InfBox{\D_2}
				{ \derXV Q : `G,x{:}A\arr B |- `b{:}A,`D }
			}
			{ \derXV \cutR P `a + x Q : `G |- `b{:}A,`D }
}
		  \Inf	
			{ \InfBox{\D_1}
				{ \derXV P : `G |- `a{:}A\arr B,`D }
			  \quad 
			  \InfBox{\D_3}
				{ \derXV R : `G,y{:}B,x{:}A\arr B |- `D }
			}
			{ \derXV \cutR P `a + x R : `G,y{:}B |- `D }
		}
		{ \derXV \imp { \cutR P `a + x Q } `b [v] y { \cutR P `a + x R } : `G,v{:}A\arr B |- `D }
	}
	{ \derPure \cut P `a + v { \imp { \cutR P `a + x Q } `b [v] y { \cutR P `a + x R } } : `G |- `D }
 \]

}

 \end{wideitemize}


\Long{

 \item[\Riv] $ \cutR P `a + x {\imp Q `b [z] y R } \red \imp { \cutR P `a + x Q } `b [z] y { \cutR P `a + z R },\ x \not= z $.  
 
\begin{wideitemize}

\item[\ref{intersection case}] Let $C = \CoI{n}$.
 \[
\Inf	{ \Inf	[\intR]
		{ \InfBox{\D^i_1}
			{ \derXV P : `G |- `a{:}C_i,`D }
		  \quad (\forall i \ele \n)
		}
		{ \derXV P : `G |- `a{:}C,`D }
	  \Inf	{ \InfBox{\D_2}
			{ \derXV Q : `G,x{:}C |- `b{:}A,`D }
		  \quad
		  \InfBox{\D_3}
			{ \derXV P : `G,x{:}C,y{:}B |- `D }
		}
		{ \derXV \imp Q `b [z] y P : `G\inter z{:}A\arr B,x{:}C |- `D }
	}
	{ \derXV \cutR P `a + x { \imp Q `b [z] y R } : `G\inter z{:}A\arr B |- `D }
 \]
 \[ \kern -3cm
\Inf	{ \Inf	{ \Inf	[\intR]
			{ \InfBox{\D^i_1}
				{ \derXV P : `G |- `a{:}C_i,`D }
			  ~ (\forall i \ele \n)
			}
			{ \derXV P : `G |- `a{:}C,`D }
		  ~
		  \InfBox{\D_2}
			{ \derXV Q : `G,x{:}C |- `b{:}A,`D }
		}
		{ \derXV \cutR P `a + x Q : `G |- `b{:}A,`D }
\raise2.25\RuleH\hbox to 3cm {\kern-3cm
	  \Inf	{ \Inf	[\intR]
			{ \InfBox{\D^i_1}
				{ \derXV P : `G |- `a{:}C_i,`D }
			  ~ (\forall i \ele \n)
			}
			{ \derXV P : `G |- `a{:}C,`D }
		  ~
		  \InfBox{\D_3}
			{ \derXV R : `G,x{:}C,y{:}B |- `D }
		}
		{ \derXV \cutR P `a + z R : `G ,y{:}B |- `D }
}
	}
	{ \derXV \imp { \cutR P `a + x Q } `b [z] y { \cutR P `a + z R } : `G\inter z{:}A\arr B |- `D }
 \]

\item[\ref{union case}] 
Impossible; excluded by $(\DaggerR)$.

\item[\ref{pure case}] \correct
 \[
\Inf	{ \InfBox{\D_1}
		{ \derXV P : `G |- `a{:}C,`D }
	  \quad
	  \Inf	{ \InfBox{\D_2}
			{ \derXV Q : `G,x{:}C |- `b{:}A,`D }
		  \dquad
		  \InfBox{\D_3}
			{ \derXV R : `G,x{:}C,y{:}B |- `D }
		}
		{ \derXV \imp Q `b [z] y R : `G\inter z{:}A\arr B,x{:}C |- `D }
	}
	{ \derXV \cutR P `a + x { \imp Q `b [z] y R } : `G\inter z{:}A\arr B |- `D }
 \]
 \[ 
\Inf	{ \Inf	{ \InfBox{\D_1}
			{ \derXV P : `G |- `a{:}C,`D }
		  \dquad
		  \InfBox{\D_2}
			{ \derXV Q : `G,x{:}C |- `b{:}A,`D }
		}
		{ \derXV \cutR P `a + x Q : `G |- `b{:}A,`D }
	  \quad
	  \Inf	{ \InfBox{\D_1}
			{ \derXV P : `G |- `a{:}C,`D }
		  \dquad
		  \InfBox{\D_3}
			{ \derXV R : `G,x{:}C,y{:}B |- `D }
		}
		{ \derXV \cutR P `a + z R : `G ,y{:}B |- `D }
	}
	{ \derXV \imp { \cutR P `a + x Q } `b [z] y { \cutR P `a + z R } : `G\inter z{:}A\arr B |- `D }
 \]

 \end{wideitemize}


 \item[\Rv] $ \cutR P `a + x { \cut Q `b + y R } \red \cut { \cutR P `a + x Q } `b + y { \cutR P `a + z R } $.
 
\begin{wideitemize}

\item[\ref{intersection case}] Let $C = \CoI{n}$.
 \[
\Inf	{ \Inf	[\intR]
		{ \InfBox{\D^i_1}
			{ \derXV P : `G |- `a{:}C_i,`D }
		  \quad (\forall i \ele \n)
		}
		{ \derXV P : `G |- `a{:}C,`D }
	  \Inf	{ \InfBox{\D_2}
			{ \derXV Q : `G,x{:}C |- `b{:}A,`D }
		  \quad
		  \InfBox{\D_3}
			{ \derXV R : `G,x{:}C,y{:}A |- `D }
		}
		{ \derXV \cut Q `b + y R : `G,x{:}C |- `D }
	}
	{ \derXV \cutR P `a + x { \cut Q `b + y R } : `G |- `D }
 \]
 \[ \kern -3cm
\Inf	{ \Inf	{ \Inf	[\intR]
			{ \InfBox{\D^i_1}
				{ \derXV P : `G |- `a{:}C_i,`D }
			  ~ (\forall i \ele \n)
			}
			{ \derXV P : `G |- `a{:}C,`D }
		  ~
		  \InfBox{\D_2}
			{ \derXV Q : `G,x{:}C |- `b{:}A,`D }
		}
		{ \derXV \cutR P `a + x Q : `G |- `b{:}A,`D }
\raise2.25\RuleH\hbox to 3cm {\kern-3cm
	  \Inf	{ \Inf	[\intR]
			{ \InfBox{\D^i_1}
				{ \derXV P : `G |- `a{:}C_i,`D }
			  ~ (\forall i \ele \n)
			}
			{ \derXV P : `G |- `a{:}C,`D }
		  ~
		  \InfBox{\D_3}
			{ \derXV R : `G,x{:}C,y{:}A |- `D }
		}
		{ \derXV \cutR P `a + z R : `G ,y{:}A |- `D }
}
	}
	{ \derXV \cut { \cutR P `a + x Q } `b + y { \cutR P `a + z R } : `G |- `D }
 \]

\item[\ref{union case}] 
Impossible; excluded by $(\DaggerR)$.

\item[\ref{pure case}] \correct
 \[
\Inf	{ \InfBox{\D_1}
		{ \derXV P : `G |- `a{:}C,`D }
	  \quad
	  \Inf	{ \InfBox{\D_2}
			{ \derXV Q : `G,x{:}C |- `b{:}A,`D }
		  \quad
		  \InfBox{\D_3}
			{ \derXV R : `G,x{:}C,y{:}A |- `D }
		}
		{ \derXV \cut Q `b + y R : `G,x{:}C |- `D }
	}
	{ \derXV \cutR P `a + x { \cut Q `b + y R } : `G |- `D }
 \]
 \[ 
\Inf	{ \Inf	{ \InfBox{\D_1}
			{ \derXV P : `G |- `a{:}C,`D }
		  \quad
		  \InfBox{\D_2}
			{ \derXV Q : `G,x{:}C |- `b{:}A,`D }
		}
		{ \derXV \cutR P `a + x Q : `G |- `b{:}A,`D }
	  \quad
	  \Inf	{ \InfBox{\D_1}
			{ \derXV P : `G |- `a{:}C,`D }
		  \quad
		  \InfBox{\D_3}
			{ \derXV R : `G,x{:}C,y{:}A |- `D }
		}
		{ \derXV \cutR P `a + z R : `G ,y{:}A |- `D }
	}
	{ \derXV \cut { \cutR P `a + x Q } `b + y { \cutR P `a + z R } : `G |- `D }
 \]

 \end{wideitemize}

}

 \end{description}

 \end{description}

 \end{proof}

}


}

 \section{Conclusions} \label{Conclusions}
We have seen that it is straightforward to define a natural notion of context assignment to the sequent calculus $\X$ that uses intersection and union types.
\Long{This system was shown natural in that we were able to show that witness expansion - the main reason to use either intersection or union - follows easily, and that both intersection and union play their natural and crucial role in that proof.}

However, \Long{this ease is not paired with soundness with respect to witness reduction. A}%
\Short{a}%
s in similar notions for the {\LC}, combining union and intersection types breaks the soundness of the system.
We have isolated the problem cases, and seen that it is exactly the non-logical behaviour of both type constructors that causes the problem.
We have looked at a number of restrictions for either {\CBN} or {\CBV} reduction that overcome this defect, but all with the loss of the witness expansion result.

This implies that it is impossible to define a semantics using types for $\X$, even for the two confluent sub-reduction systems.

 \subsection*{Acknowledgement}
I would like to thank Philippe Audebaud, Mariangiola Dezani and Alexander Summers for fruitfull discussions, and especially thank Vanessa Loprete for valuable support.

{\small \Long{\bibliography {../bieb/references} }
\Short{ 
 }
}

\newpage
\appendix

\section{A}

 \end{document}